\documentclass[aps,twocolumn, superscriptaddress]{revtex4-2}
\usepackage{bm}
\usepackage{amsmath}
\usepackage{amsthm}
\usepackage{siunitx}
\usepackage[inline]{enumitem}
\usepackage{comment}
\usepackage{xcolor}  %
\usepackage{graphicx}
\usepackage{amsfonts}
\usepackage{amssymb}
\usepackage{mathrsfs}
\usepackage{MnSymbol}
\usepackage{eufrak}
\usepackage{booktabs}
\usepackage{xcolor}
\usepackage{braket}
\usepackage{mathtools}
\usepackage{setspace}
\usepackage{enumitem}
\usepackage{simplewick}
\usepackage{simpler-wick}
\usepackage{tikz}
\colorlet{BLACK}{black}

\newtheorem{definition}{Definition}
\newtheorem{lemma}{Lemma}

\newtheorem{proposition}{Proposition}
\newtheorem{assumption}{Assumption}

\newcommand{\norm}[1]{\left\Vert{#1}\right\Vert}%
\newcommand{\norminf}[1]{\left\Vert{#1}\right\Vert_\infty}
\newcommand{\smallnorminf}[1]{\Vert{#1}\Vert_\infty}
\newcommand{\smallnorm}[1]{\Vert{#1}\Vert}
\newcommand{\K}{K}

\newcommand{\abs}[1]{\left | {#1}\right |}
\newcommand{\smallabs}[1]{| {#1} |}

\newcommand{\vecket}[1]{|{#1} \rrangle}
\newcommand{\vecbra}[1]{\llangle {#1}|}
\newcommand\numeq[1]%
  {\stackrel{\scriptscriptstyle(\mkern-1.5mu#1\mkern-1.5mu)}{=}}
\newcommand\numleq[1]%
  {\stackrel{\scriptscriptstyle(\mkern-1.5mu#1\mkern-1.5mu)}{\leq}}
\newcommand\numgeq[1]%
  {\stackrel{\scriptscriptstyle(\mkern-1.5mu#1\mkern-1.5mu)}{\geq}}
\newcommand{\wickright}{\overrightarrow{\textnormal{W}}_{\rho_E}}
\newcommand{\wickleft}{\overleftarrow{\textnormal{W}}_{\rho_E}}

\graphicspath{{Images/}}
\begin{document}

\title{\textcolor{black}{Light-cones in many-body open quantum systems with memory}}

\author{Xie-Hang Yu}
\affiliation{Max Planck Institute of Quantum Optics, Garching bei M\"unchen --- 85748, Germany}
\affiliation{Munich Center for Quantum Science and Technology (MCQST), Schellingstr. 4, D-80799 Munich, Germany. }

\author{Mark Rudner}
\affiliation{Department Physics, University of Washington, Seattle, WA -  98195, USA}
\author{Rahul Trivedi}
\email[Corresponding author: ]{rahul.trivedi@mpq.mpg.de}
\affiliation{Max Planck Institute of Quantum Optics, Garching bei M\"unchen --- 85748, Germany}
\affiliation{Munich Center for Quantum Science and Technology (MCQST), Schellingstr. 4, D-80799 Munich, Germany. }

\date{September 23, 2026}

\begin{abstract}
Memory effects in many-body quantum systems interacting with an environment significantly impact their behavior. Nevertheless, much of our understanding of many-body open quantum systems is based on models within the Born-Markov approximation where these effects are neglected. In this paper, we consider a general class of spatially local open quantum lattice models with a non-Markovian environment that is approximated as Gaussian. Our main result is a set of Lieb-Robinson bounds for such models. \textcolor{black}{For models where the environment has independent local baths, we show that the Lieb-Robinson velocity is controlled by one-norm of the memory kernel within the evolution time-interval, yielding a linear light cone whenever the memory kernel has a bounded one-norm. For models where  excitations can propagate through the environment, we prove a polynomial light cone that captures bath-mediated acceleration of information spreading. In both of these cases, we show that the derived Lieb-Robinson bounds are tight by also constructing lattice models saturating these bounds.} Finally, we consider the problem of approximating a non-Markovian model by a larger Markovian model, which contains the system spins together with only a finite number of bosonic environment modes. We establish that as a consequence of our Lieb-Robinson bounds, the number of environment modes {per system site} needed to accurately capture local observables is independent of the size of the system. 
\end{abstract}
\maketitle

\section{Introduction}
A lattice of quantum spins that only interact with each other locally are physically expected to have only a finite velocity at which correlations can propagate from one spin to another. This qualitative expectation is formalized in the form of Lieb-Robinson bounds \cite{chen2023speed}: Given two spatially local operators $A$ and $B$, they typically upper bound the norm of %
the commutator ${[A(t), B]}$, where $A(t)$ is the observable $A$ in the Heisenberg picture at time $t$. %
These bounds were {originally} derived for quantum spin lattices with only Hamiltonian interactions, with the initial results requiring an assumption of strict spatial locality \cite{lieb1972finite, hastings2004lieb}. Subsequent results relaxed this requirement to exponentially decaying interactions \cite{hastings2006spectral, nachtergaele2006lieb}, and more recently, to polynomially decaying interactions \cite{chen2019finite, chen2021optimal, foss2015nearly, else2020improved, eldredge2017fast, tran2020hierarchy, tran2021optimal, hong2021fast}. 
Lieb-Robinson bounds for lattice Hamiltonians have found an immense number of applications in %
{establishing} quasi-locality properties of dynamics \cite{kliesch2014lieb,osborne2006efficient, bravyi2006lieb, hastings2004locality} and %
{characterizing} ground state properties for gapped Hamiltonians \cite{hastings2004locality, hastings2007area, hastings2005quasiadiabatic, hastings2010quasi, hastings2021gapped, nachtergaele2019quasi}, as well as in developing classical and quantum algorithms for simulation of lattice Hamiltonians \cite{haah2021quantum, tran2019locality}.

An important question in the study of quantum many-body systems was then to generalize Lieb-Robinson bounds to open quantum systems, where the system evolution cannot be described by a Hamiltonian. Within the Markovian approximation, the open quantum dynamics of a lattice of quantum spins can often be described by a spatially local Lindbladian \cite{breuer2002theory}. Lieb Robinson bounds for such models have also been extensively studied --- in particular, they have been developed for strictly local Lindbladians \cite{poulin2010lieb, kliesch2014lieb, barthel2012quasilocality, kliesch2011dissipative}, as well as Lindbladians with exponentially and polynomially decaying interactions \cite{sweke2019lieb}. These bounds have subsequently lent rigorous insights into many-body physics of open quantum systems.
In particular, they have been {used} %
to study {the} %
stability of Lindbladian fixed points \cite{lucia2015rapid, cubitt2015stability, brandao2015area}, and developing protocols for simulating %
open quantum lattice models on both digital quantum computers \cite{kliesch2011dissipative, barthel2012quasilocality} and analog quantum simulators \cite{kashyap2024accuracy}.

However, the %
{existence and nature} of Lieb-Robinson bounds for open quantum systems, beyond the Markovian regime, remain {important open questions}. %
A key challenge in developing Lieb-Robinson bounds for non-Markovian open quantum systems is the difficulty in describing the reduced system dynamics.
In the Markovian limit, the environment can be traced out {to yield a time-local} %
Lindblad master equation for the reduced state of the system. 
In the non-Markovian setting, the environment state %
stores the history of the system, {which can in turn influence the system's later evolution} \cite{de2017dynamics, breuer2016colloquium, gribben2022using}. %

{Importantly}, the environment, {which is generally taken to be large}, %
can {essentially have an} %
infinite-dimensional {Hilbert space} %
and the system-environment interaction Hamiltonian can be unbounded {(as for example when the system is coupled to a bosonic field)}. This creates a significant challenge in obtaining a Lieb-Robinson bound. Indeed, it has been recognized for quite some time that lattice models with infinite local Hilbert space dimension need not even have a Lieb-Robinson bound.
For example, it is possible to construct Hamiltonian lattice models with unbounded local Hilbert space dimension that permit supersonic propagation of correlations \cite{eisert2009supersonic}. In some cases, a Lieb-Robinson bound can be derived despite this issue --- in particular, for models of spins interacting with a lattice of linearly coupled bosonic modes \cite{junemann2013lieb}, as well as for lattice of interacting bosonic modes \cite{schuch2011information, kuwahara2021lieb, raz2009estimating, yin2022finite, kuwahara2024effective}.

In this paper, we study Lieb-Robinson bounds for non-Markovian open quantum systems. The only assumption we make on the environment is that it, in the absence of interaction with the system, is non-interacting and can thus be described by a free quantum field. \textcolor{black}{We show that if the environment is locally independent and has finite memory (a notion that we formalize),  then the system obeys a Lieb-Robinson bound with a linear light cone. By contrast, if the environment itself has a finite propagation speed for excitations, the system may exhibit supersonic transport.} In both cases, we prove that the channel on the system generated by the non-Markovian dynamics is quasi-local: over finite times, it maps local operators to approximately local, or quasi-local, operators.

We next revisit the problem of Markovian dilations of Gaussian non-Markovian environments. For numerical simulations of non-Markovian models, as well as for theoretically {characterizing} %
the ``amount" of memory that the non-Markovian environment retains of the system, a commonly followed approach is to attempt to approximate the non-Markovian environment by a finite number of bosonic modes \cite{tamascelli2018nonperturbative, chin2010exact, pleasance2020generalized}. A standard approach to constructing these bosonic modes is the star-to-chain transformation \cite{chin2010exact} --- this transformation has been extensively studied from a theoretical standpoint, and it has been shown that it can well approximate a wide variety of non-Markovian environments \cite{trivedi2021convergence, trivedi2022description, woods2015simulating, mascherpa2017open}. However, {in} %
current analyses of this transformation the number of environment modes per system site needed to approximate the state of a many-body open quantum system scales polynomially with the total system size. For many-body systems that are geometrically local, it can be physically expected that to well approximate only local observables, the number of environment modes {per system site} needed should be independent of the system size and depend instead only on the evolution time since the local observable dynamics is not expected to depend on the entire many-body system. As an application of the Lieb Robinson bounds that we develop in this paper, we make this qualitative expectation precise.
\begin{figure}
    \centering
    \includegraphics[scale=0.25]{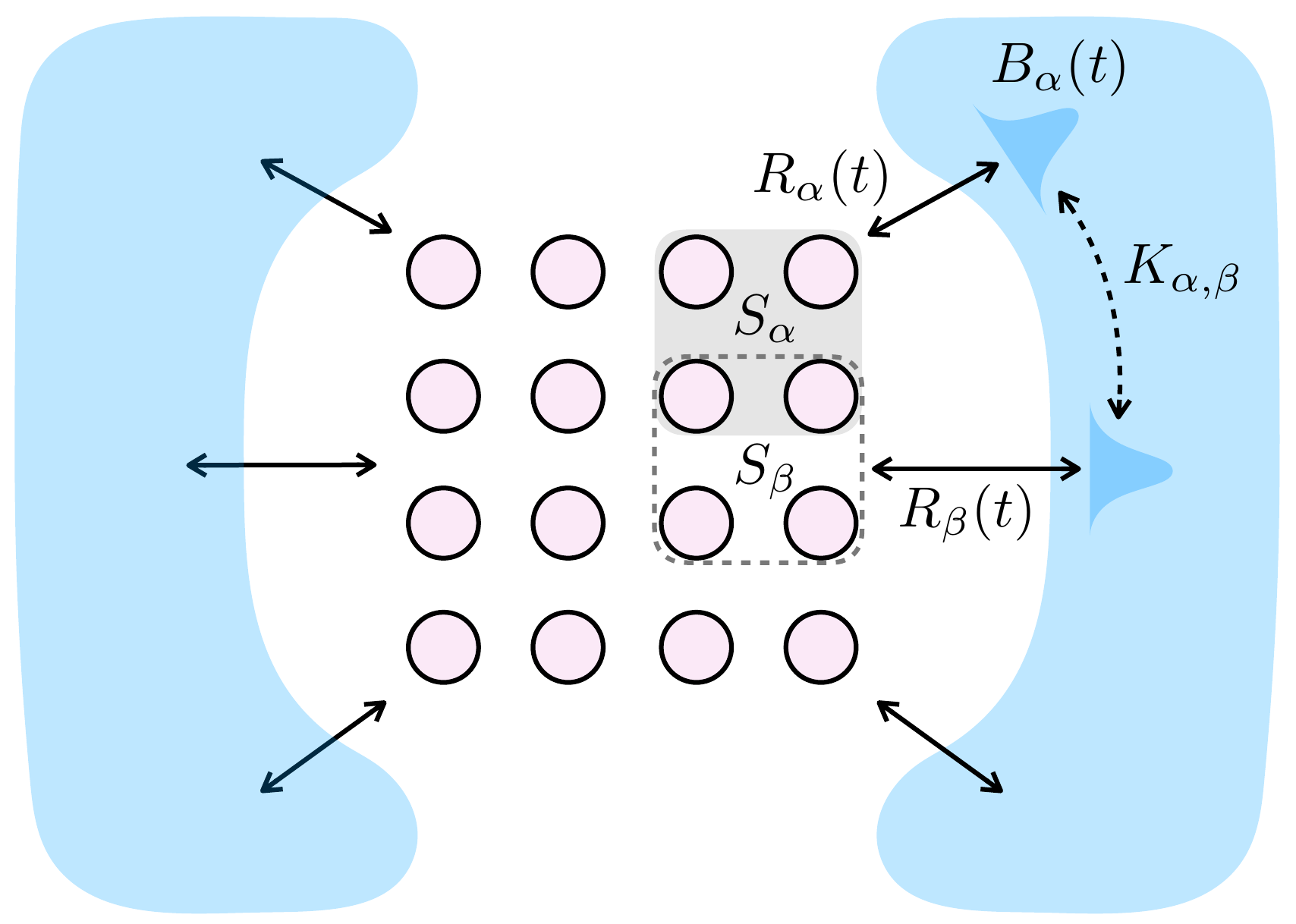}
    \caption{Schematic depiction of the non-Markovian many-body model considered in this paper. The many-body system {is} defined on a $d-$dimensional lattice, {$\Lambda$, and} interacts locally with a non-Markovian bath via the system operators $R_\alpha$ and memory kernels  $K_{\alpha,\beta}$.
    }
    \label{fig:schematic}
\end{figure}

\section{Setup: Open system with memory}
We consider a many-body system of $n$  qudits arranged on a $d-$dimensional lattice $\Lambda$ (Fig.~\ref{fig:schematic}). These qudits interact with each other locally either through a system
Hamiltonian or through interaction mediated by an environment. We will model the system-environment dynamics via the Hamiltonian
\begin{align}\label{eq:full_sys_env}
    H(t) = H_\text{S}(t) + H_\text{SE}(t),
\end{align}
where $H_\text{S}(t)$ describes the dynamics of the system in the absence of the environment and $H_\text{SE}(t)$ describes the system-environment interaction. \textcolor{black}{We will assume that $H(t)$ is already written in the interaction picture with respect to the {free} environment Hamiltonian.} 

The system Hamiltonian $H_S(t)$ will be assumed to be geometrically local{,} i.e.,
\[
H_\text{S}(t) = \sum_{\alpha} h_\alpha(t),
\]
where each $h_\alpha(t)$ is a geometrically local operator supported in the region ${S}_{\alpha} \subseteq \Lambda$. The system-environment interaction Hamiltonian $H_\text{SE}(t)$ will be assumed to be of the form
\begin{align}\label{eq:sys_env_hamiltonian}
H_\text{SE}(t) =  \sum_{\alpha} R_\alpha(t) B_{\alpha}(t) ,
\end{align}
where $R_{\alpha}(t)$ are geometrically local system operators supported on
regions ${S}_{\alpha}$, taken in the Schrödinger picture, while
$B_{\alpha}(t)$ are environment operators in the interaction picture with
respect to the free environment evolution. Throughout this paper, we will restrict ourselves to Gaussian and bosonic environments~\cite{gardiner1985input,Exactferialdi2016}, although our conclusions would also straightforwardly apply to Gaussian fermionic enviroments. More specifically,  \textcolor{black}{the operators $B_{\alpha}(t)$ will be assumed to be continuous or discrete linear superpositions of the environment's canonical annihilation and creation operators and can correspond to different quadratures of different linear combinations of modes in the Bosonic environment. 
For the system-operators appearing in the Hamiltonian, $h_\alpha(t), R_\alpha(t)$, we will only assume that they are bounded, i.e., $\norm{h_\alpha(t)}, \norm{R_\alpha(t)}\leq 1$, and do not have infinite time-variation, i.e., $\norm{h'_\alpha(t)}, \smallnorm{R_\alpha'(t)} < \infty$. Our analysis will not hinge on any further assumptions on the system-environment interaction: For instance, in contrast to previous works that have addressed Lieb-Robinson bounds in lattice Bose-Hubbard models, we will not require the system-environment interaction to be excitation number conserving (i.e., we allow for non-rotating terms) or have any other specific structue or conservation-laws}.

Assuming the environment to {\color{black}be stationary} and to {initially} (at $t = 0$) be in a Gaussian state $\rho_E$, the dynamics of the system can be described entirely in terms of the two-point correlation functions%
\begin{subequations}
\begin{align}\label{eq:main:kernel_def}
\K_{\alpha,\alpha'}(t - t') = \text{Tr}_{{E}}(B_{\alpha}(t) B_{\alpha'}(t') \rho_E).
\end{align}
{We note} that if $K_{\alpha,\alpha'}(\tau) \sim \delta(\tau)$, %
the reduced system dynamics is Markovian. As a generalization to the non-Markovian setting, we will assume that
\begin{align}\label{eq:kernel_form}
\K_{\alpha,\alpha'}(\tau) = \K_{\alpha,\alpha'; c}(\tau) + \sum_{j} k_{\alpha,\alpha'}[j] \delta(\tau - T_j),
\end{align}
\end{subequations}
where $\K_{\alpha,\alpha', c}$ are continuous functions and $k_{\alpha,\alpha', j}$ are complex constants. \textcolor{black}{These classes of memory kernels captures almost all physically relevant models studied in quantum optics and solid-state physics.
For instance, an Ohmic bath (with a high-frequency cut-off) %
{is a case with} 
a continuous ${K}_{\alpha,\alpha'}$, %
{and no} $\delta-$function contributions \cite{shi2016bound}. As another example, models in quantum optics such as lossy Jaynes Cummings models, where the environment can be described by a finite number of lossy bosonic environment modes, admit a similar description \cite{mazzola2009pseudomodes, dalton2001theory, garraway2006theory}. A bath which has reflections (for example, an atom coupled to a half-cavity with the mirror at the end) and time-delayed feedback can be described by $\delta-$function contributions in Eq.~(\ref{eq:kernel_form}), {with the times $T_j$ capturing the corresponding sharp delay times} \cite{LaserdrivenDorner2002,DynamicsTufarelli2013,grimsmo2015time, whalen2017open, pichler2016photonic}.} 
In addition to being physically relevant, the dynamics resulting from the model in Eq.~(\ref{eq:full_sys_env}) can be rigorously defined despite the $\delta-$function divergences in the kernels ${K}_{\alpha,\alpha'}$ \cite{trivedi2022description, lonigro2022generalized}.
{This makes Eq.~(\ref{eq:full_sys_env})} %
a suitable starting point for developing a Lieb-Robinson bound for non-Markovian open quantum systems.

\textcolor{black}{Finally, we remark that we assume that the environment itself is stationary (or time-independent), which is why the memory kernel in Eq.~\eqref{eq:kernel_def} can be assumed to depend only on the time-differences. However, we will allow the system Hamiltonian as well as the system operators $R_\alpha(t)$ coupling to the environment to be time dependent. Furthermore, much of the framework in this paper can be straightforwardly extended to non stationary environments i.e., when the memory kernel in Eq.~\eqref{eq:kernel_def} depends individually on the two time instants and not just on the time difference.}

\section{Lieb Robinson bound}
Consider a bounded local system obervable $O_X$ supported on a subregion $X$ of the lattice $\Lambda$. We will consider the Heisenberg picture evolution of $O_X$ from time $t'$ to time $t > t'$ under two Hamiltonians: %
(i) the full system-environment Hamiltonian $H(t)$ [Eq.~(\ref{eq:full_sys_env})], and %
(ii) the Hamiltonian $H_{X[l]}(t)$ obtained by restricting $H(t)$ to just the terms acting within a distance $l$ of $X$:
\[
H_{X[l]} = \sum_{\alpha: {S}_{\alpha} \cap X[l] \neq \emptyset} \bigg(h_\alpha(t) +  B_{\alpha}(t) R_\alpha(t)\bigg),
\]
where $X[l]$ is the set of sites that are within distance $l$ from $X$ \footnote{Here the distance is the graph (Manhattan) distance on the unit-spaced lattice $\mathbb{Z}^d$.}.
In particular, we define
\begin{align*}
&\,(\textrm{i})\ \,O_X(t, t') = U^\dagger(t, t') O_X U(t, t'), \nonumber\\
&(\textrm{ii})\ O_{X}(t, t'; l)=U_{X[l]}^\dagger(t, t') O_X U_{X[l]}(t, t'),
\end{align*}
where $U_{X[l]}(t_f, t_i) = \mathcal{T}\exp(-i\int_{t_i}^{t_f} H_{X[l]}(s') ds')$. We will characterize the deviation $O_X(t, t') - O_X(t, t'; l)$ between these two Heisenberg picture evolutions. If we only considered the system Hamiltonian without any system-environment interaction (i.e., $R_\alpha(t)= 0$), then the usual Lieb-Robinson bounds for geometrically local lattice Hamiltonians would yield that $\norm{O_X(t, t') - O_X(t, t'; l)} \leq e^{O(v_\text{LR}t - l)}$ for some system-size independent constant $v_\text{LR}$ \cite{hastings2006spectral}. Physically, this can be {understood from} the local observable's support spreading at a constant velocity $v_\text{LR}$ and thus its evolution being well approximated by the Hamiltonian restricted to its vicinity.  

However, as mentioned above, the system-environment Hamiltonian is {unbounded} --- the interaction terms in the Hamiltonian can be arbitrarily large {\it depending on the system-environment state}.
We therefore do not expect a similar bound to hold in general on the operator norm $\norm{O_X(t, t') - O_{X}(t, t'; l)}$, which is taken on the full system--environment Hilbert space $\mathcal H_S\otimes\mathcal H_E$. We instead consider the following expectation value of
$O_X(t,t')-O_X(t,t';l)$:
\begin{equation}
\begin{aligned}
\Delta_{O_X}(t,t';l)
=
\left|
\mathrm{Tr}_{SE}
\!\left[
\bigl(O_X(t,t')-O_X(t,t';l)\bigr)\rho(t')
\right]
\right|,
\label{eq:lr_target}
\end{aligned}
\end{equation}
where $\rho(t')$ denotes the system--environment state obtained by evolving
the initial product state $\sigma_S\otimes\rho_E$ for a time $t'$ under the
full system--environment Hamiltonian,
\begin{align}
\rho(t')
=
U(t',0)
\bigl(\sigma_S\otimes\rho_E\bigr)
U^\dagger(t',0),
\end{align}
with $\sigma_S$ an arbitrary initial system state and
\begin{align}
U(t_f,t_i)
=
\mathcal{T}
\exp\!\left[
-i\int_{t_i}^{t_f} H(s)\,ds
\right]
\end{align}
the unitary propagator generated by the Hamiltonian $H(t)$ defined in
Eq.~(\ref{eq:full_sys_env}). Note that even though $\rho_E$ is a Gaussian environment state, $\rho(t')$ need not be a Gaussian in the environment. 

\subsection{Independent baths}
\textcolor{black}{The first setting that we will consider is the case where the environment is composed of a set of independent baths that locally interact with the system (Fig.~\ref{fig:ind_bath_model}). More specifically, the environment Hilbert space will be a tensor product of $N$ independent baths labeled by $i\in \{1,2, \dots N\}$. Each bath itself can be rather general: In the simplest case, it can be either one or a few bosonic modes or it could be a continuum of bosonic modes i.e., a bosonic quantum field. The $i^\text{th}$ bath will interact only with system sites in a geometrically local region $S_i$ and the system-environment interaction Hamiltonian will be re-written as:
\begin{align}\label{eq:sys_env_hamiltonian_ind_bath}
    H_\text{SE}(t) = \sum_{i = 1 }^N\sum_{\nu = 1}^L R_{i}^{(\nu)}(t) B_{i}^{(\nu)}(t),
\end{align}
where $R_{i}^{(\nu)}(t)$ is a system-operator supported in $S_i$ and $B_{i}^{(\nu)}(t)$, for $\nu \in \{1, 2 \dots L\}$, is a bath operator in the $i^\text{th}$ bath. Note that this is a special case of the model in Eq.~\eqref{eq:sys_env_hamiltonian} with the index $\alpha$ being expressed as $\alpha = (i, \nu)$ and setting $B_{\alpha}(t)=B_{i}^{(\nu)}(t), R_\alpha(t) = R_i^{(\nu)}(t)$. Again, for different $\nu$, the bath operators $B_i^{(\nu)}(t)$ can correspond to different superpositions of annihilation and creation operators of the $i^\text{th}$ bath, capturing the system interacting with different quadratures and modes of the $i^\text{th}$ bath. Finally, we also assume that initially the environmental state factorizes as $\rho_E=\otimes_{i=1}^N \rho_{E,i}$, where $\rho_{E,i}$ is the initial state of the $i^\text{th}$ bath. Under these assumptions, the memory kernels introduced in Eq.~\eqref{eq:main:kernel_def} can be rewritten as}
\begin{subequations}\label{eq:independent_bath_kernel_maintext}
  \begin{equation}
K_{\alpha,\alpha'}(t-t')=\delta_{i,j}K_{i, \nu, \nu'}(t-t'),
\end{equation}
where $\alpha = (i, \nu), \beta = (j, \nu')$ and
\begin{equation}
    K_{i, \nu, \nu'}(t - t') = \mathrm{Tr}_{{E,i}}(B_{i}^{{(\nu)}}(t)B_{i}^{{(\nu')}}(t')\rho_{E, i}).
\end{equation}  
\end{subequations}

\textcolor{black}{While the independent bath model would appear to make a somewhat restrictive assumption on the environment, they are nevertheless an interesting starting point for understanding many-body non-Markovian systems. First, we remark that taking the Markovian limit of this model (i.e., $K_{i, \nu, \nu'}(\tau) \to \delta(\tau)$) {yields} a Lindblad master equation with geometrically local jump operators which have also been previously studied in the context of deriving Lieb Robinson bounds \cite{poulin2010lieb}. The independent-bath model can thus be considered to be a natural non-Markovian generalization of geometrically local Lindbladians.}

\textcolor{black}{Furthermore, these models also arise in several physically relevant settings: For instance, quantum-optical systems describing arrays of  atomic or solid-state quantum emitters interacting with free or structured electromagnetic-field ~\cite{NonmarkovianBellomo2007,SuddenLopez2008} can be described with the independent bath model when the emitter separations are larger than the dominant wavelength of the electromagnetic field. When the atomic or solid-state systems additionally interact directly with each (e.g., with Rydberg interactions as described in Refs.~\cite{liang2026frustrated, mann2025squeezing}), we obtain a many-body system described by the independent bath model. These models are also commonly used for describing site-local phonon-exciton transport~\cite{zhumodified2011} as well as for spin-boson systems with uncorrelated bosonic baths at each site~\cite{CaiIdentifying2014}.
}

\begin{figure}
    \centering
    \includegraphics[width=1.0\linewidth]{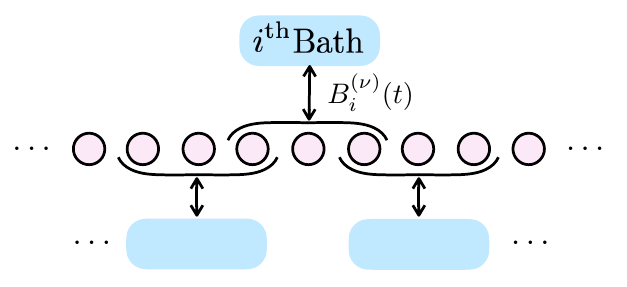}
    \caption{\textcolor{black}{Schematic depiction of an environment with independent baths. Each local bath $i$ can itself be highly multi-mode (or even a continuum of modes) and can support multiple bath operators $B_i^{(\nu)}(t)$, for $\nu \in \{1, 2 \dots L\}$, that couple locally to the system, while distinct baths have vanishing cross-correlators.}}
    \label{fig:ind_bath_model}
\end{figure}

\emph{Lieb-Robinson bound}. We now provide a Lieb-Robinson {bound} for such models---the Lieb Robinson velocity will depend on both the  supports $\mathcal{S}_\alpha$ of the operators $h_\alpha(t), R_\alpha(t)$, as well as the memory kernels $K_i^{(\nu, \nu')}(t)$. We {define} $a_0$ (the support diameter) and $\mathcal{Z}$ (the support coordination number) as positive constants satisfying
\[
\text{diam}(\mathcal{S}_{\alpha}) \leq a_0,
\]
i.e., $a_0$ is an upper bound on the diameters of the regions $\mathcal{S}_\alpha$ and
\[
\abs{\{\alpha': \mathcal{S}_\alpha \cap \mathcal{S}_{\alpha'} \neq \emptyset\}} \leq \mathcal{Z},
\]
i.e., the number of sub-regions $S_{\alpha'}$ that intersect with any one sub-region $\mathcal{S}_{\alpha}$ is upper bounded by $\mathcal{Z}$. The existence of such constants $a_0, \mathcal{Z}$ which are also system-size independent is direct consequence of the geometrical locality of the system-operators $h_\alpha(t), R_\alpha(t)$. Furthermore, for $\tau > 0$, we will also define $m(\tau)$ via
\begin{align}\label{eq:one_norm_measure}
m(\tau) = \max_{i, \nu, \nu'} \bigg(\int_{-\tau}^\tau \smallabs{K_{i, \nu, \nu'; c}(s)}ds + \sum_{j: T_j \in [-\tau, \tau]} \smallabs{k_{i, \nu, \nu'}[j]}\bigg),
\end{align}
where $K_{i, \nu, \nu'; c}(s)$ and $k_{i, \nu, \nu'}[j]$ are the continuous and $\delta$-function parts of the memory kernel $K_{i, \nu, \nu'}(s)$ as per Eq.~\eqref{eq:kernel_form}. $m(\tau)$ captures how the $1-$norm of the memory kernels grows with the time-difference $\tau$ and can be viewed as a measure of the integrated memory of the environment. As we show below, $m(t)$ is a key quantity governing the Lieb-Robinson velocity in these models.
\begin{proposition}[Lieb Robinson bounds for independent bath model]\label{prop:lr_bound_main}
    Given a local observable $O_X$ supported in $X \subseteq \Lambda$ with $\abs{X} \leq O(1)$ and times $t>t' >0$, for all initial states $\sigma_S$ of the system qudits, $\Delta_{O_X}(t, t'; l)$ defined in Eq.~(\ref{eq:lr_target}) satisfies
    \begin{align*}
        &\Delta_{O_X}(t, t'; l) \nonumber\\
    &\qquad\leq \norm{O_X} f(l)\exp(-l/a_0)(\exp(v_\textnormal{LR}(t)\abs{t -t'}/a_0) - 1),
    \end{align*}
    where, for large $l$, $f(l) \leq O(l^{d-1}) $, and
    \[
    v_\textnormal{LR}(t) = ea_0 \mathcal{Z}L + 14L^2 ea_0 \mathcal{Z}m(t).
    \]
\end{proposition}
\noindent
Recall that $L$ denotes the number of local channels and is independent of both the total system size and the evolution time.
\textcolor{black}{The proof of this proposition departs significantly from the derivation of a Lieb-Robinson bound for finite-dimensional systems due to the unbounded system-environment model. Although Lieb-Robinson bounds for systems with unbounded Hamiltonians have been previously derived, most significantly for Bose Hubbard models \cite{yin2022finite, kuwahara2024effective}, their proof strategy has been to carefully derive bounds on local excitation numbers and their moments, and then use these bounds to approximate the model by an effective finite-dimensional model. However, such a strategy cannot directly be applied to the models that we consider in this paper---the underlying reason is that in our setting each local site may be coupled to a continuous bosonic field, and the couplings are possibly $\delta-$function divergent. Thus, even within the single excitation subspace of the environment, the system-environment Hamiltonian is still unbounded. If we were to only consider models without the $\delta$-function divergences in its memory kernel and with the continuous parts being sufficiently rapidly decaying, there could be an alternative strategy to deriving Lieb-Robinson bounds by discretizing the bosonic field and truncating the particle number. However, even then, this approach would yield an overestimated Lieb-Robinson velocity that would depend on the $2-$norm of the memory kernels as opposed to its $1-$norm. As we show below with an explicit model, the dependence of the Lieb-Robinson velocity on the $1-$norm of the memory kernel, as provided in Proposition \ref{prop:lr_bound_main}, is tight and can be saturated.}

\textcolor{black}{Our key insight to deal with the unbounded and divergent system-environment interaction, which we detail in the supplement, {is to use} the Gaussian nature of the environment. Intuitively, upon performing a Taylor expansion of the evolution operators, the unbounded environment terms  can be contracted with one another via Wick's theorem and thereby replaced by {$c$-numbers} for ``suitable'' initial states. More mathematically, we first identify spaces of ``suitable'' system-environment states that contain states of the form} \textcolor{black}{$\rho_S\otimes \rho_E${ for every system state $\rho_S$ and each fixed Gaussian state $\rho_E$ satisfying our assumptions}, and we show that this ``suitable'' space remains closed under evolution by the system-environment dynamics as well as on application of the environment operators $B_{\alpha}(t)$ themselves.} Having identified and restricted ourselves to these state spaces, we then proceed with an analysis similar to that of finite-dimensional lattice models with bounded interactions.
 \begin{figure*}
    \centering
    \includegraphics[width=1.01\linewidth]{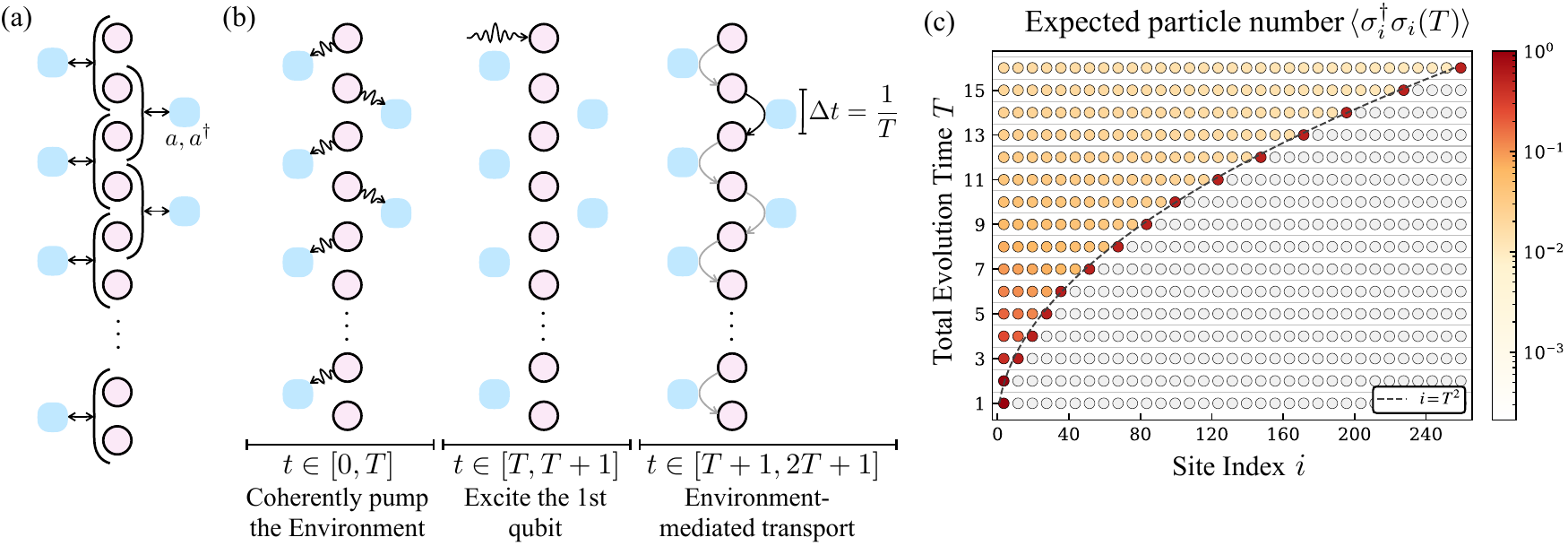}
    \caption{(a)Schematic depicting a time-dependent 1D model with supersonic transport that violates a Lieb-Robinson bound with a linear light cone. \textcolor{black}{The model is constructed by interleaving three different time-dependent system-environment Hamiltonians --- for $0 \leq t \leq T$, the bath modes are continuously driven by a constant field and accumulate $T^2$ particles on average. Then, from $T \leq t \leq  T + 1$, the first qubit in the system is excited and from $T + 1 \leq t \leq 2T + 1$, the bath oscillators, which have $T^2$ particles on average in each, are used to mediate a transport of this excitation at a velocity $\sim T$. (b)Numerical simulations of the above-mentioned protocol for different values of $T$. The heat map shows the final occupation probability distribution of excitations. The wave front is observed to propagate as $T^2$, implying a propagation velocity proportional to $T$.%
    }}
    \label{fig:supersonic}
\end{figure*}

The explicit Lieb--Robinson velocity $v_\text{LR}(t)$ that we obtain in Proposition~\ref{prop:lr_bound_main} increases with $a_0$, $\mathcal{Z}$, and $L$, which is expected since {larger values} of these parameters {increase} the effective strength and range of interactions in the model. The Lieb--Robinson velocity additionally depends on $m(t)$, the $1$-norm measure of the memory kernels defined in Eq.~\eqref{eq:one_norm_measure}, which in turn depends on the initial environment state. {The quantity $m(t)$ characterizes the cumulative strength of the bath correlations over the evolution interval.}  The dependence of $v_{\text{LR}}$ on $m(t)$ is physically expected
since, if, due to the environment, the system state can be
correlated with its past, then this can effectively increase the
strength of the local interactions by a factor proportional
to $m(t)$.

\textcolor{black}{Finally, we remark that the estimate of the Lieb-Robinson velocity can be improved if it is apriori known that some of the local baths only couple to a single site i.e., have coupling operators $R_i^{(\nu)}(t)$ that act on the same single site (for all $\nu$) and thus cannot cause transport through the system. Formally, this can be easily seen by going into the interaction picture with respect to these system-environment terms and repeating our analysis. As a consequence, in such cases, we can re-define $m(t)$ to exclude the baths which couple only to a single site.}

\emph{Supersonic transport}. Proposition \ref{prop:lr_bound_main} suggests that for models where $\mu(t)$ {continues to grow without bound}, the model can exhibit supersonic transport, i.e., in time $t$, the support of a local observable can grow to a distance $\propto t^\alpha$ for $\alpha > 1$ in the long-time limit. We now illustrate such supersonic transport explicitly in a time-dependent system-environment model and show that in time $t$, the model can transport excitations to a distance $\propto t \mu(t)$, thus asymptotically saturating the Lieb-Robinson bound \footnote{We remark that our construction differs from that of Ref.~\cite{eisert2009supersonic}, which also developed a lattice of infinite-dimensional systems with supersonic transport, in that we explicitly consider models that have a Hamiltonian of the form in Eq.~(\ref{eq:full_sys_env}) while the construction of Ref.~\cite{eisert2009supersonic} required more complex nearest neighbour interactions and did not satisfy the template of Eq.~(\ref{eq:full_sys_env})}. 
\textcolor{black}{We consider a one-dimensional chain of $N$ qubits as the system and a chain of $N - 1$ harmonic oscillators as the environment. Each pair of qubits couples to an independent harmonic oscillator: In particular, the $i^\text{th}$ bosonic mode interacts with qubits $i$ and $i + 1$ during the protocol [Fig.~\ref{fig:supersonic}(a)].} \textcolor{black}{We assume that the {initial joint state is $\ket{0}_S^{\otimes N}\otimes\ket{0}_E^{\otimes(N-1)}$} and we take the environment operators $B_{i}^{{(\nu)}}(t)$ in Eq.~(\ref{eq:independent_bath_kernel_maintext}) to be given by $B_i^{(1)}(t) = a_i$ and $B_i^{(2)}{(t)}=a{_i}^\dagger$. With this choice, we obtain that
\[
{\K_{i, \nu, \nu'}(t)} = \begin{cases}
    1 &\text{ if } \nu = 1, \nu' = 2\\
    0 &\text{ otherwise}.
\end{cases}
\]
Therefore, for this model, the $1-$norm measure defined in Eq.~\eqref{eq:one_norm_measure} is given by $m(t) = t$, and consequently Proposition \ref{prop:lr_bound_main} would suggest that $v_\text{LR}(t) \sim t$, i.e., this model should be able to exhibit transport across distances $\sim t^2$ in time $t$.}

\textcolor{black}{To explicitly design a protocol that exhibits supersonic transport in this model, we 
consider a total evolution time of $2T + 1$ where $T$ is an integer.}
\textcolor{black}{The protocol is divided into three stages [Fig.~\ref{fig:supersonic}(b)]:
\begin{enumerate}
    \item[(1)] In the first stage, which occurs during the time-interval $[0, T]$, we use the system for time $T$ to pump each environment {mode} into a coherent state with amplitude $T$. This is achieved by starting all the system qubits in the state $\ket{0}$, and applying a system-environment Hamiltonian $H_\text{SE} = i\sum_{j = 0}^{N - 1}\sigma_j^z(a_j^\dagger - a_j)$. 
    \item[(2)] In the second stage, which occurs during the time-interval $[T, T + 1]$, we switch off the system-environment interaction and drive the first system qubit to $\ket{1}$. This is accomplished with the system Hamiltonian $H_\text{S} = \pi \sigma_1^x/2 $.
    \item[(3)] In the last stage, corresponding to the time-interval $[T + 1, 2T + 1]$, we use the excited environment modes to mediate a fast transport of the excitation in the system via environment-assistend particle exchange. Concretely, this is accomplished by using the system-environment Hamiltonian
    \begin{subequations}\label{eq:excitation_propagation}
    \begin{align}
    H_\text{SE}(t) = \frac{\pi}{2} \sum_{j = 1}^{N - 1} g_j(t {-} (T +1))h_{j, j + 1},
    \end{align}
    where
    \begin{align}
        h_{j, j + 1} =  \ket{0_j 1_{j + 1}}\!\bra{1_j 0_{j + 1}}a_j + \text{h.c.},
    \end{align}
    and
    \begin{align}
    g_j(t) = \begin{cases}
        1 & \text{ if } t \in [(j -1)/T, j/T], \\
        0 & \text{ otherwise},
    \end{cases}
    \end{align}
    \end{subequations}
    i.e., we first switch on an interaction between the first environment mode, and the first two qubits for a total time $\Delta t =1/T$, then between the second environment mode and the next two qubits and so on. {When the $j^\text{th}$ interaction is active, it removes one photon from bath mode $j$ while implementing an approximate excitation transfer from qubit $j$ to qubit $j+1$.}
\end{enumerate}
}

\textcolor{black}{Since the environment was initially in the vacuum state and all the qubits in $\ket{0}$, at the end of the first stage (i.e., at time $T$), the system-environment state is $\ket{0}^{\otimes N} \otimes \ket{\alpha = T}^{\otimes (N - 1)}$, where $\ket{\alpha}$ is a coherent state with amplitude $\alpha$. In the second-stage, we  simply excite the first qubit while keeping the environment modes unchanged, yielding a system-environment state $\sigma_1^\dagger\ket{0}^{\otimes N} \otimes \ket{\alpha = T}^{\otimes (N - 1)}$ with $\sigma^\dagger=\ket{1}\bra{0}$.}

\textcolor{black}{In the last stage, this excitation is propagated through the chain of system qubits. The velocity at which the excitation is propagated can be estimated by approximating the environment modes classically: Since the environment modes are in a coherent state with amplitude $T$, this amounts to setting $a_j \to T$
\[h_{j, j + 1}\approx \frac{\pi T}{2}(\ket{0_j 1_{j + 1}}\!\bra{1_j 0_{j + 1}}  + \text{h.c.})\]
Thus, in time $1/T$, $h_{j, j + 1}$ can map $\ket{1_j 0_{j + 1}} \to \ket{0_j 1_{j + 1}}$ i.e., transfer the excitation from the $j^\text{th}$ qubit to the $(j + 1)^\text{th}$ qubit. This yields a velocity of the excitation that grows as $T$, and in time $T$, the excitation can be transported to a distance of $T^2$ thus exhibiting \emph{supersonic transport}. This is numerically exhibited in Fig.~\ref{fig:supersonic}(c) where the expected excitation number $\langle \sigma_i^\dagger \sigma_i\rangle$ is shown as a function of the site $i$ at the end of the protocol for different choices of $T$. We see that the excitation is indeed propagated to a distance of $T^2$: However, we observe that the transfer of excitation is not perfect i.e., the qubits at $i < T^2$ are not exactly in $\ket{0}$ at the end of the protocol. This can be attributed to an imperfect environment-mediated excitation transfer by $h_{j, j + 1}$ due to quantum fluctuations in the environment mode around its classical approximation. {In the Supplement Sec.~\ref{sec:supp_independent_saturation},} we rigorously analyze the impact of these fluctuations and show that they can be neglected in the limit of large $T$.
} 

We remark that all the time-dependent terms in the system-environment Hamiltonian can be switched on and off smoothly to ensure that they are differentiable with respect to time as required in our setup [see text below Eq.~(\ref{eq:sys_env_hamiltonian})]---this is addressed in the Supplement {Sec.~\ref{sec:supp_independent_saturation}}. In {that section,} we also explicitly consider the observable $O = \sigma_{T^2 + 1}^\dagger \sigma_{T^2 + 1}$, which measures the number of excitations at the $(T^2 + 1)^\text{th}$ qubit, and show that $\Delta_{O}(2T + 1, 0; l)$ [as defined in Eq.~(\ref{eq:lr_target})], is a non-zero constant if $l < T^2$, thus exceeding the conventional linear light cone and saturating the Lieb-Robinson bound provided in Proposition 1.

\subsection{Coupled memory}
\begin{figure}
    \centering
    \includegraphics[width=1.0\linewidth]{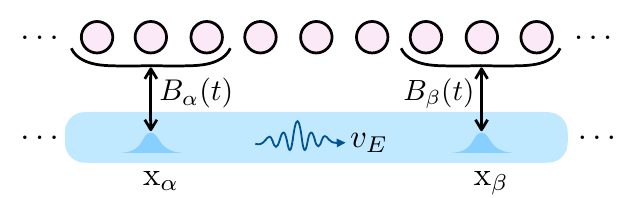}
    \caption{\textcolor{black}{Schematic depiction of an environment with coupled Bath operators. In this setting, bath operators such as $B_\alpha(t)$ and $B_\beta(t)$ are embedded in a shared environment and are associated with positions $\text{x}_\alpha$ and $\text{x}_\beta$. Excitations in the shared environment can propagate with velocity $v_E$, producing correlations between different system regions.}}
    \label{fig:nonind_bath_model}
\end{figure}
\textcolor{black}{We next consider the case in which the environment does not have independent baths and can, due to its own dynamics,  propagate excitations even in the absence of coupling to the system. More specifically, as depicted in Figure~\ref{fig:nonind_bath_model}, we associate the environment operators $B_\alpha(t)$ [defined in Eq.~\eqref{eq:sys_env_hamiltonian}] which interact with system sites in $\mathcal{S}_\alpha$ with a position $\text{x}_\alpha$ in the environment. We will assume the distance between two such positions $\text{x}_\alpha, \text{x}_\beta$ scales with the lattice distance between with regions $\mathcal{S}_\alpha, \mathcal{S}_\beta$ i.e., we will assume that there are system-size independent constants $p_\pm, q_\pm > 0$ such that
\begin{subequations}\label{eq:commensurate_distance}
\begin{align}
&d(\mathcal{S}_\alpha, \mathcal{S}_\beta) \leq p_+d(\text{x}_\alpha, \text{x}_\beta) + q_+ \text{ and},\\
&d(\text{x}_\alpha, \text{x}_\beta) \leq p_- d(\mathcal{S}_\alpha, \mathcal{S}_\beta) + q_-.
\end{align}
\end{subequations}
Unlike the independent bath model previously considered, we allow the kernels $K_{\alpha,\beta}(\tau)\neq 0$ for generic $\alpha, \beta$.
However, we assume that the free environment dynamics itself has a finite propagation velocity: This is formalized by assuming that the kernels $K_{\alpha, \beta}$ satisfy the bound
\begin{align}
&\int_0^\tau \abs{K_{\alpha,\beta}(s)}\,ds \leq\nonumber\\
&\qquad \qquad K_0 \min\left\{\exp\left(\frac{v_E|\tau|-d(\text{x}_\alpha, \text{x}_{\beta})}{a_E}\right),1\right\}, \label{eq:main_text_coupled_conditions}
\end{align}
where $a_E, K_0 > 0$. In this model, the kernels have a bounded one-norm i.e., $\int_0^\infty \abs{K_{\alpha, \beta}(\tau)}d\tau \leq K_0 < \infty$. Furthermore, over time $\tau$, the environment can only significantly propagate excitations over a distance $v_E \tau$, and the memory kernel $K_{\alpha, \beta}(\tau)$ exhibit at-least an exponential decay if $d(\text{x}_\alpha, \text{x}_\beta) \gg v_E \tau$. We remark that $v_E, a_E$ and $K_0$ may depend on the initial environment state $\rho_E$. }

\textcolor{black}{This environment model captures physical setting where, for instance, the environment is a bosonic field and is interacting with the system at different spatial points. Physical examples of such models could be an array of closely spaced atoms coupling to a free or structured electromagnetic field without making the approximation that the photons emitted by one of the atoms is not absorbed by a different atom.}

Our next result provides a Lieb-Robinson bound for this model. Surprisingly, the propagation in the environment itself can dramatically modify the nature of the light-cone in the system, which is captured in the following Proposition.

\begingroup\color{black}
\begin{proposition}
Given a local observable $O_X$ supported in $X\subseteq \Lambda$ with $\abs{X}\leq O(1)$ and with the memory-kernels $K_{\alpha, \beta}(\tau)$ satisfying Eq.~(\ref{eq:main_text_coupled_conditions}), then for any fixed $0<\kappa_{E}<1/a_{E}$ and all initial states $\sigma_S$ of the system qudits, $\Delta_{O_X}(t, t'; l)$ defined in Eq.~(\ref{eq:lr_target}) satisfies
    \begin{equation}
    \begin{aligned}
    &\Delta_{O_X}(t,t';l)\leq \\&\quad\quad\norm{O_X}f_c(t,l)\exp(\kappa_{E}(v_{\textnormal{LR,c}}(t)(t-t')-l)/p_+),
    \end{aligned}
    \end{equation}
    where $p_+$ is the constant defined in Eq.~\eqref{eq:commensurate_distance}, $f_c(t,l)\leq\mathcal{O}(t^{3d}l^{d})$, $v_{\textnormal{LR,c}}(t)\leq{O}(t^d)$ as $t,l \to \infty$.
\end{proposition}
\endgroup
\textcolor{black}{The proof of this proposition follows a similar strategy of identifying relevant system-environment states as in the proof of proposition \ref{prop:lr_bound_main}. However, unlike proposition~1, the environment is not composed of locally independent baths and its coupled structure has to be taken into account: In the Supplement {Sec.~\ref{sec:supp_coupled_bound}}, we show how to handle this additional complication and derive a Lieb-Robinson bound. We remark that Proposition~2 {provides a super-linear upper envelope,} i.e., in $d-$dimensions, a local observable's support can grow to a distance $l \sim t^{d+1}$ instead of $\sim t$ as would be predicted by  Proposition~1 for memory kernels with a bounded 1-norm. Physically, this is a consequence of the long-range interactions between the system qudits mediated by the environment. More specifically, since the environment propagates excitation at a velocity $v_E$, at time $t$ excitations at any spatial location in the environment can {in principle} coherently accumulate from a region of size $t^d$. This can then possibly enhance the effective environment-mediated interaction between the system qudits by {at most} a factor of $t^d$, giving rise to a time-dependent Lieb-Robinson velocity $v_{\textnormal{LR,c}} \sim t^d$. In the Supplement {Sec.~\ref{sec:supp_coupled_saturation}}, we make this argument precise and construct a
model that explicitly saturates a Lieb-Robinson bound,
with exponential tails outside the superlinear light-cone,
with a Lieb-Robinson velocity scaling as $t^d$.}

\section{Implication on bath approximations} 

For a many-body system, if we are interested in approximating the full system state at a given time and to a target precision{,} then the number of modes in the environment is expected to scale with the system size since we effectively demand a high precision on high-order correlators in the system. However, for a geometrically local many-body system interacting locally with independent baths, it is physically expected that local observables or few-point correlation functions in the system can be approximated well by a number of modes in each local bath that is independent of the system size, and depends only on the evolution time and the target precision.

\begin{figure}[b]
    \centering
    \includegraphics[width=0.8\linewidth]{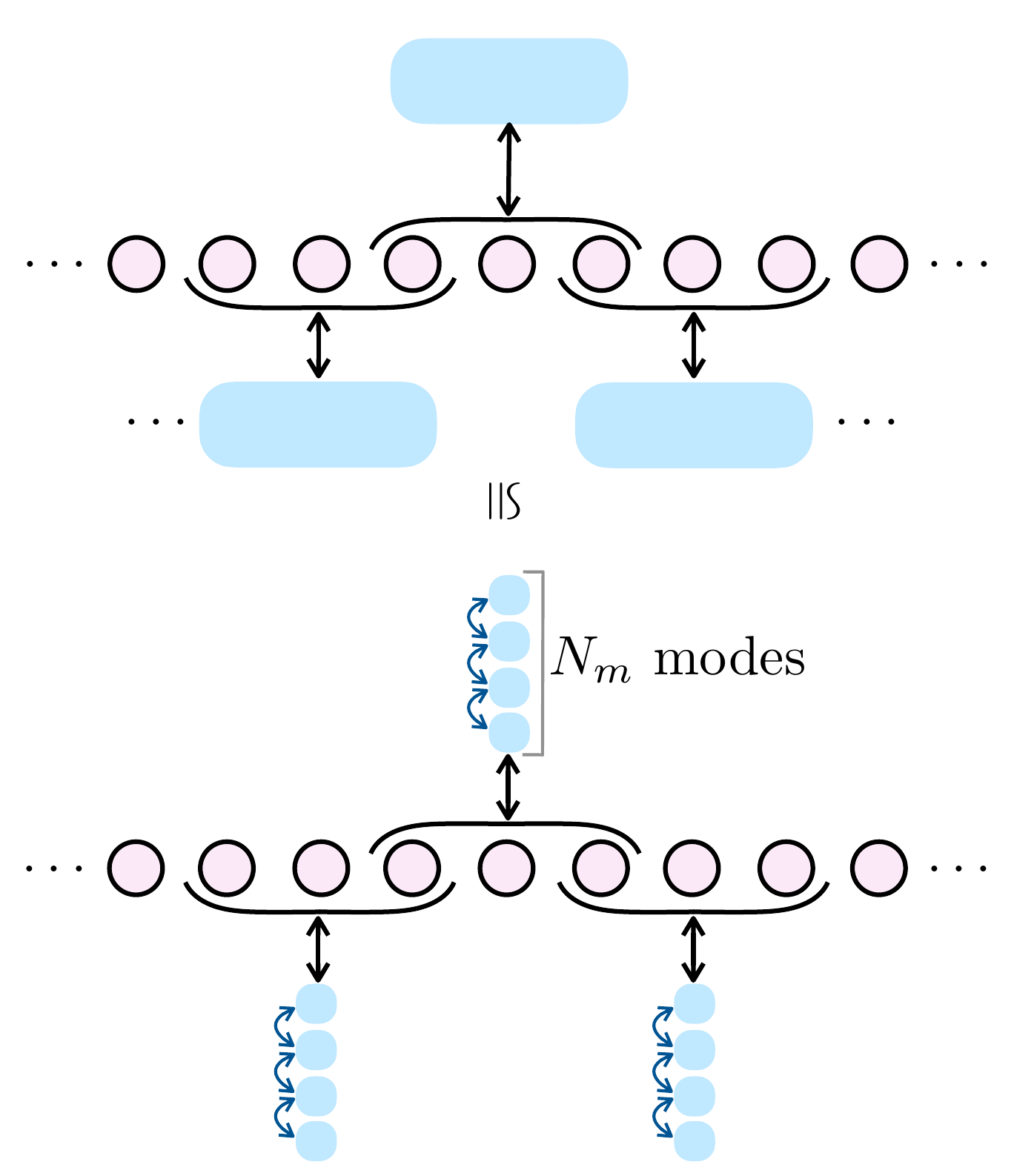}
    \caption{Schematic depiction of the star-to-chain transformation used to approximate an independent non-Markovian environment by a discrete set of bosonic modes. Each bath is replaced with $N_m$ modes, where $N_m$ controls the accuracy of this approximation, and these modes themselves 1D nearest-neighbour coupled lattice.}
    \label{fig:star_to_chain}
\end{figure}

This physical expectation can be made precise by applying the Lieb Robinson bounds from Proposition 1. We will analyze the star-to-chain transformation as the method to construct the discrete-mode approximation of each independent bath \cite{chin2010exact, woods2015simulating, trivedi2022description}. 
As depicted in Fig.~\ref{fig:star_to_chain}, each bath is approximated by a 1D chain of $N_m$ discrete bosonic modes with nearest neighbour coupling, where $N_m$ controls the accuracy of the approximation. For simplicity, we consider the case of $L=2$, while the generalization is straightforward~\cite{BlockLanczosShirakawa2014,EfficientNuomin2024,EfficientIrene2025}. Recall that $L$ defined above Eq. (\ref{eq:sys_env_hamiltonian}) is the number of terms that act on the same subregion. Now, the system-environment Hamiltonian in Eq.~(\ref{eq:sys_env_hamiltonian}) can be written as
\begin{equation}\label{eq:simple_Hamiltonian_used_in_star}
H_\text{SE}(t)=\sum_{i}\left(A_{i}^{\dagger}(t) R_{i}(t)+\text{h.c.}\right),
\end{equation}
{Equation~\eqref{eq:simple_Hamiltonian_used_in_star} corresponds to $B_i^{(-)}(t)=A_i(t)$ and $B_i^{(+)}(t)=A_i^\dagger(t)$ in Eq.~\eqref{eq:sys_env_hamiltonian}.}
Denoting the annihilation operators of the discrete bosonic modes in the star-to-chain transformation by $b_{i, 1}, b_{i, 2} \dots b_{i, N_m}$, Eq. (\ref{eq:simple_Hamiltonian_used_in_star})  is approximated instead by a Hamiltonian of the form
\begin{align*}
&{H}_{\text{SE}}^{N_m}(t)= \sum_{i}\left[ g_{i} (t) b_{{i},1}^\dagger(t) R_{i}(t) + g_{i}(t) b_{i, 1}(t) R_{i}^\dagger (t)\right],
\end{align*}
where $b_{i, 1}(t) = e^{iH_{i, E}t} b_{i, 1}e^{-iH_{i,E}t}$ with
\[
H_{i, E}=\sum_{j = 1}^{N_m} \omega_{i, j} b_{i, j}^\dagger b_{i, j} + \sum_{j = 1}^{N_m - 1} \big(t_{i, j}b_{i, j}^\dagger b_{i, j + 1} + \text{h.c.}\big).
\]
The constants $g_{i}, \omega_{i, j}, t_{i, j}$ can be chosen to best approximate the memory kernels of $H_\text{SE}(t)$ with those of $H^{N_m}_\text{SE}(t)$. We review the standard procedure to compute these coefficients from the kernels $\K_i^{(\nu,\nu')}(t)$ (Eq.~\ref{eq:independent_bath_kernel_maintext}) in the Supplement {Sec.~\ref{sec:supp_star_chain}}.

Employing the Lieb-Robinson bounds from Proposition 1, we now estimate the number of modes $N_m$ required to approximate local system observables using the star-to-chain transformation. The following proposition states our key result, and establishes that $N_m$ can be chosen to be uniform in the system size and depend only on the evolution time and the precision demanded in the local observable.
\begin{proposition}
    There exists an approximation of the non-Markovian environment with $N_m$ discrete bosonic modes per bath such that every local system observable $O$ with support that has an $O(1)$ diameter and $\norm{O}\leq 1$ can be approximated to precision $\varepsilon$ at time $t$ with
    \begin{align*}
    N_m &= \Theta\big(\big(\varepsilon^{-1}{t^{2d + 3}}\big)^{1 + o(1)}\big) + \Theta\big( \big({\varepsilon}^{-1}\log \varepsilon^{-1}\big)^{1 + o(1)}\big) +  \nonumber\\
    &\qquad \Theta\big(\big(t\kappa_0\big(\varepsilon^{-1}\Theta(t^{d + 1}) \Theta(t \log^d \varepsilon^{-1})\big)^{1 + o(1)} \big),
    \end{align*}
    where $\kappa_0 : (0, \infty) \to (0, \infty)$ is a function that satisfies
    \[
    \frac{1}{x} = \sum_{\tau \in \{0, t\}} \int_{\tau - \kappa_0(x)}^{\tau + \kappa_0(x)} \textnormal{K}_c(s) ds,
    \]
    with $\textnormal{K}_c$ the upper bound on $K_{i,c}^{(\nu,\nu')}$ for all $i$.
\end{proposition}
\noindent The proof of this proposition follows almost directly from Proposition 1 --- the evolution of the local observable up to time $t$ depends on a neighbourhood of diameter $O(t)$ around it. Consequently, it is only affected by the error in the discrete-mode approximation of the baths interacting with this neighbourhood --- we estimate this error by following previous analyses of the star-to-chain transformation \cite{woods2015simulating, trivedi2022description}. The details of our analyses can be found in the supplementary material. While we have chosen to analyze the star-to-chain transformation, we expect a similar conclusion to hold for other methods for approximating the bath, such as the pseudo-mode approximation, which could outperform the star-to-chain transformation for specific systems \cite{tamascelli2018nonperturbative, pleasance2020generalized}. 

\section{Conclusion}
Our work provides a Lieb-Robinson bound for non-Markovian many-body systems, which fundamentally provides a velocity of propagation of information in these systems. Similar to the Markovian case, we expect this bound to be important in theoretical studies of both non-equilibrium as well equilibrium properties of non-Markovian many-body systems. \textcolor{black}{For example, our bound implies robustness of local observables against perturbations to the memory kernels by leveraging techniques in~\cite{unifiedZhen2024}, with applications in simulating such systems~\cite{yu2025optimizingdigitalquantumsimulation}.}

Our work opens up several interesting questions pertaining to non-Markovian many-body physics, and provides techincal tools that could possibly be developed further to answering these questions. Since we are considering non-Markovian models, in comparison to the Markovian case, there is a possibility of creating interactions between the qubits, via the environment, that are correlated in time and not just in space. The first open question is to more sharply characterize propagation of information in non-Markovian models where the kernels decay slowly and algebraically and thus have {an unbounded time-integrated one-norm} (e.g., if the kernels $\sim t^{-\alpha}$ for $\alpha \in (0, 1)$). These models are time-local in the sense that the memory kernels go to 0 at large time differences, albeit slowly. {Our analysis does not conclusively support a linear light cone in these models, but instead suggests a polynomial one.} In the same spirit, and motivated by results in many-body Hamiltonians and Lindbladians, building upon the tools developed in this paper we could investigate the light-cone structure of models with both algebraically decaying spatial interactions which are also non-Markovian. We expect these analyses to allow us to elucidate the many-body physics that can be induced by the rich spatio-temporal correlation that may exist in non-Markovian many-body models and reveal how they differ from the more well-understood Markovian case.

\section*{Acknowledgement}
We acknowledge useful discussions with Ana Asenjo Garcia, Frederik Nathan, Alexey Gorshkov and Darrick Chang.

\section*{Funding sources}
R.T. acknowledges funding from the European Union's Horizon Europe research and innovation program under grant agreement no. 101221560 (ToNQS). RT also acknowledges funding by the Deutsche Forschungsgemeinschaft (DFG, German Research Foundation) under Germany’s Excellence Strategy – EXC-2111 – 390814868. M.R. acknowledges the Brown Investigator Award, a program of the Brown Science Foundation, the University of Washington College of Arts and Sciences, and the Kenneth K. Young Memorial Professorship for support.

\section*{Author contribution}
R.T. and M.R. conceived the problem and formulated the mathematical setting. X.Y. and R.T. developed the analytical framework and the proof techniques. All the authors contributed to the writing of the manuscript

\onecolumngrid
\newpage

{\begin{center}
\large \textbf{Supplementary Information}
\end{center}
}
\tableofcontents
\section{Notation}
Given a Hilbert space $\mathcal{H}$, we will denote by $\text{L}(\mathcal{H})$ the set of bounded linear operators from $\mathcal{H}\to \mathcal{H}$, define $\text{M}(\mathcal{H})$ to be the set of the bounded Hermitian operators from $\mathcal{H}\to \mathcal{H}$ and define $\text{D}_1(\mathcal{H})$ as the set of valid density matrices on $\mathcal{H}$. Given a set of (not necessarily commuting) operators $O_1, O_2 \dots O_n$, we will use the product notation
\[
\prod_{i = 1}^n O_i = O_1 O_2 \dots O_n \text{ and }\prod_{i = n}^1 O_i = O_n O_{n - 1} \dots O_1.
\]
We will typically use the $\dagger$ superscript to indicate the adjoint, or Hermitian conjugate, of an operator or superoperator. \textcolor{black}{Specifically, the adjoint of a superoperator $\mathcal{L}$ is defined as $\mathrm{Tr}(A\mathcal{L}(B))=\mathrm{Tr}(\mathcal{L}^\dagger (A)B)$ for $A,B\in \mathrm{L}(\mathcal{H})$.} However, in some cases, a more compact expression can be obtained by using the following notation --- for some operator or superoperator $X$, we define $X^{(-)}:=X$ and $X^{(+)}:=X^\dagger$. Furthermore we will use $\bar{+}=-$ and $\bar{-}=+$. For example, $X^{(\bar{u})}:=X^\dagger$ for $u=-$. While dealing with mixed states and their dynamics, it will often be convenient to adopt the vectorized notation, where we map operators on a (finite-dimensional) Hilbert space to state via $A = \sum_{i_l, i_r}A_{i_l, i_r}\ket{i_l}\!\bra{i_r} \to \vecket{A} = \sum_{i_l, i_r}A_{i_l, i_r} \ket{i_l, i_r}$. 
Superoperators, such as Lindbladians- or channels, will be mapped to ordinary operators in this picture. Given an operator $X \in \text{M}(\mathcal{H})$, we will define $X_l, X_r \in \text{M}(\mathcal{H}\otimes \mathcal{H})$ by $X_l\vecket{\rho} = (X\otimes I)\vecket{\rho} = \vecket{X\rho}$ and $X_r\vecket{\rho} = (I\otimes X^\text{T})\vecket{\rho} =  \vecket{\rho X}$. $X_l (X_r)$ can also be interpreted as a superoperator which left (right) multiplies its argument with $X$ i.e. $X_l(Y) = XY$ and $X_r(Y) = YX$. Furthermore, given an operator $X \in \text{L}(\mathcal{H})$, we will denote by $\mathcal{C}_X = X_l- X_r$ as the commutator superoperator corresponding to $X$ i.e. $\mathcal{C}_X \vecket{\phi} = \vecket{[X, \phi]}$.

Note also that inner products between vectorized operators are equivalent to Hilbert-Schmidth inner products between the unvectorized operators: $\vecbra{A}B\rrangle = \text{Tr}(A^\dagger B)$. Also note that $\text{Tr}(A) = \text{Tr}(I^\dagger A) =\vecbra{I}A\rrangle$. Similarly, if we consider the tensor product of two Hilbert spaces $\mathcal{H}_S \otimes \mathcal{H}_E$ and suppose $O \in \text{L}(\mathcal{H}_S \otimes \mathcal{H}_E)$, then we obtain the following notation for partial trace $\text{Tr}_E(O)= \vecbra{I_E} O\rrangle$. Finally, while we will use vectorized notation throughout this paper, all the norms, even when applied in the vectorized notation, will correspond to unvectorized norms. For e.g.~$\norm{\vecket{\sigma}}_1$ will be correspond to $\norm{\sigma}_1$ i.e.~the $1-$norm of $\sigma$ as an operator and not to the $1-$norm of $\sigma$ as a vector. \textcolor{black}{For finite-dimensional operators, $\norm{\cdot}_p$ denotes the Schatten $p$-norm when its argument is an operator; for superoperators, $\norm{\cdot}_\diamond$ denotes the diamond norm.}

For a set $\mathcal{S}$, we will denote by $\Theta_{\mathcal{S}}$ the indicator function on $\mathcal{S}$ i.e.~$\Theta_\mathcal{S}(s) = 1$ if $s \in \mathcal{S}$ and $0$ otherwise. We will denote by $L^p(\mathbb{R})$ the space of functions $f:\mathbb{R}\to \mathbb{C}$ such that $\abs{f}^p$ is integrable. For a function $f:\mathbb{R} \to \mathbb{C}$, $\norm{f}_p$ (when defined) will denote the $L^p$ norm --- in particular,
\[
\norm{f}_p = \bigg(\int_{-\infty}^\infty \abs{f(x)}^p dx\bigg)^{1/p},
\]
and $\norm{f}_\infty$ is the smallest non-negative number such that $f(x) \leq \norm{f}_\infty$ almost everywhere. Given a function $f$ and an interval $I$, we will define its norm restricted to an interval $I$ via
\[
\smallnorm{f_I}_p = \smallnorm{f \cdot  \Theta_I}_p \quad \text{where}\quad \Theta_I(t)= \begin{cases}
    1 & \text{ if } t \in I, \\
    0 & \text{ otherwise}.
\end{cases}
\]
For a vector of measurable functions $\bm{f}(t) = (f_1(t), f_2(t), \dots, f_n(t))$, we will define
\[
\norminf{\bm{f}(t)}:=\sup_{\alpha}\abs{f_\alpha(t)}, \quad \text{and}\quad \norm{\bm{f}_I}_\infty:=\sup_{\alpha}\norm{(f_{\alpha})_{I}}_\infty,
\]
where $f|_I$ is the restriction of $f$ to the interval $I$ i.e., $f|_I(t) = f(t)$ if $t \in I$ and $0$ otherwise. We will also adopt the standard notation and denote by $\text{C}^1(\mathbb{R})$ the space of continuously differentiable complex-valued functions on $\mathbb{R}$.

We will also define $\mathcal{K}_1(\mathbb{R})$ to be the space of tempered distributions $\K$ which are of the form
\begin{align}\label{eq:rep_kernel}
\K(t) = \K_c(t) + \sum_{j = 1}^M k[j] \delta(t - \tau_j),
\end{align}
where $\K_c \in \text{C}^1(\mathbb{R})$ is a continuous function and $\delta(\cdot)$ is the Dirac-delta function. The function $\K_c$ will often be called the continuous part of $\K$ and $\K - \K_c = \sum_{j = 1}^M k[j] \delta(t - \tau_j)$ will be called the atomic part of $\K$. {Identifying $\K$ with the complex measure $d\mu_{\K}(t)=\K_c(t)dt+\sum_{j=1}^M k[j]\delta_{\tau_j}(dt)$, we use the standard probability-theory definition $\mathrm{TV}(\K;I):=|\mu_{\K}|(I)$. For the representation in Eq.~\eqref{eq:rep_kernel}, this gives}
\[
\text{TV}(\K; I) = \int_I \abs{\K_c(t)}dt + {\sum_{j:\,\tau_j\in I} \abs{k[j]}.}
\]
{The total variation on the full line is}
\[
\text{TV}(\K) = \text{TV}(\K; \mathbb{R}) = \norm{\K_c}_{1} + \sum_{j = 1}^M \abs{k[j]}. 
\]
Given two functions (or distributions), $f$ and $g$, we will denote by $f\star g$ the convolution of $f$ and $g$ formally defined via
\[
(f\star g)(\tau) = \int_{-\infty}^\infty f(\tau') g(\tau - \tau') d\tau'.
\]

\emph{Summation convention}. We will often omit writing summations in more tedious calculations --- in any equation or inequality, any index (superscript or subscript) on the right-hand side which does not appear on the left hand side will be assumed to be summed over. Any index (superscript or subscript) that appears on both left-hand and right-hand side of an equation will \emph{not} be summed over. For e.g.
\begin{align*}
&F = x_i y_i \text{ is shorthand for }F = \sum_{i} x_i y_i, \\
&F = x_i \text{ is shorthand for }F = \sum_i x_i, \\
&F_k = x_i^k y^j \text{ is shortand for }F_k = \bigg(\sum_i x_i^k\bigg) \bigg(\sum_j y_j\bigg), \text{ and }\\
&F_k = x_i^k y_j z^j_k \text{ is shortand for }F_k= \bigg(\sum_i x_i^k\bigg) \bigg(\sum_{j} y_j z^j_k\bigg).
\end{align*}
Note that the summation convention that we use is close to the Einstein's summation convention with the difference being that unrepeated indices are also summed over as long as they do not appear on the left-hand side of an equation or inequality.
\section{Preliminaries}
This section reviews the setup that we analyze (section \ref{sec:supp_prelim_setup}), as well as some preliminaries that will be important in our calculations. Section \ref{sec:supp_prelim_wick} reviews the Wick's theorem and defines some notation to streamline its application --- the notation introduced in this subsection will be used throughout the supplement. As introduced in the main text, in this supplement, we also consider system-environment models where the memory kernels can have $\delta$-function divergences --- in section \ref{sec:regularization}, we show that we can instead approximate the $\delta$ function as a limit of smooth functions (i.e.,~mollifiers) while computing physically relevant quantities (such as the dynamics of system observables). This can be seen as effectively defining the dynamics associated with models whose kernels have $\delta$-function divergences, and is an important simplification for the following section. If the reader is willing to accept that this is possible, then they can skip section \ref{sec:regularization} and directly go to next section for the calculation of the Lieb-Robinson bound.
\subsection{Setup}\label{sec:supp_prelim_setup}
As introduced in the main text, we will assume that the system qudits are arranged on a $d-$dimensional lattice $\Lambda$. The qudits interact with each other either directly or through a bath. We will denote the system Hilbert space by $\mathcal{H}_S$ and the environment Hilbert space by $\mathcal{H}_E$. The system-environment Hamiltonian introduced in the main text, expressed directly in the interaction picture with respect to the environment, is given by
\begin{align}\label{eq:supp:sys_env_hamil}
    H(t) = \sum_{\alpha} h_\alpha(t) + \sum_\alpha  B_{\alpha}(t) R_{\alpha}(t),
\end{align}
where $h_\alpha(t), R_{\alpha}(t)$ are geometrically local operators supported on $\mathcal{S}_{\alpha}$ and $B_{\alpha}(t)$ are operators on the environment that are linear in its annihilation and creation operators. We also introduce the quantities, $\mathcal{Z}$ and $a_0$, that are determined by the supports $\mathcal{S}_{\alpha}$ via
\begin{align}\label{eq:constants_lattice}
\mathcal{Z} = \sup_{S_\alpha} \abs{\{\mathcal{S}_{\alpha'} : \mathcal{S}_{\alpha'} \cap \mathcal{S}_{\alpha} \neq \emptyset\}} \text{ and }a_0 = \sup_{S_\alpha} \text{diam}(\mathcal{S}_{\alpha}).
\end{align}
For geometrically local models both $a_0, \mathcal{Z}$ will be independent of the system size. The constant $a_0$ is simply an upper bound on the diameters of the supports $\mathcal{S}_{\alpha}$. The constant $\mathcal{Z}$ can be considered to be the coordination number of the interaction graph, i.e.,~it is an upper bound on the number of sets $\mathcal{S}_{\alpha'}$ that intersect with any one set $\mathcal{S}_{\alpha}$. We will also assume that $\norm{h_\alpha(t)}, \norm{R^\nu_\alpha(t)} \leq 1$ and that $h_\alpha(t), R_\alpha(t)$ are differentiable with bounded derivatives, i.e., $\smallnorm{h_\alpha'(t)}, \smallnorm{{R_\alpha}'(t)} < \infty$.

We will often need to restrict this Hamiltonian to sub-regions of the lattice $\Lambda$. We will denote the full set of values of the index $\alpha$ in Eq.~\ref{eq:supp:sys_env_hamil} by $\mathcal{A}$. Given $\mathcal{B} \subseteq \mathcal{A}$, we define
\begin{align}\label{eq:restriction_hamiltonian}
H_{\mathcal{B}}(t) = \sum_{\alpha \in \mathcal{B}}\bigg(h_\alpha(t) +  B_{\alpha}(t) R_\alpha(t)\bigg).
\end{align}
Note that $H_\mathcal{A}(t) = H(t)$. We will also define the unitary $U_{\mathcal{B}}(t, s)$ by 
\begin{align}\label{eq:restriction_unitary}
U_{\mathcal{B}}(t, s) = \mathcal{T}\exp\bigg(-i \int_s^t H_{\mathcal{B}}(s') ds'\bigg).
\end{align}
We will also define $\mathcal{U}_{\mathcal{B}}(t, s)$ to be the vectorized unitary corresponding to $U_{\mathcal{B}}(t, s)$ and $\mathcal{H}_{\mathcal{B}}(t)$ as the commutator corresponding to $H_{\mathcal{B}}(t)$:
\begin{subequations}\label{eq:def_doubled_unitaries}
\begin{align}
\mathcal{U}_{\mathcal{B}}(t, s) = U_{\mathcal{B}, l}(t, s) U^\dagger_{\mathcal{B}, r}(t, s) \text{ and }\mathcal{H}_{\mathcal{B}}(t) = H_{\mathcal{B}, l}(t) - H_{\mathcal{B}, r}(t) = \sum_{\sigma \in \{l, r\}} (-1)^\sigma H_{\mathcal{B}, \sigma}(s),
\end{align}
where, throughout this paper, we will use the convention 
\begin{align}
    (-1)^\sigma = \begin{cases}
        1 & \text{ if } \sigma = l, \\
        -1 & \text{ if } \sigma = r.
    \end{cases}
\end{align}
Note also that 
\begin{align}
\mathcal{U}_{\mathcal{B}}(t, s) = \mathcal{T}\exp\bigg(-i \int_s^t \mathcal{H}_{\mathcal{B}}(s') ds'\bigg).
\end{align}
\end{subequations}
Finally, for simplicity, we will use the notation $U(t, s) = U_\mathcal{A}(t, s), \mathcal{U}(t, s) = \mathcal{U}_\mathcal{A}(t, s)$ and $\mathcal{H}(t) = \mathcal{H}_\mathcal{A}(t)$. Given $X \subseteq \Lambda$, we will define $\mathcal{A}_X$ by
\[
\mathcal{A}_X = \{ \alpha \in \mathcal{A}: \mathcal{S}_{\alpha} \cap X \neq \emptyset\}.
\]
\textcolor{black}{Note that $H_{\mathcal{A}_X}$ will then be the sum of local terms whose support intersects with $X$.} Given a subregion $X \subseteq \Lambda$, we will also implicitly use
\[
H_X(t) = H_{\mathcal{A}_X}(t), \mathcal{H}_{X}(t) = \mathcal{H}_{\mathcal{A}_X}(t), U_X(t, s) = U_{\mathcal{A}_X}(t, s) \text{ and }\mathcal{U}_{X}(t, s)= \mathcal{U}_{\mathcal{A}_X}(t, s).
\]
We will consider the initial environment state, $\rho_E$, to be Gaussian. It can thus be described by its first and second moments of the annihilation and creation operators. In particular, we assume that
\begingroup\color{black}
\begin{align}\label{eq:gaussian_state_moments}
\text{Tr}(B_{\alpha}(t) B_{\alpha'}(t') \rho_E) = K_{\alpha, \alpha'}(t - t') \text{ where } \K_{\alpha, \alpha'} \in \mathcal{K}_1(\mathbb{R}).
\end{align}
\endgroup
While in general a Gaussian state can also have a displacement, without loss of generality, we will assume that $\text{Tr}(B_{\alpha}(t) \rho_E ) = 0$ since a displacement can be equivalently incorporated as a Hamiltonian term acting only on the system using the Mollow transform \cite{mollow1975pure, gardiner1985input}. It will be convenient to define the kernel ${K}_{\alpha, \alpha', \sigma}(t - t')$, where $\sigma \in \{l, r\}$ via
\begingroup\color{black}
\begin{align}\label{eq:kernel_def}
    \vecbra{I_E} B_{\alpha,\sigma}(t) B_{\alpha',\sigma'}(t')\vecket{\rho_E} = {K}_{\alpha, \alpha', \sigma'} (t - t'),
\end{align}
\endgroup
Importantly, we emphasize that due to the cyclic property of the trace, $\vecbra{I_E} B_{\alpha,\sigma}(t) B_{\alpha',\sigma'}(t')\vecket{\rho_E}$ is independent of $\sigma \in \{l, r\}$. The kernels $\K_{\alpha, \alpha', \sigma'}$ are related to the kernels $\K_{\alpha, \alpha'}$ defined above via
\[
\K_{\alpha, \alpha', \sigma}(\tau) = \begin{cases}
    \K_{\alpha, \alpha'}(\tau) & \text{ if }\sigma = l, \\
    \K_{\alpha', \alpha}(-\tau) & \text{ if }\sigma = r.
\end{cases}
\]
Furthermore, from this relation, it also follows that ${K}_{\alpha, \alpha', \sigma} \in \mathcal{K}_1(\mathbb{R})$ for all $\sigma \in \{r, l\}, \alpha, \alpha' \in \mathcal{A}$, i.e., $K_{\alpha,\alpha',\sigma}$ has the representation

\begin{equation}
K_{\alpha,\alpha',\sigma}(\tau)
=
K_{\alpha,\alpha',\sigma;c}(\tau)
+
\sum_{j=1}^{M}
k_{\alpha,\alpha',\sigma}[j] \delta(\tau - T_j),
\end{equation}
for some $K_{\alpha, \alpha', \sigma; c}$ and $k_{\alpha, \alpha', \sigma}[j]$.

\subsection{Wick's theorem}\label{sec:supp_prelim_wick}
We begin by introducing two operator spaces, $\mathcal{S}(\rho_E)$ and $\mathcal{Q}(\rho_E)$ which depend on the initial environment state $\rho_E$, that we will frequently encounter in our analysis. The operator space $\mathcal{S}(\rho_E)$, formally defined below, can be considered to be the set of operators generated by starting from an operator in which the environment is in $\rho_E$,  applying a sequence of time-evolution under the restricted system-environment Hamiltonians $H_{\mathcal{B}}(t)$ followed by modifying the state of the system (and not of the environment) by application of a linear super-operator.
\begin{definition}[Operator space $\mathcal{S}(\rho_E)$]\label{def:S_rho}
    An operator $\phi \in \mathcal{S}(\rho_E) \subseteq \textnormal{L}(\mathcal{H}_S \otimes \mathcal{H}_E)$ if $\exists \{\mathcal{A}_i \subseteq \mathcal{A}\}_{i \in [1:n]}, \{(s_i, t_i) : s_i, t_i \in \mathbb{R}\}_{[i:n]}$, system superoperators $\{\Omega_i : \norm{\Omega_i}_{\diamond} \leq 1\}_{i\in [1:n]}$ and system operator $\sigma_S$ with $\norm{\sigma_S}_1 \leq 1$ such that, 
    \[
    \vecket{\phi} = \prod_{i = n}^1 \Omega_i \mathcal{U}_{\mathcal{A}_i}(t_i, s_i) \vecket{\sigma_S, \rho_E}.
    \]
    Furthermore, $(\{\mathcal{A}_i\}_{i\in[1:n]}, \{(s_i, t_i)\}_{i\in[1:n]}, \{\Omega_i\}_{i \in [1:n]}, \sigma_S)$ will be called a \textbf{representation} of the operator $\phi$ and $\{(s_i, t_i)\}_{i \in [1:n]}$ will be called the \textbf{time edges} of the representation of $\phi$.
\end{definition}
\noindent The operator space $\mathcal{Q}(\rho_E)$ can be considered to be the space dual to the operator space $\mathcal{S}(\rho_E)$: Here, we start with an operator which is identity $I_E$ on the environment and then apply a sequence of Heisenberg picture time-evolutions with respect to the restricted system-environment Hamiltonians $H_{\mathcal{B}}(t)$ followed by a modification of the state of the system by application of a linear super-operator.
\begin{definition}[Operator space $\mathcal{Q}(\rho_E)$]\label{def:Q_rho}
    An operator $\theta \in \mathcal{Q}({\rho_E})$ if $\exists \{\mathcal{A}_i \subseteq \mathcal{A}\}_{i \in [1:n]}, \{(s_i, t_i) : s_i, t_i \in \mathbb{R} \}_{i \in [1:n]}$, system superoperators $\{\Omega_i : \norm{\Omega_i}_\diamond \leq 1\}_{i \in [1:n]}$ and system operator $O_S$ with $\norm{O_S} \leq 1$ such that 
\[
\vecbra{\theta} = \vecbra{O_S, I_E}\prod_{i = n}^1  \mathcal{U}_{\mathcal{A}_i}(t_i, s_i) \Omega_i.
\]
Furthermore, $(\{\mathcal{A}_i \}_{i \in [1:n]}, \{(s_i, t_i)\}_{i \in [1:n]}, \{\Omega_i\}_{i \in [1:n]}, O_S)$ will be called a \textbf{representation} of the state $\theta$ and $\{(s_i, t_i)\}_{i \in [1:n]}$ will be called the \textbf{time edges} of the {representation} of $\theta$.
\end{definition}

In all the analysis in this as well as the following sections, we will repeatedly use Wick's theorem which we briefly review below. We will directly formulate all the analysis in the vectorized picture. A particularly useful way of formulating the Wick's theorem is to break it down into \emph{Wick contractions}. For two super-operators $B_1, B_2$ linear in the environment annihilation and creation super-operators (henceforth referred to as ``linear environment superoperators"), define the Wick contraction
\[
\wick{\c1{B_1}\, \c1{B_2}} = \vecbra{I_E} B_1 B_2 \vecket{\rho_E} \text{ and } \wick{\c1{B_1}\, \c1{\text{id}}} = \vecbra{I_E} B_1 \vecket{\rho_E} = 0.
\]
More generally, if $B_1, B_2 \dots B_n$ and $B$ are linear environment superoperators, then we define the wick contraction of the product $B_1 B_2 \dots B_n$ with $B$ as
\begin{align*}
&\wick{\c1{(B_1 B_2 \dots B_n)}\, \c1{B}} = \sum_{i = 1}^n (\wick{\c1{B_i}\, \c1{B}}) B_1 B_2 \dots B_{i - 1}B_{i + 1} \dots B_n, \\
&\wick{\c1{B}\, \c1{(B_1 B_2 \dots B_n)}} = \sum_{i = 1}^n (\wick{\c1{B}\, \c1{B_i}}) B_1 B_2 \dots B_{i - 1}B_{i + 1} \dots B_n.
\end{align*}
Note that the Wick contraction of the product $B_1 B_2 \dots B_n$ with $B$ is still an environment superoperator, but is ``one degree" lower i.e. it is a sum of products of $n - 1$ linear environment superoperators.

We extend this definition by linearity: Suppose $C$ is a super-operator such that
\begin{align}\label{eq:analytic_linear_superop}
    C = \sum_\alpha S_\alpha  B_{1, \alpha} B_{2, \alpha} \dots B_{k_\alpha, \alpha},
\end{align}
where $S_\alpha$ is a system superoperator, $B_{i, \alpha}$ are linear environment superoperators i.e., $C$ is a sum of products of linear environment superoperators. Then, if $B$ is another linear ennvironment superoperator, we define
\begin{align*}
&\wick{\c1{C}\, \c1{B}} = \sum_\alpha S_\alpha \wick{\c1{(B_{1, \alpha} B_{2, \alpha} \dots B_{k_\alpha, \alpha})}\, \c1{B}}, \\
&\wick{\c1{B}\, \c1{C}} = \sum_\alpha S_\alpha \wick{\c1{B}\, \c1{(B_{1, \alpha} B_{2, \alpha} \dots B_{k_\alpha, \alpha})}}.
\end{align*}
We can finally state the Wick's theorem succintly as: Suppose $C$ and $B$ are as defined above then
\begin{align}\label{eq:wick_theorem_vanilla}
\vecbra{I_E} C B \vecket{\rho_E} = \vecbra{I_E} \wick{\c1{C}\, \c1{B}}\vecket{\rho_E}.
\end{align}
As an example of applying the Wick's theorem, consider evaluating $\vecbra{I_E} B_1 B_2 B_3 B_4 \vecket{\rho_E}$ where $B_i$ are linear environment superoperators. We proceed as follows:
\begin{align*}
    \vecbra{I_E} B_1 B_2 B_3 B_4 \vecket{\rho_E} &= \vecbra{I_E} \wick{\c1{B_1}\, \c1{(B_2 B_3 B_4)}} \vecket{\rho_E}, \nonumber \\
    &=\vecbra{I_E} (\wick{\c1{B_1} \c1{B_2}} )B_3 B_4 \vecket{\rho_E} + \vecbra{I_E} (\wick{\c1{B_1} \c1{B_3}} )B_2 B_4 \vecket{\rho_E} + \vecbra{I_E} (\wick{\c1{B_1} \c1{B_4}} )B_2 B_3 \vecket{\rho_E}, \nonumber \\
    &=\vecbra{I_E} B_1 B_2 \vecket{\rho_E} \vecbra{I_E}B_3 B_4\vecket{\rho_E} + \vecbra{I_E} B_1 B_3 \vecket{\rho_E} \vecbra{I_E}B_2 B_4\vecket{\rho_E} + \vecbra{I_E} B_1 B_4 \vecket{\rho_E} \vecbra{I_E}B_2 B_3\vecket{\rho_E}.
\end{align*}

In the proofs below, we will often need to analyze expressions of the form $\vecbra{\theta} B_1 B_2 \dots B_n \vecket{\phi}$, where $\theta \in \mathcal{Q}(\rho_E), \phi \in \mathcal{S}(\rho_E)$ and $B_1, B_2 \dots B_n$ are linear environment superoperators. As a simple example of handling such expressions, consider the case where $\vecket{\phi} = \mathcal{U}_{\mathcal{B}}(t, s)\vecket{\rho_E}$, $\vecbra{\phi} = \vecbra{I_E}$ and we are interested in analyzing $\vecbra{\theta} B \vecket{\phi}$. To apply the Wick's theorem to this expression, we proceed by a first-order Trotterization of $\mathcal{U}_B(t, s)$ followed by the application of the Wick's theorem:
\begin{align}\label{eq:SM_wick_contraction_unitary}
&\vecbra{\theta} B \vecket{\phi} = \vecbra{I_E} B \mathcal{U}_\mathcal{B}(t, s)  \vecket{\rho_E} \nonumber\\
&\numeq{1} \lim_{\varepsilon \to 0} \vecbra{I_E}B \prod_{j = T_\varepsilon - 1}^0\mathcal{U}_B(s + (j + 1)\varepsilon, s + j\varepsilon ) \vecket{\rho_E}, \nonumber \\
&\numeq{2} \lim_{\varepsilon \to 0}\vecbra{I_E}B\prod_{j=T_\varepsilon - 1}^0\left(\text{id}-\int_{s+j\varepsilon}^{s+(j+1)\varepsilon}i\mathcal{H}_{\mathcal{B}}(s')ds'\right)\vecket{\rho_E},\nonumber\\
&\numeq{3} \lim_{\varepsilon \to 0} \Theta_{\mathcal{B}}(\alpha)(-1)^\sigma\sum_{j=0}^{T_\varepsilon - 1} \prod_{j'=j+1}^{T_\varepsilon - 1}\vecbra{I_E}\left(\text{id}-\int_{s+j'\varepsilon}^{s+(j'+1)\varepsilon}i\mathcal{H}_{\mathcal{B}}(s')ds'\right)\bigg(-i\int_{s+j\varepsilon}^{s+(j+1)\varepsilon}R_{\alpha, \sigma}(s')\vecbra{I_E}B B_{\alpha, \sigma}(s')\vecket{\rho_E}ds'\bigg)\times\nonumber\\
&\quad\quad\quad\quad\quad\quad\quad\quad\quad\quad\quad\quad\quad\quad\quad\quad\quad\prod_{j'=0}^{j-1}\left(\text{id}-\int_{s+j'\varepsilon}^{s+(j'+1)\varepsilon}i\mathcal{H}_{\mathcal{B}}(s')ds'\right)\vecket{\rho_E},\nonumber\\
&\numeq{4} -i\Theta_{\mathcal{B}}(\alpha)(-1)^\sigma\int_{s}^{t}\vecbra{I_E}\mathcal{U}_{\mathcal{B}}(t,\tau)R_{\alpha, \sigma}(\tau)\mathcal{U}_{\mathcal{B}}(\tau,s)\vecket{\rho_E}\vecbra{I_E}BB_{\alpha, \sigma}(\tau)\vecket{\rho_E}d\tau.
\end{align}
where in (1), we defined $T_\varepsilon=(t-s)/\varepsilon$, in (2) we have made the approximation $\mathcal{U}_\mathcal{B}(\tau+  \varepsilon, \tau) \approx \text{id} - i\int_{\tau}^{\tau + \varepsilon}\mathcal{H}_\mathcal{B}(s') ds'$, in (3) we have substituted $\mathcal{H}_{\mathcal{B}}(s') = \Theta_{\mathcal{B}}(\alpha) (-1)^\sigma R_{\alpha, \sigma}(s') B_{\alpha, \sigma}(s')$ (note the use of the summation convention introduced in section \ref{sec:supp_prelim_setup}) and then used the Wick's theorem to contract $B$ with every possible $B_{\alpha, \sigma}(s')$ in this expansion, and in (4), we have rewritten the product of infinitesimal evolution into a time-ordered one. It is convenient to rewrite the resulting expression as
\[
\vecbra{I_E} B \mathcal{U}_{\mathcal{B}}(t, s)\vecket{\rho_E}= \vecbra{I_E} \wickright(B; \mathcal{U}_{\mathcal{B}}(t, s)\vecket{\rho_E})
\]
where we define the ``right Wick contraction" of $B$ that is linear in the environment annihilation and creation superoperators with the vectorized state $\mathcal{U}_{\mathcal{B}}(t, s)\vecket{\rho_E}$
\[
\wickright(B; \mathcal{U}_{\mathcal{B}}(t, s)\vecket{\rho_E})) = -i \Theta_{\mathcal{B}}(\alpha)(-1)^\sigma\int_{s}^t \vecbra{I_E}BB_{\alpha, \sigma}(\tau)\vecket{\rho_E} \mathcal{U}_{\mathcal{B}}(t, \tau) R_{\alpha, \sigma}(\tau) \mathcal{U}_{\mathcal{B}}(\tau, s) \vecket{\rho_E}.
\]
The vectorized operator $\wickright(B; \mathcal{U}_{\mathcal{B}}(t, s)\vecket{\rho_E})$ can be interpreted as the operator obtained on wick contracting the superoperator $B$ with every environment super-operator in an expansion of $\mathcal{U}_{\mathcal{B}}(t, s)$. We now define the ``right Wick contraction" $\wickright$ more generally: Suppose $\vecket{\phi} \in \mathcal{S}({\rho_E})$ with representation $(\{\mathcal{A}_i\}_{i \in [1:n]}, \{(s_i, t_i)\}_{[1:n]}, \{\Omega_i\}_{[1:n]}, \sigma_S)$ and $B$ is a linear environment superoperator, then
\begin{subequations}\label{eq:right_W_1}
\begin{align}
    \overrightarrow{\textnormal{W}}(B; \vecket{\phi}) = -i (-1)^{\sigma} \Theta_{\mathcal{A}_j}(\alpha) \int_{s_{j }}^{t_j} \vecbra{I_E} B B_{\alpha,\sigma}(\tau)\vecket{\rho_E} \vecket{\phi_{j, \alpha, \sigma}(\tau)} d\tau,
    \end{align}
    where
    $\vecket{\phi_{j, \alpha, \sigma}(\tau)} \in \mathcal{S}(\rho_E)$ is given by
    \begin{align}
    \vecket{\phi_{j, \alpha, \sigma}(\tau)} =  \bigg(\prod_{i = n}^{j + 1} \Omega_i \mathcal{U}_{\mathcal{A}_i}(t_i, s_{i })\bigg) \Omega_{j}\mathcal{U}_{\mathcal{A}_j}(t_j, \tau) R_{\alpha, \sigma}(\tau) \mathcal{U}_{\mathcal{A}_j}(\tau, s_{j })\bigg(\prod_{i = j - 1}^{1} \Omega_i  \mathcal{U}_{\mathcal{A}_i}(t_i, s_i)\bigg)\vecket{\sigma_S, \rho_E}.
\end{align}
\end{subequations}
    Then, we define $\overrightarrow{\textnormal{W}}(\{B_j\}_{j \in [1:m]}; \vecket{\phi}) $ recursively via
    \begin{align}\label{eq:right_W}
    &\overrightarrow{\textnormal{W}}(\{B_k\}_{k \in [1:m]}; \vecket{\phi}) = -i (-1)^{\sigma} \Theta_{\mathcal{A}_j}(\alpha) \int_{s_{j}}^{t_j} \vecbra{I_E} B_m B_{\alpha,\sigma}(\tau) \vecket{\rho_E} {\overrightarrow{\text{W}}\big(\{B_k\}_{k \in [1:m - 1]} ; \vecket{\phi_{j, \alpha, \sigma}^{\nu}(\tau)}\big)} d\tau.
    \end{align}

Similarly, just as we defined the right Wick contraction in Eqs.~\eqref{eq:right_W_1} and \eqref{eq:right_W}, we can also define a ``left Wick contraction" $\wickleft$ that Wick contracts a linear environment super-operator with a vectorized operator in the dual space $\mathcal{Q}(\rho_E)$. Suppose $\theta \in \mathcal{Q}({\rho_E})$ with representation $(\{\mathcal{A}_i\}_{i \in [1:n]}, \{(s_i, t_i)\}_{[1:n]}, \{\Omega_i\}_{[1:n]}, O_S)$ and $B$ is a superoperator that is a linear combination of annihilation and creation operators, then we define 
    \begin{align}\label{eq:left_W_1}
    \wickleft(\vecbra{\theta}; B) = -i(-1)^{\sigma}  \Theta_{\mathcal{A}_j}(\alpha) \int_{s_{j}}^{t_j} \vecbra{I_E}B_{\alpha,\sigma}(\tau)B\vecket{\rho_E} \vecbra{\theta_{j, \alpha, \sigma}(\tau)} d\tau,
    \end{align}
    where 
    $\vecbra{\theta_{j, \alpha, \sigma}(\tau)} \in \mathcal{Q}(\rho_E)$ are given by
    \[
    \vecbra{\theta_{j, \alpha, \sigma}(\tau)} =  \vecbra{O_S, I_E}\bigg(\prod_{i = n}^{j + 1} \Omega_i \mathcal{U}_{\mathcal{A}_i}(t_i, s_{i})\bigg) \Omega_{j}\mathcal{U}_{\mathcal{A}_j}(t_j, \tau) R_{j, \sigma}(\tau) \mathcal{U}_{\mathcal{A}_j}(\tau, s_{j})\bigg(\prod_{i = j - 1}^{1} \Omega_i  \mathcal{U}_{\mathcal{A}_i}(t_i, s_{i})\bigg).
    \]
    Then, we define $\wickleft(\vecbra{\theta}; \{B_j\}_{j \in [1:m]}) $ recursively via
    \begin{align}\label{eq:left_W}
    &\wickleft\big(\vecbra{\theta}; \{B_k\}_{k \in [1:m]}\big) = -i  \Theta_{\mathcal{A}_j}(\alpha) \int_{s_{j}}^{t_j}  \vecbra{I_E}B_{\alpha,\sigma}(\tau) B_m \vecket{\rho_E} \wickleft\big(\vecbra{\theta_{j,\alpha, \sigma}(\tau)}; \{B_k\}_{k \in [1:m - 1]}\big) d\tau.
    \end{align}
With the right and left Wick contraction operations at hand, we can then state the Wick's theorem as follows, which is a direct generalization of Eq. (\ref{eq:SM_wick_contraction_unitary}):
\begin{lemma}[Wick's theorem]\label{lemma:wick_thm_gaussian}
    Suppose $\vecket{\phi} \in \mathcal{S}(\rho_E)$, $\vecbra{\theta} \in \mathcal{Q}(\rho_E)$ and $B_1, B_2 \dots B_n$ are superoperators that are linear combination of annihilation and creation operators, then
    \begin{equation}
        \vecbra{\theta}B_1 B_2 \dots B_n \vecket{\phi} = \sum_{\substack{\mathcal{I}_l, \mathcal{I}_r, \mathcal{I} \\ \mathcal{I}_l \cup \mathcal{I}_r\cup \mathcal{I} = [1:n] \\ 
        \abs{\mathcal{I}} \text{ is even}}} \sum_{p \in \mathcal{P}_{\mathcal{I}}} \bigg(\prod_{(i, j) \in p} \vecbra{I_E}B_i B_j \vecket{\rho_E}\bigg) \wickleft\big(\vecbra{\theta}; \{B_j\}_{j\in \mathcal{I}_l}\big)\wickright\big(\{B_j\}_{j\in \mathcal{I}_r}; \vecket{\phi}\big),
    \end{equation}
    where $\mathcal{P}_{\mathcal{I}}$ is the set of all pairings of the elements of $\mathcal{I}$, and each pairing $p \in \mathcal{P}_I$ is expressed as a list of tuples $\{(i_1, j_1), (i_2, j_2) \dots \}$ where we assume $i_n < j_n$.
\end{lemma}
\noindent The following examples illustrate the application of lemma \ref{lemma:wick_thm_gaussian} and helps unpack the notation in the lemma:
\begin{align*}
    &\vecbra{\theta}B \vecket{\phi} = \wickleft(\vecbra{\theta}; B) \vecket{\phi} + \vecbra{\theta} \wickright(B; \vecket{\phi}), \\
    &\vecbra{\theta}B_1 B_2 \vecket{\phi} = \wickleft\big(\vecbra{\theta}; \{B_1, B_2\}\big)\vecket{\phi} + \vecbra{\theta} \wickright\big(\{B_1, B_2\}; \vecket{\phi}\big) + \vecbra{I_E}B_1 B_2 \vecket{\rho_E} \vecbra{\theta}\phi\rrangle + \nonumber\\
    &\qquad \qquad \qquad \qquad \qquad \qquad  \wickleft(\vecbra{\theta}; B_1) \wickright(B_2; \vecket{\phi}) + \wickleft(\vecbra{\theta}; B_2) \wickright(B_1; \vecket{\phi}).
\end{align*}

\subsection{Analyzing the mollification}\label{sec:regularization}
\begingroup\color{black}
Throughout this section, we will assume that the kernels ${K}_{\alpha, \alpha'}$ {have} locally bounded total variation i.e., within any finite interval, their total variation is finite. More, specifically, we make the following assumption.

\begin{assumption}[Local integrability of the memory kernels]\label{assump:local_int}
$\forall s, t \in \mathbb{R}$ with $s\leq t$ and $\alpha, \alpha',\sigma$
\[
\textnormal{TV}(K_{\alpha, \alpha',\sigma}; [s, t]) < \infty.
\]
\end{assumption}

\noindent\emph{Remark}. We remark that  we do not assume that $K_{\alpha, \alpha', \sigma}$ has a bounded total variation. Examples of kernels which are locally bounded but not bounded would include kernels of the form $1/(\abs{t}+ T_0)^\eta$ where $T_0 > 0$ and $\eta < 1$.
\endgroup

 To define the regularized model, we will require a function $\eta:\mathbb{R} \to [0, \infty)$ that is smooth, compactly supported in $[-1, 1]$ and is normalized such that $\norm{\eta}_{1} = 1$. An explicit example of such a function would be the standard mollifier
\begin{subequations}\label{eq:standard_mollifiers}
\begin{align}
\eta(x) = \begin{cases}
    A_0 e^{-1/(1 - \abs{x}^2)} & \text{ if } |x| < 1, \\
    0 & \text{ otherwise},
\end{cases}
\end{align}
where $A_0$ is a constant chosen to satisfy $\norm{\eta}_{1} = 1$. We can then define a family of mollifiers $\eta_\delta$ dependent on a (small) regularizing parameter $\delta > 0$:
\begin{align}
\eta_\delta(x) = \frac{1}{\delta}\eta\bigg(\frac{x}{\delta}\bigg),
\end{align}
\end{subequations}
which, like $\eta(x)$, is again a smooth, non-negative function compactly supported in $[-\delta, \delta]$ and with $\norm{\eta_\delta}_{1} = 1$. Given $\delta > 0$, we first define a ``mollified" system-environment Hamiltonian corresponding to Eq.~\ref{eq:sys_env_hamiltonian}, ${H}^\delta(t)$, via
\begin{subequations}\label{eq:mollified_sys_env_hamiltonian}
\begin{align}
{H}^\delta(t) = \sum_{\alpha} h_\alpha(t) + \sum_{\alpha} B_{\alpha}^{\delta}(t) R_{\alpha}(t),
\end{align}
where
\begin{align}\label{eq:mollified_bath_operators}
B_{\alpha}^{\delta}(t) = \int_{-\infty}^\infty \eta_\delta(t - s) B_{\alpha}(s) ds.
\end{align}
\end{subequations}
Furthermore, this definition of $B_{\alpha}^{\delta}(t)$ defines a set of kernels $\K^{\delta, \delta'}_{\alpha, \alpha', \sigma'}$, paralleling Eq.~\ref{eq:kernel_def}. More specifically, we define $\K^{ \delta, \delta'}_{\alpha, \alpha', \sigma'}$ via
\begin{subequations}
\begin{align}
    \K^{ \delta, \delta'}_{\alpha,  \alpha', \sigma'}(t - t') = \vecbra{I_E} B_{\alpha,\sigma}^{\delta}(t) B_{\alpha',\sigma'}^{\delta'}(t')\vecket{\rho_E}.
\end{align}
For notational convenience, we will often use $K_{\alpha, \alpha', \sigma'}^{\delta} := K_{\alpha, \alpha', \sigma'}^{\delta, \delta}$. Using Eq.~\eqref{eq:mollified_bath_operators}, we obtain that
\begin{align}
\K^{\delta, \delta'}_{\alpha, \alpha', \sigma'}(t - t') = \int_{-\infty}^{\infty}\int_{-\infty}^{\infty}  \eta_{\delta}(t - s) \K_{\alpha, \alpha', \sigma'}(s - s') \eta_{\delta'}(t' - s') ds' ds =  \big(\eta_{\delta, \delta'}\star \K_{\alpha, \alpha', \sigma'}\big)(t - t'),
\end{align}
with
\begin{align}\label{eq:convolved_mollifier}
    \eta_{\delta, \delta'}(\tau) = (\eta_\delta \star \eta_{\delta'})(\tau) = \int_{-\infty}^\infty \eta_\delta(\tau + s')\eta_{\delta'}(s') ds'.
\end{align}
\end{subequations}
It can be noted that $\eta_{\delta, \delta'}$ is a smooth, non-negative function that is compactly supported in the interval $[-(\delta + \delta'), \delta + \delta']$ and $\norm{\eta_{\delta, \delta'}}_1 = 1$. The kernel $\K^{\delta, \delta'}_{\alpha, \sigma'}$ can therefore be considered to be the mollification of the kernel $\K_{\alpha, \sigma'}$ with mollifier $\eta_{\delta, \delta'}$. The following lemma provides a useful property of the mollified kernel $\K^{\delta, \delta'}_{\alpha, \sigma'}$.
\begin{lemma}\label{lemma:mollified_kernel}
    The kernel $\K^{\delta, \delta'}_{\alpha, \alpha', \sigma'}$ is a smooth and locally bounded function, and satisfies
    \begin{enumerate}
        \item[(a)] $\forall \delta, \delta' > 0$,
        \begin{equation}
        |\K^{\delta, \delta'}_{\alpha, \alpha', \sigma'}(\tau)| \leq \min\bigg(\frac{1}{\delta}, \frac{1}{\delta'}\bigg)\norm{\eta}_\infty  \textnormal{TV}(\text{K}_{\alpha, \alpha', \sigma'}; [\tau - (\delta + \delta'), \tau + (\delta + \delta')]),
        \end{equation}
        \item[(b)] $\forall \delta, \delta' > 0$,
        \begin{equation}
        \textnormal{TV}\big(\K^{\delta, \delta'}_{\alpha, \alpha'; \sigma'}; [s, t]\big)\leq \textnormal{TV}(\text{K}_{\alpha, \alpha', \sigma'}; [s - (\delta + \delta'), t + (\delta + \delta')]). \label{eq:total_variation_mollified}
        \end{equation}
    \end{enumerate}
\end{lemma}
\begin{proof}
 If $\K_{\alpha, \alpha', \sigma'}$ has a representation given by 
 \begin{align*}
 \K_{\alpha, \alpha'; \sigma'}(\tau) = \K_{\alpha, \alpha', \sigma'; c}(\tau) + \sum_{j} k_{\alpha, \alpha', \sigma'; j}\delta(\tau - T_j),
 \end{align*}
then 
\begin{align*}
    \K^{\delta,\delta'}_{\alpha, \alpha', \sigma'}(\tau) &= \int_{-\infty}^{\infty} \K_{\alpha, \alpha', \sigma'; c}(\tau') \eta_{\delta, \delta'}(\tau - \tau') d\tau' + \sum_{j}k_{\alpha, \alpha', \sigma'; j}\eta_{\delta, \delta'}(\tau - T_j), \nonumber \\
    &=\int_{\tau-(\delta+\delta')}^{\tau +(\delta + \delta')}\K_{\alpha, \alpha', \sigma';c}(\tau') \eta_{\delta, \delta'}(\tau -\tau') d\tau' + \sum_{j: T_j \in[\tau - (\delta + \delta'), \tau + (\delta + \delta')]}\smallabs{k_{\alpha, \alpha', \sigma'; j}}\eta_{\delta,\delta'}(\tau-T_j),
\end{align*}
where we have used the fact that $\eta_{\delta, \delta'}$ is supported in $[-(\delta + \delta'), \delta + \delta']$.
    It is immediately clear from this expression and the smoothness of $\eta_{\delta, \delta'}$ that $\K^{\delta, \delta'}_{\alpha, \sigma'}$ is a  smooth function. Next, to upper bound $\smallabs{\K^{\delta,\delta'}_{\alpha, \sigma'}(\tau)}$, note that
    \begin{align}\label{eq:upper_bound_inf_norm_manip}
        \smallabs{\K^{ \delta,\delta'}_{\alpha, \alpha', \sigma'}(\tau)} &\leq \int_{\tau - (\delta + \delta')}^{\tau + (\delta + \delta')} \smallabs{\K_{\alpha, \alpha', \sigma', c}(\tau')}{\eta_{\delta, \delta'}(\tau - \tau')}d\tau' + \sum_{j} \smallabs{k_{\alpha, \alpha', \sigma'; j}}{\eta_{\delta, \delta'}(\tau - T_j)}, \nonumber\\
        &\leq \norm{\eta_{\delta, \delta'}}_\infty \bigg(\int_{\tau - (\delta + \delta')}^{\tau + (\delta + \delta')}\smallabs{\K_{\alpha, \alpha', \sigma', c}(\tau')}d\tau' + \sum_{j: T_j \in [\tau - (\delta + \delta'), \tau + (\delta + \delta')]} \smallabs{k_{\alpha, \alpha', \sigma'; j}}\bigg) ,\nonumber \\
        &\leq \norm{\eta_{\delta, \delta'}}_\infty \text{TV}(\text{K}_{\alpha, \alpha', \sigma'}; [\tau - (\delta + \delta'), \tau + (\delta + \delta')]), 
    \end{align}
    where, the reader can note that $u$ is the function from Assumption \ref{assump:local_int}.
    Next, we note that since $\eta_{\delta, \delta'} = \eta_{\delta}\star \eta_{\delta'}$, $\norm{\eta_{\delta,\delta'}}_\infty \leq \norm{\eta_\delta}_1 \norm{\eta_{\delta'}}_\infty \leq \norm{\eta}_\infty / \delta'$. A similar bound holds on switching $\delta$ and $\delta'$ and consequently, we obtain $\norm{\eta_{\delta,\delta'}}_\infty \leq \norm{\eta}_\infty \min(1/\delta, 1/\delta')$. From this bound together with Eq.~\ref{eq:upper_bound_inf_norm_manip}, we obtain the bound on $\smallnorm{\K^{ \delta, \delta'}_{\alpha, \alpha', \sigma'}}_\infty$ in the lemma statement.

    Proceeding similarly, we note that
    \begin{align}
        &\int_s^t \smallabs{\K^{ \delta,\delta'}_{\alpha, \alpha', \sigma'}(\tau)}d\tau =\int_{\tau = s}^t \int_{\tau' = -\infty}^\infty \smallabs{\K_{\alpha, \alpha', \sigma'; c}(\tau')} \eta_{\delta, \delta'}(\tau - \tau') d\tau' d\tau + \sum_j \smallabs{k_{\alpha, \alpha', \sigma'; j}}\int_{\tau = s}^t  \eta_{\delta, \delta'}(\tau - T_j) d\tau, \nonumber\\
        &\qquad=\int_{\tau = s}^t \int_{\tau'=s-(\delta+\delta')}^{t + (\delta+ \delta')} \smallabs{\K_{\alpha, \alpha', \sigma'; c}(\tau')} \eta_{\delta, \delta'}(\tau - \tau') d\tau' d\tau + \sum_{j: T_j \in [s - (\delta + \delta'), t + (\delta + \delta')]} \smallabs{k_{\alpha, \alpha', \sigma'; j}}\int_{\tau = s}^t  \eta_{\delta, \delta'}(\tau - T_j) d\tau, \nonumber \\
        &\qquad\leq \int_{\tau' = s - (\delta + \delta')}^{t + (\delta + \delta')}\smallabs{\K_{\alpha, \alpha', \sigma'; c}(\tau')}d\tau' + \sum_{j: T_j \in [s - (\delta + \delta'), t + (\delta + \delta')]}\smallabs{k_{\alpha, \alpha', \sigma'; j}}, \nonumber\\
        &\qquad= \textnormal{TV}\big(\K_{\alpha, \alpha', \sigma'}; [s- (\delta + \delta'), t + (\delta + \delta')]\big),
    \end{align}
    which establishes the bound on $\text{TV}(\K^{\delta, \delta'}_{\alpha, \alpha', \sigma'}; [s, t])$ in the lemma statement.
\end{proof}

Next, we expect the mollified kernel $\K^{\delta, \delta'}_{\alpha, \alpha', \sigma'}$ to approximate the unmollified kernel $\K_{\alpha, \alpha', \sigma'}$ as a distribution. To make this expectation precise and to define a quantitative distance between two distributions, we would need to introduce a space of `test' functions. For our purposes, as will become clear in our analysis below, the appropriate function space is the space of piecewise continous and differentiable functions, which we define below.

\begin{definition}[Function space $\text{PC}^{1}_{t^*}(\mathbb{R})$]
    Given $t^* = \{t_1^*, t_2^* \dots t_n^*\}$, where $t_1^* < t_2^* \dots < t_n^*$, any $f:\mathbb{R} \to \mathbb{C} \in \text{PC}_{t^*}^{1}(\mathbb{R})$ has a representation
\[
f(t) = \begin{cases}
    0 & \text{for }t \in (-\infty, t_1^*), \\
    f_1(t) & \text{for }t \in (t_1^*, t_2^*), \\
    f_2(t) & \text{for }t \in (t_2^*, t_3^*), \\
    \vdots \\
    f_{n - 1}(t) & \text{for }t \in (t_{n  - 1}^*, t_n^*), \\
    0 & \text{for }t \in (t_n^*, \infty),
\end{cases}
\]
for some $f_1, f_2 \dots f_{n - 1} \in \textnormal{C}^1(\mathbb{R})$ and at the points of discontinuity, the function $f$ can assume any arbitrary value.
\end{definition}
\noindent Consider a kernel $\K \in \mathcal{K}_1(\mathbb{R})$ and a test function $f \in \text{PC}_{t^*}(\mathbb{R})$---due to the discontinuities in $f$ and delta functions in $\K$, there is some ambiguity in how to precisely apply $\K$ to $f$. More specifically, we need to choose a convention on how to integrate a delta function across a discontinuity. We make this convention concrete in the following definition.

\begin{definition}[Applying $\K \in \mathcal{K}_1(\mathbb{R})$ to $f \in \textnormal{PC}^{1}_{t^*}(\mathbb{R})$] 
    Suppose $\K$ has the decomposition $\K = \K_c + \sum_{j=1}^M k_j \delta(t - T_j)$, then the action of $\K$ on $f$, which we will denote by $\int_{-\infty}^\infty \K(\tau) f(\tau) d\tau$, is defined as
\begin{align*}
 \int_{-\infty}^\infty \K(\tau) f(\tau) d\tau = \int_{-\infty}^\infty \K_c(\tau) f(\tau) d\tau + \sum_{j = 1}^M \frac{k_j}{2}\big(f(T_j^-) + f(T_j^+)\big)
\end{align*}
\end{definition}
\noindent Importantly, we assume that, at a point of discontinuity of a test function, the action of a $\delta$-function is equivalent to averaging its left limit and right limit. While this choice is arbitrary, it is consistent with our goal of considering $\delta-$functions in a kernel as the limit of \emph{symmetric} mollifiers and is also the convention that is commonly used to deal with such kernels in practice in the open quantum systems literature. 

We next introduce a notion of how close two kernels $\K, \K' \in \mathcal{K}_1(\mathbb{R})$ are, or alternatively, a notion of how small the kernel $\K - \K' \in \mathcal{K}_1(\mathbb{R})$ is.
\begin{definition}\label{def:kernel_distance}
Given $t^* = \{t_1^*, t_2^* \dots t_M^*\}$, a kernel $ \K \in \mathcal{K}_1(\mathbb{R})$ is $(\lambda_0, \lambda_1)-$small on the test function space $\textnormal{PC}_{t^*}^{1}(\mathbb{R})$ if $\forall f \in \textnormal{PC}_{t^*}^{1}(\mathbb{R})$ with $\norm{f}_\infty, \norm{f'}_\infty < \infty$,
\[
\abs{\int_{-\infty}^{\infty}  \K(t) f(t) dt} \leq \lambda_0 \norm{f}_\infty + \lambda_1 \norm{f'}_\infty.
\]
Two kernels $\K, \K' \in \mathcal{K}_1(\mathbb{R})$ are $(\lambda_0, \lambda_1)$-close on $\textnormal{PC}_{t^*}^{1}(\mathbb{R})$ if the kernel $\K - \K'$ is $(\lambda_0, \lambda_1)-$small on $\textnormal{PC}_{t^*}^{1}(\mathbb{R})$.
\end{definition}
\noindent Next, we establish that the kernels $\K_{\alpha, \alpha', \sigma'}$ and $\K_{\alpha, \alpha',  \sigma'}^{\delta, \delta'}$ are close in the sense of definition \ref{def:kernel_distance}.
\begin{lemma}\label{lemma:kernel_appx}
    Suppose $\K \in \mathcal{K}_1(\mathbb{R})$ is a kernel with continuous part $\K_c$ and $\delta$-functions at $\{T_1, T_2 \dots\}$ and suppose $r_\delta:\mathbb{R} \to \mathbb{R}$ is a smooth, non-negative even function that is compactly supported on $[-\delta, \delta]$ and satisfies $\norm{r_\delta}_1 = 1$. Consider the test function space $\textnormal{PC}_{t^*}^1(\mathbb{R})$ and suppose that $\delta < \Delta \tau_\textnormal{min} / 2$, where 
    \[
    \Delta \tau_\textnormal{min} = \min_{t, s \in \mathbb{T}, t\neq s} \abs{t - s} \text{ where }\mathbb{T} = \{T_1, T_2 \dots\} \cup \{t_1^*, t_2^* \dots \},
    \]
    i.e., $\Delta \tau_\textnormal{min}$ is the minimum non-zero spacing in the set $\mathbb{T}$ formed by the locations of the $\delta-$functions in $\K$ and the points of discontinuity of $\textnormal{PC}_{t^*}^1(\mathbb{R})$. Then $\K$ and $\K \star r_\delta$ are $(\lambda_0(\delta), \lambda_1(\delta))-$close, where
    \[
    \lambda_0(\delta) =  2\sum_{t \in t^*}\textnormal{TV}\big(\K_c; [t - \delta, t+\delta]\big) \text{ and }\lambda_1(\delta) = \delta\ \textnormal{TV}(\K; [t_1^* - \delta, t_n^* + \delta]).
    \]
\end{lemma}
\begin{proof}
For notational convenience, given $\xi > 0$, we will define the set $\mathbb{T}_\xi \subseteq \mathbb{R}$ via
\begin{align}
    \mathbb{T}_\xi(t^*) = [t_1^* - \xi, t_1^* + \xi] \cup [t_2^* - \xi, t_2^* + \xi] \cup \dots \cup [t_n^* - \xi, t_n^* + \xi].
\end{align}
    Suppose $\K$ has a representation
    \[
    \K(\tau) = \K_c(\tau) + \sum_{j = 1}^M k_j\delta(\tau - T_j),
    \]
    where $\K_c$ is a continuous function. We will also use the notation $\K^\delta = \K \star r_\delta$, which can then be expressed as
    \[
    \K^\delta(\tau) = \K_c^\delta(\tau)+ \sum_{j = 1}^M k_jr_\delta(\tau - T_j).
    \]
    For $f \in \textnormal{PC}_{t^*}^1(\mathbb{R})$, we obtain that
    \begin{align}\label{eq:error_decomp}
        &\abs{\int_{-\infty}^\infty \K(\tau)f(\tau) d\tau - \int_{-\infty}^\infty \K^\delta(\tau)f(\tau) d\tau}, \nonumber\\
        &\qquad \qquad \leq \underbrace{\abs{\int_{-\infty}^\infty \K_c(\tau)f(\tau) d\tau - \int_{-\infty}^\infty \K^\delta_c(\tau)f(\tau) d\tau}}_{\Delta_c[f]: \text{Error due to the continuous part}} + \underbrace{\sum_{j = 1}^n \abs{k_j} \abs{\frac{1}{2}(f(T_j^+) + f(T_j^-)) - (f\star r_\delta)(T_j)}}_{\Delta_d[f]: \text{Error due to the atomic part}}, \nonumber \\
    \end{align}
    where $\K_{c}^{\delta} = \K_{c} \star r_\delta$. We can separately analyze the errors due to the continuous part contribution, $\Delta_c[f]$, and the $\delta-$function contribution, $\Delta_d[f]$.

    \textbf{Error from the continuous part:} We begin by rewriting the expression for $\Delta_c[f]$ as
    \begin{align}\label{eq:cont_error_rewrite}
        \Delta_c[f] = \abs{\int_{-\infty}^\infty \K_c(\tau){f}(\tau) d\tau - \int_{-\infty}^\infty \K_c^\delta(\tau)f(\tau) d\tau} \leq\int_{t_1^* - \delta}^{t_n^* +\delta} \abs{\K_c(\tau)}\abs{{f}(\tau) - (f\star r_\delta)(\tau)}d\tau,
    \end{align}
    where we have used that $f \in \text{PC}_{t^*}^1(\mathbb{R})$ is, by definition, supported in $[t_1^*, t_n^*]$ and hence $f \star r_\delta$ is supported in $[t_1^* - \delta, t_n^*+\delta]$. This upper bound on $\Delta_c[f]$ can be further upper bounded by analyzing the point-wise error $f(\tau) - (f\star r_\delta)(\tau)$ between $f$ and its mollification $f\star r_\delta$. Consider first analyzing this error at $\tau$ which are $\delta$-away from all points of discontinuity of $f$ i.e.~suppose $\tau \notin \mathbb{T}_\delta(t^*)$ where $\mathbb{T}_\delta(t^*) = [t^*_1 - \delta, t^*_1 + \delta] \cup [t^*_2 - \delta, t^*_2 + \delta] \dots \cup [t^*_n - \delta, t^*_n + \delta]$, then
    \begin{align}\label{eq:mollification_pwise_error}
    \abs{f(\tau) - (f \star r_\delta)(\tau)} \numeq{1} \abs{\int_{-\delta}^{\delta}\big(f(\tau) - f(\tau + \tau')\big)r_\delta(\tau') d\tau'} \numleq{2} \norm{f'}_\infty \int_{-\delta}^\delta \abs{\tau'}r_\delta(\tau') d\tau' \leq \norm{f'}_\infty \delta,
    \end{align}
    where, in (1), we have used the fact that $\int_{-\delta}^\delta r_\delta(\tau') d\tau' = 1$ to rewrite $f(\tau) = \int_{-\delta}^\delta f(\tau) r_\delta(\tau') d\tau'$. In (2), we have used the fact that for $\tau \in \mathbb{T}_\delta(t^*)$ and $\tau' \in [-\delta, \delta]$, $f(\tau + \tau')$ has no points of non-differentiability and thus $\abs{f(\tau) - f(\tau + \tau')} \leq \norm{f'}_\infty \tau'$. Then, starting from Eq.~\ref{eq:cont_error_rewrite}, we obtain that
    \begin{align}\label{eq:final_error_cont}
    \Delta_c[f] &\leq \delta \norm{f'}_\infty  \int_{[t_1^* - \delta, t_n^* + \delta]\setminus\mathbb{T}_\delta(t^*)} \abs{\K_c(\tau)} d\tau + \int_{\mathbb{T}_\delta(t^*)} \abs{\K_c(\tau)} \abs{f(\tau) - (f\star r_\delta)(\tau)}d\tau, \nonumber \\
    &\leq \delta \norm{f'}_\infty \text{TV}(\K_c; [t_1^* - \delta, t_n^* + \delta]) + \int_{\mathbb{T}_\delta(t^*)} \abs{\K_c(\tau)} \big(\abs{f(\tau)} + \abs{(f\star r_\delta)(\tau)}\big)d\tau, \nonumber \\
    &\leq \delta \norm{f'}_\infty \text{TV}(\K_c; [t_1^* - \delta, t_n^* + \delta]) + 2\norm{f}_\infty \text{TV}(\K_c; \mathbb{T}_\delta(t^*)).
    \end{align}

    \textbf{Error from the atomic part:} Next, we consider the error $\Delta_d[f] = \sum_{j = 1}^n \abs{k_j}\Delta_{d, j}[f]$ where
    \[
    \Delta_{d, j}[f] = \abs{\frac{1}{2}\big(f(T_j^+) + f(T_j^-)\big) - (f\star r_\delta)(T_j)}.
    \]
    We note that since $f$ is supported in the interval $[t_1^*, t_n^*]$, and thus $f \star r_\delta$ is supported on $[t_1^* - \delta, t_n^* +\delta]$, it follows that
    \[
    \Delta_{d, j}[f] = 0 \text{ if } T_j \notin [t_1^* - \delta, t_n^* + \delta].
    \]
    Furthermore, since we assume that $\forall t \in \{T_1, T_2 \dots \}, s\in \{t_1^*, t_2^* \dots \}$ with $ t\neq s: \abs{t - s} > 2\delta$, there are only two possibilities---either $T_j \in t^*$, or $T_j \notin \mathbb{T}_{2\delta}(t^*)$. Consider first the case where $T_j \notin \mathbb{T}_{2\delta}$---in this case, $f$ is continuous and differentiable in $[T_j - \delta, T_j + \delta]$ and thus
    \[
    \Delta_{d, j}[f] = \abs{f(T_j) - (f\star r_\delta)(T_j)} \leq \norm{f'}_\infty \delta,
    \]
where we have used Eq.~\ref{eq:mollification_pwise_error}. Next, consider the case if $T_j \in t^*$, then 
    \begin{align*}
        \Delta_{d, j}[f] &= \abs{\frac{1}{2}(f(T_j^+) + f(T_j^-)) - \int_{-\delta}^\delta r_\delta(\tau)f(T_j - \tau)d\tau}, \\ &\leq \abs{\frac{1}{2}f(T_j^+) - \int_{-\delta}^0 r_\delta(\tau) f(T_j - \tau)d\tau} + \abs{\frac{1}{2}f(T_j^-) - \int_{0}^\delta r_\delta(\tau) f(T_j - \tau)d\tau}, \nonumber \\
        &\leq \abs{\int_{-\delta}^0 r_\delta(\tau) (f(T_j^+) - f(T_j - \tau))d\tau} + \abs{ \int_{0}^\delta r_\delta(\tau) (f(T_j^-) - f(T_j - \tau))d\tau}, \nonumber \\
        &\numleq{1} \sum_{s \in \{-, +\}}\norm{f'}_\infty \int_{-\delta}^\delta \abs{\tau}\abs{r_\delta(\tau)}d\tau \leq \norm{f'}_\infty \delta,
    \end{align*}
    where, in (1), we have used the fact that $f$ is separately differentiable in both the intervals $(T_j, T_j + \delta)$ and $(T_j - \delta, T_j)$. Thus, in both the cases, $T_j \in t^*$ or $T_j \notin \mathbb{T}_{2\delta}(t^*)$, we obtain that $\Delta_{d, j}[f] \leq \norm{f'}_\infty \delta$. Therefore, we obtain that the error from the atomic part, $\Delta_d[f]$, can be upper bounded by
    \begin{align}\label{eq:final_error_atomic}
    \Delta_d[f] \leq \sum_{j = 1}^n\abs{k_j} \Delta_{d, j}[f]  \leq \norm{f'}_\infty \delta \sum_{j: T_j \in [t_1^* - \delta, t_n^* + \delta]} \abs{k_j}.
    \end{align}
Finally, combining Eqs.~\ref{eq:error_decomp}, \ref{eq:final_error_cont} and \ref{eq:final_error_atomic}, we obtain that
    \begin{align*}
    &\abs{\int_{a}^b \K(\tau) f(\tau) d\tau - \int_{a}^b \K^\delta(\tau)f(\tau) d\tau } \nonumber\\
    &\qquad \leq \delta \norm{f'}_\infty \bigg(\text{TV}(\K_c; [t_1^* - \delta, t_n^* + \delta]) + \sum_{j: T_j \in [t_1^* - \delta, t_n^* + \delta]} \abs{k_j}\bigg) + 2\norm{f}_\infty \text{TV}(\K_c, \mathbb{T}_\delta(t^*)),\nonumber \\
    &\qquad \leq \delta \norm{f'}_\infty \text{TV}(\K; [t_1^* - \delta, t_n^* + \delta]) + 2\norm{f}_\infty\sum_{t \in t^*} \text{TV}(\K_c; [t - \delta, t + \delta]).
    \end{align*}
    which proves the lemma statement.
\end{proof}

The previous lemma establishes that a kernel and its mollification are close to each other in a meaningful sense. Next, we use this fact to argue that the quantum dynamics of the un-mollified and mollified models are also close. As with the un-mollified model, given $\mathcal{B} \subseteq \mathcal{A}$, we will define
\[
{H}^\delta_\mathcal{B}(t) = \sum_{\alpha \in \mathcal{B}}h_\alpha(t) + \sum_{\alpha \in \mathcal{B}}  B_{\alpha}^{\delta}(t) R_{\alpha}(t) \text{ and }{U}^\delta_{\mathcal{B}}(t, s) = \mathcal{T}\exp\bigg(-i\int_{-\infty}^\infty {H}^\delta_\mathcal{B}(\tau) d\tau\bigg).
\]
The main result that we will establish is given by the following lemma.
\begin{proposition}\label{lemma:mollification}
    For any system operators $O_S, \sigma_S$, system superoperators $\{\Omega_i : \norm{\Omega_i}_\diamond \leq 1\}_{i \in [1:n]}$, lattice subregions $\{\mathcal{A}_i\}_{i \in [1:n]}$ as well as times $\{(s_i, t_i)\}_{i \in [1:n]}$, then
    \begin{align}\label{eq:system_GF}
    \lim_{\delta \to 0}\ \vecbra{I_E, O_S} \bigg(\prod_{i = 1}^n {\mathcal{U}}^\delta_{\mathcal{A}_i}(t_i, s_i) \Omega_i \bigg) \vecket{\rho_E, \sigma_S},
    \end{align}
    exists.
\end{proposition}
\begin{proof}[Proof sketch]
    {\color{black}The proof of this proposition is a straightforward but tedious analysis: Here, for clarity and readability, we will simply sketch out the basic ideas in the proof in the simple case of proving that
    \[
    \lim_{\delta \to 0}g(\delta) \text{ where } g(\delta) = \vecbra{I_E, O_S} \mathcal{U}_{\mathcal{A}}^\delta(t, s) \vecket{\rho_E, \sigma_S},
    \]
    exists and provide the full analysis in appendix \ref{appendix:prop_4}. To prove the existence of this limit, we will show that $g(\delta)$ for any set of decreasing $\delta$ forms a Cauchy sequence. In other words, we have to show that $\abs{g(\delta_1) - g(\delta_2)} \to 0$ as $\delta_1, \delta_2 \to 0$. We begin by writing
    \begin{align}
        \abs{g(\delta_1) - g(\delta_2)} &= \abs{\int_s^t\vecbra{I_E, O_S}  \mathcal{U}_{\mathcal{A}}^{\delta_1}(t, \tau)\Delta B_{\alpha, \sigma}(\tau)R_{\alpha, \sigma}(\tau)\mathcal{U}_{\mathcal{A}}^{\delta_2}(\tau, s)\vecket{\rho_E, \sigma_S}} \text{ where }\Delta B_{\alpha, \sigma}(\tau) = B^{\delta_1}_{\alpha, \sigma}(\tau) - B^{\delta_2}_{\alpha, \sigma}(\tau).
    \end{align}
    Next, we use the Wick's theorem to obtain
    \begin{align}\label{eq:g_delta_1_minus_delta_2}
        &\abs{g(\delta_1) - g(\delta_2)}  \leq \Theta_{\mathcal{A}}(\alpha)\Theta_{\mathcal{A}}(\alpha')\bigg(\nonumber\\
        &\qquad \abs{\int_s^t \int_\tau^t \vecbra{I_E}B_{\alpha', \sigma'}^{\delta_1}(\tau') \Delta B_{\alpha, \sigma}(\tau) \vecket{\rho_E}\vecbra{I_E, O_S}\mathcal{U}^{\delta_1}_{\mathcal{A}}(t, \tau')R_{\alpha', \sigma'}(\tau')\mathcal{U}^{\delta_1}_{\mathcal{A}}(\tau', \tau)R_{\alpha, \sigma}(\tau) \mathcal{U}_{\mathcal{A}}^{\delta_2}(\tau, s)\vecket{\rho_E, \sigma_S}d\tau' d\tau} + \nonumber\\
        &\qquad \abs{\int_s^t \int_s^\tau \vecbra{I_E} \Delta B_{\alpha, \sigma}(\tau) B_{\alpha', \sigma'}^{\delta_2}(\tau') \vecket{\rho_E}\vecbra{I_E, O_S}\mathcal{U}^{\delta_1}_{\mathcal{A}}(t, \tau)R_{\alpha, \sigma}(\tau)\mathcal{U}^{\delta_1}_{\mathcal{A}}(\tau, \tau')R_{\alpha', \sigma'}(\tau') \mathcal{U}_{\mathcal{A}}^{\delta_2}(\tau', s)\vecket{\rho_E, \sigma_S}d\tau' d\tau}\bigg).
    \end{align}
    Consider now the kernels $\vecbra{I_E} B_{\alpha', \sigma'}^{\delta_1}(\tau')\Delta B_{\alpha, \sigma}(\tau)\vecket{\rho_E}$ and $\vecbra{I_E}\Delta B_{\alpha, \sigma}(\tau) B^{\delta_2}_{\alpha', \sigma'}(\tau')\vecket{\rho_E}$ --- we can argue using lemma \ref{lemma:kernel_appx} that these are ``small". For instance, 
    \begin{align}\label{eq:kernel_diff_1}
        \vecbra{I_E} B_{\alpha', \sigma'}^{\delta_1}(\tau')\Delta B_{\alpha, \sigma}(\tau)\vecket{\rho_E} &= \vecbra{I_E} B^{\delta_1}_{\alpha', \sigma'}(\tau') B^{\delta_1}_{\alpha, \sigma}(\tau) \vecket{\rho_E} - \vecbra{I_E} B^{\delta_1}_{\alpha', \sigma'}(\tau') B^{\delta_2}_{\alpha, \sigma}(\tau) \vecket{\rho_E}, \nonumber \\
        & = K^{\delta_1, \delta_1}_{\alpha', \alpha, \sigma}(\tau' - \tau) - K^{\delta_1, \delta_2}_{\alpha', \alpha, \sigma}(\tau' - \tau),
    \end{align}
    and similarly
    \begin{align}\label{eq:kernel_diff_2}
        \vecbra{I_E} \Delta B_{\alpha, \sigma}(\tau) B_{\alpha', \sigma'}^{\delta_2}(\tau')\vecket{\rho_E} = K_{\alpha, \alpha', \sigma'}^{\delta_1, \delta_2}(\tau - \tau') - K_{\alpha, \alpha', \sigma'}^{\delta_2, \delta_2}(\tau - \tau').
    \end{align}
    Now, for a given a test-function space $\text{PC}_{t^*}^1(\mathbb{R})$, it follows from lemma \ref{lemma:kernel_appx}, it then follows that for sufficiently small $\delta_1, \delta_2$, both $K_{\alpha', \alpha, \sigma}^{\delta_1, \delta_1}$ and $K_{\alpha', \alpha, \sigma}^{\delta_1, \delta_2}$ are close to $K_{\alpha', \alpha, \sigma}$ when applied to a test function. Using the triangle inequality together with Eqs.~\eqref{eq:kernel_diff_1} and \eqref{eq:kernel_diff_2}, it can then be deduced that $\vecbra{I_E} B_{\alpha', \sigma'}^{\delta_1}(\tau')\Delta B_{\alpha, \sigma}(\tau)\vecket{\rho_E}$ and $\vecbra{I_E} \Delta B_{\alpha, \sigma}(\tau) B_{\alpha', \sigma'}^{\delta_2}(\tau')\vecket{\rho_E}$ are close to 0 when applied on this test function space. Since the terms in Eq.~\eqref{eq:g_delta_1_minus_delta_2} can be reformulated as applying the kernels $\vecbra{I_E} B_{\alpha', \sigma'}^{\delta_1}(\tau')\Delta B_{\alpha, \sigma}(\tau)\vecket{\rho_E}$ and $\vecbra{I_E} \Delta B_{\alpha, \sigma}(\tau) B_{\alpha', \sigma'}^{\delta_2}(\tau')\vecket{\rho_E}$ on a piecewise continuous function (see appendix \ref{appendix:prop_4} for more details), we then obtain that $\abs{g(\delta_1) - g(\delta_2)}$ is small and in particular goes to $0$ as $\delta_1, \delta_2 \to 0$.}
\end{proof}

Formally, this lemma allows us to ``define" expressions of the form of Eq.~\ref{eq:system_GF} but for the original unregularized model i.e.~as a matter of definition, we will set
\begin{align}\label{eq:GF}
\vecbra{I_E, O_S} \bigg(\prod_{i = 1}^n {\mathcal{U}}_{\mathcal{A}_i}(t_i, s_i) \bigg) \vecket{\rho_E, \sigma_S} := \lim_{\delta \to 0}\ \vecbra{I_E, O_S} \bigg(\prod_{i = 1}^n {\mathcal{U}}^\delta_{\mathcal{A}_i}(t_i, s_i) \bigg) \vecket{\rho_E, \sigma_S}.
\end{align}
We point out that, from a more mathematical standpoint, we might be interested in defining the unitary group $U(t, s) = \mathcal{T}\exp(-i\int_s^t H(\tau) d\tau)$ instead of defining only expressions in Eq.~\ref{eq:GF}. This is a subtle issue due to the $\delta-$function contribution in the kernels, especially when the system itself is infinite-dimensional, but such a unitary group can in-fact be defined using the mathematical machinery of quantum stochastic calculus and its generalizations \cite{trivedi2022description, parthasarathy2012introduction}. In this paper, since our focus is to derive a Lieb-Robinson bound, we will not concern ourselves with this mathematical issue but use the definition in Eq.~\ref{eq:GF} instead of attempting to define the unitary group.

\section{Recap: Lieb Robinson bound for bounded lattice Hamiltonians}\label{sec:recap}
In this section,  we first briefly recap the derivation of the Lieb-Robinson bound for bounded lattice Hamiltonian following Ref.~\cite{hastings2006spectral} --- we will explicitly point out where this derivation fails for the models considered in this paper. This will also serve to guide the remainder of the proof of proposition 1.

Here, we consider a $d$-dimensional lattice $\Lambda$ of qudits that locally interact with each other only through a bounded Hamiltonian --- this is equivalent to setting $R_\alpha(s) = 0$ in the model in Eq.~\ref{eq:supp:sys_env_hamil} to obtain
\[
H(s) = \sum_{\alpha} h_\alpha(s),
\]
where, similar to Eq.~\ref{eq:supp:sys_env_hamil}, $h_\alpha(s)$ is supported on a geometrically local set $\mathcal{S}_{\alpha}$.  Furthermore, we will assume that $\norm{h_\alpha(s)} \leq 1$ and $h_\alpha(s)$ are differentiable with $\norm{h_\alpha'(s)}\leq \infty$. We will denote by $\mathcal{M}_H$ the set of all such Hamiltonians --- $\mathcal{M}_H$ will depend only on the supports $\mathcal{S}_{\alpha}$ and as will be clear below, the Lieb-Robinson bounds that we derive will also depend only on the supports $\mathcal{S}_{\alpha}$ and thus work for any Hamiltonian from $\mathcal{M}_H$. We consider $\Delta_{O_X}(t, t'; l)$ defined in the main-text: Given an operator $O_X$ supported on the region $X$,
\[
\Delta_{O_X}(t, t'; l) = \abs{\text{Tr}(U^\dagger (t, t') O_X U(t, t') - U_{X[l]}^\dagger (t, t') O_X U_{X[l]}(t, t') ) \rho(t'))},
\]
where $X[l] = \{x \in \Lambda : d(x, X) \leq l\}$, $U_{X[l]}(t, t')$ is as defined in the main text and  $\rho(t') = U(t', 0) \sigma_S U^\dagger(t', 0)$ is the state on evolving an initial state of the qudits $\sigma_S$ for time $t'$. We will break up the analysis of $\Delta_{O_X}(t, t; l)$ into several steps.

\textbf{Step 1: Defining a commutator bound.} For the purpose of obtaining a bound on $\Delta_{O_X}(t, t'; l)$, it is convenient to first obtain a related commutator bound --- given $X, Y \subseteq\Lambda$ and $\mathcal{B}\subseteq \mathcal{A}$, we define $\gamma_\mathcal{B}^{X, Y}(t, t')$ by
\[
\gamma^{X, Y}_\mathcal{B}(t, t') = \sup_{O_X, \mathcal{L}_Y} \frac{\norm{\mathcal{L}_Y^\dagger (U_\mathcal{B}^\dagger(t, t')O_X U_\mathcal{B}(t, t'))}}{\norm{O_X} \norm{\mathcal{L}_Y}_\diamond}.
\]
Here, the supremum is taken over all operators $O_X$ supported on $X$ and superoperators $\mathcal{L}_Y$ supported on $Y$ which also satisfy $\mathcal{L}_Y^\dagger(I) = 0$. Now, $\Delta_{O_X}(t, t'; l)$ can be upper-bounded in by $\gamma^{X, Y}_\mathcal{B}(t, t')$ --- to see this, first note that
\begin{align*}
    \Delta_{O_X}(t, t'; l) &\leq \smallnorm{U^\dagger(t, t') O_X U(t, t') - U_{{X[l]}}^\dagger(t, t') O_X U_{{X[l]}}(t, t')}, \nonumber \\
    &= \smallnorm{O_X  - U(t, t')U_{{X[l]}}(t', t) O_X U_{{X[l]}}(t, t')U(t', t)}, \nonumber \\
    &=\norm{\int_{t'}^t U(t, s) [H(s) - H_{{X[l]}}(s), U_{{X[l]}}(s, t)O_X U_{{X[l]}}(t, s)] U(s, t) ds}.
\end{align*}
Suppose $X_c[l]$ to be the set of sites which are contained in sub-regions $\mathcal{S}_{\alpha}$ overlapping with $X[l]$: $X_c[l] = \bigcup_{\alpha:\mathcal{S}_{\alpha} \cap X[l]\neq \emptyset} \mathcal{S}_{\alpha}$. Then, since $U_{{X[l]}}(t', s) O_X U_{{X[l]}}(s, t')$ is supported only on $X_c[l]$, $[h_\alpha(s), U_{{X[l]}}(t', s) O_X U_{{X[l]}}(s, t')] = 0$ for any $\alpha : \mathcal{S}_{\alpha} \cap X_c[l] = \emptyset$. Therefore, we obtain that
\[
[H(s) - H_{{X[l]}}(s), U_{{X[l]}}(s, t) O_X U_{{X[l]}}(t, s)] = \sum_{\substack{\alpha \notin \mathcal{A}_{X[l]} \\ \mathcal{S}_{\alpha} \cap X_c[l]\neq \emptyset}} [h_\alpha(s), U_{{X[l]}}(s, t) O_X U_{{X[l]}}(t, s)].
\]
Finally, we note that
\begin{align}\label{eq:delta_bound_comm}
\Delta_{O_X}(t, t'; l) \leq  \sum_{\substack{\alpha \notin \mathcal{A}_{X[l]} \\ \mathcal{S}_{\alpha} \cap X_c[l]\neq \emptyset}} \int_{t'}^t \norm{ [h_\alpha(s), U_{{X[l]}}(s, t) O_X U_{{X[l]}}(t, s)]}ds \leq 2 \sum_{\substack{\alpha \notin \mathcal{A}_{X[l]} \\ \mathcal{S}_{\alpha} \cap X_c[l]\neq \emptyset}} \int_{t'}^t \gamma^{X, \mathcal{S}_{\alpha}}_{\mathcal{A}_{X[l]}}(t, s) ds.
\end{align}
Thus, obtaining a bound on $\gamma^{X, Y}_{\mathcal{B}}(t, t')$ also furnishes a bound on $\Delta_{O_X}(t, t'; l)$. We remark that this bound (Eq.~\ref{eq:delta_bound_comm}) implicitly assumes that $\gamma^{X, \mathcal{S}_{\alpha}}_{\mathcal{A}_{X[l]}}(t, t')$ is integrable with respect to $t'$ --- while this might seem to be a reasonable assumption, we will provide a rigorous proof in Step 3.

To understand the complications that arise when considering the full unbounded system-environment model, it will be more insightful to re-express $\gamma^{X, Y}_\mathcal{B}(t, t')$ as
\begin{align}\label{eq:fd:gamma_def}
\gamma^{X, Y}_\mathcal{B}(t, t') = \sup_{\phi, O_X, \mathcal{L}_Y} \frac{\abs{\text{Tr}(\mathcal{L}_Y^\dagger (U_\mathcal{B}^\dagger(t, t')O_X U_\mathcal{B}(t, t'))\phi)}}{\norm{O_X}\norm{\mathcal{L}_Y}_\diamond \norm{\phi}_1} = \sup_{\phi, O_X, \mathcal{L}_Y } \frac{\smallabs{\vecbra{O_X}\mathcal{U}_\mathcal{B}(t, t')\mathcal{L}_Y\vecket{\phi}}}{\norm{O_X}\norm{\mathcal{L}_Y}_\diamond \norm{\phi}_1},
\end{align}
where we have used the fact that, for any operator $A$, $\norm{A} = \sup_{\phi : \norm{\phi}_1 \leq 1} \abs{\text{Tr}(A\phi)}$. In this expression for $\gamma^{X, Y}_\mathcal{B}(t, t')$, we can interpret $\phi$ to be the \emph{state} (even though it is not constrained to be a positive operator) over which the expected value of the operator $\mathcal{L}_Y^\dagger(U_\mathcal{B}^\dagger(t, t') O_X U_\mathcal{B}(t, t'))$ is being maximized over. For bounded Hamiltonians, we can maximize this expected value over all $\phi$ with the only constraint being its normalization (i.e.~$\norm{\phi}_1 \leq 1$) --- while dealing with the unbounded system-environment model, a key challenge would be to identify ``good" spaces of operators $\phi$ to define a quantity similar to $\gamma_\mathcal{B}^{X, Y}(t, t')$.

\textbf{Step 2: Deriving a first-order expansion for the commutator bound.} The general strategy to derive an upper bound on $\gamma^{X, Y}_{\mathcal{B}}(t, t')$ is to consider it for a slightly larger time interval $[t' - \varepsilon, t]$, $\gamma_{\mathcal{B}}^{X, Y}(t, t' - \varepsilon)$, and relate it to $\gamma_{\mathcal{B}}^{X, Y}(t, t')$ by performing a first order expansion of the propagator $\mathcal{U}_{\mathcal{B}}(t, t')$. Throughout the following analysis, without loss of generality, will assume the normalization $\norm{\mathcal{L}_Y}_\diamond = 1, \norm{O_X} = 1, \norm{\phi}_1 = 1$. We begin by noting that
\begin{align}
    \abs{\vecbra{O_X}\mathcal{U}_\mathcal{B}(t, t' - \varepsilon) \mathcal{L}_Y \vecket{\phi}} =\abs{\vecbra{O_X}\mathcal{U}_\mathcal{B}(t, t')\mathcal{U}_\mathcal{B}(t', t'-\varepsilon) \mathcal{L}_Y \mathcal{U}_\mathcal{B}(t' - \varepsilon, t') \mathcal{U}_\mathcal{B}(t', t'-\varepsilon)\vecket{\phi}}.
\end{align}
Now, to $O(\varepsilon^2)$, we expect $\mathcal{U}_\mathcal{B}(t', t'-\varepsilon) \mathcal{L}_Y \mathcal{U}_\mathcal{B}(t' - \varepsilon, t')$ to be equal to $\mathcal{U}_{\mathcal{B}\cap \mathcal{A}_Y}(t', t' - \varepsilon) \mathcal{L}_Y \mathcal{U}_{\mathcal{B}\cap \mathcal{A}_Y}(t'-\varepsilon, t')$ i.e., by only considering the terms in $\mathcal{H}_\mathcal{B}(s)$ that intersect with $Y$. This can be made precise, via a $1^\text{st}$ order Dyson series expansion:
\begin{align}\label{eq:fd:first_order_exp}
&\mathcal{U}_{\mathcal{B}}(t', t'-\varepsilon) \mathcal{L}_Y \mathcal{U}_{\mathcal{B}}(t' - \varepsilon, t') \nonumber\\
&\qquad = \mathcal{L}_Y - i\int_{t' - \varepsilon}^{t'} [\mathcal{H}_{\mathcal{B}}(s'), \mathcal{L}_Y] ds' + i^2 \int_{t' - \varepsilon}^{t'} \int_{s'}^{t'}\mathcal{U}_\mathcal{B}(t', s'')[\mathcal{H}_\mathcal{B}(s''), [\mathcal{H}_{\mathcal{B}}(s'), \mathcal{L}_Y]]\mathcal{U}_\mathcal{B}(s'', t') ds'' ds', \nonumber\\
&\qquad \numeq{1} \mathcal{L}_Y - i \int_{t' - \varepsilon}^{t'} [\mathcal{H}_{\mathcal{B} \cap \mathcal{A}_Y}(s'), \mathcal{L}_Y] ds' + i^2 \int_{t' - \varepsilon}^{t'} \int_{s'}^{t'}\mathcal{U}_\mathcal{B}(t', s'')[\mathcal{H}_\mathcal{B}(s''), [\mathcal{H}_\mathcal{B}(s'), \mathcal{L}_Y]]\mathcal{U}_\mathcal{B}(s'', t') ds'' ds', \nonumber \\
&\qquad \numeq{2} \mathcal{U}_{\mathcal{B}\cap \mathcal{A}_Y}(t', t'-\varepsilon) \mathcal{L}_Y \mathcal{U}_{\mathcal{B}\cap \mathcal{A}_Y}(t' - \varepsilon, t') + i^2 \int_{t' - \varepsilon}^{t'} \int_{s'}^{t'}\mathcal{U}_\mathcal{B}(t', s'')[\mathcal{H}_\mathcal{B}(s''), [\mathcal{H}_\mathcal{B}(s'), \mathcal{L}_Y]]\mathcal{U}_\mathcal{B}(s'', t') ds'' ds' - \nonumber\\
&\qquad \qquad \qquad \qquad \qquad i^2 \int_{t' -\varepsilon}^{t'} \int_{s'}^{t'}\mathcal{U}_{\mathcal{B}\cap \mathcal{A}_Y}(t', s'') [\mathcal{H}_{\mathcal{B} \cap \mathcal{A}_Y}(s''), [\mathcal{H}_{\mathcal{B} \cap  \mathcal{A}_Y}(s'), \mathcal{L}_Y]] \mathcal{U}_{\mathcal{B}\cap \mathcal{A}_Y}(s'', t')ds'' ds'.
\end{align}
In (1), we have used that if $\alpha \in \mathcal{B}\setminus \mathcal{A}_Y$ then $\mathcal{S}_{\alpha} \cap Y = \emptyset$ and consequently $[\mathcal{C}_{h_\alpha(s')}, \mathcal{L}_Y]= 0$. Since $\mathcal{H}_\mathcal{B}(s') =  \mathcal{H}_{\mathcal{B} \cap \mathcal{A}_Y}(s') + \mathcal{H}_{\mathcal{B}\setminus \mathcal{A}_Y}(s')$ and consequently $[\mathcal{L}_Y, \mathcal{H}_\mathcal{B}(s')] = [\mathcal{L}_Y, \mathcal{H}_{\mathcal{B}\cap \mathcal{A}_Y}(s')]$. In (2), we have used a first order Dyson series expansion of $\mathcal{U}_{\mathcal{B}\cap \mathcal{A}_Y}(t', t'-\varepsilon)\mathcal{L}_Y\mathcal{U}_{\mathcal{B}\cap \mathcal{A}_Y}(t' - \varepsilon, t')$. Therefore, from Eq.~\ref{eq:fd:first_order_exp}, we obtain that
\begin{align}\label{eq:fd:taylor_expansion_1}
\abs{\vecbra{O_X}\mathcal{U}_\mathcal{B}(t, t'-\varepsilon) \mathcal{L}_Y \vecket{\phi} - \vecbra{O_X}\mathcal{U}_\mathcal{B}(t, t')\mathcal{U}_{\mathcal{B}\cap \mathcal{A}_Y}(t', t'-\varepsilon) \mathcal{L}_Y \vecket{\tilde{\phi}} } \leq 8 (\abs{\mathcal{B}}^2 + \abs{\mathcal{B}\cap \mathcal{A}_Y}^2)\varepsilon^2,
\end{align}
where $\vecket{\tilde{\phi}} = \mathcal{U}_{\mathcal{B}\cap \mathcal{A}_Y}(t' - \varepsilon, t') \mathcal{U}_\mathcal{B}(t', t'-\varepsilon)\vecket{\phi}$ (note that $\smallnorm{\tilde{\phi}}_1 = \norm{\phi}_1 = 1$). Next, we again perform a $1^\text{st}$ order Dyson expansion of $\mathcal{U}_{\mathcal{B}\cap \mathcal{A}_Y}(t', t' -\varepsilon)$ in $\vecbra{O_X} \mathcal{U}_{\mathcal{B}}(t, t') \mathcal{U}_{\mathcal{B}\cap \mathcal{A}_Y}(t', t'-\varepsilon) \mathcal{L}_Y \vecket{\tilde{\phi}}$ to obtain
\begin{align}
&\vecbra{O_X} \mathcal{U}_{\mathcal{B}}(t, t') \mathcal{U}_{\mathcal{B}\cap \mathcal{A}_Y}(t', t'-\varepsilon) \mathcal{L}_Y \vecket{\tilde{\phi}} = \vecbra{O_X}\mathcal{U}_{\mathcal{B}}(t, t') \mathcal{L}_Y \vecket{\tilde{\phi}} - i\int_{t' -\varepsilon}^{t'}\vecbra{O_X}\mathcal{U}_{\mathcal{B}}(t, t') \mathcal{H}_{\mathcal{B}\cap \mathcal{A}_Y}(s') \mathcal{L}_Y \vecket{\tilde{\phi}} ds'
  + \nonumber\\
&\qquad \qquad \qquad(-i)^2\int_{t' - \varepsilon}^{t'}\int_{s'}^{t'}\vecbra{O_X}\mathcal{U}_\mathcal{B}(t, t') \mathcal{U}_{\mathcal{B}\cap \mathcal{A}_Y}(t', s'')\mathcal{H}_{\mathcal{B}\cap \mathcal{A}_Y}(s'') \mathcal{H}_{\mathcal{B}\cap \mathcal{A}_Y}(s')\mathcal{L}_Y \vecket{\tilde{\phi}} ds'' ds',
\end{align}
from which we obtain that
\begin{align}\label{eq:fd:taylor_expansion_2}
    &\abs{\vecbra{O_X} \mathcal{U}_{\mathcal{B}}(t, t') \mathcal{U}_{\mathcal{B}\cap \mathcal{A}_Y}(t', t'-\varepsilon) \mathcal{L}_Y \vecket{\tilde{\phi}} -\bigg( \vecbra{O_X}\mathcal{U}_{\mathcal{B}}(t, t') \mathcal{L}_Y \vecket{\tilde{\phi}} - i\int_{t' -\varepsilon}^{t'}\vecbra{O_X}\mathcal{U}_{\mathcal{B}}(t, t') \mathcal{H}_{\mathcal{B}\cap \mathcal{A}_Y}(s') \mathcal{L}_Y \vecket{\tilde{\phi}} ds' \bigg) } \leq \nonumber\\
    &\qquad \qquad \qquad \qquad  \qquad \qquad \qquad \qquad  \qquad \qquad \qquad \qquad  \qquad \qquad \qquad \qquad  \qquad \qquad \qquad \qquad  2  \abs{\mathcal{B}\cap \mathcal{A}_Y}^2 \varepsilon^2.
\end{align}
From Eqs.~\ref{eq:fd:taylor_expansion_1} and \ref{eq:fd:taylor_expansion_2} and the triangle inequality, we obtain that (here, $C_0 = 8 \abs{\mathcal{B}}^2 + 10 \abs{\mathcal{B}\cap \mathcal{A}_Y}^2 $):
\begin{align}
\abs{\vecbra{O_X}\mathcal{U}_\mathcal{B}(t, t' - \varepsilon) \mathcal{L}_Y \vecket{\phi}} &\leq \abs{\vecbra{O_X}\mathcal{U}_\mathcal{B}(t, t') \mathcal{L}_Y \vecket{\tilde{\phi}}} + \int_{t ' -\varepsilon}^{t'} \abs{\vecbra{O_X} \mathcal{U}_{\mathcal{B}}(t, t') \mathcal{H}_{\mathcal{B}\cap \mathcal{A}_Y}(s')\mathcal{L}_Y \vecket{\tilde{\phi}} } ds' + C_0 \varepsilon^2,  \nonumber\\
& \leq \abs{\vecbra{O_X} \mathcal{U}_\mathcal{B}(t, t')\mathcal{L}_Y \vecket{\tilde{\phi}}} + \sum_{\alpha \in \mathcal{B} \cap \mathcal{A}_Y} \int_{t' - \varepsilon}^{t'} \abs{\vecbra{O_X} \mathcal{U}_\mathcal{B}(t, t') \mathcal{C}_{h_{\alpha}(s')} \mathcal{L}_Y \vecket{\tilde{\phi}}}ds' +C_0 \varepsilon^2, \nonumber \\
& \numleq{1}  \gamma^{X, Y}_\mathcal{B}(t, t') + 2  \sum_{\alpha \in \mathcal{A}_Y}\gamma_{\mathcal{B}}^{X, \mathcal{S}_{\alpha}}(t, t') \varepsilon  + C_0 \varepsilon^2 ,
\end{align}
where in (1) we have used the definition of $\gamma^{X, Y}_{\mathcal{B}}(t, t')$ and used the upper bounds:
\begin{align*}
&\abs{\vecbra{O_X} \mathcal{U}_\mathcal{B}(t, t')\mathcal{L}_Y \vecket{\phi}} \leq \gamma_\mathcal{B}^{X, Y}(t, t')\norm{\mathcal{L}_Y}_\diamond \norm{O_X} \norm{\phi}_1 \leq \gamma_{\mathcal{B}}^{X, Y}(t, t') \text{ and } \\
&\abs{\vecbra{O_X} \mathcal{U}_\mathcal{B}(t, t') \mathcal{C}_{h_{\alpha}(s')} \mathcal{L}_Y \vecket{\tilde{\phi}}} \leq \gamma_{\mathcal{B}}^{X, \mathcal{S}_{\alpha}}(t, t') \smallnorm{\mathcal{C}_{h_\alpha(s')}}_\diamond \norm{O_X}\smallnorm{\mathcal{L}_Y (\tilde{\phi})}_1 \leq 2 \gamma_{\mathcal{B}}^{X, \mathcal{S}_{\alpha}}(t, t').
\end{align*}
In (1), we have also simultaneously upper bounded the sum over $\alpha \in \mathcal{B}\cap \mathcal{A}_Y$ with a (larger) sum over $\alpha \in \mathcal{A}_Y$. Since this bound holds for any $O_X, \mathcal{L}_Y$ and $\vecket{\phi}$, again from the definition of $\gamma^{X, Y}_\mathcal{B}(t, t')$, we then obtain that
\begin{align}\label{eq:fd:trotter_expansion}
\gamma_\mathcal{B}^{X, Y}(t, t' - \varepsilon) \leq \gamma_\mathcal{B}^{X, Y}(t, t') + 2 \varepsilon \sum_{\alpha \in \mathcal{A}_Y}\gamma_\mathcal{B}^{X, \mathcal{S}_{\alpha}}(t, t') + C_0 \varepsilon^2.
\end{align}

\textbf{Step 3: Lieb-Robinson integral inequation.} For a given $t > 0$ and $\mathcal{B}$, it is convenient to reformulate Eq.~\ref{eq:fd:trotter_expansion} in terms of $\Gamma_{X, Y}(\tau) = \gamma^{X, Y}_{\mathcal{B}}(t, t - \tau)$:
\begin{align}\label{eq:fd:trotter_expansion_new}
    \Gamma_{X, Y}(\tau + \varepsilon) \leq \Gamma_{X, Y}(\tau) + 2\varepsilon \sum_{\alpha \in \mathcal{A}_Y} \Gamma_{X, \mathcal{S}_{\alpha}}(\tau) + C_0 \varepsilon^2.
\end{align}
We will next convert it into an integral equation --- for this, we choose $\varepsilon = \tau /  N$ and define $\tau_n = n \varepsilon$ for $n \in [0:N]$.  Then, from Eq.~\ref{eq:fd:trotter_expansion_new}, we obtain that
\begin{align}
    \Gamma_{X, Y}(\tau_n) \leq \Gamma_{X, Y}(\tau_{n - 1}) + 2\varepsilon\sum_{\alpha \in \mathcal{A}_Y}\Gamma_{X, \mathcal{S}_{\alpha}}(\tau_{n - 1}) + C_0 \bigg(\frac{\tau}{N}\bigg)^2,
\end{align}
and consequently, we have that
\begin{align}
    \Gamma_{X, Y}(\tau) \leq \Gamma_{X, Y}(0) +2\varepsilon \sum_{\alpha \in \mathcal{A}_Y} \sum_{n = 0}^{N - 1}\Gamma_{X, \mathcal{S}_{\alpha}}(\tau_{n - 1}) + C_0\frac{\tau^2}{N},
\end{align}
or, in the limit of $N \to \infty$,
\begin{align}\label{eq:almost_integral_equation_fd}
\Gamma_{X, Y}(\tau) \leq \Gamma_{X, Y}(0) +2\sum_{\alpha \in \mathcal{A}_Y}\lim_{N \to \infty}\bigg(\varepsilon \sum_{n = 0}^{N - 1}\Gamma_{X, \mathcal{S}_{\alpha}}(\tau_{n - 1})\bigg).
\end{align}
We would like to identify $\lim_{N \to \infty} \varepsilon \sum_{n =0}^{N - 1}\Gamma_{X, \mathcal{S}_{\alpha}}(\tau_{n - 1})$ with $\int_0^\tau \Gamma_{X, \mathcal{S}_{\alpha}}(\tau') d\tau'$. However, for this to hold, we must ensure that $\Gamma_{X, \mathcal{S}_{\alpha}}(\tau)$ is integrable with respect to $\tau$. For the finite-dimensional case that we are currently considering, this is relatively straightforward to show since we can establish that $\Gamma_{X, \mathcal{S}_{\alpha}}(\tau)$ is in fact a continuous function of $\tau$. To see this, note that assuming $\norm{O_X} = \norm{\phi}_1 = \norm{\mathcal{L}_Y}_\diamond = 1$,
\begin{align*}
&\abs{\vecbra{O_X}\mathcal{U}_\mathcal{B}(t, t') \vecket{\phi} - \vecbra{O_X}\mathcal{U}_\mathcal{B}(t, t' - \varepsilon) \vecket{\phi}} = \abs{\vecbra{O_X} \mathcal{U}_\mathcal{B}(t, t') (I - \mathcal{U}_B(t', t'-\varepsilon)) \vecket{\phi}}, \nonumber \\
&\qquad \qquad = \abs{\vecbra{O_X}\mathcal{U}_\mathcal{B}(t, t') \bigg(\int_{t' - \varepsilon}^{t'}\mathcal{H}_\mathcal{B}(s)  \mathcal{U}_\mathcal{B}(s, t' - \varepsilon)ds\bigg) \vecket{\phi} }, \nonumber\\
&\qquad \qquad \leq \int_{t' - \varepsilon}^{t'}\norm{\mathcal{H}_\mathcal{B}(s)}_\diamond ds \leq 2\varepsilon \abs{\mathcal{B}}.
\end{align*}
Therefore, we obtain that
\begin{align*}
&\gamma_{\mathcal{B}}^{X, Y}(t, t' - \varepsilon) = \sup_{\phi, O_X, \mathcal{L}_Y } \frac{\smallabs{\vecbra{O_X}\mathcal{U}_\mathcal{B}(t, t' - \varepsilon)\mathcal{L}_Y\vecket{\phi}}}{\norm{O_X}\norm{\mathcal{L}_Y}_\diamond \norm{\phi}_1} \leq2\varepsilon \abs{\mathcal{B}} + \sup_{\phi, O_X, \mathcal{L}_Y } \frac{\smallabs{\vecbra{O_X}\mathcal{U}_\mathcal{B}(t, t')\vecket{\phi}}}{\norm{O_X} \norm{\phi}_1 \norm{\mathcal{L}_Y}_\diamond} \leq \gamma^{X, Y}_{\mathcal{B}}(t, t') + 2\varepsilon\abs{\mathcal{B}} \text{ and } \nonumber\\
&\gamma_{\mathcal{B}}^{X, Y}(t, t') = \sup_{\phi, O_X, \mathcal{L}_Y} \frac{\abs{\vecbra{O_X}\mathcal{U}_\mathcal{B}(t, t')\mathcal{L}_Y \vecket{\phi}}}{\norm{O_X}\norm{\mathcal{L}_Y}_\diamond \norm{\phi}_1} \leq 2\varepsilon\abs{\mathcal{B}} + \sup_{\phi, O_X, \mathcal{L}_Y} \frac{\abs{\vecbra{O_X}\mathcal{U}_\mathcal{B}(t, t' - \varepsilon)}}{\norm{O_X} \norm{\mathcal{L}_Y}_\diamond \norm{\phi}_1} \leq \gamma^{X, Y}_\mathcal{B}(t, t' - \varepsilon) + 2\varepsilon \abs{\mathcal{B}},
\end{align*}
which, taken together, imply that
\[
\abs{\gamma^{X, Y}_\mathcal{B}(t, t') - \gamma^{X, Y}_\mathcal{B}(t, t' - \varepsilon)} \leq 2\varepsilon \abs{\mathcal{B}}.
\]
Since $\gamma^{X, Y}_\mathcal{B}(t, t')$ is a continuous function of $t'$, $\Gamma_{X, Y}(\tau)$ is also a continuous function of $\tau$. Thus, $\Gamma_{X, Y}(\tau)$ is an integrable function of $\tau$ (and so is $\gamma^{X, Y}_{\mathcal{B}}(t, t')$ integrable with respect to $t'$ --- note that this also justifies Eq.~\ref{eq:delta_bound_comm}). From Eq.~\ref{eq:almost_integral_equation_fd}, we then obtain
\begin{align}\label{eq:fd:int_eq_final}
\Gamma_{X, Y}(\tau) \leq \Gamma_{X, Y}(0) + 2\sum_{\alpha \in \mathcal{A}_Y} \int_0^\tau \Gamma_{X, \mathcal{S}_{\alpha}}(\tau') d\tau'\leq \Gamma_{X, Y}(0) + 2L\sum_{S_{\alpha} \in \mathcal{A}_Y} \int_0^\tau \Gamma_{X, \mathcal{S}_{\alpha}}(\tau') d\tau'.
\end{align}
Furthermore, we can note that if $X \cap Y = \emptyset$, $\Gamma_{X, Y}(0) = 0$ and for any $X, Y$, $\Gamma_{X, Y}(t) \leq 1$.

\textbf{Step 4: Solving the Lieb-Robinson integral inequation.} Finally, we can solve Eq.~\ref{eq:fd:int_eq_final} to obtain the Lieb-Robinson bound. We explicitly provide the solution of this integral inequality as a lemma, since we will reuse it in the non-Markovian case.
\begin{lemma}[Lieb Robinson recursion, Ref.~\cite{hastings2006spectral}]\label{lemma:lr_recursion}
    Suppose, for any $X, Y \subseteq \Lambda$, $\nu_{X, Y}(t) \in [0, 1]$ satisfies the integral inequality
    \[
    \nu_{X, Y}(t) \leq \nu_{X, Y}(0) + \omega_0 \sum_{{\alpha} : \mathcal{S}_{\alpha} \cap Y \neq \emptyset} \int_0^t \nu_{X, \mathcal{S}_{\alpha}}(s) ds,
    \]
    together with the initial condition $\nu_{X, Y}(0) = 0$ for $X \cap Y = \emptyset$, then
    \[
     \nu_{X, Y}(t) \leq \abs{Y}(\exp(e\omega_0 \mathcal{Z}t) - 1)\exp(-d_{X, Y}/a_0),
     \]
     where $\mathcal{Z}, a_0$ are defined in Eq.~\ref{eq:constants_lattice}.
\end{lemma}
\begin{proof}
   This lemma underlies the proof of standard Lieb Robinson bounds, and is well known in the literature on Lieb Robinson bounds. We only provide a proof of this statement for completeness and for the convenience of the reader.

   We begin by considering $X, Y \subseteq \Lambda$ with $X \cap Y = \emptyset$, then $\Gamma_{X, Y}(0) = 0$ and we obtain that
   \begin{align}\label{eq:gamma_XY_as_gamma_alpha}
   \nu_{X, Y}(t) \leq \omega_0 \sum_{{\alpha} : \mathcal{S}_{\alpha} \cap Y \neq \emptyset}\int_0^t \nu_{X, \mathcal{S}_{\alpha}}(s)ds \leq \omega_0 \mathcal{Z}\abs{Y} \sup_{{\alpha_0} : \mathcal{S}_{\alpha_0} \cap Y \neq \emptyset} \int_0^t \nu_{X, \mathcal{S}_{\alpha_0}}(s) ds.
   \end{align}
   Now, suppose $k$ is the smallest integer such that $\exists \mathcal{S}_{\alpha_0}, \mathcal{S}_{\alpha_1} \dots \mathcal{S}_{\alpha_{k}}$ such that $Y \cap \mathcal{S}_{\alpha_0} \neq \emptyset, \mathcal{S}_{\alpha_1}\cap \mathcal{S}_{\alpha_2} \neq \emptyset \dots \mathcal{S}_{\alpha_{k - 1}}\cap \mathcal{S}_{\alpha_{k}} \neq \emptyset $ and $\mathcal{S}_{\alpha_{k}}\cap X \neq \emptyset$. Clearly, since $k$ is the smallest such integer, for any $\mathcal{S}_{\beta_0}, \mathcal{S}_{\beta_1} \dots \mathcal{S}_{\beta_{k-1}}$ such that $Y \cap \mathcal{S}_{\beta_0} \neq \emptyset, \mathcal{S}_{\beta_0}\cap \mathcal{S}_{\beta_1} \neq \emptyset, \mathcal{S}_{\beta_1}\cap \mathcal{S}_{\beta_2} \neq \emptyset \dots \mathcal{S}_{\beta_{k - 2}} \cap \mathcal{S}_{\beta_{k - 1}} \neq \emptyset$ then  $\mathcal{S}_{\beta_i}\cap X = \emptyset$ for all $i \in [0:(k - 1)]$. Consequently, for all such $\mathcal{S}_{\beta_0}, \mathcal{S}_{\beta_1} \dots \mathcal{S}_{\beta_{k - 1}}$, $\Gamma_{X, \mathcal{S}_{\beta_i}}(0) = 0$.  Therefore, we obtain that for $j \in [0:(k - 1)]$
   \[
   \nu_{X, \mathcal{S}_{\beta_{j}}}(t) \leq  \omega_0 \sum_{\beta_{j + 1}: S_{\beta_j}\cap S_{\beta_{j + 1}}\neq \emptyset} \int_0^t \nu_{X, S_{\beta_{j + 1}}}(s) ds,
   \]
   from which we obtain that
   \begin{align*}
   \nu_{X, S_{\beta_0}}(s) &\leq \omega_0^k \sum_{{\beta_{1}}: \mathcal{S}_{\beta_1}\cap \mathcal{S}_{\beta_0} \neq \emptyset} \bigg(\sum_{{\beta_{2}}: \mathcal{S}_{\beta_2}\cap {\beta_1} \neq \emptyset} \dots \bigg(\sum_{\mathcal{S}_{\beta_{k}}: \mathcal{S}_{\beta_{k}}\cap \mathcal{S}_{\beta_{k - 1}} \neq \emptyset} \int_0^s \int_0^{s_1} \dots \int_0^{s_{k - 1}} \nu_{X, S_{\beta_k}}(s_k) ds_k ds_{k - 1}\dots ds_1\bigg)\bigg)\bigg), \nonumber\\
   &\leq \omega_0^k \mathcal{Z}^k \int_0^s \int_0^{s_1} \dots \int_0^{s_{k - 1}}  ds_k ds_{k - 1}\dots ds_1 = \frac{(\omega_0 \mathcal{Z} s)^{k }}{k!}.
   \end{align*}
   Then, using Eq.~\ref{eq:gamma_XY_as_gamma_alpha}, we obtain that
   \begin{align}\label{eq:upper_bound_gamma_XY}
   \nu_{X, Y}(t) \leq \abs{Y} \frac{(\omega_0 \mathcal{Z}t)^{k + 1}}{(k + 1)!} \leq  \abs{Y} e^{-(k + 1)} \frac{(e\omega_0 \mathcal{Z}t)^{k + 1}}{(k + 1)!} \leq \abs{Y}e^{-(k + 1)} \sum_{k' \geq 1}^{\infty}\frac{(e\omega_0 \mathcal{Z}t)^{k'}}{k'!} \leq \abs{Y}e^{-(k + 1)}(\exp(e\omega_0 \mathcal{Z}t) - 1).
   \end{align}
   Furthermore, we also note that the integer $k$ is lower bounded by a function of the distance $d_{X, Y}$. In particular, suppose $z_0 \in Y \cap \mathcal{S}_{\beta_0}, z_1 \in \mathcal{S}_{\beta_0}\cap \mathcal{S}_{\beta_1} \dots z_{k + 1} \in \mathcal{S}_{\beta_k}\cap X$. Then, since $z_j, z_{j +1 } \in \mathcal{S}_{\beta_j}$ for $j \in [0:k]$, $d(z_j, z_{j + 1}) \leq \text{diam}(\mathcal{S}_{\beta_j}) \leq a_0$.
   \[
   d_{X, Y} \leq d(z_0, z_{k + 1}) \leq d(z_0, z_1) + d(z_1, z_2) \dots d(z_k, z_{k + 1})  \leq (k + 1) a_0.
   \]
   Finally, observing that the upper bound on the right hand side of Eq.~\ref{eq:upper_bound_gamma_XY} is a decreasing function of $k$, we obtain that
   \[
   \nu_{X, Y}(t) \leq \abs{Y}(\exp(e\omega_0 \mathcal{Z}t) - 1)\exp(-d_{X, Y}/a_0),
   \]
   which establishes the lemma statement.
\end{proof}
\noindent Applying this lemma, we then obtain that
\[
\gamma^{X, Y}_{\mathcal{B}}(t, t') = \Gamma_{X, Y}(t - t') \leq \abs{Y}\big(\exp(2Le\mathcal{Z}\abs{t - t'}) - 1\big)\exp(-d_{X, Y}/a_0),
\]
and which, together with Eq.~\ref{eq:delta_bound_comm} yields an upper bound on $\Delta_{O_X}(t, t'; l)$.

\section{Common framework for proving non-Markovian Lieb-Robinson bounds}\label{sec:common_framework}
\begingroup\color{black}
We first establish the basic ingredients of the non-Markovian Lieb–Robinson proof, which are common to both the independent-bath and coupled-bath settings. The specific applications to these two settings will be discussed in the next section.
\endgroup We will work with the mollified model introduced in the previous section (Eq.~\eqref{eq:mollified_sys_env_hamiltonian}).
To simplify notation, we drop the explicit dependence on $\delta$ throughout the following sections; for example, $B_{\alpha}(t)$ denotes $B_{\alpha}^{\delta}(t)$ and $K_{\alpha,\alpha',\sigma'}(t)$ denotes $K_{\alpha,\alpha',\sigma'}^{\delta}(t)$. \begingroup\color{black}Since we will work only with the mollified model, the kernels $K_{\alpha,\alpha',\sigma'}(t)$ will be smooth and locally bounded functions of $t$. For notational convenience, given an interval $I \subseteq \mathbb{R}$, we will define
\[
\norminf{K_I}:=\sup_{\tau \in I}\sup_{\alpha,\alpha',\sigma'}\abs{K_{\alpha,\alpha',\sigma'}(\tau)}.
\]
It will also be convenient to define $\mu^{\textnormal{R}}_\alpha:\mathbb{R} \times \mathbb{R} \to \mathbb{R}$ via
\begin{align}\label{eq:mu}
    \mu^{\textnormal{R}}_\alpha(a, b) = \sum_{\alpha', \sigma'} \int_a^b  \smallabs{\K_{\alpha, \alpha', \sigma'}(\tau)} d\tau,
\end{align}
and $\mu^{\textnormal{L}}_{\alpha,\sigma}(a,b):\mathbb{R}\times\mathbb{R}\to\mathbb{R}$ via
\begin{align}\label{eq:kappa}
\mu^{\textnormal{L}}_{\alpha,\sigma}(a,b)=\sum_{\alpha'}\int_a^b|K_{\alpha',\alpha,\sigma}(\tau')|d\tau'.
\end{align}
Note that since $\K_{\alpha, \alpha',\sigma'}(t)$ already corresponds to the mollified model, it is a smooth function and consequently there is no ambiguity in the integral in above equation. We next observe the following simple properties of $\mu^{\textnormal{R}}_\alpha$:
\begin{itemize}
    \item For any $a, b$, $\mu^{\textnormal{R}}_\alpha(a, b) = -\mu^{\textnormal{R}}_\alpha(b, a)$.
    \item If $[a', b'] \subseteq [a, b]$, $\mu^{\textnormal{R}}_\alpha(a', b') \leq \mu^{\textnormal{R}}_\alpha(a, b)$.
    \item For any $a \leq b \leq c$, $\mu^{\textnormal{R}}_\alpha(a, c) = \mu^{\textnormal{R}}_\alpha(a, b) + \mu^{\textnormal{R}}_\alpha(b, c)$,
\end{itemize}
and the same property holds for $\mu^{\textnormal{L}}_{\alpha,\sigma}$.
Given $a, b\subseteq \mathbb{R}$ with $a < b$, it will also be convenient to define ${\mu}^{\textnormal{R}}_{\alpha, [a, b]}:\mathbb{R} \to \mathbb{R}$ and ${\mu}^{\textnormal{L}}_{\alpha,\sigma,[a, b]}:\mathbb{R} \to [0, \infty)$ via
\begin{align}\label{eq:mu_a_b}
{\mu}^{\textnormal{R}}_{\alpha, [a, b]}(t)  = \mu^{\textnormal{R}}_\alpha(t- b, t - a) \text{ and } {\mu}^{\textnormal{L}}_{\alpha, \sigma, [a, b]}(t)  = \mu^{\textnormal{L}}_{\alpha,\sigma}(a - t, b-t),
\end{align}
and by $\bm{\mu}^{\textnormal{R}}_{[a, b]} = \{\mu^{\textnormal{R}}_{1, [a, b]}, \mu^{\textnormal{R}}_{2, [a, b]} \dots\}$ as well as the same definition for the vector $\bm{\mu}^{\textnormal{L}}_{[a,b]}$. Furthermore, we note that from the properties of $\mu^{\textnormal{R}}_{\alpha}(a, b)$, it follows that $\forall a \leq b \leq c$
\[
\mu^{\textnormal{R}}_{[a, c]} = {\mu}^{\textnormal{R}}_{[a, b]} + {\mu}^{\textnormal{R}}_{[b, c]},
\]
and the same relation holds for $\mu^{\textnormal{L}}$.
Furthermore, it also follows from these definitions that $\forall a < b$ and $\tau$:
\[
\sum_{\alpha', \sigma'}\int_a^b \smallabs{\K_{\alpha, \alpha',  \sigma'}(\tau - \tau')} d\tau' \leq {\mu}^{\textnormal{R}}_{\alpha, [a, b]}(\tau) \text{ and } \sum_{\alpha'}\int_a^b \smallabs{K_{\alpha', \alpha, \sigma}(\tau' - \tau)}d\tau' \leq {\mu}^{\textnormal{L}}_{\alpha,\sigma, [a,b]}(\tau).
\]
It will also be convenient to collect $\mu^\text{R}_\alpha(a, b), \mu^\text{R}_{\alpha, [a, b]}, \mu^\text{L}_{\alpha, \sigma}(a, b)$ and $\mu^\text{L}_{\alpha, \sigma, [a, b]}$ into vectors over the indices $\alpha$ and $\sigma$:
\begin{align*}
&\bm{\mu}^{\textnormal{R}}(a, b) = \{\mu_{\alpha}^\text{R}(a, b)\}_{\alpha \in \mathcal{A}},\quad \bm{\mu}^{\textnormal{R}}_{[a, b]} = \{\mu_{\alpha, [a, b]}^\text{R}\}_{\alpha \in \mathcal{A}},  \\
&\bm{\mu}^{\textnormal{L}}(a, b) = \{\mu_{\alpha, \sigma}^{\text{L}}(a, b)\}_{\alpha \in \mathcal{A}, \sigma \in \{l, r\}}, \quad \bm{\mu}^{\textnormal{L}}_{[a, b]} = \{\mu_{\alpha,\sigma, [a, b]}^\text{L}\}_{\alpha \in \mathcal{A}, \sigma \in \{l, r\}}.
\end{align*}
We also remind the reader that $\smallnorm{\bm{\mu}^{\textnormal{R}}(a,b)}_\infty=\max_{\alpha}\{\mu^{\textnormal{R}}_\alpha(a,b)\}$, $\smallnorm{\bm{\mu}^{\textnormal{L}}(a,b)}_\infty=\max_{\alpha,\sigma}\{\mu^{\textnormal{L}}_{\alpha,\sigma}(a,b)\}$, and $\smallnorm{{\bm\mu}^{\textnormal{R}}_{[a, b]}(t)}_\infty = \norm{\bm\mu^{\textnormal{R}}(t- b, t - a)}_\infty \text{ and } \smallnorm{\bm{\mu}^{\textnormal{L}}_{[a, b]}(t)}_\infty  = \smallnorm{\bm\mu^{\textnormal{L}}(a - t, b-t)}_\infty,$
\endgroup
In the remaining subsections of the common framework, we develop the estimates that will be used in both bath models. We will build upon the derivation of the Lieb-Robinson bound for bounded lattice Hamiltonians reviewed in the previous section, but we will have to modify several different parts of the proof. We briefly sketch the major changes below:
\begin{enumerate}[leftmargin=1.4cm]
    \item[\textbf{Step 1}:] The counterpart of $\gamma^{X, Y}_\mathcal{B}(t, t')$ for the non-Markovian model has to be defined differently from Eq.~\ref{eq:fd:gamma_def} --- in particular, since the environment is an infinite-dimensional quantum field, it cannot be defined as a supremum over all $\phi$ (as in Eq.~\ref{eq:fd:gamma_def}) and still be expected to follow a Lieb-Robinson bound. We will devote subsection \ref{subsec:id_op_space} to identifying certain spaces of operators over which this supremum can be taken and still admit a Lieb Robinson bound.
    \item[\textbf{Step 2}:] Since the full system-environment Hamiltonian is unbounded, bounding the error in first order Dyson series expansion of $\gamma^{X, Y}_{\mathcal{B}}(t, t')$ with respect to $t'$ requires more careful analysis.
    \item[\textbf{Step 3}:] The key point in step 3, although simple, is that $\gamma^{X, Y}_{\mathcal{B}}(t, t')$ can be shown to be  continuous with respect to $t'$. For the full non-Markovian model, we will find that this isn't necessarily true---however, using a more sophisticated argument, we will establish that (the counter-part of) $\gamma^{X, Y}_{\mathcal{B}}(t, t')$ will be integrable with respect to $t'$, even if it is not continuous. This will turn out to be sufficient for our proof of Lieb-Robinson bounds to go through.
    \item[\textbf{Step 4}:] This step will go through without any major modifications.
\end{enumerate}

\subsection{Identifying operator spaces}\label{subsec:id_op_space}

We briefly recap the definition of $\Delta_{O_X}(t, t'; l)$ for the non-Markovian model from the main text ---  we will be given an operator $O_X$ supported on $X \subseteq \Lambda$. We will then consider, for given $t' > 0$, a state $\rho(t') = U(t', 0) (\sigma_S \otimes \rho_E) U(0, t')$ generated by evolving $\sigma_S \otimes \rho_E$ for time $t'$, where $\sigma_S$ is an arbitrary system state and $\rho_E$ is the initial Gaussian environment state. From time $t'$ to $t$, we will consider evolution under two Hamiltonians --- the full Hamiltonian $H(t)$ [Eq.~\ref{eq:supp:sys_env_hamil}], and the Hamiltonian $H_{\mathcal{A}_{X[l]}}(t)$ obtained by restricting $H(t)$, as per Eq.~\ref{eq:restriction_hamiltonian}, to
\[
X[l] = \{x \in \Lambda : d(x, X) \leq l\}
\]
i.e.~to a distance $l$ around the support $X$ of the observable $O_X$. The quantity that we will upper bound will be the error between the local observable $O_X$ under the states obtained from these two evolutions i.e.
\begin{align}\label{eq:delta_Ox_redef}
\Delta_{O_X}(t, t'; l) = \abs{\text{Tr}\big( (U(t', t) O_X U(t, t') - U_{{X[l]}}(t', t) O_X U_{{X[l]}}(t, t') ) \rho(t') \big) },
\end{align}
where $U_{X[l]}(t, t')$ will be the propagator of the Hamiltonian $H_{X[l]}(t)$. We begin with the following straightforward lemma which will motivate our definition of a commutator bound.
\begingroup\color{black}
\begin{lemma}[Truncation error decomposition]\label{lemma:delta_o_comm_bd}
The error $\Delta_{O_X}(t, t'; l)$ can be bounded as
\[
\Delta_{O_X}(t,t';l)\leq \Delta_{O_X}^{\textnormal{R}}(t,t';l)+\Delta^{\textnormal{L}}_{O_X}(t,t';l),
\]
where the right-contraction contribution is
\begin{align}
    &\Delta^{\textnormal{R}}_{O_X}(t,t';l) = \sum_{\substack{\alpha \notin \mathcal{A}_{X[l]} \\ \mathcal{S}_{\alpha} \cap {X}_c[l] \neq \emptyset}} \bigg(\int_{t'}^t \abs{\vecbra{O_X, I_E} \mathcal{U}_{X[l]}(t, s) \mathcal{C}_{h_\alpha(s)}\vecket{\rho(s)}} ds +\nonumber\\
    &\qquad \qquad \qquad \qquad \qquad \qquad\sum_{\alpha',\sigma'} \int_{s = t'}^t \int_{s' = 0}^s \smallabs{\K_{\alpha, \alpha',\sigma'}(s - s')} \smallabs{\vecbra{O_X, I_E} \mathcal{U}_{X[l]}(t, s) \mathcal{C}_{R_{\alpha}(s)}\vecket{\rho_{\alpha', \sigma'}(s, s')}} ds' ds\bigg),
\end{align}
and the left-contraction contribution is bounded by
\begin{align}
    \Delta_{O_X}^{\textnormal{L}}(t,t';l)
    =&\sum_{\substack{\alpha\notin\mathcal{A}_{X[l]}\\ \alpha'\in\mathcal{A}_{X[l]}, \sigma}}\int_{s=t'}^{t}\int_{s'=s}^{t}\abs{K_{\alpha',\alpha,\sigma}(s'-s)}\times\nonumber\\
    &\qquad\qquad\qquad\abs{\vecbra{O_{X},I_{E}}\mathcal{U}_{X[l]}(t,s')\mathcal{C}_{R_{\alpha'}(s')}\vecket{\tilde{\rho}_{\alpha, \sigma}(s, s')}}ds'ds.
\end{align}
where $X_c[l] = \bigcup_{\alpha: \mathcal{S}_{\alpha} \cap X[l] \neq \emptyset} \mathcal{S}_{\alpha}$, $\vecket{\rho(s)}  = \mathcal{U}(s, 0) \vecket{\rho(0)}$, $\vecket{\rho_{\alpha', \sigma'}(s, s')} = \mathcal{U}(s, s') R_{\alpha', \sigma'}(s')\mathcal{U}(s', 0) \vecket{\rho(0)}$ and $\vecket{\tilde{\rho}_{\alpha, \sigma}(s, s')} = \mathcal{U}_{X[l]}(s', s) R_{\alpha, \sigma}(s) \mathcal{U}(s, 0) \vecket{\rho(0)}$.
\end{lemma}
\begin{proof}
In vectorized notation, the error $\Delta_{O_X}(l)$ can be expressed as
\begin{align}\label{eq:SM_expression_of_delta_Ox}
\Delta_{O_X}(t, t'; l) &= \abs{\vecbra{O_X, I_E} \big(\mathcal{U}(t, t') - \mathcal{U}_{X[l]}(t, t')\big)\vecket{\rho(t')}},\nonumber \\
&= \abs{\int_{t'}^{t}\vecbra{O_X, I_E} \mathcal{U}_{X[l]}(t, t')\bigg( \frac{d}{ds} \big(\mathcal{U}_{X[l]}(t', s) \mathcal{U}(s, t')\big)  \bigg)\vecket{\rho(t')} ds}, \nonumber\\
&=\abs{\int_{t'}^t\vecbra{O_X, I_E}  \mathcal{U}_{X[l]}(t, s) \big(\mathcal{H}_{X[l]}(s) - \mathcal{H}(s) \big) \mathcal{U}(s, 0) \vecket{\rho(0)} ds}.
\end{align}
Next, we decompose $\mathcal{H}(s) - \mathcal{H}_{X[l]}(s)$ as
\begin{align*}
\mathcal{H}(s) - \mathcal{H}_{X[l]}(s) &= \sum_{\alpha \notin \mathcal{A}_{X[l]}} \bigg(\mathcal{C}_{h_\alpha(s)} + \sum_{\sigma}(-1)^\sigma B_{\alpha,\sigma}(s) R_{\alpha, \sigma}(s)\bigg),\nonumber\\
\end{align*}
where we have used $X_c[l] = \bigcup_{\alpha : \mathcal{S}_{\alpha} \cap X[l]\neq \emptyset} \mathcal{S}_{\alpha}$ (or, said in words, $X_c[l]$ is the union of all sub-regions $\mathcal{S}_{\alpha}$ which overlap with $X[l]$).
For notational convenience, define
\[
\vecbra{\Theta(s)}:=\vecbra{O_X,I_E}\mathcal{U}_{X[l]}(t,s),\qquad
\vecket{\rho(s)}:=\mathcal{U}(s,0)\vecket{\rho(0)} .
\]
If $\mathcal{S}_{\alpha}\cap X_c[l]=\emptyset$, then
\begin{align*}
\vecbra{\Theta(s)}\mathcal{C}_{h_\alpha(s)}=
\vecbra{\Theta(s)}\mathcal{C}_{R_\alpha(s)}=0.
\end{align*}
Lemma \ref{lemma:wick_thm_gaussian} gives
\begin{align}
&\sum_{\sigma}(-1)^\sigma
\vecbra{\Theta(s)}B_{\alpha,\sigma}(s)R_{\alpha,\sigma}(s)\vecket{\rho(s)} =
\sum_{\sigma}(-1)^\sigma
\wickleft(\vecbra{\Theta(s)};B_{\alpha,\sigma}(s))R_{\alpha,\sigma}(s)\vecket{\rho(s)}+
\nonumber\\
&\qquad\qquad \qquad\qquad \qquad\qquad\qquad\qquad \qquad\qquad \qquad
\sum_{\sigma}(-1)^\sigma
\vecbra{\Theta(s)}R_{\alpha,\sigma}(s)\wickright(B_{\alpha,\sigma}(s);\vecket{\rho(s)}).
\end{align}
The second term in the last expression satisfies
\begin{align}
&\left|
\sum_{\sigma}(-1)^\sigma
\vecbra{\Theta(s)}R_{\alpha,\sigma}(s)\wickright(B_{\alpha,\sigma}(s);\vecket{\rho(s)})
\right| \leq
\sum_{\alpha',\sigma'}\int_0^s
\abs{K_{\alpha,\alpha',\sigma'}(s-s')}
\abs{\vecbra{\Theta(s)}\mathcal{C}_{R_\alpha(s)}
\vecket{\rho_{\alpha',\sigma'}(s,s')}}ds',
\end{align}
and hence contributes only when $\mathcal{S}_{\alpha}\cap X_c[l]\neq\emptyset$. A similar calculation can be applied to $\mathcal{C}_{h_\alpha(s)}$.  For the first term, using the definition of $\wickleft(\cdot;\cdot)$,
\begin{align}
&\left|
\wickleft(\vecbra{\Theta(s)};B_{\alpha,\sigma}(s))R_{\alpha,\sigma}(s)\vecket{\rho(s)}
\right|\nonumber\\
&\quad\leq
\sum_{\substack{\alpha'\in\mathcal{A}_{X[l]}\\}}
\int_s^t\left|K_{\alpha',\alpha,\sigma}(s'-s)\right|
\left|\vecbra{O_X,I_E}\mathcal{U}_{X[l]}(t,s')
\mathcal{C}_{R_{\alpha'}(s')}\mathcal{U}_{X[l]}(s',s)
R_{\alpha,\sigma}(s)\vecket{\rho(s)}\right|ds'.
\end{align}
We emphasize that in the last line, the contribution of $\alpha\notin\mathcal{A}_{X[l]}$ cannot be restricted to $\mathcal{S}_{\alpha}\cap X_c[l]\neq\emptyset$ in general.  Combining the above estimates with Eq.~\eqref{eq:SM_expression_of_delta_Ox}, we obtain
\begin{align*}
\Delta_{O_X}(t,t';l)
&\leq
\sum_{\substack{\alpha\notin\mathcal{A}_{X[l]}\\ \mathcal{S}_{\alpha}\cap X_c[l]\neq\emptyset}}\int_{t'}^t
\left|\vecbra{\Theta(s)}\mathcal{C}_{h_\alpha(s)}\vecket{\rho(s)}\right|ds\nonumber\\
&\quad+
\sum_{\substack{\alpha\notin\mathcal{A}_{X[l]}\\ \mathcal{S}_{\alpha}\cap X_c[l]\neq\emptyset}}
\sum_{\alpha',\sigma'}\int_{t'}^t\int_0^s
\abs{K_{\alpha,\alpha',\sigma'}(s-s')}
\abs{\vecbra{\Theta(s)}\mathcal{C}_{R_\alpha(s)}
\vecket{\rho_{\alpha',\sigma'}(s,s')}}ds'ds\nonumber\\
&\quad+
\sum_{\substack{\alpha\notin\mathcal{A}_{X[l]}\\ \alpha'\in\mathcal{A}_{X[l]}\\ \sigma}}
\int_{t'}^t\int_s^t
\abs{K_{\alpha',\alpha,\sigma}(s'-s)}
\abs{\vecbra{O_X,I_E}\mathcal{U}_{X[l]}(t,s')
\mathcal{C}_{R_{\alpha'}(s')}\mathcal{U}_{X[l]}(s',s)
R_{\alpha,\sigma}(s)\vecket{\rho(s)}}ds'ds\nonumber\\
&=\Delta_{O_X}^{\textnormal{R}}(t,t';l)+\Delta^{\textnormal{L}}_{O_X}(t,t';l).
\end{align*}
\end{proof}
\endgroup

\emph{Closed operator spaces}. Lemma \ref{lemma:delta_o_comm_bd} reduces the first truncation contribution $\Delta_{1,O_X}(t,t';l)$ to commutators, e.g.~terms of the form
\begin{align*}
&\abs{\vecbra{O_X, I_E}\mathcal{U}_{X[l]}(t, s) \mathcal{C}_{h_\alpha(s)} \vecket{\rho(s)}} = \abs{\text{Tr}([U_{X[l]}(s, t)O_X U_{X[l]}(t, s), h_\alpha(s)] \rho(s))} \text{ and }\nonumber\\
&\abs{\vecbra{O_X, I_E} \mathcal{U}_{X[l]}(t, s) \mathcal{C}_{R_\alpha(s)}\vecket{\rho_{\alpha', \sigma'}(s, s')}} = \abs{\text{Tr}([U_{X[l]}(s, t) O_X U_{X[l]}(t, s), R_\alpha(s)] \rho_{\alpha', \sigma'}(s, s'))}.
\end{align*}
An approach to bounding these terms would be to attempt to place a bound on the norms of the commutators i.e.~$\norm{[U_{X[l]}(s, t) O_X U_{X[l]}(t, s), h_\alpha(s)]}$ and $\norm{[U_{X[l]}(s, t) O_X U_{X[l]}(t, s), R_\alpha(s)]}$. In the typical setting of a bounded lattice Hamiltonian, these norms can be upper-bounded by a quantity that exponentially decreases with the distance between $X$ and $\mathcal{S}_{\alpha}$. However, for the system-environment models considered in this paper, we do not expect such a bound to exist due to the infinite-dimensionalty of the environment and consequently, we need to specifically exploit the properties of the states (i.e.~$\rho(s)$ and $\rho_{\alpha', \sigma'}(s, s')$) with respect to which the commutator expectations are being taken. This motivates us to define a space of operators that are generated from the initial environment gaussian (definition \ref{def:S_rho_zeta}) which also has mathematically useful closure properties under time-evolution (lemma \ref{lemma:closure_right_wick}) that enable a Lieb-Robinson bound.

\begin{definition}[Operator space $\mathcal{S}_{\bm{\zeta}}(\rho_E)$]\label{def:S_rho_zeta}
Suppose $\bm{\zeta} = \{\zeta_\alpha\}_{\alpha \in \mathcal{A}}$ with $\zeta_\alpha: \mathbb{R} \to [0, \infty)$, then $\mathcal{S}_{\bm{\zeta}}(\rho_E)$ is the set of operators $\phi$ such that $\phi \in \mathcal{S}(\rho_E)$ [Def.~\ref{def:S_rho}] and has a representation with time edges $\{(s_i, t_i)\}_{i \in [1:n]}$ such that
\[
\sum_{i = 1}^n \abs{\sum_{\alpha', \sigma'} \int_{s_i}^{t_i}  \smallabs{\K_{\alpha, \alpha', \sigma'}(\tau-\tau')}d\tau'} = \sum_{i = 1}^n \smallabs{\mu^{\textnormal{R}}_{\alpha, [s_i, t_i]}(\tau)} \leq \zeta_\alpha(\tau) \ \forall \ \tau \in \mathbb{R}.
\]
\end{definition}

\begin{lemma}[Closure of $\mathcal{S}_{\bm{\zeta}}(\rho_E)$]\label{lemma:closure_right_wick} The set $\mathcal{S}_{\bm{\zeta}}(\rho_E)$ satisifies the following properties:
    \begin{enumerate}
    \item[(a)] For any $\mathcal{B} \subseteq \mathcal{A}$ and $0 < s < t$,
    \[
    {U}_\mathcal{B} (t, s) \mathcal{S}_{\bm \zeta}(\rho_E) U_\mathcal{B}(s, t)\text{ and }  {U}_\mathcal{B} (s, t) \mathcal{S}_{\bm \zeta}(\rho_E) {U}_\mathcal{B} (t, s) \subseteq \mathcal{S}_{\bm\zeta + \bm{\mu}^{\textnormal{R}}_{[s, t]}}(\rho_E),
    \]
    where ${\mu}^{\textnormal{R}}_{[s, t]}$ is defined in Eq.~\eqref{eq:mu_a_b}.
    \item[(b)] Suppose $\Omega$ is a system super-operator, then
    \[
    \Omega \mathcal{S}_{\bm\zeta}(\rho_E) \subseteq \norm{\Omega}_\diamond \mathcal{S}_{\bm\zeta}(\rho_E).
    \]
    \item[(c)] Suppose $f$ is either of the two functions: $f(\phi) = \norm{\phi}_1$ or $f(\phi) = \abs{\vecbra{\theta}\phi \rrangle}$ for some $\theta \in \mathcal{Q}(\rho_E)$, then
    \[
     f\big( \overrightarrow{\textnormal{W}}_{\rho_E}(\{B_{\alpha_j,\sigma_j}(\tau_j)\}_{j \in \{1:m\}}; \phi)\big) \leq \zeta_{\alpha_1}(\tau_1)\zeta_{\alpha_2}(\tau_2)\dots \zeta_{\alpha_m}(\tau_m)\bigg( \sup_{\phi \in \mathcal{S}_{\bm\zeta}(\rho_E)} f(\phi)\bigg).
    \]
    \end{enumerate}
    \end{lemma}

    \begin{proof} (a) Assume that $s < t$ and note that if $\phi \in \mathcal{S}_{\bm\zeta}(\rho_E)$, then there is a representation of $\phi$ with time-edges $\{(s_i, t_i)\}_{i \in [1:n]}$, then $\mathcal{U}_X(t, s) \vecket{\phi}$ has a representation with time-edges $\{(s_i', t_i')\}_{i \in [1:(n + 1)]} =  \{(s_i, t_i)\}_{i \in [1:n]} \cup \{(s, t)\}$. We then obtain that $\forall \tau,  \alpha$
    \begin{align*}
    \sum_{i = 1}^{n + 1} \sum_{\alpha', \sigma'} \smallabs{\mu^{\textnormal{R}}_{\alpha, [s_i', t_i']}(\tau)} = \bigg(\sum_{i = 1}^{n} \smallabs{\mu^{\textnormal{R}}_{\alpha, [s_i, t_i]}(\tau)}\bigg) + \mu^{\textnormal{R}}_{\alpha, [s, t]}(\tau) \leq \zeta_\alpha(\tau) + \mu^{\textnormal{R}}_{\alpha, [s, t]}(\tau),
    \end{align*}
    which implies that $\mathcal{U}_\mathcal{B}(t, s) \vecket{\phi} \in \mathcal{S}_{\bm\zeta+\bm{\mu}_{[s,t]}^R}(\rho_E)$. A similar argument holds for $U_\mathcal{B}(s, t) \mathcal{S}_\zeta(\rho_E) U_\mathcal{B}(t, s)$.

    (b) follows directly from definition and sub-multiplicity of the diamond norm.

    (c) For notational convenience, let us define
    \[
    \Lambda_m(\tau_1, \tau_2 \dots \tau_m) = \sup_{\substack{\phi \in \mathcal{S}_{\bm\zeta}(\rho_E) \\ \{\alpha_i, \sigma_i\}_{i \in [1:m]}}} f\big( \overrightarrow{\textnormal{W}}_{\rho_E}(\{B_{\alpha_j,\sigma_j}(\tau_j)\}_{j \in [1:m]}; \vecket{\phi})\big).
    \]
    with $\Lambda_0 = \sup_{\phi \in \mathcal{S}_{\bm\zeta}(\rho_E)} f(\phi)$. Consider first $m = 1$. Suppose $\vecket{\phi} \in \mathcal{S}_{\bm\zeta}(\rho_E)$ has a representation
    \[
    \vecket{\phi} = \prod_{i = n}^1\Omega_i \mathcal{U}_{\mathcal{A}_i}(t_i, s_i) \vecket{\sigma_S, \rho_E},
    \]
    then using Eqs.~(\ref{eq:right_W_1}a) and (\ref{eq:right_W_1}b), we obtain that
    \[
    \overrightarrow{\text{W}}_{\rho_E}\big(B_{\alpha_1,\sigma_1}(\tau_1); \vecket{\phi})\big) = -i\sum_{\substack{j,\sigma' \\ \alpha'\in \mathcal{A}_j}}(-1)^{\sigma'} \int_{s_j}^{t_j} \K_{\alpha_1, \alpha',\sigma'}(\tau_1 -\tau')\vecket{\phi_{j,\alpha', \sigma'}(\tau')}d\tau',
    \]
    where
    \[
    \vecket{\phi_{j,\alpha, \sigma}(\tau)} = \bigg(\prod_{i =n}^{j + 1}\Omega_i \mathcal{U}_{\mathcal{A}_i}(t_i, s_i)\bigg) \Omega_j \mathcal{U}_{\mathcal{A}_j}(t_j, \tau) R_{\alpha, \sigma}\mathcal{U}_{\mathcal{A}_j}(\tau, s_j) \bigg(\prod_{i =j-1}^1 \Omega_i \mathcal{U}_{\mathcal{A}_i}(t_i,s_i)\bigg) \vecket{\sigma_S, \rho_E}.
    \]
    We first show that in $\vecket{\phi_{j, \alpha',\sigma'}(\tau)} \in \mathcal{S}_{\bm\zeta}(\rho_E)$: Observe that $\vecket{\phi_{j, \alpha',\sigma'}(\tau)} \in \mathcal{S}(\rho_E)$  with time-edges $\{(s_i', t_i')\}_{i \in [1:(n+1)]} = \{(s_1, t_1)\dots (s_{j - 1}, t_{j -1}), (s_j, \tau), (\tau, t_j), (s_{j +1}, t_{j +1})\dots (s_n, t_n) \}$. Furthermore, note that since $\tau$ is in the interval between $s_j$ and $t_j$, either $s_j \leq \tau \leq t_j$ if $s_j \leq t_j$ or $t_j \leq \tau \leq s_j$ if $t_j \leq s_j$. Therefore, both $\mu_{\alpha, [s_j, \tau]}(t)$ and $\mu_{\alpha, [\tau, t_j]}(t)$ have the same sign and consequently,
    \[
    \smallabs{\mu^{\textnormal{R}}_{\alpha, [s_j, \tau]}(t)} + \smallabs{\mu^{\textnormal{R}}_{\alpha, [\tau, t_j]}(t)} = \smallabs{\mu^{\textnormal{R}}_{\alpha, [s_j, \tau]}(t) + \mu^{\textnormal{R}}_{\alpha, [\tau, t_j]}(t)} = \smallabs{\mu^{\textnormal{R}}_{\alpha, [s_j, t_j]}(t)}.
    \]
    Using this, we then obtain
    \begin{align*}
      \sum_{i=1}^{n+1}\smallabs{\mu_{\alpha,[s_i', t_i']}^{\textnormal{R}}(t)}=&\bigg(\sum_{i=1}^{j-1}\smallabs{\mu_{\alpha,[s_i,t_i]}^{\textnormal{R}}(t)}+\sum_{i=j+1}^n \smallabs{\mu_{\alpha,[s_i,t_i]}^{\textnormal{R}}(t)}\bigg)+ \smallabs{\mu_{\alpha,[\tau,t_j]}^{\textnormal{R}}(t)}+\smallabs{\mu_{\alpha,[s_j,\tau]}^{\textnormal{R}}(t)},\\
      =& \bigg(\sum_{i=1}^{j-1}\mu_{\alpha,[s_i,t_i]}^{\textnormal{R}}(t)+\sum_{i=j+1}^n \mu_{\alpha,[s_i,t_i]}^{\textnormal{R}}(t)\bigg)+ \smallabs{\mu_{\alpha,[s_j,t_j]}^{\textnormal{R}}(t)},\\
      \leq& \zeta_\alpha(t),
    \end{align*}
    This establishes that $\vecket{\phi_{j, \alpha', \sigma'}(\tau)} \in \mathcal{S}_{\bm{\zeta}}(\rho_E)$. Therefore, we have from the triangle inequality (which is satisfied for $f$)
    \begin{align*}
    f\big (\wickright(B_{\alpha_1,\sigma_1}(\tau_1); \vecket{\phi}) \big) &\leq \sum_{\substack{j,\sigma'\\\alpha'\in\mathcal{A}_j}}{\int_{s_j}^{t_j} \smallabs{\K_{\alpha_1,\alpha',  \sigma'}(\tau_1 - \tau')}f({\phi_{j, \alpha',\sigma'}(\tau')}) d\tau'}, \nonumber\\
    &\leq \bigg(\sup_{\phi \in \mathcal{S}_{\bm\zeta}(\rho_E)} f(\phi)\bigg)\sum_{j = 1}^n \sum_{\alpha'\in\mathcal{A}_j, \sigma'}{\int_{s_j}^{t_j}\smallabs{\K_{\alpha_1,\alpha' \sigma'}(\tau_1 - \tau')}d\tau'}, \nonumber  \\
    &\leq \Lambda_0 \sum_{j = 1}^n \smallabs{\mu_{\alpha_1, [s_j, t_j]}(\tau_1)} \leq \Lambda_0 \zeta_{\alpha_1}(\tau_1).
    \end{align*}
    Next, we proceed similarly using Eq.~\ref{eq:right_W}: Recall that
    \begin{align*}
        \overrightarrow{\text{W}}_{\rho_E}(\{B_{\alpha_i,\sigma_i}(\tau_i)\}_{i\in[1:m]}; {\phi}) = -i\sum_{\substack{j,\sigma'\\\alpha'\in\mathcal{A}_j}}(-1)^{\sigma'} \int_{s_j}^{t_j} \K_{\alpha_m, \alpha', \sigma'}(\tau_m - \tau) \overrightarrow{\text{W}}_{\rho_E}(\{B_{\alpha_i,\sigma_i}(\tau_i)\}_{i \in [1:m - 1]};{\phi_{j, \alpha',\sigma'}(\tau)}) d\tau,
    \end{align*}
    where $\vecket{\phi_{j, \alpha, \sigma}(\tau)}$ is as defined above. Using the fact that $f$ satisfies a triangle-inequality,
    \begin{align*}
        f\big(\overrightarrow{\text{W}}_{\rho_E}(\{B_{\alpha_i,\sigma_i}(\tau_i)\}_{i\in[1:m]}; {\phi})\big) &\leq \sum_{\substack{j,\sigma'\\\alpha'\in\mathcal{A}_j}} \int_{s_j}^{t_j}\smallabs{\K_{\alpha_m, \alpha', \sigma'}(\tau_m - \tau)} f\big(\overrightarrow{\text{W}}_{\rho_E}(\{B_{\alpha_i,\sigma_i}(\tau_i)\}_{i\in[1:m-1]}; \phi_{j, \alpha', \sigma'}(\tau))\big) d\tau, \nonumber\\
        &\leq \Lambda_{m - 1}(\tau_1, \tau_2 \dots \tau_{m - 1}) \sum_{j = 1}^n\sum_{\alpha'\in\mathcal{A}_j, \sigma'}\int_{s_j}^{t_j}\smallabs{\K_{\alpha_m,\alpha', \sigma'}(\tau_m - \tau)}d\tau, \nonumber \\
        &\leq \Lambda_{m - 1}(\tau_1, \tau_2 \dots \tau_{m - 1})\sum_{j = 1}^n \smallabs{\mu_{\alpha_m, [s_j, t_j]}(\tau_m)} \leq \Lambda_{m-1}(\tau_1, \tau_2\dots \tau_{m-1}) \zeta_{\alpha_m}(\tau_m),
    \end{align*}
    and therefore $\Lambda_m(\tau_1, \tau_2 \dots \tau_m) \leq \zeta_{\alpha_m}(\tau_m) \Lambda_{m - 1}(\tau_1, \tau_2 \dots \tau_{m - 1})$. Recursing this bound, it follows that $\Lambda_m(\tau_1, \tau_2 \dots \tau_m) \leq \zeta_{\alpha_1}(\tau_1) \zeta_{\alpha_2}(\tau_2)\dots \zeta_{\alpha_m}(\tau_m) \Lambda_0$, which proves the lemma statement.
    \end{proof}

\textcolor{black}{\noindent\emph{Remark}. We note that lemma \ref{lemma:closure_right_wick}c can be physically interpreted as effectively bounding the action of a bath operator $B_{\alpha,\sigma}(\tau)$ on a state in the space $\mathcal{S}_{\bm{\zeta}}(\rho_E)$. Since we will always use Wick's theorem to handle the Bath super-operators $B_{\alpha,\sigma}(\tau)$ on a given state, the right Wick contraction can be understood to be the relevant way of applying the super-operator $B_{\alpha,\sigma}(\tau)$ on a given state (as opposed to multiplying with the super-operator $B_{\alpha,\sigma}(\tau)$). Then, lemma \ref{lemma:closure_right_wick}(c) essentially says that on the space $\mathcal{S}_{\bm\zeta}(\rho_E)$, the bath super-operator $B_{\alpha,\sigma}(\tau)$, despite being unbounded as an operator on the entire Hilbert space, is effectively bounded in the sense of Wick contractions. } \\

It is also convenient to define an operator space $\mathcal{Q}_{\bm\zeta}(\rho_E)$ that acts like the dual of $\mathcal{S}_{\bm\zeta}(\rho_E)$.
\begin{definition}
Given $\bm\zeta=\{\zeta_{\alpha,\sigma}\}$ with $\zeta_{\alpha,\sigma}: \mathbb{R} \to [0, \infty)$, an operator $\vecbra{\theta} \in \mathcal{Q}_{\bm\zeta}({\rho_E})$ if $\vecbra{\theta} \in \mathcal{Q}(\rho_E)$ [Def.~\ref{def:Q_rho}] and has a representation with time edges $\{(s_i, t_i)\}_{i \in [1:n]}$ such that
\[
\sum_{j=1}^n \sum_{ \alpha'} \abs{\int_{s_j}^{t_j} |\K_{\alpha',\alpha, \sigma}(\tau' - \tau)| d\tau'}=\sum_{j=1}^n\mu^{\textnormal{L}}_{\alpha,\sigma,[s_j,t_j]}(\tau) \leq \zeta_{\alpha,\sigma}(\tau)\ \forall \ \tau \in \mathbb{R}.
\]
\end{definition}

We can also obtain a lemma similar to lemma \ref{lemma:closure_right_wick}c  that will be useful in the calculations in the following subsections.
    \begin{lemma}\label{lemma:closure_left_wick}
        Suppose $f$ is either of the two functions: $f(\theta) = \norm{\theta}_1$ or $f(\theta) = \abs{\vecbra{\theta}\phi \rrangle}$ for some $\phi \in \mathcal{S}(\rho_E)$ , then
    \[
     f\big( \overleftarrow{\textnormal{W}}(\vecbra{\theta}; \{B_{\alpha_j,\sigma_j}(\tau_j)\}_{j \in \{1:m\}})\big) \leq 2^m \zeta_{\alpha_1,\sigma_1}(\tau_1)\zeta_{\alpha_2,\sigma_2}(\tau_2)\dots \zeta_{\alpha_m,\sigma_m}(\tau_m) \sup_{\theta \in \mathcal{Q}_\zeta(\rho_E)} f(\theta),
    \]
    for all $\theta \in \mathcal{Q}_{\bm\zeta}(\rho_E)$, $ \{\alpha_i \in \mathcal{A}\}_{i \in [1:m]},  \{\sigma_i \in \{l, r\}\}_{i \in [1:m]}, \{\tau_i \in \mathbb{R}\}_{i\in [1:m]}$.
    \end{lemma}
    \begin{proof}
        The proof of this lemma is the same as the proof of lemma \ref{lemma:closure_right_wick}c, which we refer the reader to. The factor $2$ comes from the summation of $\sigma'$ in the contraction $\mathrm{Tr}(B_{\alpha',\sigma'}(\tau')B_{\alpha,\sigma}(\tau)\rho_E)$.
    \end{proof}

\subsection{Defining the commutator bound}\label{sec:defining_commutator_bound}

Having defined an appropriate operator space and motivated by lemma \ref{lemma:delta_o_comm_bd}, we will use $\gamma_{\bm\zeta, \mathcal{B}}^{X, Y}(t, t')$ defined below.
\begin{definition}
    Given $\mathcal{B}\subseteq \mathcal{A}$ and $X, Y \subseteq \Lambda$ and $\bm{\zeta}=\{\zeta_\alpha\}$ with $\zeta_\alpha: \mathbb{R} \to [0, \infty)$, define
\begin{align}\label{eq:commutator_def}
\gamma_{\bm\zeta, \mathcal{B}}^{X, Y}(t, t') =\sup_{\substack{\phi \in \mathcal{S}_{\bm\zeta}(\rho_E) \\ \norm{O_X}, \norm{\mathcal{L}_Y}_\diamond\leq 1}} \abs{\textnormal{Tr}(\mathcal{L}_Y^\dagger\big(U_\mathcal{B}(t', t) O_X U_\mathcal{B}(t, t')\big) \phi)} =  \sup_{\substack{\phi \in \mathcal{S}_{\bm\zeta}(\rho_E) \\ \norm{O_X}, \norm{\mathcal{L}_Y}_\diamond\leq 1}}  \abs{\vecbra{O_X, I_E}\mathcal{U}_\mathcal{B}(t, t') \mathcal{L}_{Y} \vecket{\phi}},
\end{align}
where the supremum is taken over system operators $O_X$ supported on $X \subseteq \Lambda$ and system superoperators $\mathcal{L}_Y$, supported on $Y$, which additionally satisfies $\mathcal{L}_Y^\dagger(I) = 0$.
\end{definition}
A bound on $\gamma_{\bm\zeta, \mathcal{B}}^{X, Y}(t, t')$ directly bounds the right-contraction part $\Delta_{1,O_X}(t,t';l)$ in Lemma \ref{lemma:delta_o_comm_bd}, which we make precise in the next lemma.

\begin{lemma}\label{lemma:Delta_bd_gamma}
For any $t > t'$ and $l > 0$, $\Delta_{O_X}^{\textnormal{R}}$ defined in Lemma \ref{lemma:delta_o_comm_bd} satisfies
\begin{subequations}\label{eq:error_in_terms_commutator}
\begin{align}\label{eq:error_in_terms_commutator_1}
\Delta_{O_X}^{\textnormal{R}}(t, t'; l) \leq 2\norm{O_X}\sum_{\substack{\alpha \notin \mathcal{A}_{X[l]} \\ \mathcal{S}_{\alpha} \cap X_c[l] \neq \emptyset}}\left(1+\mu^{\textnormal{R}}_\alpha(0,t)\right) \int_{s = t'}^t  \gamma^{X, \mathcal{S}_{\alpha}}_{\bm{\mu}^{\textnormal{R}}_{[0, s]}, \mathcal{A}_{X[l]}}(t, s)  ds,
\end{align}
and $\Delta_{O_X}^{\textnormal{L}}$ also defined in Lemma \ref{lemma:delta_o_comm_bd} is given by
\begin{align}\label{eq:error_in_terms_commutator_2}
    \Delta^{\textnormal{L}}_{O_X}(t, t';l)\leq 2 \norm{O_X} \sum_{\alpha' \in \mathcal{A}_{X[l]}} \sum_{\alpha \notin \mathcal{A}_{X[l]}, \sigma} \int_{s' = t'}^t \int_{s = t'}^{s'}\smallabs{K_{\alpha', \alpha, \sigma}(s'- s)}\gamma^{X, S_{\alpha'}}_{\bm{\mu}^{\textnormal{R}}_{[0, s']}, \mathcal{A}_{X[l]}}(t, s') ds' ds.
\end{align}
\end{subequations}
\end{lemma}
\begin{proof}
Since $\vecket{\rho(s)}, \vecket{\rho_{\alpha', \sigma'}(s,s')}$ defined in lemma \ref{lemma:delta_o_comm_bd} are in the operator space $\mathcal{S}_{\bm{\mu}^{\textnormal{R}}_{[0, s]}}(\rho_E)$ and $\smallnorm{\mathcal{C}_{h_\alpha(s)}}_\diamond \leq 2 \smallnorm{h_\alpha(s)} \leq 2$,
\[
\smallabs{\vecbra{O_X, I_E} \mathcal{U}_{X[l]}(t, s) \mathcal{C}_{h_\alpha}\vecket{\rho(s)}}, \smallabs{\vecbra{O_X, I_E} \mathcal{U}_{X[l]}(t, s) \mathcal{C}_{R_\alpha}\vecket{\rho_{\alpha', \sigma'}(s, s')}} \leq 2 \norm{O_X}\gamma^{X, \mathcal{S}_{\alpha}}_{\bm\mu^{\textnormal{R}}_{[0, s]}, \mathcal{A}_{X[l]}}(t, s).
\]
Therefore, it follows that
\begin{align}
    \Delta^{\textnormal{R}}_{O_X}(t, t'; l) \leq &\sum_{\substack{\alpha \notin \mathcal{A}_{X[l]} \\ \mathcal{S}_{\alpha} \cap {X}_c[l] \neq \emptyset}} \bigg(2\norm{O_X}\int_{t'}^t \gamma^{X, S_\alpha}_{\bm\mu^{\textnormal{R}}_{[0, s]}, \mathcal{A}_{X[l]}}(t, s) ds +\nonumber\\
    &\qquad \qquad \qquad \qquad \qquad 2\norm{O_X}\int_{s = t'}^t \bigg(\sum_{\alpha',\sigma'} \int_{s' = 0}^s \smallabs{\K_{\alpha, \alpha',\sigma'}(s - s')}ds'\bigg) \gamma^{X, S_\alpha}_{\bm{\mu}^{\textnormal{R}}_{[0, s], \mathcal{A}_{X[l]}}}(t, s) ds\bigg), \nonumber \\
    &\leq 2\norm{O_X}\sum_{\substack{\alpha \notin \mathcal{A}_{X[l]} \\  \mathcal{S}_{\alpha} \cap {X}_c[l] \neq \emptyset}} \int_{t'}^t (1 + \mu_\alpha^{\textnormal{R}}(0, s))\gamma^{X, S_\alpha}_{\bm{\mu}^{\textnormal{R}}_{[0, s]}, \mathcal{A}_{X[l]}}(t, s) ds.
    \nonumber\\
\end{align}
Finally, using that $\mu_\alpha^{\textnormal{R}}(0, s) \leq \mu_\alpha^{\textnormal{R}}(0, t)$ since $s \leq t$, we obtain the bound on $\Delta^{\textnormal{R}}_{O_X}(t, t';l)$ in the lemma statement. The bound on $\Delta^{\textnormal{L}}_{O_X}(t, t'; l)$ follows from a similar analysis by noting that $\vecket{\tilde{\rho}_{\alpha, \sigma}(s, s')}$ defined in lemma \ref{lemma:delta_o_comm_bd} is in the operator space $\mathcal{S}_{\bm{\mu}^{\textnormal{R}}_{[0, s']}}(\rho_E)$.
\end{proof}

Next, we bring the reader's attention to two mathematical requirements that are needed to ensure that all the calculations in the following subsection are legitimate. The corresponding lemmas proving that these requirements are satisfied can be found in appendix \ref{appendix:math_considerations}. If the reader is willing to accept that these manipulations are legitimate and does not want to be encumbered by mathematical details of their proof, then they can take these to be true and skip ahead to the next subsection. These considerations arise uniquely for the problem studied in this paper, and do not arise while deriving Lieb Robinson bounds for finite-dimensional lattice models.

\emph{Consideration 1}. In the calculations in the following subsection, we will often use the fact that for a smooth and bounded function $g:\mathbb{R} \to \mathbb{C}$ integrals of the form
\[
\int_{s_0'}^{s_1'}g(s)\gamma^{X, Y}_{\bm\zeta + \beta_0 \bm{\mu}^{\textnormal{R}}_{[s_0, s]}, \mathcal{B}}(s_1, s) ds,
\]
where $s_0 < s_1$, $\bm{\zeta}=\{\zeta_\alpha\}$ with $\zeta_\alpha:\mathbb{R}\to [0, \infty)$ being locally bounded functions and $\beta_0 > 0$, is well defined. Note that this is already implicit in the bound provided in Eq.~\eqref{eq:error_in_terms_commutator}. Furthermore, we will also use the fact that for a smooth and integrable function $k:\mathbb{R} \to \mathbb{C}$, $h:\mathbb{R} \to \mathbb{C}$ defined by
\[
h(s) = \int_{s_0'}^{s_1'}k(s - s') \gamma^{X, Y}_{\bm\zeta + \beta_0 \bm{\mu}_{[s_0, s']}^R, \mathcal{B}}(s_1, s') ds',
\]
is well defined and integrable. Both of these considerations are addressed in lemma \ref{lemma:int_gamma} in appendix \ref{appendix:math_considerations}.

\emph{Consideration 2}. in the derivation of the Lieb Robinson bound, a key tool that we will use is \emph{lower Lipschitz boundedness} of  $\gamma^{X, Y}_{\bm\zeta + \beta_0 \bm{\mu}^{\textnormal{R}}_{[s_0, s]}, \mathcal{B}}(s_1, s)$ i.e., on perturbing $s \to s - \varepsilon$ in $\gamma^{X, Y}_{\bm\zeta + \beta_0 \bm{\mu}^{\textnormal{R}}_{[s_0, s]}, \mathcal{B}}(s_1, s)$, it does not increase significantly beyond its unperturbed value. More specifically, we would need
\begin{align}\label{eq:target_property_cont}
\gamma^{X, Y}_{\bm\zeta + \beta_0 \bm{\mu}^{\textnormal{R}}_{[s_0, s - \varepsilon]}, \mathcal{B}}(s_1, s - \varepsilon) \leq \gamma^{X, Y}_{\bm\zeta + \beta_0 \bm{\mu}^{\textnormal{R}}_{[s_0, s]}, \mathcal{B}}(s_1, s) + O(\varepsilon).
\end{align}
Had $\gamma^{X, Y}_{\bm\zeta + \beta_0 \bm{\mu}^{\textnormal{R}}_{[s_0, \cdot]}, \mathcal{B}}(s_1, \cdot)$ been a smooth function, this would simply have followed from its Taylor expansion. However, it is not necessary for $\gamma^{X, Y}_{\bm\zeta + \beta_0 \bm{\mu}^{\textnormal{R}}_{[s_0, \cdot]}, \mathcal{B}}(s_1, \cdot)$ to even be a continuous function (even though, as described above and to be proved below, it is integrable). In appendix \ref{appendix:math_considerations}, by proving the following lemma, we show that despite this non-smoothness, it still satisfies Eq.~\ref{eq:target_property_cont} which will turn out to be enough for deriving a Lieb Robinson bound.

\begin{lemma}\label{lemma:upper_bound_lemma_stat} For any $X, Y \subseteq \Lambda$, $s_0 < s < t$, $\beta_0 > 1$ and for $\varepsilon < \abs{s - s_0}$
\begin{align}\label{eq:upper_bound_lemma_stat}
\gamma^{X, Y}_{\bm\zeta + \beta_0 \bm{\mu}^{\textnormal{R}}_{[s_0, s - \varepsilon]}, \mathcal{B}}(t, s - \varepsilon) \leq \gamma^{X, Y}_{\bm\zeta + \beta_0 \bm{\mu}^{\textnormal{R}}_{[s_0, s]}, \mathcal{B}}(t, s ) + C \varepsilon
\end{align}for some finite number $C$ which depends only on $\abs{\mathcal{B}}$, $\smallnorm{\bm\mu^{\textnormal{L}}(-t, t)}_\infty, \smallnorm{\bm\mu^{\textnormal{R}}(-t, t)}_\infty$ and $\smallnorm{\bm\zeta_{[0, t]}}_\infty$.
\end{lemma}

\subsection{Deriving the Lieb-Robinson recursion}
To obtain a recursive inequality for the commutator bound for the models under consideration, as we did for the case of bounded lattice Hamiltonians (section \ref{sec:recap}), we need to estimate the difference between $\gamma^{X,Y}_{\bm\zeta,\mathcal B}(t,t'-\varepsilon)$ and $\gamma^{X,Y}_{\bm\zeta,\mathcal B}(t,t')$. At a high-level, the procedure is similar to that in the case of bounded lattice Hamiltonians: Consider $\vecbra{O_X, I_E} \mathcal{U}_\mathcal{B}(t, t' - \varepsilon)\mathcal{L}_Y \vecket{\phi}$ for some $\vecket{\phi} \in \mathcal{S}_{\bm\zeta}(\rho_E)$ and perform the following sets of approximations:
\begin{align*}
\vecbra{O_X, I_E} \mathcal{U}_\mathcal{B}(t, t' - \varepsilon)\mathcal{L}_Y \vecket{\phi} &= \vecbra{O_X, I_E} \mathcal{U}_\mathcal{B}(t, t')\mathcal{U}_\mathcal{B}(t', t' - \varepsilon)\mathcal{L}_Y \mathcal{U}_\mathcal{B}(t' - \varepsilon, t')\mathcal{U}_\mathcal{B}(t', t' - \varepsilon)\vecket{\phi} \nonumber \\
&\approx \vecbra{O_X, I_E}\mathcal{U}_{\mathcal{B}}(t, t')\mathcal{U}_{ \mathcal{B}\cap \mathcal{A}_Y }(t', t' - \varepsilon)\mathcal{L}_{Y}\mathcal{U}_{ \mathcal{B}\cap \mathcal{A}_Y }(t' - \varepsilon, t')\mathcal{U}_{\mathcal{B}}(t' - \varepsilon, t')\vecket{{\phi}} + O(\varepsilon^2) \nonumber \\
&\approx \vecbra{O_X, I_E}\mathcal{U}_{\mathcal{B}}(t, t')\bigg(I- i \int_{t' - \varepsilon}^{t'}\mathcal{H}_{\mathcal{B}\cap \mathcal{A}_Y}(s)ds\bigg)\mathcal{L}_{Y}\mathcal{U}_{\mathcal{B} \cap \mathcal{A}_Y }(t' - \varepsilon, t')\mathcal{U}_{\mathcal{B}}(t' - \varepsilon, t')\vecket{{\phi}} + O(\varepsilon^2).
\end{align*}
Now, clearly the state $\vecket{\tilde{\phi}}$ defined via
\[
\vecket{\tilde{\phi}} = \mathcal{U}_{\mathcal{A}_Y \cap \mathcal{B}}(t' - \varepsilon, t') \mathcal{U}_{\mathcal{B}}(t' - \varepsilon, t') \vecket{\phi} \in \mathcal{S}_{\bm{\zeta} + 2\bm{\mu}^{\textnormal{R}}_{[t' - \varepsilon, t']}}(\rho_E).
\]
The next lemma makes precise that the above approximations are valid to $O(\varepsilon^2)$. The subtlety that needs to be handled here is the unbounded system-environment interaction Hamiltonian, which is not present in the case of bounded lattice Hamiltonians.

\begin{lemma}\label{lemma:taylor_expansion_bound}
Suppose $\vecket{\phi} \in \mathcal{S}_{\bm\zeta}(\rho_E)$, $\varepsilon<t'$, $O_X$ is an observable supported in region $X$ and $\mathcal{L}_Y$ is a system superoperator supported on the region $Y$, then
    \[
    \abs{\vecbra{O_X, I_E} \mathcal{U}_\mathcal{B}(t, t' - \varepsilon)\mathcal{L}_Y \vecket{\phi} - \vecbra{O_X, I_E} \mathcal{U}_\mathcal{B}(t, t')\bigg(I - i\int_{t' - \varepsilon}^{t'}\mathcal{H}_{ \mathcal{B}\cap \mathcal{A}_Y }(s)ds\bigg)\mathcal{L}_Y \vecket{\tilde{\phi}}} \leq C_0\norm{\mathcal{L}_Y}_\diamond \norm{O_X} \varepsilon^2,
    \]
    for some $\vecket{\tilde{\phi}} \in \mathcal{S}_{\bm\zeta + 2\bm\mu^{\textnormal{R}}_{[t' - \varepsilon, t']}}(\rho_E)$ and for $C_0$ that depends only on the support of the region $Y$, $\norminf{\bm\zeta_{[0,t]}}$, $\norminf{\bm{\mu}^{\textnormal{R}}(-t,t)}$, $\norminf{\bm{\mu}^{\textnormal{L}}(-t,t)}$  and $\norminf{K}$.
\end{lemma}
\noindent The proof of this lemma is a relatively straightforward but tedious application of a second-order Taylor expansion and is provided in appendix \ref{appendix:taylor_expansion}.

We now combine this Taylor expansion estimate with this independent-bath support property to derive one of the key results that we will use repeatedly in proving Lieb-Robinson bound: An inequality relating $\gamma^{X, Y}_{\bm{\zeta}, \mathcal{B}}(t, t' - \varepsilon)$ to $\gamma^{X, Y}_{\bm{\zeta}, \mathcal{}}(t, t')$.
\begin{lemma}[Differential upper bound]\label{lemma:small_gamma_recursion}
For sufficiently small $\varepsilon$ and $t' < t$,
\begin{align}
    &\gamma_{\bm\zeta, \mathcal{B}}^{X, Y}(t, t'-\varepsilon) - \gamma_{\bm\zeta + 2\bm\mu^{\textnormal{R}}_{[t' - \varepsilon, t']}, \mathcal{B}}^{X, Y}(t, t') \leq \nonumber\\
    &\qquad  \varepsilon \sum_{\alpha: \mathcal{S}_{\alpha} \cap Y \neq \emptyset} \bigg( 2(1 + 2 \norm{\bm\zeta_{[t'-\varepsilon,t']}}_\infty)\gamma^{X, \mathcal{S}_{\alpha}}_{\bm\zeta + 2\bm\mu^{\textnormal{R}}_{[t'-\varepsilon, t']}, \mathcal{B}}(t, t') + 2 \sum_{\alpha' 
    \in \mathcal{B}, \sigma} \int_{ t'}^t \abs{K_{\alpha', \alpha, \sigma}(s'-t')} \gamma^{X, \mathcal{S}_{\alpha'}}_{\bm\zeta + 2\bm\mu^{\textnormal{R}}_{[t' - \varepsilon, s']}, \mathcal{B}}(t, s') ds'\bigg) + C_1 \varepsilon^2,
\end{align}
where $C_1$ depends only on $\abs{\mathcal{B}}$, $\norm{\bm\zeta_{[0,t]}}_\infty, \norminf{\bm\mu^{\textnormal{R}}(-t,t)}, \norminf{\bm\mu^{\textnormal{L}}(-t,t)}$ and $\norm{K_{[0, t]}}_\infty$.
\end{lemma}
\begin{proof}
From lemma \ref{lemma:taylor_expansion_bound}, we obtain that
\begin{align*}
    \gamma^{X, Y}_{\bm\zeta, \mathcal{B}}(t, t' - \varepsilon) & \leq \sup_{\substack{\phi \in \mathcal{S}_{\bm\zeta}(\rho_E) \\ \norm{O_X} \leq 1, \norm{\mathcal{L}_Y}_\diamond \leq 1}}\abs{\vecbra{O_X, I_E}\mathcal{U}_\mathcal{B}(t, t') \mathcal{L}_Y \vecket{\tilde{\phi}}} +   \sup_{\substack{\phi \in \mathcal{S}_{\bm\zeta}(\rho_E) \\\norm{O_X} \leq 1, \norm{\mathcal{L}_Y}_\diamond \leq 1}} \abs{\vecbra{O_X, I_E} \mathcal{U}_\mathcal{B}(t, t') \bigg(\int_{t' - \varepsilon}^{t'}\mathcal{H}_{\mathcal{A}_Y\cap \mathcal{B}}(s) ds\bigg) \mathcal{L}_Y\vecket{\tilde{\phi}}} + C_0 \varepsilon^2, \nonumber \\
    &\numleq{1} \gamma^{X, Y}_{\bm\zeta + 2\bm\mu^{\textnormal{R}}_{[t' - \varepsilon, t']}, \mathcal{B}}(t, t') + T_\varepsilon(t, t') + C_0 \varepsilon^2,
\end{align*}
where $C_0$ is as defined in lemma \ref{lemma:taylor_expansion_bound} and depends only $\abs{\mathcal{B}}$, $\norminf{\bm\zeta_{[0,t]}}$, $\norminf{\bm{\mu}^{\textnormal{R}}(-t,t)}$, $\norminf{\bm{\mu}^{\textnormal{L}}(-t,t)}$  and $\norminf{K}$. In the above, in (1) we have used the fact that $\vecket{\tilde{\phi}} \in \mathcal{S}_{\bm\zeta + 2\bm\mu^{\textnormal{R}}_{[t' - \varepsilon, t']}}(\rho_E)$ and defined
\begin{align*}
    T_\varepsilon(t, t') = \sup_{\substack{\tilde{\phi} \in \mathcal{S}_{\bm\zeta + 2\bm\mu^{\textnormal{R}}_{[t'-\varepsilon, t']}}(\rho_E) \\ \norm{O_X}\leq 1, \norm{\mathcal{L}_Y}_\diamond \leq 1}} \abs{\vecbra{O_X, I_E} \mathcal{U}_\mathcal{B}(t, t') \bigg(\int_{t' - \varepsilon}^{t'}\mathcal{H}_{\mathcal{A}_Y \cap \mathcal{B}}(s) ds\bigg)\mathcal{L}_Y\vecket{\tilde{\phi}}}.
\end{align*}
Next, we upperbound $T_\varepsilon(t, t')$. We start by noting that
\begin{align}\label{eq:T_1_1}
&\abs{\vecbra{O_X, I_E}\mathcal{U}_\mathcal{B}(t, t')\bigg(\int_{t'-\varepsilon}^{t'}\mathcal{H}_{\mathcal{A}_Y \cap \mathcal{B}}(s)ds\bigg)\mathcal{L}_Y\vecket{\tilde{\phi}}} \leq \sum_{\alpha: S_\alpha \cap Y \neq \emptyset}\int_{t' - \varepsilon}^{t'} \abs{ \vecbra{O_X, I_E} \mathcal{U}_\mathcal{B}(t, t') \mathcal{C}_{h_\alpha(s)} \mathcal{L}_Y\vecket{\tilde{\phi}}} ds +\nonumber \\
&\qquad \qquad \qquad \qquad \qquad \qquad \qquad \qquad \sum_{\alpha: S_\alpha \cap Y \neq \emptyset}\abs{(-1)^\sigma \int_{t' - \varepsilon}^{t'} \vecbra{O_X, I_E} \mathcal{U}_\mathcal{B}(t, t')B_{\alpha,\sigma}(s) R_{\alpha, \sigma}(s) \mathcal{L}_{Y} \vecket{\tilde{\phi}} ds}\bigg).
\end{align}
We now analyze the right hand side of Eq.~\ref{eq:T_1_1} term by term. Note first that,
\begin{align}\label{eq:T_1_2}
\int_{t' - \varepsilon}^{t'}\abs{\vecbra{O_X, I_E} \mathcal{U}_\mathcal{B}(t, t') \mathcal{C}_{h_\alpha(s)} \mathcal{L}_{Y}\vecket{\tilde{\phi}}} ds &\leq 2\varepsilon \norm{\mathcal{L}_Y}_\diamond \gamma^{X, \mathcal{S}_{\alpha}}_{\bm\zeta + 2\bm\mu^{\textnormal{R}}_{[t' - \varepsilon,t']}, \mathcal{B}}(t, t'),
\end{align}
since by lemma \ref{lemma:closure_right_wick}b, $\mathcal{L}_{Y}\vecket{\tilde{\phi}} \in \norm{\mathcal{L}_Y}_\diamond \mathcal{S}_{\bm\zeta + 2\bm\mu^{\textnormal{R}}_{[t'-
\varepsilon, t']}}(\rho_E)$ and $\norm{\mathcal{C}_{h_\alpha}}_\diamond \leq 2\norm{h_\alpha}\leq 2$. Next, we use the Wick's theorem to obtain
\begin{align}\label{eq:manipulation_term_1}
&\abs{\int_{t' - \varepsilon}^{t'} \vecbra{O_X, I_E} \mathcal{U}_\mathcal{B}(t, t') B_{\alpha,\sigma}(s) R_{\alpha, \sigma}(s)\mathcal{L}_{Y} \vecket{\tilde{\phi}}ds} \nonumber\\
&\numleq{1} \abs{\sum_{\alpha'\in \mathcal{B}, \sigma}\int_{s = t' -\varepsilon}^{t'}\int_{s' = t'}^t K_{\alpha',\alpha,\sigma}(s' - s) \vecbra{O_X, I_E}\mathcal{U}_\mathcal{B}(t, s') \mathcal{C}_{R_{\alpha'}(s')} \mathcal{U}_\mathcal{B}(s', t')R_{\alpha, \sigma}(s) \mathcal{L}_{Y} \vecket{\tilde{\phi}} ds' ds} + \nonumber\\
& \qquad \qquad \qquad \qquad \qquad \qquad \qquad \qquad \qquad \qquad \qquad   \abs{ \int_{t' - \varepsilon}^{t'} \vecbra{O_X, I_E} \mathcal{U}_\mathcal{B}(t, t')\mathcal{C}_{R_{\alpha}(s)} \mathcal{L}_{Y}\overrightarrow{\text{W}}(B_{\alpha,\sigma}(s), \vecket{\tilde{\phi}})ds} \nonumber \\
&  \numleq{2} 2\norm{\mathcal{L}_Y}_\diamond \norm{O_X}\bigg(\bigg(\sum_{\alpha'\in \mathcal{B}, \sigma}\int_{s = t' - \varepsilon}^{t'} \int_{s' = t'}^t \abs{K_{\alpha',\alpha,\sigma}(s'-s)} \gamma^{X, S_{\alpha'}}_{\bm\zeta + 2\bm\mu^{\textnormal{R}}_{[t'-\varepsilon,s']}, \mathcal{B}}(t, s')ds' ds\bigg) + \nonumber\\
&\qquad\qquad\qquad\qquad\qquad\qquad\qquad\qquad\qquad\varepsilon (\norminf{\bm\zeta_{[t'-\varepsilon,t']}} + 2\sup_{s'\in[t'-\varepsilon,t']}\norminf{\bm\mu^{\textnormal{R}}_{[t' - \varepsilon, t' ]}(s')}) \gamma_{\bm\zeta + 2\bm\mu^{\textnormal{R}}_{[t' - \varepsilon,t']}, \mathcal{B}}^{X, \mathcal{S}_{\alpha}}(t, t')\bigg) \nonumber\noindent\\
&  \numleq{3} \norm{\mathcal{L}_Y}_\diamond \norm{O_X}\bigg(\bigg(\sum_{\alpha'\in \mathcal{B}, \sigma} \int_{s = t' - \varepsilon}^{t'} \int_{s' = t'}^t \abs{K_{\alpha',\alpha,\sigma}(s'-s)} \gamma^{X, \mathcal{S}_{\alpha'}}_{\bm\zeta + 2\bm\mu^{\textnormal{R}}_{[t' - \varepsilon, s']}, \mathcal{B}}(t, s') ds' ds\bigg) + 2 \varepsilon  \norm{\bm\zeta_{[t'-\varepsilon,t']}}_\infty \gamma^{X, \mathcal{S}_{\alpha}}_{\bm\zeta + 2\bm\mu^{\textnormal{R}}_{[t'-\varepsilon, t']}, \mathcal{B}}(t, t') + \nonumber \\
& \qquad \qquad \qquad \qquad \qquad \qquad \qquad \qquad \qquad \qquad \qquad \qquad \qquad \qquad \qquad \qquad \qquad 8\varepsilon^2 \abs{\mathcal{B}}\norm{K_{[0, t]}}_\infty\bigg),
\end{align}
where in (1) we have used the fact that $\wickright(B_{\alpha,\sigma}(s), \vecket{\tilde{\phi}})$ is independent of $\sigma$ to set $(-1)^\sigma R_{\alpha, \sigma}(s) \mathcal{L}_Y \wickright(B_{\alpha,\sigma}(s), \vecket{\tilde{\phi}}) = \mathcal{C}_{R_\alpha(s)} \mathcal{L}_Y \wickright(B_{\alpha,l}(s), \vecket{\tilde{\phi}})$. In (2) we have used the closure property in lemma \ref{lemma:closure_right_wick}c to upper bound
\[
\abs{\vecbra{O_X, I_E} \mathcal{U}_\mathcal{B}(t, t')\mathcal{C}_{X_{\alpha}} \mathcal{L}_{Y}\overrightarrow{\text{W}}(B_{\alpha,l}(s), \vecket{\tilde{\phi}})} \leq (\smallnorm{\bm\zeta(s)}_\infty + 2\smallnorminf{\bm\mu^{\textnormal{R}}_{[t' - \varepsilon, t']}(s)}) \abs{\vecbra{O_X, I_E}\mathcal{U}_\mathcal{B}(t, t')C_{X_{\alpha, \sigma}} \mathcal{L}_{Y} \vecket{\tilde{\phi}}},
\]
and to upper bound $\small{\vecbra{O_X, I_E}\mathcal{U}_\mathcal{B}(t, s') \mathcal{C}_{R_{\alpha'}} \mathcal{U}_\mathcal{B}(s', t')X_{\alpha, \sigma} \mathcal{L}_{Y} \vecket{\tilde{\phi}}}$, we have used lemma \ref{lemma:closure_right_wick}a to identify that $\mathcal{U}_\mathcal{B}(s', t') X_{\alpha, \sigma} \mathcal{L}_Y \vecket{\tilde{\phi}} \in \mathcal{S}_{\bm\zeta + 2\bm\mu^{\textnormal{R}}_{[t' - \varepsilon, t']} +\bm\mu^{\textnormal{R}}_{[t', s']} } (\rho_E)\subseteq \mathcal{S}_{\bm\zeta + 2\bm\mu^{\textnormal{R}}_{[t' - \varepsilon, s']}  }(\rho_E)$. In (3) we have used the fact that $\smallnorm{\bm\mu^{\textnormal{R}}_{[t'-\varepsilon, t']}(s')}_\infty \leq 2\varepsilon \abs{\mathcal{B}}\norm{K_{[0, t]}}_\infty $. Note also that in (2), we have implicitly used the result stated in consideration 1 in the previous section (more specifically, lemma \ref{lemma:int_gamma}b from the appendix \ref{appendix:math_considerations}) to insure that the double integral is well defined.

Next, we further manipulate the first term in Eq.~\ref{eq:manipulation_term_1} by first performing a change of variables $s \to s$ and $s' \to s + \tau$ to obtain
\begin{align}
    &\int_{s = t' - \varepsilon}^{t'} \int_{s' = t'}^t \abs{K_{\alpha',\alpha,\sigma}(s' - s)} \gamma_{\bm\zeta + 2\bm\mu^{\textnormal{R}}_{[t' - \varepsilon, s']}, \mathcal{B}}^{X, \mathcal{S}_{\alpha'}}(t, s') ds' ds = \int_{s = t' - \varepsilon}^{t'} \int_{\tau = t' - s}^{t - s} \abs{K_{\alpha',\alpha,\sigma}(\tau)} \gamma_{\bm\zeta + 2\bm\mu^{\textnormal{R}}_{[t' - \varepsilon, s + \tau]}, \mathcal{B}}^{X, \mathcal{S}_{\alpha'}}(t, s + \tau)ds d\tau.
\end{align}
Next, since the integrand is non-negative, we upper bound the integration region $t' - \varepsilon \leq s \leq t', t' - s \leq \tau \leq t - s$ as depicted in Fig.~\ref{fig:integration_region} and obtain
\begin{align*}
    &\int_{t' - \varepsilon}^{t'} \int_{t'}^t \abs{K_{\alpha',\alpha,\sigma}(s' - s)} \gamma_{\bm\zeta + 2\bm\mu^{\textnormal{R}}_{[t' - \varepsilon, s']}, \mathcal{B}}^{X, \mathcal{S}_{\alpha'}}(t, s') ds' ds  \nonumber \\
    &\qquad \leq \int_{\tau = 0}^{t - t'} \int_{s = t' - \varepsilon}^{t'}\abs{K_{\alpha',\alpha,\sigma}(\tau)} \gamma_{\bm\zeta + 2\bm\mu^{\textnormal{R}}_{[t' - \varepsilon, s + \tau]}, \mathcal{B}}^{X, \mathcal{S}_{\alpha'}}(t, s + \tau)ds d\tau +  \int_{\tau = t - t'}^{t - t' + \varepsilon} \int_{s = t'-\varepsilon}^{t - \tau}  \abs{K_{\alpha',\alpha,\sigma}(\tau)} \gamma_{\bm\zeta + 2\bm\mu^{\textnormal{R}}_{[t' - \varepsilon, s + \tau]}, \mathcal{B}}^{X, \mathcal{S}_{\alpha'}}(t, s + \tau) ds d\tau.
\end{align*}
\begin{figure}[t]
    \centering
    \includegraphics[width=0.9\textwidth]{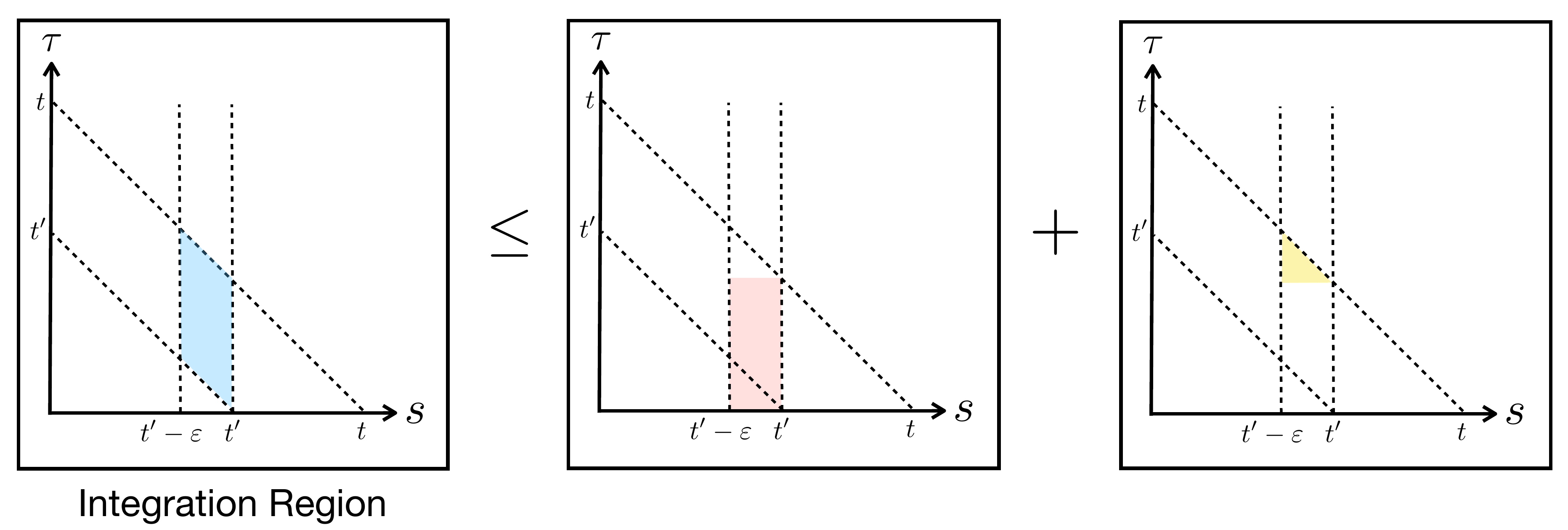}
    \caption{The integration region for the double integral $t' - \varepsilon \leq s \leq t', t' - s \leq \tau \leq t - s$ is shown in blue. The integral of a non-negative function over this region is upper bounded by the sum of integrals of this function in the two regions shown on the right.}
    \label{fig:integration_region}
\end{figure}
Now, using lemma \ref{lemma:upper_bound_lemma_stat} with $\beta_0 = 2$, we obtain that
\begin{align*}
    &\sum_{\alpha'}\int_{\tau = 0}^{t - t'} \int_{s = t' - \varepsilon}^{t'}\abs{K_{\alpha',\alpha,\sigma}(\tau)} \gamma_{\bm\zeta + 2\bm\mu^{\textnormal{R}}_{[t' - \varepsilon, s + \tau]}, \mathcal{B}}^{X, \mathcal{S}_{\alpha'}}(t, s + \tau)ds d\tau \nonumber\\
    &\qquad \qquad \leq \varepsilon\sum_{\alpha'}\int_{\tau = 0}^{t - t'}\abs{K_{\alpha',\alpha,\sigma}(\tau)} \gamma^{X, \mathcal{S}_{\alpha'}}_{\bm\zeta + 2\bm\mu^{\textnormal{R}}_{[t' - \varepsilon,  t'+\tau]}, \mathcal{B}}(t, t' + \tau)d\tau + C\varepsilon^2 \norminf{\bm\mu^{\textnormal{L}}(0,t)},
\end{align*}
where $C$ is the constant defined in lemma \ref{lemma:upper_bound_lemma_stat} that depends only on $\abs{\mathcal{B}}, \smallnorm{\bm\mu^{\textnormal{L}}(0,t)}_\infty, \smallnorm{\bm\zeta_{[0,t]}}_\infty$ and $\smallnorm{\bm\mu^{\textnormal{R}}(-t,t)}_\infty$.
Furthermore, since $\gamma^{X, Y}_{\zeta, \mathcal{B}}(t, t') \leq 1$, we also have that
\begin{align}\label{eq:SM_small_contribution_from_second_term}
    \sum_{\alpha'}\int_{\tau = t - t'}^{t - t' + \varepsilon} \int_{s = t' -\varepsilon}^{t - \tau} \abs{K_{\alpha',\alpha,\sigma}(\tau)} \gamma^{X, \mathcal{S}_{\alpha'}}_{\bm\zeta + 2\bm\mu^{\textnormal{R}}_{[t' - \varepsilon, s + \tau]}, \mathcal{B}}(t, s + \tau) ds \ d\tau \leq \varepsilon\int_{t-t'}^{t-t'+\varepsilon}\sum_{\alpha'}\abs{K_{\alpha',\alpha,\sigma}(\tau)}d\tau\leq \varepsilon^2\abs{\mathcal{B}}\smallnorm{K_{[0, t]}}_\infty.
\end{align}
Consequently, we obtain
\begin{align}\label{eq:T_1_3}
    &\sum_{\alpha'}\int_{t' - \varepsilon}^{t'} \int_{t'}^t \abs{K_{\alpha',\alpha,\sigma}(s' - s)} \gamma_{\bm\zeta + 2\bm\mu^{\textnormal{R}}_{[t' - \varepsilon, s']}, \mathcal{B}}^{X, \mathcal{S}_{\alpha'}}(t, s') ds' ds \nonumber\\
    &\quad\quad\quad\leq \varepsilon \sum_{\alpha'}\int_{\tau = 0}^{t - t'}\abs{K_{\alpha',\alpha,\sigma}(\tau)} \gamma_{\bm\zeta + 2\bm\mu^{\textnormal{R}}_{[t' - \varepsilon, t'+\tau]}, \mathcal{B}}^{X, \mathcal{S}_{\alpha'}}(t, t' + \tau) d\tau + (C\smallnorm{\bm\mu^{\textnormal{L}}(0,t)}_\infty+\abs{\mathcal{B}}\smallnorm{K_{[0, t]}}_\infty)\varepsilon^2.
\end{align}
From Eqs.~\ref{eq:T_1_1}, \ref{eq:T_1_2}, \ref{eq:manipulation_term_1} and \ref{eq:T_1_3}, we then obtain that
\[
T_\varepsilon^{(1)}(t, t') \leq \sum_{\alpha: \mathcal{S}_{\alpha} \cap Y \neq \emptyset} \bigg( 2(1 + 2 \norm{\bm\zeta_{[t'-\varepsilon,t']}}_\infty)\gamma^{X, \mathcal{S}_{\alpha'}}_{\bm\zeta + 2\bm\mu^{\textnormal{R}}_{[t'-\varepsilon, t']}, \mathcal{B}}(t, t') + 2\sum_{\alpha',\sigma} \int_{ t'}^t \abs{K_{\alpha',\alpha,\sigma}(s'-t')} \gamma^{X, \mathcal{S}_{\alpha'}}_{\bm\zeta + 2\bm\mu^{\textnormal{R}}_{[t' - \varepsilon, s']}, \mathcal{B}}(t, s') ds'\bigg) + c_0\varepsilon^2,
\]
where $c_0 = 2C\smallnorm{\bm\mu^{\textnormal{L}}(0,t)}_\infty+10\abs{\mathcal{B}}\smallnorm{K_{[0, t]}}_\infty$. This proves the lemma statement.
\end{proof}

\section{Lieb-Robinson bounds for independent-bath model}

\begingroup\color{black}
We now apply the common framework to the case where the environment is composed of independent baths. Recall from the main text that, in this model, the index $\alpha$ labeling the system-environment interaction operators has the form $\alpha = (i, \nu)$ where $\nu \in \{1, 2 \dots L\}$. Thus, the environment operators are indexed as $B_\alpha(t) \cong B_{i}^{(\nu)}(t)$, where the index $i$ labels the bath corresponding to $\alpha$ and $\nu$ labels different possible bath operators within the $i^\text{th}$ bath. The independence of the different baths is equivalent to assuming that
\[
[B_i^{(\nu)}(t), B_{i'}^{(\nu')}(t')] = 0 \text{ for } i \neq i',
\]
and assuming that the initial environment state is a product state on the different baths i.e. $\rho_E = \otimes_i \rho_{E, i}$ where $\rho_{E, i}$ is a state of the $i^\text{th}$ bath. Therefore, the kernels are uncorrelated between different baths i.e.,
\[
K_{\alpha, \alpha'}(t) \cong K_{i, \nu, \nu'}(t)\delta_{i, i'}.
\]
Finally, we will also assume that, given $\alpha = (i, \nu)$, the support $\mathcal{S}_\alpha $ depends only $i$ i.e. $\mathcal{S}_\alpha \cong \mathcal{S}_i$ or, in words, the $i^\text{th}$ bath only interacts with qubits in the geometrically local region $S_i$.

For the rest of this section, it will be convenient to use $m(\tau)$ defined in the main text which we recap below: Suppose the kernels $K_{i, \nu, \nu'}(t)$ have representation given by
\[
K_{i, \nu, \nu'}(t)= K_{i, \nu, \nu';c}(t) + \sum_j k_{i, \nu, \nu'}[j]\delta(t - T_j),
\]
where $K_{{i, \nu, \nu'}; c}$ is a continuous function and $k_{i, \nu, \nu'}\in \mathbb{C}$. Then, for $\tau > 0$, define $m(\tau)$ via
\begin{align}\label{eq:def_m_delta}
m(\tau) = \sup_{i, \nu, \nu'}\text{TV}(K_{i, \nu', \nu}; [-\tau, \tau]) = \sup_{i, \nu, \nu'} \bigg[\int_{-\tau}^{\tau}\smallabs{K_{{i, \nu, \nu'}; c}(\tau)}d\tau + \sum_{j: T_j\in[-\tau, \tau ]}\smallabs{k_{i, \nu, \nu'}[j]} \bigg]
\end{align}
We also note that $m(\tau)$ is a function which is right-continuous i.e.,
\begin{align}
\lim_{\delta \to 0^+} m(\tau + \delta) = \sup_{i, \nu, \nu'} \lim_{\delta \to 0^+ }\bigg[\int_{-(\tau + \delta)}^{\tau+ \delta}\smallabs{K_{i, \nu, \nu'; c}(\tau)}d\tau + \sum_{j: T_j\in[-(\tau +\delta), \tau + \delta]}\smallabs{k_{i, \nu, \nu'; c}[j]} \bigg] = m(\tau),
\end{align}
where we can swap the limit $\delta \to 0$ and the supremum over $i, \nu, \nu'$ since the supremum is only over a finite set.

Importantly, we remind the reader that the framework developed in the previous section (Section \ref{sec:common_framework}), including the definitions of $\bm{\mu}^{\text{L}/\text{R}}(a, b), \gamma^{X, Y}_{\zeta, \mathcal{B}}(t, t')$ and associated lemmas, are all for a $\delta-$mollified model i.e., for a model with the kernels
\[
K_{i, \nu, \nu'}^\delta = (\eta_\delta \star \eta_\delta)\star K_{i, \nu, \nu'},
\]
with $\eta_\delta$ being the mollifier defined in Section \ref{sec:regularization}. This has to be accounted in the analysis of this section towards proving proposition 1, so we will re-instate indexing by $\delta$ in the following proofs.

First, we develop some useful upper bounds on $\smallnorm{\mu^{\text{L}}(0, t)}_\infty$, $\smallnorm{\mu^{\text{R}}(0, t)}_\infty$ by specializing the general definitions provided in Section \ref{sec:common_framework}. 
\begin{lemma}\label{lemma:mu_upper_bounds}
For small enough $\delta$ and $\tau \geq 0$, $\mu_\alpha^\textnormal{R}(0, \tau), \mu_\alpha^\textnormal{R}(-\tau, \tau)$, $\mu_{\alpha, \sigma}^\textnormal{L}(0, \tau)$ and $\mu_{\alpha, \sigma}^\textnormal{L}(-\tau, \tau)$ as defined in Eq.~\eqref{eq:mu} and Eq.~\eqref{eq:kappa} with the $\delta-$mollified kernels satisfy
\[
\mu_\alpha^\textnormal{R}(-\tau, \tau) \leq \mu_\alpha^\textnormal{R}(-\tau, \tau) \leq 2Lm(\tau + \delta) \quad \text{and}\quad \mu_{\alpha, \sigma}^\textnormal{L}(-\tau, \tau) \leq \mu_{\alpha, \sigma}^\textnormal{L}(-\tau, \tau) \leq L m(\tau + \delta).
\]
\end{lemma}
\begin{proof}
Suppose $\alpha = (i, \nu)$, then from Eq.~\eqref{eq:mu} and lemma~\ref{lemma:mollified_kernel}, it follows that
\begin{align}
\mu_{(i, \nu)}^{\textnormal{R}}(-\tau, \tau) \leq \sum_{\nu', \sigma'}\int_{-\tau}^\tau \smallabs{K^{\delta}_{i, \nu, \nu', \sigma'}(s)} ds\leq \sum_{ \nu', \sigma'} \textnormal{TV}(K_{i, \sigma'}^{(\nu, \nu')}, [-(\tau + \delta), (\tau + \delta)]) \leq 2L m(\tau + \delta).
\end{align}
Similarly,
\begin{align}
    \mu_{(i, \nu), \sigma}^{\textnormal{L}}(-\tau, \tau) \leq \sum_{\nu'}\int_{-\tau}^\tau \smallabs{K^{\delta}_{i, \nu', \nu, \sigma}(s)} ds \leq  \sum_{ \nu'} \textnormal{TV}(K_{i, \nu', \nu, \sigma}^{\delta}, [-(\tau +\delta), \tau + \delta]) \leq L m(\tau + \delta).
\end{align}
Finally, noting that $\mu_\alpha^{\text{R}}(0, \tau) \leq \mu_\alpha^{\text{R}}(-\tau, \tau)$ and $\mu_{\alpha, \sigma}^{\text{L}}(0, \tau) \leq \mu_{\alpha, \sigma}^{\text{R}}(-\tau, \tau)$, we obtain the bounds in the lemma statement.
\end{proof}

We now prove a Lieb-Robinson bound on $\gamma^{X, Y}_{\zeta, B}(t, t')$ by building on lemma \ref{lemma:small_gamma_recursion}. Again, the reader can be reminded that $\gamma^{X, Y}_{\zeta, B}(t, t')$ has also been defined for the $\delta-$mollified kernel (see section \ref{sec:common_framework}).
\endgroup

\begin{lemma}\label{lemma:gamma_lr}
    For $0\leq t' < t$,
    \[
    \gamma^{X, Y}_{\bm\zeta, \mathcal{B}}(t, t') \leq \abs{Y}\bigg(\exp\bigg(\frac{v_\textnormal{LR}(\bm\zeta, t)\abs{t - t'}}{a_0}\bigg) - 1\bigg) \exp\bigg(- \frac{d_{X, Y}}{a_0}\bigg),
    \]
    where
    \[
    v_\textnormal{LR}(\bm\zeta, t) = ea_0 \mathcal{Z}(L + 2L\norm{\bm\zeta_{[0,t]}}_\infty + 10L^2m(t + \delta)),
    \]
    with $\mathcal{Z}, a_0$ being defined in Eq.~\eqref{eq:constants_lattice} and $m_\delta(\tau)$ being defined in Eq.~\eqref{eq:def_m_delta}
\end{lemma}
\begin{proof}
    Given $\norminf{\bm\zeta_{[0,t]}}<\infty$, $0 < t' < t$ and for $\tau \in (0, t - t')$, consider $\Gamma_{X, Y}(\tau)$  defined by
    \[
    \Gamma_{X, Y}(\tau) = \gamma^{X, Y}_{\bm\zeta + 2\bm\mu^{\textnormal{R}}_{[t', t - \tau]}, \mathcal{B}}(t, t - \tau).
    \]
    From lemma \ref{lemma:small_gamma_recursion}, we then obtain that for sufficiently small $\varepsilon$
    \begin{align}\label{eq:recursion_Gamma}
        &\Gamma_{X, Y}(\tau + \varepsilon)-\Gamma_{X,Y}(\tau) \nonumber\\
        &\leq \varepsilon \sum_{\alpha : \mathcal{S}_{\alpha} \cap Y \neq \emptyset}\bigg(2\bigg(1 + 2\norminf{\bm\zeta_{[0,t]}}+4\sup_{s\in[t-\tau-\varepsilon,t-\tau]}\norminf{\bm\mu^{\textnormal{R}}_{[t',t-\tau-\varepsilon]}(s)} \bigg) \Gamma_{X, \mathcal{S}_{\alpha}}(\tau) + 2\sum_{\alpha',\sigma} \int_{t - \tau}^{t} \smallabs{K^\delta_{\alpha',\alpha,\sigma}(s' - t + \tau)}\Gamma_{X, \mathcal{S}_{\alpha'}}(t - s') ds'\bigg) + C_1 \varepsilon^2\nonumber \\
        &\leq \varepsilon \sum_{\alpha : \mathcal{S}_{\alpha} \cap Y \neq \emptyset}\bigg(2\bigg(1 + 2\norminf{\bm\zeta_{[0,t]}}+4\norminf{\bm\mu^{\textnormal{R}}(0,t)}\bigg) \Gamma_{X, \mathcal{S}_{\alpha}}(\tau) + 2 \underbrace{\sum_{\alpha',\sigma}\int_{0}^{\tau} \smallabs{K^\delta_{\alpha',\alpha,\sigma}(\tau - \tau')}\Gamma_{X, \mathcal{S}_{\alpha'}}(\tau') d\tau'}_{\Lambda_{X, \mathcal{S}_{\alpha}}(\tau)}\bigg) + C_1\varepsilon^2.
    \end{align}
    This inequality can be converted into its integral form. More explicitly, for a given $\tau$ we choose $\varepsilon = \tau / N$ and define $\tau_n = n\varepsilon$. Then, from Eq.~\ref{eq:recursion_Gamma}, we obtain
    \[
    \Gamma_{X, Y}(\tau_n) \leq \Gamma_{X, Y}(\tau_{n - 1}) + \varepsilon\bigg(\sum_{\alpha : \mathcal{S}_{\alpha} \cap Y \neq \emptyset} 2\bigg(1 + 2\norminf{\bm\zeta_{[0,t]}}+4\norminf{\bm\mu^{\textnormal{R}}(0,t)}\bigg)\Gamma_{X, \mathcal{S}_{\alpha}}(\tau_{n - 1}) + 2 \Lambda_{X, \mathcal{S}_{\alpha}}(\tau_{n - 1}) \bigg) + \frac{C_1 \tau^2}{N^2},
    \]
    and consequently
    \[
    \Gamma_{X, Y}(\tau_N) \leq \Gamma_{X, Y}(\tau_0) +  \sum_{\alpha : \mathcal{S}_{\alpha} \cap Y \neq \emptyset}\bigg[2\bigg(1 + 2\norminf{\bm\zeta_{[0,t]}}+4\norminf{\bm\mu^{\textnormal{R}}(0,t)}\bigg)\bigg(\frac{\tau}{N}\sum_{n = 0}^{N - 1}\Gamma_{X, \mathcal{S}_{\alpha}}(\tau_n)\bigg) + 2\bigg(\frac{\tau}{N}\sum_{n = 0}^{N - 1}\Lambda_{X, \mathcal{S}_{\alpha}}(\tau_n)\bigg)\bigg] + \frac{C_1 \tau^2}{N}.
    \]
    Taking the limt of $N\to \infty$ in this inequality, and the fact that from lemma \ref{lemma:int_gamma}a and b it follows that $\Gamma_{X, \mathcal{S}_{\alpha}}$ and $\Lambda_{X, \mathcal{S}_{\alpha}}$ are both integrable functions, we obtain that
    \begin{align}\label{eq:integral_ineq}
    \Gamma_{X, Y}(\tau) &\leq \Gamma_{X, Y}(0) + \sum_{\alpha : \mathcal{S}_{\alpha} \cap Y \neq \emptyset}\bigg[2\big(1 + 2\norminf{\bm\zeta_{[0,t]}}+4\norminf{\bm\mu^{\textnormal{R}}(0,t)}\big) \int_0^\tau \Gamma_{X, \mathcal{S}_{\alpha}}(\tau') d\tau' + 2 \int_0^\tau \Lambda_{X, \mathcal{S}_{\alpha}}(\tau') d\tau'\bigg] \nonumber\\
        &\numleq{1} \Gamma_{X, Y}(0) + L \sum_{i : \mathcal{S}_{i} \cap Y \neq \emptyset}\bigg[2\big(1 + 2\norminf{\bm\zeta_{[0,t]}}+4\norminf{\bm\mu^{\textnormal{R}}(0,t)}\big) \int_0^\tau \Gamma_{X, \mathcal{S}_{i}}(\tau') d\tau' + 2 \int_0^\tau \Lambda_{X, \mathcal{S}_{i}}(\tau') d\tau'\bigg],
    \end{align}
    where in (1) we have used the fact that, for $\alpha = (i, \nu)$, $\mathcal{S}_\alpha = \mathcal{S}_i$. Next, we observe that, if $\alpha = (i, \nu)$ and $\alpha' = (i', \nu')$ then $K_{\alpha', \alpha, \sigma}(s) = \delta_{i, i'} K_{i}^{(\nu, \nu')}(s)$. Together with the fact that $\mathcal{S}_\alpha$ depends only on $i$, we obtain that
    \begin{align}\label{eq:int_Lambda_manipulated}
        \int_0^\tau \Lambda_{X, \mathcal{S}_i} &= \int_0^\tau \Lambda_{X, \mathcal{S}_{\alpha}}(\tau') d\tau' \nonumber\\
        &= \sum_{\alpha',\sigma}\int_{\tau' = 0}^\tau \int_{\tau'' =0}^{ \tau'}\abs{K^\delta_{\alpha',\alpha,\sigma}(\tau' - \tau'')} \Gamma_{X, \mathcal{S}_{\alpha'}}(\tau'') d\tau'' d\tau' \nonumber\\
        &=\sum_{i',\nu', \sigma}\int_{\tau' = 0}^\tau \int_{\tau'' =0}^{ \tau'}\abs{K_{i,\nu, \nu',\sigma}^{\delta}(\tau' - \tau'')} \Gamma_{X, \mathcal{S}_{i'}}(\tau'')\delta_{i, i'} d\tau'' d\tau' \nonumber \\
        &= \int_{\tau'' = 0}^{\tau}\bigg[\sum_{\nu',\sigma}\int_{\tau' = \tau''}^{\tau}\abs{K_{i,\nu, \nu',\sigma}^{\delta}(\tau' - \tau'')}\bigg] \Gamma_{X, \mathcal{S}_{i}}(\tau'') d\tau'' \nonumber\\
        &\leq 2\norminf{\bm\mu^{\textnormal{L}}(0,t)} \int_0^\tau \Gamma_{X, \mathcal{S}_{i}}(\tau') d\tau'.
    \end{align}
    From Eqs.~\eqref{eq:integral_ineq} and \eqref{eq:int_Lambda_manipulated}, we obtain that
    \begin{align}\label{eq:SM_independent_final_iterative_formula_Gamma}
    \Gamma_{X, Y}(\tau) \leq & \Gamma_{X, Y}(0) + \omega_0\sum_{{\alpha}: \mathcal{S}_{\alpha}\cap Y \neq \emptyset} \int_0^\tau \Gamma_{X, S_i}(\tau') d\tau',
    \end{align}
    where
    \[
    \omega_0 = 2\big(L + 2L\norminf{\bm\zeta_{[0,t]}}+4L\norminf{\bm\mu^{\textnormal{R}}(0,t)}+2L\norminf{\bm\mu^{\textnormal{L}}(0,t)}\big).
    \]
    From lemma \ref{lemma:mu_upper_bounds} it then follows that $\omega_0=2(L+2L\smallnorminf{\bm\zeta_{[0,t]}}+10L^2m(t + \delta))$.
    Now, applying lemma \ref{lemma:lr_recursion} and setting $\tau = t - t'$, we obtain the lemma statement.
\end{proof}
\begin{proof}[Proof of proposition 1 from main text]
    Let us first analyze $\Delta_{O_X}(t, t'; l)$, as defined in the main text, for the $\delta-$mollified model. We will make use of lemma \ref{lemma:delta_o_comm_bd}: Note that, for the independent bath model,  $\Delta_{O_X}^{\text{L}}(t,t';l)$, defined in lemma \ref{lemma:delta_o_comm_bd}, is 0 since if, for any set $Y$, $\alpha \in \mathcal{A}_Y = \{\alpha: \mathcal{S}_\alpha \cap Y = \emptyset\}$ and $\alpha' \notin \mathcal{A}_Y$, then $\alpha$ and $\alpha'$ correspond to different baths and thus $K_{\alpha, \alpha'} = 0$. Hence, the total truncation error $\Delta_{O_X}(t, t'; l)$ equals the right-contraction contribution $\Delta_{O_X}^{\text{R}}(t,t';l)$ bounded in lemma \ref{lemma:Delta_bd_gamma} via
    \begin{align}\label{eq:delta_ox_bound_recap_prop1}
    \Delta_{O_X}(t, t'; l) = \Delta_{O_X}^{\text{R}}(t, t'; l) \leq 2\norm{O_X}\sum_{\substack{\alpha \notin \mathcal{A}_{X[l]} \\ \mathcal{S}_{\alpha} \cap X_c[l] \neq \emptyset}}\left(1+\mu^{\textnormal{R}}_\alpha(0,t)\right) \int_{s = t'}^t  \gamma^{X, \mathcal{S}_{\alpha}}_{\bm{\mu}^{\textnormal{R}}_{[0, s]}, \mathcal{A}_{X[l]}}(t, s)  ds,
    \end{align}
    Noting the fact that if $\bm \zeta = \bm\mu^{\text{R}}_{[0, s]}$ for $0\leq s \leq t$, then 
    \begin{align*}
    \smallnorm{\bm{\zeta}_{[0, t]}}_\infty = \max_\alpha \sup_{s' \in[0, t]}\mu^{\text{R}}_{\alpha, [0, s]}(s') = \max_\alpha \sup_{s' \in[0, t]}\mu^{\text{R}}_{\alpha}(s' - s, s') \leq \max_\alpha \mu^{\text{R}}_{\alpha}(-t, t) \leq 2L m(t + \delta),
    \end{align*}
    where in the last step we have used lemma \ref{lemma:mu_upper_bounds}. With this estimate together with lemma \ref{lemma:gamma_lr}, it follows that
    \[
    \gamma^{X, \mathcal{S}_\alpha}_{\bm\mu^\text{R}_{[0, s]}, \mathcal{A}_{X[l]}}(t, s) \leq \abs{\mathcal{S}_\alpha} \exp\bigg(\frac{v_\text{LR}(t)|t - s|}{a_0}\bigg)\exp\bigg(-\frac{d_{X, \mathcal{S}_\alpha}}{a_0}\bigg),
    \]
    where 
    \[
    v_{\text{LR}}(t) = ea_0 \mathcal{Z}L ( 1 + 14 L m(t + \delta)).
    \] 
    Returning to Eq.~\eqref{eq:delta_ox_bound_recap_prop1} and using lemma \ref{lemma:mu_upper_bounds} to upper bound $\mu_\alpha^\text{R}(0, t)$ by $2L m(t + \delta)$
    \[
    \Delta_{O_X}(t, t'; l) \leq 2\norm{O_X}(1 + 2L m_\delta(t)) \sum_{\substack{\alpha \notin \mathcal{A}_{X[l]} \\ \mathcal{S}_{\alpha} \cap X^c[l] \neq \emptyset}} \abs{\mathcal{S}_{\alpha}}\int_{s = 0}^{t - t'} \exp\bigg(\frac{v_\text{LR}(t)s - d_{X, \mathcal{S}_{\alpha}}}{a_0} \bigg) ds.
    \]
    Now, we note that by definition of $\mathcal{A}_{X[l]}$ (lemma \ref{lemma:delta_o_comm_bd}), for any $\alpha \notin \mathcal{A}_{X[l]}$ , $d_{X, \mathcal{S}_{\alpha}} \geq l$. Furthermore, in $d$ dimensions and given the fact that $\text{diam}(\mathcal{S}_{\alpha}) \leq a_0$, $\abs{\mathcal{S}_{\alpha}} \leq (a_0+1)^d$. Therefore,
    \[
    \sum_{\substack{\alpha \notin \mathcal{A}_{X[l]} \\ \mathcal{S}_{\alpha} \cap X^c[l] \neq \emptyset}} \abs{\mathcal{S}_{\alpha}}\int_{s = 0}^{t - t'} \exp\bigg(\frac{v_\text{LR}(t)s - d_{X, \mathcal{S}_{\alpha}}}{a_0} \bigg) ds \leq \frac{a_0(a_0+1)^{d}}{v_\text{LR}(t)}\exp\bigg(-\frac{l}{a_0}\bigg)\bigg(\exp\bigg(\frac{v_\text{LR}(t)\abs{t - t'}}{a_0}\bigg) - 1\bigg)\bigg(\sum_{\substack{\alpha \notin \mathcal{A}_{X[l]} \\ \mathcal{S}_{\alpha} \cap X^c[l] \neq \emptyset}} 1\bigg).
    \]
    Now, note that the number of lattice sites in the set $X^c[l] \setminus X[l]$ is $(\text{diam}(X) + 2l + a_0)^d - (\text{diam}(X)+2l)^d$ and since any one lattice site can be in at-most $\mathcal{Z}$ sub-regions $\mathcal{S}_{\alpha}$, we obtain that
    \[
    \abs{\{\alpha : \alpha \notin \mathcal{A}_{X[l]} \text{ and }\mathcal{S}_{\alpha}\cap X^c[l] \neq \emptyset\}} \leq \mathcal{Z}\big[ (\text{diam}(X) + 2l + a_0)^d - (\text{diam}(X)+2l)^d\big].
    \]
    Together with the bound $v_\text{LR}(t) \geq ea_0\mathcal{Z}L (1 + 2L m_\delta(t))$, we obtain that
    \begin{align}\label{eq:final_bound_Delta_Ox_mollified}
    \Delta_{O_X}(t, t'; l) \leq \norm{O_X}\frac{2(a_0+1)^{d}}{eL} \big[ (\text{diam}(X) + 2l + a_0)^d - (\text{diam}(X)+2l)^d\big] \exp\bigg(-\frac{l}{a_0}\bigg)\bigg(\exp\bigg(\frac{v_\text{LR}(t)\abs{t - t'}}{a_0}\bigg) - 1\bigg).
    \end{align}
    Note that $\big[ (\text{diam}(X) + 2l + a_0)^d - (\text{diam}(X)+2l)^d\big] \leq O(l^{d -1 })$ for large $l$.

    Finally, we remark that, thus far, we have analyze $\Delta_{O_X}(t, t'; l)$ for  $\delta-$mollified model. However, as was established in Proposition \ref{lemma:mollification}, $\Delta_{O_X}(t, t'; l)$ converges to its value for the actual model as $\delta \to 0$. Thus, taking the limit of $\delta \to 0^+$ in Eq.~\eqref{eq:final_bound_Delta_Ox_mollified}, which entails replacing $m(t + \delta)$ by $m(t)$ (due to the right-continuity of $m$) on the right-hand side in $v_\text{LR}(t)$ proves the proposition. 
\end{proof}

{\color{black}This completes the independent-bath proof of Proposition 1. The following refinement remains within the independent-bath setting and improves the estimate of the Lieb-Robinson velocity. Consider the setting where the some of the baths only couple to a single site in the system, then physically it is expected that the Lieb-Robinson velocity in the model is independent of the memory kernel of these baths. To see how this can be proved, define the Hamiltonian
\[
\tilde{H}_{\text{SE}}(t) = \sum_{i \in \mathcal{I}} \sum_{\nu = 1}^L R_{i}^{(\nu)}(t) B_i^{(\nu)}(t),
\]
where $\mathcal{I}$ indexes the set of baths which couple only to a single site in the system i.e., for each $i \in \mathcal{I}$, $R_i^{(1)}(t), R_i^{(2)}(t) \dots R_i^{(\nu)}(t)$ are supported on the same single site. In this case, to provide a Lieb-Robinson bound, we simply move into the interaction picture with respect to $\tilde{H}_{\text{SE}}(t)$: This then yields the Hamiltonian
\[
\tilde{H}(t) = \sum_\alpha \tilde{h}_\alpha(t) + \sum_{i\notin \mathcal{I}} \sum_{\nu = 1}^L \tilde{R}^{(\nu)}_i(t) B_i^{(\nu)}(t),
\]
where
\begin{align}
    \tilde{R}^{(\nu)}_i(t) = \tilde{U}(0, t) R_i^{(\nu)}(t)\tilde{U}(t, 0) \quad \text{and}\quad \tilde{h}_\alpha(t) =  \tilde{U}(0, t) h_\alpha(t) \tilde{U}(t, 0),
\end{align}
where $\tilde{U}(t, s)$ is the unitary generated by $\tilde{H}_\text{SE}(t)$ defined above. Since $\tilde{H}_\text{SE}(t)$ only contains terms supported on single system sites, the operators $\tilde{R}^{(\nu)}_i(t)$ and $\tilde{h}_\alpha(t)$ have the same support on the system sites as ${R}^{(\nu)}_i(t)$ and ${h}_\alpha(t)$. Furthermore, the bath operators appearing in the rotated Hamiltonian $\tilde{H}(t)$ still commute with $\tilde{R}^{(\nu)}_i(t)$ and $\tilde{h}_\alpha(t)$ simply because they correspond to different baths than the ones in $\tilde{H}_\text{SE}(t)$. Consequently, the analysis in done thus far can instead be performed with $\tilde{H}(t)$ and would yield a Lieb-Robinson bound as in proposition 1 but with a modified $m(\tau)$ given below
\[
m(\tau) = \sup_{i \notin \mathcal{I}, \nu, \nu'} \text{TV}(K_{i, \nu, \nu'}; [-\tau, \tau]) \leq \sup_{i, \nu, \nu'} \text{TV}(K_{i, \nu, \nu'}; [-\tau, \tau]).
\]
Thus, moving into interaction picture with respect to the system-environment terms corresponding to baths only coupling with single system sites allows us to improve the upper bound on the Lieb-Robinson velocity derived in proposition 1. 
}

\section{Case of coupled bath}\label{sec:supp_coupled_bound}

\begingroup\color{black}
We now refer to the common framework developed in the previous section, but no longer assume that different bath supports are independent. Instead, we assume that the environment itself has a light cone,
i.e., the kernels $K_{\alpha,\alpha',\sigma}$ are small when $\mathcal{S}_{\alpha}$
and $\mathcal{S}_{\alpha'}$ are very far apart. To formalize this notion, we
associate each environmental operator $B_{\alpha}(t)$ with a distinct
position $\text{x}_{\alpha}$ in the environment. We require the distances between the points $\text{x}_{\alpha}$ to be commensurate with the lattice
distance between the regions $\mathcal{S}_{\alpha}$ i.e., the farther apart
the regions $\mathcal{S}_{\alpha}$ and $\mathcal{S}_{\beta}$ on the lattice the farther
apart should the points $\text{x}_{\alpha}$ and $\text{x}_{\beta}$ be. We formalize this by requiring the existence of constants $p_{\pm},q_{\pm}>0$ such that
\begin{equation}\label{eq:SM_environment_distance}
d(\mathcal{S}_{\alpha},\mathcal{S}_{\beta})\leq p_{+}d(\text{x}_{\alpha},\text{x}_{\beta})+q_{+}\text{ and }d(\text{x}_{\alpha},\text{x}_{\beta})\leq p_{-}d(\mathcal{S}_{\alpha},\mathcal{S}_{\beta})+q_{-}.
\end{equation}
We further require that $d(\text{x}_{\alpha},\text{x}_{\beta})$ satisfies the triangle inequality, whereas  $d(S_{\alpha},S_{\beta})$ (being a distance between two sets) may not.

We will now impose the following condition on the kernels $K_{\alpha,\alpha',\sigma}(s)$
to capture a finite velocity of propagation in the environment.

\begin{assumption}[Light-cone in the environment]\label{assump:light_cone_env}
Suppose the kernels $K_{\alpha, \alpha'}(\tau)$ have the representation
\[
K_{\alpha, \alpha'}(\tau) = K_{\alpha, \alpha'; c}(\tau) + \sum_j k_{\alpha, \alpha'}[j] \delta(\tau - T_j).
\]
We assume that there exists a non-negative constants $K_{0}, a_E, \lambda_E > 0$ such that for all $\tau > 0$
\begin{equation}
\textnormal{TV}(K_{\alpha, \alpha'}; [-\tau, \tau]) =\int_{-\tau}^{\tau}|K_{\alpha,\alpha';c}(s)|ds + \sum_{T_j \in [-\tau, \tau]} \smallabs{k_{\alpha, \alpha'}[j]} \leq K_{0}\min\big({\exp\left(\lambda_{E}|\tau|-d(\textnormal{x}_{\alpha},\textnormal{x}_{\alpha'})/a_{E}\right)},1\big)
\end{equation}
where $d(\textnormal{x}_{\alpha},\textnormal{x}_{\alpha'})$ is defined in Eq. (\ref{eq:SM_environment_distance}).
\end{assumption}
\noindent In the limit of $a_{E}\to0$
and $\lambda_{E}\to0$ and assuming that $\textnormal{x}_\alpha \neq \textnormal{x}_{\alpha'}$ if $\alpha \neq \alpha'$, we obtain an independent bath model i.e., $K_{\alpha, \alpha'}= 0$ unless $\alpha= \alpha'$. More generally, the parameter $a_{E}\lambda_{E}$ can be interpreted as the propagation velocity of excitations in the environment, and when non-zero, it can modify the speed at which information propagates in the system.
We first provide a lemma which replaces the upper bound in assumption \ref{assump:light_cone_env} by a different upper-bound which is more convenient to handle in the following Lieb-Robinson proof.
\begin{lemma}\label{lemma:modified_kernel_upper_bounds}
Suppose the kernels $K_{\alpha, \alpha'}$ satisfy assumption \ref{assump:light_cone_env}. For any fixed $0<\kappa_E<1/a_{E}$, there exist $\eta > d$ (where $d$ is the lattice dimension) and $0<\lambda_0 < \min(1/q_{-}, 1/q_{+})$ (where $q_\pm$ are introduced in Eq.~\eqref{eq:SM_environment_distance}) such that for $\tau > 0$
\begin{align}
\textnormal{TV}(K_{\alpha, \alpha'}; [-\tau, \tau])=\int_{-\tau}^{\tau}|K_{\alpha,\alpha';c}(s)|ds + \sum_{T_j \in [-\tau, \tau]} \smallabs{k_{\alpha, \alpha'}[j]} \leq K_0 \min\bigg(\frac{\exp(\lambda_E \abs{\tau} - \kappa_E d(\textnormal{x}_\alpha, \textnormal{x}_{\alpha'}))}{(1 + \lambda_0d(\textnormal{x}_\alpha, \textnormal{x}_{\alpha'}))^\eta}, 1 \bigg)
\end{align}
\end{lemma}
\begin{proof}
We begin by picking $\kappa_E$ as any constant such that $\kappa_E < 1/a_E$. Defining $\epsilon = 1/a_E - \kappa_E$, it then follows that for any $\eta > 0$
\[
\exp(-d(\text{x}_\alpha, \text{x}_{\alpha'})/a_E) = \exp(-\kappa_E d(\text{x}_\alpha, \text{x}_{\alpha'})) \exp(-\epsilon d(\text{x}_\alpha, \text{x}_{\alpha'})) \numleq{1} \frac{\exp(-\kappa_E d(\text{x}_\alpha, \text{x}_{\alpha'}))}{(1 + \epsilon d(\text{x}_\alpha, \text{x}_{\alpha'})/\eta )^\eta},
\]
where in (1) we have used the inequality $(1+x)^a\leq e^{ax} \ \forall \ a, x \geq0$. Using the fact that if $A \leq B$, then $\min(A, 1) \leq \min(B, 1)$, we then obtain
\[
\int_{-\tau}^{\tau}|K_{\alpha,\alpha';c}(s)|ds + \sum_{T_j \in [-\tau, \tau]} \smallabs{k_{\alpha, \alpha'}[j]} \leq K_0 \min\bigg(\frac{\exp(\lambda_E \abs{\tau} - \kappa_E d(\textnormal{x}_\alpha, \textnormal{x}_{\alpha'}))}{(1 + \lambda_0d(\textnormal{x}_\alpha, \textnormal{x}_{\alpha'}))^\eta}, 1 \bigg),
\]
where $\lambda_0= \epsilon / \eta$. Finally, note that thus far $\eta$ is a free parameter: We pick $\eta > \max(d, \epsilon q_-, \epsilon q_+)$ to satisfy the constraints in the lemma statement.
\end{proof}
\noindent We also provide a technical lemmas that will be useful in the following analysis. The proof of this lemma is provided in the appendix \ref{app:coupled_bath_summations}.
\begin{lemma}\label{lemma:coupled_both_summations}
The following summations inequalities are true
\begin{enumerate}
    \item[(a)] $\exists D_0, D_{\partial 0}>0$ and $D_{\partial 1}\geq0$ dependent only on the lattice dimension $d$ and the constants $a_0, \mathcal{Z}, p_\pm, q_\pm$ such that for all $\alpha$ and $l \geq 1$,
    \[
    \sum_{\alpha': d(\textnormal{x}_\alpha, \textnormal{x}_{\alpha'}) \leq l} 1 \leq D_0 l^d \quad \textnormal{and}\quad \sum_{\alpha': l \leq d(\textnormal{x}_\alpha, \textnormal{x}_{\alpha'}) < l + 1} 1 \leq D_{\partial 0} l^{d - 1} + D_{\partial 1}l^d.
    \]
    Furthermore, $D_{\partial 1} = 0$ if $p_+ p_- = 1$.
    \item[(b)] For $\eta > d$, where $d$ is the dimensionality of the lattice, then $\exists B_\eta > 0$ dependent only on $\eta$, $d$, $\mathcal{Z}$, $a_0$, $p_+$ and $q_+$ such that for all $\alpha$
\begin{align}
\sum_{\alpha'} \frac{1}{(1 + d(\textnormal{x}_\alpha, \textnormal{x}_{\alpha'}))^\eta} \leq \frac{B_\eta}{2^\eta}.
\end{align}
\end{enumerate}
\end{lemma}
\noindent Next, we establish upper bounds on $\mu^\text{R}_\alpha(-t, t)$ and $\mu^\text{L}_{\alpha, \sigma}(-t, t)$. We remind the reader that these are defined in section \ref{sec:common_framework}, and are defined with the $\delta$-mollified kernels.
\begin{lemma}\label{lemma:coupled_bath_finiteness_on_memory}
For $\tau > 0$, the quantity $\norminf{\bm{\mu}^{\textnormal{R}}(-\tau,\tau)}$,
$\norminf{\bm{\mu}^{\textnormal{L}}(-\tau,\tau)}$ are finite for fixed $\tau$ and satisfy the upper bounds
\begin{equation}
\norminf{\bm{\mu}^{\textnormal{R}}(-\tau,\tau)}\leq 2(C_{E}(\tau + 2\delta)^{d} + B_E)\quad \text{and}\quad \norminf{\bm{\mu}^{\textnormal{L}}(-\tau, \tau)}\leq C_{E}(\tau + 2\delta)^{d} + B_E,
\end{equation}
where $C_{E}=K_0 D_0 \lambda_E^d/\kappa_E^d$ and $B_{E}= K_0 B_\eta/(2^\eta \min(1, \lambda_0)^\eta)$ with $D_0$ being defined in lemma \ref{lemma:coupled_both_summations}.
\end{lemma}
\begin{proof}
Recall from lemma \ref{lemma:mollified_kernel} that the $\delta-$mollified kernels, $K_{\alpha, \alpha', \sigma'}^{\delta, \delta} = (\eta_\delta \star \eta_\delta)\star K_{\alpha, \alpha', \sigma'}$, where $\eta_\delta$ is defined in section \ref{sec:regularization}, satisfy
\begin{align}
\int_{-\tau}^\tau \smallabs{K_{\alpha, \alpha', \sigma'}^{\delta}(s)}ds &\leq \int_{-(\tau + 2\delta)}^{\tau + 2\delta}|K_{\alpha,\alpha';c}(s)|ds + \sum_{T_j \in [-(\tau + 2\delta), \tau + 2\delta]} \smallabs{k_{\alpha, \alpha'}[j]} \nonumber\\
&\leq K_0 \min\bigg(\frac{\exp(\lambda_E (\tau + 2\delta) - \kappa_E d(\textnormal{x}_\alpha, \textnormal{x}_{\alpha'}))}{(1 + \lambda_0d(\textnormal{x}_\alpha, \textnormal{x}_{\alpha'}))^\eta}, 1 \bigg).
\end{align}
Since $\mu^\text{R}_{\alpha}(-\tau, \tau)$ is defined for the mollified kernels, we have that
\begin{align}
\norminf{\bm{\mu}^{\textnormal{R}}(-\tau,\tau)} & =\max_{\alpha}\sum_{\alpha',\sigma'}\int_{-\tau}^{\tau}\smallabs{K^{ \delta}_{\alpha,\alpha',\sigma'}(s)}ds,\nonumber \\
 & \leq2K_{0}\max_{\alpha}\sum_{\alpha'}\min\bigg(\frac{\exp(\lambda_E (\tau + 2\delta) - \kappa_E d(\textnormal{x}_\alpha, \textnormal{x}_{\alpha'}))}{(1 + \lambda_0d(\textnormal{x}_\alpha, \textnormal{x}_{\alpha'}))^\eta}, 1\bigg)\nonumber \\
 & \leq2K_{0}\max_\alpha \bigg(\sum_{\alpha': d(\text{x}_{\alpha},\text{x}_{\alpha'})\leq \lambda_E (\tau + 2\delta)/\kappa_E}1+ \sum_{\alpha':d(\text{x}_{\alpha},\text{x}_{\alpha'})\geq \lambda_E (\tau + 2\delta)/\kappa_E}\frac{\exp(\lambda_E (\tau + 2\delta) - \kappa_E d(\textnormal{x}_\alpha, \textnormal{x}_{\alpha'}))}{(1 + \lambda_0d(\textnormal{x}_\alpha, \textnormal{x}_{\alpha'}))^\eta}\bigg)\nonumber \\
 &\leq 2K_0 \max_\alpha \bigg(\sum_{\alpha': d(\text{x}_{\alpha},\text{x}_{\alpha'})\leq \lambda_E (\tau + 2\delta)/\kappa_E}1+ \frac{1}{(\min(1, \lambda_0))^\eta}\sum_{\alpha':d(\text{x}_{\alpha},\text{x}_{\alpha'})\geq \lambda_E (\tau + 2\delta)/\kappa_E}\frac{1}{(1 + d(\textnormal{x}_\alpha, \textnormal{x}_{\alpha'}))^\eta}\bigg) \nonumber \\
 & \numleq{1} 2K_{0}\left(\frac{D_{0}\lambda_E^d}{\kappa_E^d} (\tau + 2\delta)^{d} + \frac{B_\eta}{2^\eta (\min(1, \lambda_0))^\eta}\right).
\end{align}
where in (1) we have used lemma \ref{lemma:coupled_both_summations}. Substituting $C_{E}=K_0 D_0 \lambda_E^d/\kappa_E^d$ and $B_{E}= K_0 B_\eta/(2^\eta \min(1, \lambda_0)^\eta)$, we obtain
\begin{equation}
\norminf{\bm{\mu}^{\textnormal{R}}(-\tau,\tau)}\leq2(C_{E}(\tau + 2\delta)^{d} + B_E).
\end{equation}
A similar analysis can be used to derive an upper bound on $\norminf{\bm{\mu}^{\textnormal{L}}(-\tau,\tau)}.$
\end{proof}

We now start from the same Taylor estimate Lemma \ref{lemma:taylor_expansion_bound}, and establish the recursion relation. 

\begin{lemma}\label{lemma:coupled_bath_integration_expression_gamma_lr}Given $0 \leq t'<t$ and $\tau \in (0, t- t')$, define
$\Gamma_{X,Y}(\tau)=\gamma_{\bm{\zeta}+2\bm{\mu}_{[t',t-\tau]}^{R}}^{X,Y}(t,t-\tau)$.
For coupled baths, we derive the following integral inequality:
\begin{equation}
\Gamma_{X,Y}(\tau)\leq\Gamma_{X,Y}(0)+\omega(t)\sum_{\alpha:S_{\alpha}\cap Y\neq\emptyset}\int_{0}^{\tau}\Gamma_{X,S_{\alpha}}(\tau')d\tau'+2\sum_{\alpha,\alpha',\sigma:S_{\alpha}\cap Y\neq\emptyset}\int_{0}^{\tau}\Gamma_{X,S_{\alpha'}}(\tau')d\tau'\int_{0}^{\tau-\tau'}\abs{K^{(\delta)}_{\alpha',\alpha,\sigma}(\tau'')}d\tau'',\label{eq:SM_coupled_bath_integration_inequality}
\end{equation}
where
\begin{equation}
\omega(t)=2(1+2\smallnorminf{\bm{\zeta}_{[0,t]}}+4\smallnorminf{\bm{\mu}^{\textnormal{R}}(-t,t)}).\label{eq:SM_definition_omega}
\end{equation}
\end{lemma}
\begin{proof}

From lemma \ref{lemma:small_gamma_recursion}, for sufficiently small $\varepsilon$, we have
\begin{align}
 & \Gamma_{X,Y}(\tau+\varepsilon)-\Gamma_{X,Y}(\tau)\nonumber \\
 & \leq\varepsilon\sum_{\alpha:S_{\alpha}\cap Y\neq\emptyset}\left(2(1+2\norminf{\bm{\zeta}_{[0,t]}}+4\sup_{s\in[0,t]}\norminf{\bm{\mu}_{[t',t-\tau]}^{R}(s)})\Gamma_{X,S_{\alpha}}(\tau)+2\sum_{\alpha',\sigma}\int_{t-\tau}^{t}\abs{K^{\delta}_{\alpha',\alpha,\sigma}(s'-t+\tau)}\Gamma_{X,S_{\alpha'}}(t-s')ds'\right)+C_{1}\varepsilon^{2}\nonumber \\
 & \leq\varepsilon\sum_{\alpha:S_{\alpha}\cap Y\neq\emptyset}\left(2(1+2\norminf{\bm{\zeta}_{[0,t]}}+4\norminf{\bm{\mu}^{\textnormal{R}}(-t,t)})\Gamma_{X,S_{\alpha}}(\tau)+2\underbrace{\sum_{\alpha',\sigma}\int_{0}^{\tau}\abs{K^{\delta}_{\alpha',\alpha,\sigma}(\tau-\tau')}\Gamma_{X,S_{\alpha'}}(\tau')d\tau'}_{\Lambda_{X,S_{\alpha}}(\tau)}\right)+C_{1}\varepsilon^{2}.
\end{align}
Here, we remind the reader that lemma \ref{lemma:small_gamma_recursion} was derived for the $\delta-$mollified kernel, and therefore we have used $K^{(\delta)}_{\alpha', \alpha, \sigma}$ in the above inequality. Following steps similar to those in Lemma \ref{lemma:gamma_lr}, we take $\varepsilon=\tau/N$ for sufficiently large $N$ and arrive at the integral inequality for $\Gamma_{X,Y}(\tau)$:
\begin{equation}
\begin{aligned}
\Gamma_{X,Y}(\tau)\leq{}&\Gamma_{X,Y}(0)+\sum_{\alpha:S_{\alpha}\cap Y\neq\emptyset}2(1+2\norminf{\bm{\zeta}_{[0,t]}}+4\norminf{\bm{\mu}^{\textnormal{R}}(-t,t)})\int_{0}^{\tau}\Gamma_{X,S_{\alpha}}(\tau')d\tau'\\
&+2\sum_{\alpha:S_{\alpha}\cap Y\neq\emptyset}\int_{0}^{\tau}\Lambda_{X,S_{\alpha}}(\tau')d\tau'.
\end{aligned}
\end{equation}
Finally, we note that
\begin{equation}
\begin{aligned}\int_{0}^{\tau}\Lambda_{X,S_{\alpha}}(\tau')d\tau' & =\sum_{\alpha',\sigma}\int_{\tau'=0}^{\tau}\int_{\tau''=0}^{\tau'}\abs{K^{\delta}_{\alpha',\alpha,\sigma}(\tau'-\tau'')}\Gamma_{X,S_{\alpha'}}(\tau'')d\tau''d\tau'\\
 & =\sum_{\alpha',\sigma}\int_{\tau'=0}^{\tau}\Gamma_{X,S_{\alpha'}}(\tau')d\tau'\int_{0}^{\tau-\tau'}\abs{K^{\delta}_{\alpha',\alpha,\sigma}(\tau'')}d\tau'',
\end{aligned}
\end{equation}
which completes the proof.

\end{proof}

Compared with Lemma \ref{lemma:gamma_lr}, the integral inequality derived above exhibits a tail across different system sites, originating from the coupling between distinct baths. To derive the corresponding Lieb–Robinson bound, we thus need to sum over contributions from the entire system. To streamline the remaining proof, we first provide a solution of a modified Lieb-Robinson integral equation, which we will later apply to provide a Lieb-Robinson bound. We begin with the proof of a simple but useful inequality.
\begin{lemma}\label{lemma:F_useful_property}
    For $\tau > 0$, define $F_{\alpha, \beta}(\tau)$ via
    \begin{align}
            F_{\alpha, \beta}(\tau) = \frac{\exp(\lambda_E \tau - \kappa_E d(\textnormal{x}_\alpha, \textnormal{x}_\beta))}{(1 + d(\textnormal{x}_\alpha, \textnormal{x}_\beta))^\eta},
    \end{align}
    then for $\tau_1, \tau_2 > 0$
    \begin{align}
        \sum_{\gamma}F_{\alpha,\gamma}(\tau_{1})F_{\gamma,\beta}(\tau_2) \leq B_\eta F_{\alpha, \beta}(\tau_1 +\tau_2),
    \end{align}
    where $B_\eta$ is the constant defined in lemma \ref{lemma:coupled_both_summations}(b).
\end{lemma}
\begin{proof}
For $\tau_1, \tau_2 > 0$
\begin{align}\label{eq:F_useful_property}
   \sum_{\gamma}F_{\alpha,\gamma}(\tau_{1})F_{\gamma,\beta}(\tau_2) & \leq \sum_{\gamma}\frac{\exp(\lambda_{E}(\tau_{1}+\tau_{2}))\exp[-\kappa_E(d(\text{x}_{\alpha},\text{x}_{\gamma})+d(\text{x}_{\gamma},\text{x}_{\beta}))]}{(1+d(\text{x}_{\alpha},\text{x}_{\gamma}))^{\eta}(1+d(\text{x}_{\gamma},\text{x}_{\beta}))^{\eta}}\nonumber\\
   &= \sum_{\gamma}\frac{\exp(\lambda_{E}(\tau_{1}+\tau_{2}))\exp[-\kappa_E(d(\text{x}_{\alpha},\text{x}_{\gamma})+d(\text{x}_{\gamma},\text{x}_{\beta}))]}{(1+d(\text{x}_{\alpha},\text{x}_{\beta}))^{\eta}}\frac{(1+d(\text{x}_{\alpha},\text{x}_{\beta}))^{\eta}}{(1+d(\text{x}_{\alpha},\text{x}_{\gamma}))^{\eta}(1+d(\text{x}_{\gamma},\text{x}_{\beta}))^{\eta}}\nonumber\\
 & \numleq{1}\frac{\exp(\lambda_{E}(\tau_{1}+\tau_{2}))\exp[-\kappa_E d(\text{x}_{\alpha},\text{x}_{\beta})]}{(1+d(\text{x}_{\alpha},\text{x}_{\beta}))^{\eta}}\sum_{\gamma}\frac{(1+d(\text{x}_{\alpha},\text{x}_{\gamma})+d(\text{x}_{\gamma},\text{x}_{\beta}))^{\eta}}{(1+d(\text{x}_{\alpha},\text{x}_{\gamma}))^{\eta}(1+d(\text{x}_{\gamma},\text{x}_{\beta}))^{\eta}}\nonumber\\
 & \numleq{2} F_{\alpha,\beta}(\tau_{1}+\tau_{2})\sum_{\gamma}2^{\eta-1}\bigg[\frac{1}{(1+d(\text{x}_{\alpha},\text{x}_{\gamma}))^{\eta}}+\frac{1}{(1+d(\text{x}_{\gamma},\text{x}_{\beta}))^{\eta}}\bigg]\nonumber\\
 & \numleq{3} B_{\eta}F_{\alpha,\beta}(\tau_{1}+\tau_{2}),
\end{align}
where in (1) we have used the triangle inequality ($d(\text{x}_\alpha, \text{x}_\beta) \leq d(\text{x}_\alpha, \text{x}_\gamma) + d(\text{x}_\gamma, \text{x}_\beta)$), in (2) we have used the inequality $(1 + a + b)^\eta \leq ((1 + a) + (1 + b))^\eta \leq 2^{\eta - 1}((1 + a)^\eta + (1 + b)^\eta)$ for $a, b > 0$ and $\eta > 1$, and in (3) we have used lemma \ref{lemma:coupled_both_summations}.
\end{proof}
\begin{lemma}[Modified Lieb-Robinson recursion]\label{lemma:modified_lr_recursion}
Suppose for $X \subseteq \Lambda$, $\nu_{X, \alpha}(t) \in [0, 1]$ satisfies the integral inequality
\begin{align}\label{eq:modified_recursion_eq_1}
    \nu_{X, \alpha}(\tau) \leq \nu_{X, \alpha}(0) + {\omega}_0 \sum_{\beta} \int_0^\tau F_{\alpha, \beta}(\tau - \tau') \nu_{X, \beta}(\tau') d\tau',
\end{align}
where $\nu_{X, \alpha}(0) = 0$ if $ X \cap \mathcal{S}_\alpha = \emptyset$, $\omega_0 > 0$ and $F_{\alpha, \beta}(\tau)$ is defined in Lemma~\ref{lemma:F_useful_property}. Furthermore, suppose $\chi_{X, Y}(t)$ satisfies the inequality
\begin{align}
    \chi_{X, Y}(\tau) \leq \omega_1 \sum_{\alpha:\mathcal{S}_\alpha \cap Y \neq \emptyset} \int_0^\tau \nu_{X, \alpha}(\tau')d\tau' + \omega_2 \sum_{\alpha', \alpha:\mathcal{S}_\alpha \cap Y \neq \emptyset} \int_0^\tau F_{\alpha, \alpha'}(\tau -\tau') \nu_{X, {\alpha'}}(\tau') d\tau',
\end{align}
with $\omega_{1},\omega_{2}\geq0.$
Then, assuming that $d(X, Y) > 2a_0$, we obtain that
\begin{align}
    \chi_{X, Y}(\tau) \leq \chi_0 \exp(( \lambda_E + \omega_0 B_\eta)\tau - \kappa_E d(X, Y)/p_+),
\end{align}
where
\[
\chi_0 = \abs{X}\abs{Y}\mathcal{Z}^{2} \exp\bigg( \frac{\kappa_E(2a_0 + q_+)}{p_+}\bigg)\bigg[\frac{\omega_1 \omega_0}{(\lambda_E + \omega_0 B_\eta)^{2}} + \frac{\omega_2}{\omega_0 B_\eta + \lambda_E}\bigg].
\]
\end{lemma}
\begin{proof}
We can directly solve the recursion for $\nu_{X, \alpha}$: Recursing over Eq.~\eqref{eq:modified_recursion_eq_1}, we obtain that
\begin{equation}\label{eq:modified_rec_nu}
\nu_{X,{\alpha}}(\tau)\leq\nu_{X,{\alpha}}(0)+\sum_{\beta}\int_{0}^{\tau}G_{\alpha,\beta}(\tau-\tau')\nu_{X,{\beta}}(0)d\tau',
\end{equation}
where
\begin{align}
G_{\alpha,\beta}(\tau) & ={\omega}_0 F_{\alpha,\beta}(\tau)+\sum_{n=1}^{\infty}{\omega_0}^{n+1}\sum_{\gamma_{1},\gamma_2\cdots,\gamma_{n}}\int_{\tau\geq s_{1}\geq s_{2}\geq\cdots\geq s_{n}\geq 0}F_{\alpha,\gamma_{1}}(\tau-s_{1})F_{\gamma_{1},\gamma_{2}}(s_{1}-s_{2})\cdots F_{\gamma_{n},\beta}(s_{n})ds_{1}ds_{2}\cdots ds_{n}\nonumber\\
 & \numleq{1}{\omega}_0 F_{\alpha,\beta}(\tau)+\sum_{n=1}^{\infty}{\omega_0}^{n+1}B_{\eta}^{n}F_{\alpha,\beta}(\tau)\int_{\tau\geq s_{1}\geq s_{2}\geq\cdots\geq s_{n}\geq 0} ds_1 ds_2 \dots ds_n\nonumber\\
 &\leq{\omega}_0 F_{\alpha,\beta}(\tau)+\sum_{n=1}^{\infty}{\omega}_0^{n+1}B_{\eta}^{n}F_{\alpha,\beta}(\tau)\frac{\tau^{n}}{n!}\nonumber\\
 & = {\omega}_0F_{\alpha,\beta}(\tau)\exp(\omega_0 B_{\eta}\tau).
\end{align}
where in (1) we have used lemma~\ref{lemma:F_useful_property}. Next, substituting Eq.~\eqref{eq:modified_rec_nu} into the inequality for $\chi_{X, Y}(\tau)$ and using the fact that $\nu_{X, \alpha}(0) = 0$ if $X \cap \mathcal{S}_\alpha= \emptyset$, we obtain that
\begin{align}\label{eq:coupled_recursion_split_into_terms}
    &\chi_{X, Y}(\tau) \leq \omega_1 \bigg[\sum_{\alpha:\mathcal{S}_\alpha \cap Y \neq \emptyset} \nu_{X, \alpha}(0) \tau +  \sum_{\alpha, \beta:\mathcal{S}_\alpha \cap Y \neq \emptyset} \int_0^\tau \int_{0}^{\tau'} G_{\alpha, \beta}(\tau' - \tau'') \nu_{X, \beta}(0) d\tau'' d\tau'\bigg] +\nonumber\\
    &\qquad \qquad \ \ \omega_2\bigg[\sum_{\alpha', \alpha:\mathcal{S}_\alpha \cap Y \neq \emptyset} \int_0^\tau F_{\alpha, \alpha'}(\tau - \tau') \nu_{X, \alpha'}(0) d\tau'\bigg.\nonumber\\ &\qquad\qquad\qquad\bigg. + \sum_{\alpha', \alpha, \beta: \mathcal{S}_\alpha \cap Y \neq \emptyset} \int_0^\tau \int_0^{\tau'} F_{\alpha,\alpha'}(\tau - \tau') G_{\alpha', \beta}(\tau' - \tau'') \nu_{X, \beta}(0) d\tau'' d\tau'\bigg].
\end{align}
We now analyze the right hand side of this inequality term by term. Note that since $d(X, Y) > 2a_0$, there is no $\alpha$ such that $\mathcal{S}_\alpha \cap X \neq \emptyset$ and $\mathcal{S}_\alpha \cap Y \neq \emptyset$. Thus,
\begin{align}\label{eq:coupled_recursion_term_1}
    \sum_{\alpha:\mathcal{S}_\alpha \cap Y \neq \emptyset} \nu_{X, \alpha}(0) \tau \leq \sum_{\alpha: \mathcal{S}_\alpha \cap Y \neq \emptyset, \mathcal{S}_\alpha \cap X \neq \emptyset} \tau = 0.
\end{align}
Next, note that if $\mathcal{S}_\alpha\cap X \neq \emptyset$ and $\mathcal{S}_\beta \cap Y \neq \emptyset$, then $d(\mathcal{S}_\alpha, \mathcal{S}_\beta) \geq d(X, Y) - 2a_0$ and therefore, from Eq.~\eqref{eq:SM_environment_distance}, $d(\text{x}_\alpha, \text{x}_\beta) \geq (d(X, Y) - 2a_0 -q_+)/p_+$. This implies
\begin{align}\label{eq:summation_coupled_recursion_solution}
    \sum_{\alpha: \mathcal{S}_\alpha \cap Y \neq \emptyset} \ \sum_{\beta: \mathcal{S}_\beta \cap X \neq \emptyset} \frac{\exp(-\kappa_E d(\text{x}_\alpha, \text{x}_\beta))}{(1 + d(\text{x}_\alpha, \text{x}_\beta))^\eta} \leq \abs{X}\abs{Y}\mathcal{Z}^2 \exp(\kappa_E(2a_0 + q_+)/p_+) \exp(-\kappa_E d(X, Y)/p_+).
\end{align}

It then follows that
\begin{align}\label{eq:coupled_recursion_term_2}
    &\sum_{\alpha, \beta:\mathcal{S}_\alpha \cap Y \neq \emptyset} \int_0^\tau \int_{0}^{\tau'} G_{\alpha, \beta}(\tau' - \tau'') \nu_{X, \beta}(0) d\tau'' d\tau'\nonumber\\
    &\qquad\leq \omega_0 \sum_{\alpha:\mathcal{S}_\alpha \cap Y \neq \emptyset} \sum_{\beta: \mathcal{S}_\beta \cap X \neq \emptyset} \int_0^\tau \int_{0}^{\tau'} \frac{\exp((\lambda_E + \omega_0 B_\eta )(\tau' - \tau'') - \kappa_E d(\text{x}_\alpha, \text{x}_\beta))}{(1 + d(\text{x}_\alpha, \text{x}_\beta))^\eta} d\tau'' d\tau' \nonumber \\
    &\qquad = \frac{\omega_0[\exp((\lambda_E+\omega_0B_\eta)\tau)-1-(\lambda_{E}+\omega_{0}B_{\eta})\tau]}{(\lambda_{E} + \omega_{0} B_{\eta})^{2}} \sum_{\alpha:\mathcal{S}_\alpha \cap Y \neq \emptyset} \sum_{\beta: \mathcal{S}_\beta \cap X \neq \emptyset} \frac{\exp(- \kappa_E d(\text{x}_\alpha, \text{x}_\beta))}{(1 + d(\text{x}_\alpha, \text{x}_\beta))^\eta}  \nonumber \\
    &\qquad \numleq{1} \abs{X}\abs{Y}\mathcal{Z}^2 \frac{\omega_0 \exp(\kappa_E(2a_0+q_+)/p_+)}{(\lambda_E + \omega_0 B_\eta)^{2}} \exp((\lambda_E +\omega_0 B_\eta)\tau) \exp(-\kappa_E d(X, Y)/p_+),
\end{align}
where in (1) we have used Eq.~\eqref{eq:summation_coupled_recursion_solution}. A similar analysis would also yield that
\begin{align}\label{eq:coupled_recursion_term_3}
    \sum_{\alpha', \alpha:\mathcal{S}_\alpha \cap Y \neq \emptyset} \int_0^\tau F_{\alpha, \alpha'}(\tau - \tau') \nu_{X, \alpha'}(0) d\tau' \leq \abs{X}\abs{Y}\mathcal{Z}^2 \frac{\exp(\kappa_E(2a_0+q_+)/p_+)}{\lambda_E} (\exp(\lambda_E\tau)-1)  \exp(-\kappa_E d(X, Y)/p_+).
\end{align}
Finally, consider
\begin{align}\label{eq:coupled_recursion_term_4}
    &\sum_{\alpha', \alpha, \beta: \mathcal{S}_\alpha \cap Y \neq \emptyset} \int_0^\tau \int_0^{\tau'} F_{\alpha,\alpha'}(\tau - \tau') G_{\alpha', \beta}(\tau' - \tau'') \nu_{X, \beta}(0) d\tau'' d\tau'\nonumber \\
    &\qquad\leq \omega_0 \sum_{\alpha: \mathcal{S}_\alpha \cap Y \neq \emptyset} \sum_{\beta: \mathcal{S}_\beta \cap X \neq \emptyset}\sum_{\alpha'} \int_0^\tau \int_0^{\tau'} \exp(\omega_0 B_\eta (\tau' - \tau'')) F_{\alpha, \alpha'}(\tau -\tau') F_{\alpha', \beta}(\tau' - \tau'')d\tau'' d\tau' \nonumber \\
    &\qquad \numleq{1} \omega_0 B_\eta \sum_{\alpha: \mathcal{S}_\alpha \cap Y \neq \emptyset} \sum_{\beta: \mathcal{S}_\beta \cap X \neq \emptyset} \int_0^\tau \int_0^{\tau'} \exp(\omega_0 B_\eta (\tau' - \tau'')) F_{\alpha, \beta}(\tau - \tau'')d\tau'' d\tau' \nonumber\\
    &\qquad \leq  \sum_{\alpha: \mathcal{S}_\alpha \cap Y \neq \emptyset} \sum_{\beta: \mathcal{S}_\beta \cap X \neq \emptyset} \int_0^\tau [\exp(\omega_0 B_\eta (\tau - \tau'))-1] F_{\alpha, \beta}(\tau - \tau')d\tau' \nonumber \\
    &\qquad \leq \left[\frac{\exp((\lambda_E +\omega_0 B_\eta)\tau)-1}{\lambda_{E} + \omega_{0} B_{\eta}}-\frac{\exp(\lambda_{E}\tau)-1}{\lambda_{E}}\right] \sum_{\alpha:\mathcal{S}_\alpha \cap Y \neq \emptyset} \sum_{\beta:\mathcal{S}_\beta \cap X \neq \emptyset} \frac{\exp(-\kappa_E d(\textnormal{x}_\alpha, \textnormal{x}_\beta))}{(1 + d(\textnormal{x}_\alpha, \textnormal{x}_\beta))^\eta} \nonumber \\
    &\qquad \numleq{2} \abs{X}\abs{Y}\mathcal{Z}^2\exp(\kappa_E(2a_0 + q_+)/p_+)\left[\frac{\exp((\lambda_E+\omega_0B_\eta)\tau)-1}{\lambda_{E}+\omega_{0}B_{\eta}}-\frac{\exp(\lambda_{E}\tau)-1}{\lambda_{E}}\right]\exp(-\kappa_E d(X, Y)/p_+),
\end{align}
where in (1) we have used lemma \ref{lemma:F_useful_property} and in (2) we have used Eq.~\eqref{eq:summation_coupled_recursion_solution}. 

Combining Eqs.~\eqref{eq:coupled_recursion_split_into_terms}, \eqref{eq:coupled_recursion_term_1}, \eqref{eq:coupled_recursion_term_2}, \eqref{eq:coupled_recursion_term_3} and \eqref{eq:coupled_recursion_term_4} and combining the two terms proportional to $\omega_{2}$, we obtain the lemma statement.
\end{proof}

\begin{lemma}\label{lemma:coupled_bath_commuting_LR_bound}
For the coupled bath model, it follows that
\begin{equation*}
\gamma_{\bm{\zeta}, \mathcal{B}}^{X,Y}(t,t')\leq\abs{X}\abs{Y}\Gamma_0(t)\exp\left(\frac{\kappa_E}{p_+}\bigg({v_{\textnormal{LR}}(\bm\zeta, t)(t-t')-d(X,Y)}\bigg)\right),
\end{equation*}
where $v_\textnormal{LR}(\bm\zeta,t)$ and $\Gamma_0(t)$ are given by
\begin{align*}
    v_{\textnormal{LR}}(\bm{\zeta}, t) = v_0 + V_0 (\smallnorm{\bm{\zeta}_{[0, t]}}_\infty + 2 \smallnorm{\bm\mu^\textnormal{R}(-t, t)}_\infty)
\end{align*}
with
\begin{align*}
          v_0 = \frac{p_+\lambda_E}{\kappa_E}  +\frac{2p_+B_\eta \exp(\kappa_E q_{-} + 2\lambda_E \delta )}{\kappa_E \min(\lambda_{0},1-\lambda_{0}q_{-})^{\eta}}(1 + 2K_0\mathcal{Z}) \quad \text{and}\quad V_0 = \frac{4p_+B_\eta \exp(\kappa_E q_{-} + 2\lambda_E \delta )}{\kappa_E \min(\lambda_{0},1-\lambda_{0}q_{-})^{\eta}},
\end{align*}
and
\begin{align*}
    \Gamma_0(t) = \exp\bigg(\frac{2\kappa_{E}a_{0}}{p_{+}}\bigg)+\mathcal{Z}^{2}\exp\bigg( \frac{\kappa_E(2a_0 + q_+)}{p_+}\bigg)\bigg[\frac{\omega(t) \tilde{\omega}(t)}{\lambda_{E}^{2}} + \frac{\tilde{\omega}_1(t)}{\lambda_{E}}\bigg],
\end{align*}
with
\begin{align*}
    \omega(t) = 2(1+2\smallnorminf{\bm{\zeta}_{[0,t]}}+4\smallnorminf{\bm{\mu}^{\textnormal{R}}(-t,t)}),
\end{align*}
and
\begin{align*}
     \tilde{\omega}(t) =  \frac{\exp(\kappa_E q_{-} + 2\lambda_E \delta )}{\min(\lambda_{0},1-\lambda_{0}q_{-})^{\eta}}\left(\omega(t)+4\mathcal{Z}K_{0}\right) \quad \text{and}\quad \tilde{\omega}_1(t) = 4K_0 \frac{\exp(2\lambda_E \delta)}{(\min(1, \lambda_0))^\eta}
\end{align*}
\end{lemma}
\begin{proof}
Setting $Y=S_{\alpha}$ in Lemma \ref{lemma:coupled_bath_integration_expression_gamma_lr} yields
\begin{equation}
\Gamma_{X,S_{\alpha}}(\tau)\leq\Gamma_{X,S_{\alpha}}(0)+\omega(t)\sum_{\beta:S_{\beta}\cap S_{\alpha}\neq\emptyset}\int_{0}^{\tau}\Gamma_{X,S_{\beta}}(\tau')d\tau'+2\sum_{\alpha',\beta,\sigma:S_{\beta}\cap \mathcal{S}_{\alpha}\neq\emptyset}\int_{0}^{\tau}\Gamma_{X,S_{\alpha'}}(\tau')d\tau'\int_{0}^{\tau-\tau'}\abs{K^\delta_{\alpha',\beta,\sigma}(\tau'')}d\tau'',
\end{equation}
where
\begin{equation}
\omega(t)=2(1+2\smallnorminf{\bm{\zeta}_{[0,t]}}+4\smallnorminf{\bm{\mu}^{\textnormal{R}}(-t,t)}).
\end{equation}
We can write the above expression in a more compact form as
\begin{equation}\label{eq:SM_coupled_bath_inequality_Gamma}
\Gamma_{X,S_{\alpha}}(\tau)\leq\Gamma_{X,S_{\alpha}}(0)+\sum_{\beta}\int_{0}^{\tau}f_{\alpha,\beta}(\tau-\tau')\Gamma_{X,S_{\beta}}(\tau')d\tau',
\end{equation}
where
\begin{equation}
f_{\alpha,\beta}(\tau)=\omega(t)\Theta(S_{\beta}\cap S_{\alpha}\neq\emptyset)+2\sum_{\beta',\sigma:S_{\beta'}\cap S_{\alpha}\neq\emptyset}\int_{0}^{\tau}\abs{K^\delta_{\beta,\beta',\sigma}(\tau')}d\tau'.
\end{equation}
Next, we derive a simpler upper bound on $f_{\alpha, \beta}(\tau)$ which allows us to use lemma \ref{lemma:modified_lr_recursion}. Note that  if $S_{\beta'}\cap S_{\alpha} \neq \emptyset$, the $d(S_{\alpha},S_{\beta'})=0$: Therefore, from Eq. \eqref{eq:SM_environment_distance}, we obtain $d(\text{x}_{\alpha},\text{x}_{\beta'})\leq q_{-}$.
Consequently, for any $\beta'$,
\begin{equation}
d(\text{x}_{\beta},\text{x}_{\beta'})\geq d(\text{x}_{\beta},\text{x}_{\alpha})-d(\text{x}_{\beta'},\text{x}_{\alpha})\geq d(\text{x}_{\beta},\text{x}_{\alpha})-q_{-},
\end{equation}
and hence
\begin{equation}
\begin{aligned}
\sum_{\beta',\sigma:S_{\beta'}\cap S_{\alpha}\neq\emptyset}\int_{0}^{\tau}\abs{K^\delta_{\beta,\beta',\sigma}(\tau')}d\tau'& \numleq{1} \sum_{\beta', \sigma:S_{\beta'}\cap S_{\alpha} \neq \emptyset} \text{TV}(K_{\beta, \beta', \sigma}; [-(\tau + 2\delta), (\tau + 2\delta)])\\
&\leq K_{0}\frac{\exp\left(\lambda_{E}(\tau + 2\delta)-\kappa_E(d(\text{x}_{\beta},\text{x}_{\alpha}) - q_-)\right)}{\left(1+\lambda_{0}(d(\text{x}_{\beta},\text{x}_{\alpha})-q_{-})\right)^{\eta}} \sum_{\beta', \sigma:S_{\beta'}\cap S_\alpha \neq \emptyset}1\\
&\numleq{2}2\mathcal{Z}K_{0}\frac{\exp(\kappa_E q_{-} + 2\lambda_E \delta )}{\min(\lambda_{0},1-\lambda_{0}q_{-})^{\eta}}F_{\alpha, \beta}(\tau),
\end{aligned}
\end{equation}
where $F_{\alpha, \beta}(\tau)$ is defined in lemma \ref{lemma:F_useful_property}. In (1), we have used lemma \ref{lemma:mollified_kernel} and in (2), we have used that $\lambda_0 q_- < 1$ which can be satisfied as per lemma \ref{lemma:modified_kernel_upper_bounds}. Furthermore, note that
\begin{align}
    \Theta(S_{\alpha}\cap S_{\beta}\neq\emptyset) \leq \frac{\exp(\lambda_{E}(\tau + 2\delta)-\kappa_E p_{-}d(S_{\alpha},S_{\beta}))}{(1+\lambda_{0}p_{-}d(S_{\alpha},S_{\beta}))^{\eta}},
\end{align}
since if $\mathcal{S}_\alpha \cap \mathcal{S}_\beta \neq \emptyset$ or, equivalently, $d(S_\alpha, S_\beta) = 0$ then the left-hand side is 1 while the right-hand side is $\exp(\lambda_E (\tau + 2\delta))>1$ while if $\mathcal{S}_\alpha \cap \mathcal{S}_\beta = \emptyset$ then the left-hand side is 0 while the right-hand side is non-negative. Using Eq.~\eqref{eq:SM_environment_distance}, we obtain that $p_- d(S_\alpha, S_\beta) \geq d(\text{x}_\alpha, \text{x}_\beta)-q_-$ and thus
\begin{align}
    \Theta(S_{\alpha}\cap S_{\beta}\neq\emptyset) \leq \frac{\exp(\kappa_E q_{-} + 2\lambda_E \delta )}{\min(\lambda_{0},1-\lambda_{0}q_{-})^{\eta}}F_{\alpha, \beta}(\tau).
\end{align}
Combining the upper bounds derived above, we then obtain
\begin{equation}\label{eq:definition_tilde_omega}
f_{\alpha,\beta}(\tau)\leq\tilde{\omega}(t)F_{\alpha,\beta}(\tau) \quad \text{where}\quad \tilde{\omega}(t)=\frac{\exp(\kappa_E q_{-} + 2\lambda_E \delta )}{\min(\lambda_{0},1-\lambda_{0}q_{-})^{\eta}}\left(\omega(t)+4\mathcal{Z}K_{0}\right).
\end{equation}
Thus, Eq.~\eqref{eq:SM_coupled_bath_inequality_Gamma} can be rewritten as
\begin{align}\label{eq:lr_modified_recursion_form_1}
    \Gamma_{X, S_\alpha}(\tau) \leq \Gamma_{X, S_\alpha}(0) + \tilde{\omega}(t) \sum_\beta \int_0^\tau F_{\alpha, \beta}(\tau - \tau') \Gamma_{X, S_\beta}(\tau') d\tau'.
\end{align}
Next, we note that if $\tau > 0$, then
\begin{align}
    \int_0^\tau \smallabs{K_{\alpha', \alpha, \sigma}^\delta(s)}ds \leq \text{TV}(K_{\alpha', \alpha, \sigma}; [-(\tau + 2\delta), (\tau + 2\delta)]) \leq K_0 \frac{\exp(\lambda_E (\tau + 2\delta) - \kappa_E d(\text{x}_\alpha, \text{x}_{\alpha'}))}{(1 + \lambda_0 d(\text{x}_\alpha, \text{x}_{\alpha'}))^\eta} \leq \frac{\tilde{\omega}_1(t)}{4}F_{\alpha, \alpha'}(\tau),
\end{align}
where $\tilde{\omega}_1(t) = 4K_0 \exp(2\lambda_E \delta)/(\min(1, \lambda_0))^\eta$. Thus from lemma \ref{lemma:coupled_bath_integration_expression_gamma_lr}, we obtain that
\begin{align}\label{eq:lr_modified_recursion_form_2}
    \Gamma_{X, Y}(\tau) \leq \Gamma_{X, Y}(0) + \omega(t)\sum_{\alpha:S_\alpha \cap Y \neq \emptyset}\int_0^\tau \Gamma_{X, S_\alpha}(\tau') d\tau' + \tilde{\omega}_1(t) \sum_{\alpha, \alpha': S_{\alpha}\cap Y \neq \emptyset} \int_0^\tau F_{\alpha, \alpha'}(\tau - \tau') \Gamma_{X, S_{\alpha'}}(\tau') d\tau'.
\end{align}
For $d(X,Y)>2a_{0}$, using lemma \ref{lemma:modified_lr_recursion} with Eqs.~\eqref{eq:lr_modified_recursion_form_1} and \eqref{eq:lr_modified_recursion_form_2}, we obtain that
\begin{subequations}
\begin{align}
    \Gamma_{X, Y}(\tau) \leq \abs{X}\abs{Y}\Gamma(t)\exp\bigg(\frac{\kappa_E}{p_+}(v_{\text{LR}}(\bm\zeta, t) \tau-d(X, Y))\bigg),
\end{align}
where
\begin{align}
    v_{\text{LR}}(\bm{\zeta}, t) = v_0 + V_0 \big(\smallnorm{\bm{\zeta}_{[0, t]}}_\infty + 2\smallnorm{\bm\mu^\text{R}(-t, t)}_{\infty}\big)
\end{align}
with
\begin{align*}
          v_0 = \frac{p_+\lambda_E}{\kappa_E}  +\frac{2p_+B_\eta \exp(\kappa_E q_{-} + 2\lambda_E \delta )}{\kappa_E \min(\lambda_{0},1-\lambda_{0}q_{-})^{\eta}}(1 + 2K_0\mathcal{Z}) \quad \text{and}\quad V_0 = \frac{4p_+B_\eta \exp(\kappa_E q_{-} + 2\lambda_E \delta )}{\kappa_E \min(\lambda_{0},1-\lambda_{0}q_{-})^{\eta}},
\end{align*}
and
\begin{align}
    \Gamma(t) = \mathcal{Z}^{2}\exp\bigg( \frac{\kappa_E(2a_0 + q_+)}{p_+}\bigg)\bigg[\frac{\omega(t) \tilde{\omega}(t)}{(\lambda_E + \tilde{\omega}(t) B_\eta)^{2}} + \frac{\tilde{\omega}_1(t)}{\tilde{\omega}(t) B_\eta + \lambda_E}\bigg].
\end{align}
\end{subequations}
To obtain the precise bound in the lemma statement, we note that $\gamma^{X, Y}_{\bm \zeta}(t, t') = \Gamma_{X, Y}(t - t')$. Furthermore, $\tilde{\omega}(t) \geq 0$ implies $\Gamma(t) \leq \Gamma_0(t)$, where $\Gamma_0(t)$ is defined in the lemma statement. For $d(X,Y)\leq2a_{0}$, the claimed bound follows from $\gamma^{X,Y}_{\bm\zeta,\mathcal B}(t,t')\leq1$ and $\Gamma_{0}(t)\geq\exp(2\kappa_{E}a_{0}/p_{+})$. Thus the estimate holds for all $X$ and $Y$.

\end{proof}

\begin{proof}[Proof of proposition 2]
We will again analyze the observable error $\Delta_{O_X}(t, t'; l)$ for the $\delta-$mollified model, and take the limit of $\delta \to 0$ at the end. By Lemma \ref{lemma:delta_o_comm_bd}, for $t>t'$ the truncation error decomposes as
\[
\Delta_{O_X}(t,t';l)\leq \Delta_{O_X}^{\text{R}}(t,t';l)+\Delta^\text{L}_{O_X}(t,t';l).
\]
We will separately upper-bound the both the terms on the right-hand side of this inequality. It will be convenient to define $v_{\text{LR}, c}(t)$ via
\[
v_{\text{LR}, c}(t) = v_0 + 6V_0 C_E (t + 2\delta)^d + 6V_0 B_E,
\]
with $v_0$ and $V_0$ being as defined in lemma \ref{lemma:coupled_bath_commuting_LR_bound} and $B_E, C_E$ are defined in lemma \ref{lemma:coupled_bath_finiteness_on_memory}. We remark that $v_{\text{LR}, c}(t) = \Theta(t^d)$ as $t\to \infty$. Furthermore, since from lemma \ref{lemma:coupled_bath_finiteness_on_memory}, $\smallnorm{\bm\mu^{\text{R}}_{[0, t]}}_\infty \leq 2 (C_E (t + 2\delta)^d + B_E)$, it follows that the Lieb-Robinson velocity defined in lemma \ref{lemma:coupled_bath_commuting_LR_bound} satisfies
\[
v_{\text{LR}}(\bm \zeta, t) \leq v_{\text{LR}, c}(t) \quad \text{for}\quad \bm\zeta = \bm \mu_{[0, s]}^{\text{R}} \quad \text{with}\quad 0 \leq s \leq t.
\]
Similarly, we define
\[
\begin{aligned}\Gamma_{0, c}(t) ={}& \exp\bigg(\frac{2\kappa_{E}a_{0}}{p_{+}}\bigg)\\ &+\frac{\mathcal Z^{2}\exp(\kappa_{E}(2a_{0}+q_{+})/p_{+}+\kappa_E q_- + 2\lambda_E\delta)}{(\min(1, \lambda_0, 1 - \lambda_0 q_-))^\eta}\bigg[\frac{\omega_c^2(t) + 4K_0 \mathcal{Z}\omega_c(t)}{\lambda_{E}^{2}}+\frac{4 K_0}{\lambda_{E}}\bigg]\\ \text{where}\quad &\omega_c(t) = 2 + 24(C_E (t + 2\delta)^d + B_E),
\end{aligned}
\]
and remark that, also as a consequence of lemma \ref{lemma:coupled_bath_finiteness_on_memory}, $\Gamma_0(t) \leq \Gamma_{0, c}(t)$. Furthermore, it can be noted that $\Gamma_{0, c}(t) = \Theta(t^{2d})$ as $t \to \infty$.

In terms of $v_{\text{LR}, c}(t)$ and $\Gamma_{0, c}(t)$, we can then easily see that lemma \ref{lemma:coupled_bath_commuting_LR_bound} implies that for all $s \in [0, t]$ and $t > t'$,
\begin{align}\label{eq:commutator_bound_coupled_useful}
    \gamma^{X, Y}_{\bm \mu^\text{R}_{[0, s]}, \mathcal{A}_{X[l]}}(t, t') \leq \abs{X}\abs{Y}\Gamma_{0, c}(t) \exp\bigg(\frac{\kappa_E}{p_+}(v_{\text{LR}, c}(t)(t - t') - d(X, Y))\bigg).
\end{align}
\noindent\underline{Upper bounding $\Delta_{O_X}^{\text{R}}(t, t'; l)$}. From lemma \ref{lemma:Delta_bd_gamma},
\begin{equation}
\Delta_{O_X}^{\text{R}}(t,t',l)\leq2\norm{O_{X}}\sum_{\substack{\alpha\notin\mathcal{A}_{X[l]}\\S_{\alpha}\cap X_{c}[l]{\neq\emptyset}}}(1+\mu^{\textnormal{R}}_{\alpha}(0,t))\int_{s=t'}^{t}\gamma_{\bm{\mu}_{[0,s]}^{R},\mathcal{A}_{X[l]}}^{X,S_{\alpha}}(t,s)ds.
\end{equation}
From lemma \ref{lemma:coupled_bath_finiteness_on_memory}, it follows that $\mu^{\textnormal{R}}_{\alpha}(0,t)\leq2(C_E(t + 2\delta)^{d} + B_E)$.
Futhermore, Eq.~\eqref{eq:commutator_bound_coupled_useful} yields
\begin{align}
\int_{s=t'}^{t}\gamma_{\bm{\mu}_{[0,s]}^\text{R},\mathcal{A}_{X[l]}}^{X,S_{\alpha}}(t,s)ds & \leq \abs{\mathcal{S}_\alpha} \abs{X} \Gamma_{0, c}(t)\int_{s=t'}^{t}\exp\left(\frac{\kappa_E}{p_+}({v_{\textnormal{LR,c}}(t)(t-s)-d(X,S_{\alpha})})\right)ds\nonumber \\
 & \leq (a_{0}+1)^{d}\abs{X} \frac{p_+\Gamma_{0, c}(t)}{\kappa_E v_{\text{LR}, c}(t)}\exp\left(\frac{\kappa_E}{p_+}({v_{\textnormal{LR,c}}(t)(t-t')-d(X,S_{\alpha})})\right)\nonumber \\
 &\leq (a_{0}+1)^{d}\abs{X} \frac{\Gamma_{0, c}(t)}{\lambda_E}\exp\left(\frac{\kappa_E}{p_+}({v_{\textnormal{LR,c}}(t)(t-t')-d(X,S_{\alpha})})\right),
\end{align}
where we have used $\abs{S_{\alpha}}\leq(a_{0}+1)^{d}$ and that $v_{\text{LR}, c}(t) \geq v_0 \geq p_+ \lambda_E/\kappa_E$.
Thus, we obtain that
\begin{equation}
\begin{aligned}\Delta_{O_{X}}^\text{R}(t,t';l) & \leq\frac{2(a_{0}+1)^{d}}{\lambda_E}\abs{X}\norm{O_{X}}(1+ 2B_E + 2C_E(t+2\delta)^d)\Gamma_{0,c}(t)\\ &\quad\times\exp\left(\frac{\kappa_E v_{\textnormal{LR,c}}(t)(t-t')}{p_{+}}\right)\sum_{\substack{\alpha\notin\mathcal{A}_{X[l]}\\
S_{\alpha}\cap X_{c}[l]\neq\emptyset
}
}\exp\left(-\frac{\kappa_E d(X,S_{\alpha})}{p_{+}}\right).\end{aligned}
\end{equation}
For these supports, $d(X,S_\alpha)>l$ and $S_{\alpha}$ intersects $X[l+a_{0}]\setminus X[l]$. Counting this boundary by a union of lattice shells centered at sites of $X$ gives, for a constant $C_{d}>0$ depending only on $d$,
\begin{align}
\sum_{\substack{\alpha\notin\mathcal{A}_{X[l]}\\
S_{\alpha}\cap X_{c}[l]\neq\emptyset
}
} \exp\left(-\frac{\kappa_E d(X,S_{\alpha})}{p_{+}}\right)\leq C_{d}\mathcal Z\abs{X}(1+a_{0})(1+l+a_0)^{d-1} \exp\bigg(-\frac{\kappa_E l}{p_+}\bigg).
\end{align}
Therefore, we obtain that
\begin{align}\label{eq:coupled_both_upper_bound_delta_R}
    \Delta_{O_X}^\text{R}(t, t'; l) \leq \norm{O_X} f^\text{R}(l) g^\text{R}(t) \exp\left(\frac{\kappa_E}{p_{+}}(v_{\text{LR}, c}(t)(t - t') - l)\right),
\end{align}
where $f^\text{R}(l) \leq O(l^{d-1})$ as $l \to \infty$ and $g^\text{R}(t) \leq O(t^{3d})$ as $t \to \infty$. \\

\noindent\underline{Upper bounding $\Delta_{O_X}^{\text{L}}(t, t'; l)$}.
Returning to lemma \ref{lemma:Delta_bd_gamma}:
\begin{align}\label{eq:coupled_error_in_terms_commutator_2}
    \Delta^{\textnormal{L}}_{O_X}(t, t';l)\leq 2 \norm{O_X} \sum_{\alpha' \in \mathcal{A}_{X[l]}} \sum_{\alpha \notin \mathcal{A}_{X[l]}, \sigma} \int_{s' = t'}^t \int_{s = t'}^{s'}\smallabs{K^\delta_{\alpha', \alpha, \sigma}(s'- s)}\gamma^{X, S_{\alpha'}}_{\bm{\mu}^{\textnormal{R}}_{[0, s']}, \mathcal{A}_{X[l]}}(t, s') ds\,ds'.
\end{align}
First integrate over $s$ at fixed $s'$. The total-variation bound gives
\[
\int_{t'}^{s'}\smallabs{K^\delta_{\alpha',\alpha,\sigma}(s'-s)}\,ds
\leq \frac{K_0\exp(2\lambda_{E}\delta)}{\min(1,\lambda_0)^\eta}F_{\alpha,\alpha'}(s'-t'),
\]
where $F_{\alpha,\alpha'}$ is defined in lemma \ref{lemma:F_useful_property}. Since $\lambda_E \leq \kappa_E v_{\textnormal{LR},c}(t)/p_+$, the remaining time integral satisfies
\begin{align}
    \int_{s' = t'}^t \exp(\lambda_E(s' - t')) \exp\bigg(\frac{\kappa_Ev_{\text{LR}, c}(t)}{p_+}(t - s')\bigg) ds' \leq t \exp\bigg(\frac{\kappa_Ev_{\text{LR}, c}(t)}{p_+}(t - t')\bigg).
\end{align}
Therefore, we obtain that
\begin{equation}
\begin{aligned}
\Delta^{\textnormal{L}}_{O_X}(t, t';l)&\leq \frac{4K_0\exp(2\lambda_{E}\delta)(a_0+1)^d t\Gamma_{0, c}(t)}{\min(1, \lambda_0)^\eta}\abs{X} \norm{O_X}\exp\bigg(\frac{\kappa_Ev_{\text{LR}, c}(t)}{p_+}(t - t')\bigg)\\ &\quad\times\sum_{\alpha' \in \mathcal{A}_{X[l]}} \sum_{\alpha \notin \mathcal{A}_{X[l]}}  \frac{\exp(-\kappa_E (p_{+}d(\text{x}_{\alpha}, \text{x}_{\alpha'}) + d({S}_{\alpha'}, X))/p_+)}{(1 + d(\text{x}_{\alpha}, \text{x}_{\alpha'}))^\eta} .
\end{aligned}
\end{equation}
Next, using Eq.~\eqref{eq:SM_environment_distance} together with the lattice triangle inequality and $\mathrm{diam}(S_{\alpha'})\leq a_{0}$, we obtain
\begin{align}
    p_+d(\text{x}_\alpha, \text{x}_{\alpha'}) + d(S_{\alpha'}, X) \geq d({S}_{\alpha}, S_{\alpha'}) + d(S_{\alpha'}, X) - q_+ \geq d(S_\alpha, X) -(q_+ + a_0),
\end{align}
Enforcing $\alpha' \in \mathcal{A}_{X[l]}, \alpha \notin \mathcal{A}_{X[l]}$, we then obtain $p_+d(\text{x}_\alpha, \text{x}_{\alpha'}) + d(S_{\alpha'}, X) \geq l - (q_+ + a_0)$. Therefore,
\begin{align}
    &\sum_{\alpha' \in \mathcal{A}_{X[l]}} \sum_{\alpha \notin \mathcal{A}_{X[l]}}  \frac{\exp(-\kappa_E (p_{+}d(\text{x}_{\alpha}, \text{x}_{\alpha'}) + d({S}_{\alpha'}, X))/p_+)}{(1 + d(\text{x}_{\alpha}, \text{x}_{\alpha'}))^\eta} \nonumber\\
    &\qquad \leq \exp(-\kappa_E l/p_+)\exp(\kappa_E (q_+ + a_0)/p_+)\sum_{\alpha' \in \mathcal{A}_{X[l]}} \sum_{\alpha \notin \mathcal{A}_{X[l]}}  \frac{1}{(1 + d(\text{x}_\alpha, \text{x}_{\alpha'}))^\eta} \nonumber \\
    &\qquad \numleq{1} \frac{B_\eta}{2^\eta} \exp(-\kappa_E l/p_+)\exp(\kappa_E (q_+ + a_0)/p_+) \sum_{\alpha \in \mathcal{A}_{X[l]}} 1\nonumber\\
    &\qquad \leq \frac{B_\eta}{2^\eta} \exp(-\kappa_E l/p_+)\exp(\kappa_E (q_+ + a_0)/p_+) \mathcal{Z}\abs{X}(2l+1)^d,
\end{align}
where in (1) we have used lemma \ref{lemma:coupled_both_summations}. Combining the above estimates, we then obtain
\begin{align}\label{eq:coupled_both_upper_bound_delta_L}
    \Delta_{O_X}^\text{L}(t, t'; l) \leq \norm{O_X}f^\text{L}(l) g^\text{L}(t)  \exp\left(\frac{\kappa_E}{p_{+}}(v_{\text{LR}, c}(t)(t - t') - l)\right),
\end{align}
where $f^\text{L}(l) \leq O(l^d)$ as $l \to \infty$ and $g^\text{L}(t) \leq O(t^{2d + 1})$.

Combining the upper bounds in Eqs.~\eqref{eq:coupled_both_upper_bound_delta_R} and \eqref{eq:coupled_both_upper_bound_delta_L}, we obtain that
\[
\Delta_{O_X}(t, t'; l) \leq \norm{O_X}f(l)g(t)  \exp\left(\frac{\kappa_E}{p_{+}}(v_{\text{LR}, c}(t)(t - t') - l)\right),
\]
where $f(l) = f^\text{R}(l)+f^\text{L}(l) \leq O(l^d)$ as $l \to \infty$ and $g(t) = \max(g^\text{R}(t), g^\text{L}(t)) \leq O(t^{3d})$ as $t \to \infty$. Finally, recall that this bound holds for the $\delta-$mollified kernel. However, from proposition \ref{lemma:mollification}, it follows that the limit of $\Delta_{O_X}(t, t'; l)$ as $\delta \to 0$ exists. It can be checked that the limit of $\delta \to 0$ of the right hand side in the above inequality also exists: Taking this limit is equivalent to setting $\delta = 0$ in the expressions for $f(l), g(t)$ and $v_{\text{LR}, c}(t)$. This proves the proposition.
\end{proof}
\endgroup

\section{{\color{black}Saturating the Lieb-Robinson bounds: Violation of the linear light-cone}}
This section gives independent-bath and coupled-bath examples that saturate the velocity scalings in Propositions~1 and~2, respectively. In the independent-bath example, the memory measure $m(t)$ defined in Eq.~\eqref{eq:one_norm_measure} grows linearly with $t$, and the transport velocity has the same scaling. The coupled-bath example establishes the tightness of the exponential-tail scaling in one dimension.
\subsection{Independent bath model}\label{sec:supp_independent_saturation}
Here we analyze the explicit example with independent baths, when the memory kernels $\K_{\alpha,\alpha'}$ have {\color{black}windowed total variation that grows linearly in time, which induces a supersonic transportation and will violate a linear Lieb-Robinson light cone.} We first provide the system-environment Hamiltonian that was sketched out in the main text. The system will be a 1D chain of $n$ qubits, and the environment will have $n - 1$ baths, each of which will have only one bosonic mode. Specifically, we label the bath operators by $\alpha=(i,\nu)$ and set $B_{(i,\nu)}(t)=a_i^{(\nu)}$ for all $t$, where $i\in\{1,\ldots,n-1\}$ and $\nu\in\{-,+\}$, with $a_i^{(-)}=a_i$ and $a_i^{(+)}=a_i^\dagger$. The $i^\text{th}$ bath will interact with the $i^\text{th}$ and $(i + 1)^\text{th}$ qubit. {We use $\sigma_i:=\ket{0_i}\!\bra{1_i}$, $\sigma_i^\dagger:=\ket{1_i}\!\bra{0_i}$, $\sigma_i^x:=\sigma_i+\sigma_i^\dagger$, and $\sigma_i^z:=\ket{0_i}\!\bra{0_i}-\ket{1_i}\!\bra{1_i}$.} All bath modes start in the vacuum. Define $K_i(\tau):=K_{(i,-),(i,+)}(\tau)=1$; the other vacuum two-point functions vanish. The one-sided memory budget and the two-sided memory measure are therefore
\begingroup\color{black}
\begin{equation}\label{eq:SM_independent_bath_saturating_LR_v}
\mu^{\mathrm{R}}_{(i,-)}(0,t)=\mathrm{TV}(K_i;[0,t])=t,
\qquad
m(t)=\mathrm{TV}(K_i;[-t,t])=2t.
\end{equation}
In particular, at the observation time $t=2T+1$, these quantities equal $2T+1$ and $4T+2$, respectively. Proposition~\ref{prop:lr_bound_main} consequently gives $v_{\mathrm{LR}}(2T+1)=O(T)$. We will construct a protocol with transport velocity $v_{\mathrm{S}}(T)=\Omega(T)$, showing that this scaling is tight up to constants. Here $f(T)=\Omega(g(T))$ means that $f(T)\geq c g(T)$ for some $c>0$ and all sufficiently large $T$.
\endgroup

\begingroup\color{black}
Fix a positive integer $T$ and the observation time $2T+1$. The Hamiltonian $H(t)$ is defined in three stages:
\begin{equation}
H(t)=\begin{cases}
H_{\mathrm{pump}}(t), & 0\leq t<T;\\
H_{\mathrm{ex}}(t), & T\leq t<T+1;\\
H_{\mathrm{prop}}(t), & T+1\leq t<2T+1.
\end{cases}
\end{equation}
\endgroup
where, as the names suggest, $H_\text{pump}(t)$ pumps the environment oscillators with photons, $H_\text{ex}(t)$ puts an excitation in the system and $H_\text{prop}(t)$ propagates the excitation. In order to ensure that the time-dependent system Hamiltonian and the system coupling operators are smooth, we will use a pulse function $\xi: \mathbb{R}\to \mathbb{R}$ in their construction --- $\xi$ can be chosen to be any smooth function which satisfies
\begin{equation}
\xi(t) = 0 \text{ for all } t \notin [0, 1] \text{ and }\int_0^1 \xi(s) ds = \frac{\pi}{2}.\label{eq:SM_definition_of_smooth_xis}
\end{equation}

\begingroup\color{black}
\emph{Description of $H_\mathrm{pump}(t)$}. Physically, the Hamiltonian $H_\mathrm{pump}(t)$ serves as a constant driving field for each environmental mode, creating a coherent state with high average photon numbers. More specifically,
\begin{equation}
H_\mathrm{pump}(t)=i\xi\left(\frac{t}{T}\right)\frac{2}{\pi}\sum_{i=1}^{n-1}\sigma_i^{z}\left(a_i^\dagger-a_i\right),
\end{equation}
where $\xi(s)$ is defined in Eq. (\ref{eq:SM_definition_of_smooth_xis}).
If the system is initially in $\ket{0}^{\otimes n}$ and the environment is initially in the vacuum state, we directly identify that this pumping Hamiltonian will drive the environmental mode to the coherent state $\ket{\alpha=T}$ as
\begin{equation}
\begin{aligned}
\ket{\psi(T)}
&=\mathcal{T}\exp\!\left(-i\int_0^T H_{\mathrm{pump}}(s)\,ds\right)
\left(\ket{0}_S^{\otimes n}\otimes\ket{0}_E^{\otimes(n-1)}\right)\\
&=\prod_{i=1}^{n-1}\exp\!\left(T\sigma_i^z(a_i^\dagger-a_i)\right)
\left(\ket{0}_S^{\otimes n}\otimes\ket{0}_E^{\otimes(n-1)}\right)\\
&=\ket{0}_S^{\otimes n}\otimes\ket{\boldsymbol{\alpha}_T}_E,
\qquad
\ket{\boldsymbol{\alpha}_T}_E:=\bigotimes_{i=1}^{n-1}\ket{\alpha=T}_{E_i}.
\end{aligned}
\end{equation}
Here $\ket{\alpha=T}_{E_i}$ is the coherent state of mode $i$ with real displacement $T$; the scalar coherent amplitude $\alpha$ is distinct from the bath-index label $(i,\nu)$.
\endgroup

\begingroup\color{black}
\emph{Description of $H_\mathrm{ex}(t)$}. $H_\mathrm{ex}(t)$ acts only on the first qubit and applies a $\sigma_1^{x}$ rotation on it.
\begin{equation}
H_\mathrm{ex}(t)=\xi(t-T)\sigma_1^{x}.
\end{equation}
This pulse is supported on the singleton $\{1\}$ and creates an excitation on the first qubit. Up to a global phase, the resulting state is $\ket{1_10_2\cdots0_n}\otimes\ket{\boldsymbol{\alpha}_T}_E$.
\endgroup

\begingroup\color{black}
\emph{Description of $H_\mathrm{prop}(t)$}.This Hamiltonian leverages the excitations in the environment to rapidly propagate the excitations in the system from site $i=1$ to site $i=T^2+1$. We construct $H_\text{prop}(t)$ by first switching on an interaction between qubits 1, 2 and the first bath, then between qubits 2, 3 and the second bath and so on:
\begin{equation}\label{eq:SM_independent_propH}
H_{\mathrm{prop}}(t)=\sum_{i=1}^{n-1}g_{i}(t)\ket{0_{i}1_{i+1}}{\bra{1_i0_{i+1}}}\otimes a{_i}+\text{h.c.}=:\sum_{i=1}^{n-1}g_i(t)h_{i,i+1},
\end{equation}
with $g_i(t)=\xi\bigl(T(t-T-1)-(i-1)\bigr)$. Intuitively, each environmental mode has an average photon number $T^2$. For a Fock state with exactly the same photon number, the propagation of excitation by $1$ site using  Jaynes--Cummings coupling Eq. (\ref{eq:SM_independent_propH}) requires $1/T$ time. When $T$ is very large, the photon number fluctuation in the coherent state will be suppressed, and the coherent state should have a similar effect as the Fock state with the same average photon number. Therefore, we can propagate the excitation by $T^2$ sites in $T$ time. More concretely, we have
\begin{lemma}
    {Let $n\geq T^2+1$.} If the system is initially in $\ket{0}^{\otimes n}$ and the environment is initially in the vacuum state, {the occupation probability at site $T^2+1$ after evolution under $H(t)$ remains bounded below by a nonzero constant.} In other words, {defining} $\ket{\psi(t)}=U(t,0)(\ket{0}^{\otimes n}\otimes\ket{0}^{\otimes (n-1)})$, we have, for large enough $T$,
    \[
    \mathrm{Tr}\left(\sigma_{T^2+1}^{{\dagger}}\sigma_{T^2+1} \ket{\psi(2T+1)}\bra{\psi(2T+1)}\right)\geq C,
    \]
with $C$ a constant which does not depend on $T$ and $n$.
\label{Lemma:independent_bath_supersonic_saturate}
\end{lemma}

\begingroup\color{black}
Before proving this lemma, we establish a one-step Jaynes--Cummings estimate.
\begin{lemma}
Consider the Jaynes--Cummings Hamiltonian
\begin{equation}
H=\sigma^\dagger a+\sigma a^{\dagger},
\end{equation}
where $\sigma=\ket{0}\!\bra{1}$ acts on the qubit and $a$ acts on a single bosonic mode. Tensor factors below are ordered as qubit then mode. For an interaction time $\Delta t>0$, define the transition-amplitude magnitude
\begin{equation}
A_{\mathrm{tr}}:=\left|\left(\bra{1}\otimes\bra{\alpha}\right)e^{-iH\Delta t}\left(\ket{0}\otimes\ket{\alpha}\right)\right|.\label{eq:SM_transition_probability_indepen_satura}
\end{equation}
If $\alpha\Delta t=\pi/2$ and $\alpha>0$, then
\begin{equation}
A_{\mathrm{tr}}\geq1-\frac{C_{\mathrm{JC}}}{\alpha^2},
\qquad C_{\mathrm{JC}}:=\frac12.
\label{eq:SM_JC_amplitude_lower_bound}
\end{equation}
\label{Lemma:SM_satu_Fact_about_JC_model}
\end{lemma}
\begin{proof}
Let $D(\alpha)$ be the displacement operator acting on the bosonic mode, so that $D(\alpha)\ket{0}=\ket{\alpha}$ and $D(\alpha)^\dagger aD(\alpha)=a+\alpha$. We suppress its identity action on the qubit. In the displaced frame,
\begin{equation}
H':=D(\alpha)^\dagger H D(\alpha)=H_0+V,
\qquad
H_0:=\alpha\sigma^x,
\qquad
V:=\sigma^\dagger a+\sigma a^\dagger.
\end{equation}
Writing $\ket{\Phi_0}:=\ket{0}\otimes\ket{0}$, the choice $\Delta t=\pi/(2\alpha)$ gives
\begin{equation}
e^{-iH_0\Delta t}\ket{\Phi_0}=-i\ket{1}\otimes\ket{0}=: \ket{\Phi_{\mathrm{id}}}.
\end{equation}
Notice that we can now rewrite $A_{\mathrm{tr}}$ as 
\[
A_{\mathrm{tr}}=\left|\bra{\Phi_{\mathrm{id}}}e^{-iH'\Delta t}\ket{\Phi_0}\right|.
\]
Applying Duhamel's identity we obtain 
\begin{equation}
\left(e^{-iH'\Delta t}-e^{-iH_0\Delta t}\right)\ket{\Phi_0}
=-i\int_0^{\Delta t}e^{-iH'(\Delta t-s)}V e^{-iH_0s}\ket{\Phi_0}\,ds.
\end{equation}
Moreover,
\begin{equation}
e^{-iH_0s}\ket{\Phi_0}
=\cos(\alpha s)\ket{0,0}-i\sin(\alpha s)\ket{1,0},
\qquad
V e^{-iH_0s}\ket{\Phi_0}
=-i\sin(\alpha s)\ket{0}\otimes\ket{1},
\end{equation}
and hence
\begin{equation}
\left\|e^{-iH'\Delta t}\ket{\Phi_0}-\ket{\Phi_{\mathrm{id}}}\right\|
\leq\int_0^{\Delta t}|\sin(\alpha s)|\,ds
=\frac{1}{\alpha}.
\end{equation}
Both vectors are normalized, so
\begin{equation}
A_{\mathrm{tr}}
=\left|\langle\Phi_{\mathrm{id}}|e^{-iH'\Delta t}|\Phi_0\rangle\right|
\geq\operatorname{Re}\!\left\langle\Phi_{\mathrm{id}}\middle|e^{-iH'\Delta t}\middle|\Phi_0\right\rangle
=1-\frac12\left\|e^{-iH'\Delta t}\ket{\Phi_0}-\ket{\Phi_{\mathrm{id}}}\right\|^2
\geq1-\frac{1}{2\alpha^2},
\end{equation}
which proves the claim.
\end{proof}
\endgroup
\begingroup\color{black}
We now prove Lemma~\ref{Lemma:independent_bath_supersonic_saturate}.
\begin{proof}
After the pumping and excitation stages, the joint state is, up to a global phase,
\begin{equation}
\ket{\psi(T+1)}=\ket{q_0}\otimes\ket{\boldsymbol{\alpha}_T}_E,
\qquad
\ket{q_0}:=\ket{1_10_2\cdots0_n}.
\end{equation}
For $j=1,\ldots,T^2$, define
\begin{equation}
U_j:=\exp\!\left(-i\frac{\pi}{2T}h_{j,j+1}\right),
\qquad
\ket{q_j}:=\ket{0_1\cdots0_j1_{j+1}0_{j+2}\cdots0_n}.
\end{equation}
The pulse supports are ordered, so
\begin{equation}
U_\mathrm{prop}(2T+1,T+1)=U_{T^2}\cdots U_2U_1.
\label{eq:SM_U_prop_satu_in}
\end{equation}
It is easy to verify that
\begin{align}
&\abs{\left(\bra{q_{T^2}}\otimes{}_E\bra{\boldsymbol{\alpha}_T}\right)
U_{T^2}\cdots U_1
\left(\ket{q_0}\otimes\ket{\boldsymbol{\alpha}_T}_E\right)}\nonumber\\
&\quad\numeq{1}\abs{{}_E\bra{\boldsymbol{\alpha}_T}
\left(\prod_{j=T^2}^{1}\bra{q_j}U_j\ket{q_{j-1}}\right)
\ket{\boldsymbol{\alpha}_T}_E}\nonumber\\
&\quad\numeq{2}\prod_{j=1}^{T^2}
\abs{\left(\bra{0_j1_{j+1}}\otimes{}_{E_j}\bra{\alpha=T}\right)
U_j
\left(\ket{1_j0_{j+1}}\otimes\ket{\alpha=T}_{E_j}\right)},
\label{eq:SM_exact_forward_factorization}
\end{align}
where (1) follows from the fact that a branch that does not transfer the excitation from $j$ to $j+1$ is annihilated by the projection onto $\ket{q_{T^2}}$ and (2) follows from the fact that $U_j$ only acts on the sites $j,j+1$ and the $j^{\text{th}}$ bath.
On the two-dimensional subspace spanned by $\ket{1_j0_{j+1}}$ and $\ket{0_j1_{j+1}}$, $h_{j,j+1}$ is the Jaynes--Cummings Hamiltonian of Lemma~\ref{Lemma:SM_satu_Fact_about_JC_model}, and thus each factor in Eq.~\eqref{eq:SM_exact_forward_factorization} is bounded below by Eq.~\eqref{eq:SM_JC_amplitude_lower_bound}, with $\alpha=T$ and effective interaction time $\Delta t=\int g_j(t)\,dt=\pi/(2T)$. The physical duration of each pulse is $1/T$.

Let $P_{T^2+1}$ denote the final occupation probability of site $T^2+1$. Projecting the environment onto $\ket{\boldsymbol{\alpha}_T}_E$ can only decrease that probability. Hence Lemma~\ref{Lemma:SM_satu_Fact_about_JC_model}, with $\alpha=T$, yields
\begin{align}
P_{T^2+1}
&\geq\left|
\left(\bra{q_{T^2}}\otimes{}_E\bra{\boldsymbol{\alpha}_T}\right)
U_{\mathrm{prop}}(2T+1,T+1)\ket{\psi(T+1)}
\right|^2\nonumber\\
&\geq\left(1-\frac{C_{\mathrm{JC}}}{T^2}\right)^{2T^2}.
\end{align}
For $T^2\geq2C_{\mathrm{JC}}$, the inequality $\log(1-y)\geq-3y/2$ with $y\leq 1/2$ gives
\begin{equation}
P_{T^2+1}\geq e^{-3C_{\mathrm{JC}}}=:C>0,
\label{eq:SM_independent_saturation_observable_against_prop1}
\end{equation}
where $C$ is independent of $T$ and $n$. This proves the lemma.
\end{proof}
\endgroup
According to Lemma \ref{Lemma:independent_bath_supersonic_saturate}, we transport an excitation by distance $T^2$ within time $2T+1$. Therefore, the particle transport velocity is
\begin{equation}
v_{\mathrm{S}}(T)=\Omega(T).
\end{equation}
Compared with Eq. (\ref{eq:SM_independent_bath_saturating_LR_v}), we see that the derived Lieb-Robinson velocity is tight up to some constant. This example with {\color{black}windowed environmental memory growing linearly in time} also violates the conventional linear light-cone structure, as we formally show below:

\begin{lemma}\label{lemma:supersonic_truncated}
    Let $X=\{T^2+1\}$ and $O_X=\sigma_{T^2+1}^\dagger\sigma_{T^2+1}$. Starting from the all-zero system state and the bath vacuum, evolution under $H_{X[l]}(t)$ with $0\leq l<T^2$ leaves the reduced system state equal to $(\ket{0}\!\bra{0})^{\otimes n}$ at time $2T+1$.
\end{lemma}
\begin{proof}
    Since $l<T^2$, the singleton support $\{1\}$ of $H_{\mathrm{ex}}(t)$ does not intersect $X[l]$, so this pulse is absent from $H_{X[l]}(t)$. The retained pumping terms leave the all-zero system state unchanged. Thus, at time $T+1$, all system qubits remain in $\ket{0}$. Every propagation term annihilates the all-zero system state, so this remains true through time $2T+1$.
\end{proof}

\noindent \textbf{Violation of the conventional linear light-cone structure}: We now consider the conventional linear light-cone structure, which can be obtained by considering a similar formula in Proposition~1 but with a constant $v_{\text{LR}}$. Using lemmas \ref{Lemma:independent_bath_supersonic_saturate} and \ref{lemma:supersonic_truncated}, we obtain that for an initial system state $\rho_S = (\ket{0}\!\bra{0})^{\otimes n}$, environment state $\rho_E = (\ket{0}\!\bra{0})^{\otimes {n-1}} $, and observable $O_X=\sigma_{T^2+1}^\dagger \sigma_{T^2+1}$
\[
\text{Tr}( O_X U(2T + 1, 0) (\rho_S \otimes \rho_E)  U(0, 2T + 1)) \geq C, \text{ and }\text{Tr}(O_X U_{X[l]}(2T + 1, 0) (\rho_S \otimes \rho_E)  U_{X[l]}(0, 2T + 1)) = 0,
\]
where we have assumed $l < T^2$. This immediately implies that $\Delta_{O_X}(2T+1,0;l)\geq C$ for any such $l$, thus violating the linear light-cone in Proposition 1 if $v_{\text{LR}}(t)$ is uniformly bounded for all $t$.
\endgroup

\begingroup\color{black}
\subsection{Coupled baths}\label{sec:supp_coupled_saturation}

Proposition~2 gives $v_{\mathrm{LR,c}}(t)=O(t)$ in one dimension. We now construct a coupled-bath example that saturates this scaling in the exponential tail of the truncation error. We consider a one-dimensional chain of $n$ qubits, indexed by $i=1,\ldots,n$, coupled to a bosonic field initially in the vacuum. Fix a positive integer $T$. In the Schr\"odinger picture the total Hamiltonian is
\begin{equation}
H_{\mathrm{tot}}(t)=H_{\mathrm{SE}}^{\mathrm{Sch}}(t)+H_E+H_S(t),
\end{equation}
where
\[
H_E=\int_{\mathbb R}\omega a_\omega^\dagger a_\omega\,d\omega,
\qquad
H_{\mathrm{SE}}^{\mathrm{Sch}}(t)=\sum_{i=1}^{n-1}
\bigl(R_i(t)A_i^\dagger+R_i^\dagger(t)A_i\bigr).
\]
The spatial field operators are
\begin{equation}
A_x=\int_{\mathbb R}\frac{d\omega}{\sqrt{2\pi}}a_\omega e^{-i\omega x},
\qquad x\in\mathbb R,
\end{equation}
with $[a_\omega,a_{\omega'}^\dagger]=\delta(\omega-\omega')$ and hence $[A_x,A_y^\dagger]=\delta(x-y)$. Each $R_i(t)$ acts on qubits $i,i+1$. In the interaction picture with respect to $H_E$, we write $H(t)=H_{\mathrm{SE}}(t)+H_S(t)$, where
\begin{equation}
H_{\mathrm{SE}}(t)=\sum_{i=1}^{n-1}\left(R_{i}(t)A_{i+t}^{\dagger}+{R}_{i}^{\dagger}(t)A_{i+t}\right).
\end{equation}
As in the preceding subsection, use the bath label $\alpha=(i,\nu)$, with $\nu\in\{-,+\}$, and set $B_{(i,\nu)}(t)=A_{i+t}^{(\nu)}$ with $A_x^{(-)} = A_x$ and $A_x^{(+)} = A_x^\dagger$. The vacuum kernel has a strict light cone:
\begin{equation*}
K_{(i,-),(i',+)}(\tau)=\delta\bigl(\tau-(i'-i)\bigr).
\end{equation*}
Proposition~2 implies a Lieb-Robinson velocity bounded by
\begin{equation}
v_{\mathrm{LR},c}(t)=O(t).\label{eq:SM_satu_de_LR_upperbound}
\end{equation}
In the remainder of this subsection, we will construct an example that saturates the Lieb-Robinson bound in proposition 2. In particular, we will show that if a Lieb-Robinson bound, with exponential tails outside the light-cone, exists for this model then it must be $v_{\mathrm{LR},c}(t)=\Omega(t)$.

\begingroup\color{black}
We use the same smooth pulse \(\xi\) as in the preceding subsection. Fix an integer \(T\), take \(n\geq T^2+2T+2\), and start from \(\ket{0}_S^{\otimes n}\otimes\ket{0}_E\).
The protocol again has three stages:
\begin{equation}
H(t)=\begin{cases}
H_{\mathrm{pump}}(t), & 0\leq t<T,\\
H_{\mathrm{ex}}(t), & T\leq t<T+1,\\
H_{\mathrm{prop}}(t), & T+1\leq t<2T+1.
\end{cases}
\end{equation}
For the pumping stage, choose
\begin{equation}
H_{\mathrm{pump}}(t)
=i\frac{2}{\pi}\xi\!\left(\frac{t}{T}\right)
\sum_{j=1}^{n-1}\sigma_j^z
\left(a_{j+t}^\dagger-a_{j+t}\right).
\end{equation}
Since \(\sigma_j^z\ket{0_j}=\ket{0_j}\), this leaves the system unchanged and displaces the field to a coherent state \(\ket{f}\) satisfying
\begin{equation}
A_x\ket{f}=f(x)\ket{f},
\qquad
f(x)=\frac{2}{\pi}\sum_{j=1}^{n-1}
\xi\!\left(\frac{x-j}{T}\right).
\label{eq:SM_coupled_smooth_field}
\end{equation}
In the following, the field coordinates used during the propagation stage satisfy
\begin{equation}
T+2\leq x\leq T^2+2T+1.
\end{equation}
Thus, the choice \(n\geq T^2+2T+2\)  removes boundary truncation from the sum in Eq.~\eqref{eq:SM_coupled_smooth_field}. A shifted Riemann-sum estimate gives
\begin{equation}
\left|\frac{f(x)}{T}-1\right|
\leq\frac{2\|\xi'\|_1}{\pi T},
\label{eq:SM_coupled_flat_field}
\end{equation}
uniformly on this interval, where we used \(\int_0^1\xi(u)\,du=\pi/2\).

The excitation pulse is
\begin{equation}
H_{\mathrm{ex}}(t)=\xi(t-T)\sigma_1^x,
\end{equation}
and it is supported on the singleton $\{1\}$. It maps the first qubit to $\ket{1}$, up to a global phase. For the propagation stage, choose
\begin{align}
H_{\mathrm{prop}}(t)
&=\sum_{i=1}^{T^2}g_i(t)
\left(J_i a_{i+t}+J_i^\dagger a_{i+t}^\dagger\right),\label{eq:SM_coupled_smooth_prop}\\
g_i(t)&=\xi\!\left(T(t-T-1)-(i-1)\right),\nonumber\\
J_i&=\ket{0_i1_{i+1}}\!\bra{1_i0_{i+1}}.\nonumber
\end{align}
Thus, on an active bond during propagation, the coefficient in the general coupling Hamiltonian is $R_i(t)=g_i(t)J_i^\dagger$.
The $i^\text{th}$ pulse is supported on
\begin{equation}
t_i\leq t\leq t_i+\frac1T,
\qquad
t_i:=T+1+\frac{i-1}{T}.
\end{equation}
Conditioned on the excitation having reached site \(i\), its state at the beginning of the $i^\text{th}$ pulse is \(\ket{1_i0_{i+1}}\). In the effective basis
\begin{equation}
\ket{0_{\mathrm{eff}}}:=\ket{1_i0_{i+1}},
\qquad
\ket{1_{\mathrm{eff}}}:=\ket{0_i1_{i+1}},
\end{equation}
write $\sigma_{\mathrm{eff}}:=\ket{0_{\mathrm{eff}}}\!\bra{1_{\mathrm{eff}}}$ and $\sigma_{\mathrm{eff}}^x:=\sigma_{\mathrm{eff}}+\sigma_{\mathrm{eff}}^\dagger$. On this subspace $J_i=\sigma_{\mathrm{eff}}^\dagger$; the pulse index on the effective operators is suppressed.

\begingroup\color{black}
As in the preceding subsection, we can work in a displaced frame. Define the field displacement operator
\[
D[f]:=\exp\!\left(\int_{\mathbb R}\!\left[f(x)a_x^\dagger-f(x)^*a_x\right]dx\right),
\]
which acts only on the environment and satisfies
\[
D[f]\ket{0}_E=\ket{f},
\qquad
D[f]^\dagger a_xD[f]=a_x+f(x).
\]
Recall that $f$ is defined in Eq.~\eqref{eq:SM_coupled_smooth_field} and is real because $\xi$ is real.
Immediately before the propagation stage, the joint state is, up to a global phase,
\(\ket{\psi(T+1)}=\ket{1_10_2\cdots0_n}\otimes\ket{f}\). In the displaced frame this becomes
\[
D[f]^\dagger\ket{\psi(T+1)}
=\ket{1_10_2\cdots0_n}\otimes\ket{0}_E.
\]
During the $i^\text{th}$ pulse, the Hamiltonian on the effective qubit and its input field is
\(H_i(t)=g_i(t)(\sigma_{\mathrm{eff}}^\dagger a_{i+t}+\sigma_{\mathrm{eff}} a_{i+t}^\dagger)\). Its displaced-frame Hamiltonian is
\begin{equation}
\widetilde H_i(t):=D[f]^\dagger H_i(t)D[f]
=g_i(t)f(i+t)\sigma_{\mathrm{eff}}^x
+g_i(t)\left(\sigma_{\mathrm{eff}}^\dagger a_{i+t}+\sigma_{\mathrm{eff}} a_{i+t}^\dagger\right).
\end{equation}
The displaced field is in the vacuum and satisfies
\([a_{i+t},a_{i+s}^\dagger]=\delta(t-s)\). Let $\rho_i(t)$ denote the normalized reduced state of the effective qubit during pulse $i$, conditioned on arrival at site $i$ at time $t_i$. The displacement acts only on the environment and therefore leaves this reduced state unchanged. Tracing out the displaced vacuum field gives~\cite{gardiner1985input}
\endgroup
\begin{align}
\frac{d\rho_i}{dt}
&=-ig_i(t)f(i+t)[\sigma_{\mathrm{eff}}^x,\rho_i]
+g_i(t)^2\mathcal D[\sigma_{\mathrm{eff}}](\rho_i),\label{eq:SM_coupled_delta_master}\\
\mathcal D[\sigma_{\mathrm{eff}}](\rho)
&:=\sigma_{\mathrm{eff}} \rho\sigma_{\mathrm{eff}}^\dagger-\frac12\{\sigma_{\mathrm{eff}}^\dagger\sigma_{\mathrm{eff}},\rho\}.\nonumber
\end{align}

\begingroup\color{black}
For the $i^\text{th}$ pulse, introduce the dimensionless local time \(u:=T(t-t_i)\). Then $0\leq u\leq1$, $t=t_i+u/T$, and $g_i(t_i+u/T)=\xi(u)$. Define the rescaled state by $\widetilde\rho_i(u):=\rho_i(t_i+u/T)$. Thus, Eq.~\eqref{eq:SM_coupled_delta_master} becomes
\endgroup
\begin{equation}
\frac{d\widetilde\rho_i}{du}
=-i\xi(u)\frac{f(i+t_i+u/T)}{T}[\sigma_{\mathrm{eff}}^x,\widetilde\rho_i]
+\frac{\xi(u)^2}{T}\mathcal D[\sigma_{\mathrm{eff}}](\widetilde\rho_i).
\label{eq:SM_coupled_rescaled_master}
\end{equation}
Let \(\rho_{\mathrm{id}}(u)\) solve the ideal equation
\begin{equation}
\frac{d\rho_{\mathrm{id}}}{du}=-i\xi(u)[\sigma_{\mathrm{eff}}^x,\rho_{\mathrm{id}}],
\qquad
\widetilde\rho_i(0)=\rho_{\mathrm{id}}(0)=\ket{0_{\mathrm{eff}}}\!\bra{0_{\mathrm{eff}}}.
\end{equation}
The normalization of \(\xi\) gives
\begin{equation}
\rho_{\mathrm{id}}(1)=\ket{1_{\mathrm{eff}}}\!\bra{1_{\mathrm{eff}}}.
\end{equation}
Let \(\mathcal E_i(u,s)\) denote the propagator of the actual rescaled master equation. Duhamel's identity gives
\begin{align}
\widetilde\rho_i(1)-\rho_{\mathrm{id}}(1)
=\int_0^1\mathcal E_i(1,u)\Bigg[
&-i\xi(u)\left(\frac{f(i+t_i+u/T)}{T}-1\right)
[\sigma_{\mathrm{eff}}^x,\rho_{\mathrm{id}}(u)]\nonumber\\
&+\frac{\xi(u)^2}{T}\mathcal D[\sigma_{\mathrm{eff}}](\rho_{\mathrm{id}}(u))
\Bigg]du.
\label{eq:SM_coupled_duhamel}
\end{align}
The propagator \(\mathcal E_i(1,u)\) is trace-norm contractive on Hermitian operators. Using
\begin{equation}
\|[\sigma_{\mathrm{eff}}^x,\rho]\|_1\leq2\|\rho\|_1,
\qquad
\|\mathcal D[\sigma_{\mathrm{eff}}](\rho)\|_1\leq2\|\rho\|_1,
\end{equation}
together with Eq.~\eqref{eq:SM_coupled_flat_field} and \(\|\rho_{\mathrm{id}}(u)\|_1=1\), we obtain directly from Eq.~\eqref{eq:SM_coupled_duhamel}
\begin{align}
\|\widetilde\rho_i(1)-\rho_{\mathrm{id}}(1)\|_1
&\leq\int_0^1\left[
2|\xi(u)|\left|\frac{f(i+t_i+u/T)}{T}-1\right|
+\frac{2\xi(u)^2}{T}\right]du\nonumber\\
&\leq\frac{1}{T}
\left(\frac{4\|\xi'\|_1\|\xi\|_1}{\pi}
+2\|\xi\|_2^2\right)
=:\frac{C_\xi}{T}.
\label{eq:SM_coupled_step_error}
\end{align}
For the target projector, Eq.~\eqref{eq:SM_coupled_step_error} means that, conditioned on arrival at site \(i\), the transfer succeeds with probability at least \(1-C_\xi/T\). A failed transfer cannot reach the final site through any later pulse, and different pulses couple to disjoint vacuum field increments. The same one-pulse bound therefore applies successively to all \(T^2\) transfers. Let \(\rho_S(t)\) denote the reduced state of the system at time \(t\). The final occupation therefore obeys
\begin{equation}
\operatorname{Tr}\!\left[\rho_S(2T+1)\sigma_{T^2+1}^\dagger\sigma_{T^2+1}\right]
\geq\left(1-\frac{C_\xi}{T}\right)^{T^2}
\geq e^{-2C_\xi T}
\label{eq:SM_coupled_tail_lower}
\end{equation}
for \(T\geq2C_\xi\). This estimate supplies an $e^{-O(T)}$ lower bound on the final occupation probability; it does not establish a $T$-independent lower bound.

Finally, choose
\begin{equation}
X=\{T^2+1\},
\qquad
O_X=\sigma_{T^2+1}^\dagger\sigma_{T^2+1},
\qquad
l=T^2-1.
\end{equation}
The truncated Hamiltonian does not contain \(H_{\mathrm{ex}}\), so its expectation of \(O_X\) is zero. Hence
\begin{equation}
\Delta_{O_X}(2T+1,0;T^2-1)\geq e^{-2C_\xi T}.
\end{equation}
Fix $a_E'>a_E$ as in Proposition~2. Its bound gives
\begin{equation}
e^{-2C_\xi T}
\leq f_c(2T+1,T^2-1)
\exp\!\left[
\frac{(2T+1)v_{\mathrm{LR,c}}(2T+1)-(T^2-1)}
{p_+a_E'}\right].
\end{equation}
For \(d=1\), \(\log f_c(2T+1,T^2-1)=O(\log T)\). Taking logarithms therefore yields
\begin{equation}
(2T+1)v_{\mathrm{LR,c}}(2T+1)
\geq T^2-1-2p_+a_E'C_\xi T-O(\log T),
\end{equation}
and consequently
\begin{equation}
v_{\mathrm{LR,c}}(2T+1)=\Omega(T).
\end{equation}
Together with Eq.~\eqref{eq:SM_satu_de_LR_upperbound}, this proves that the scaling of the exponential tail is tight up to constants in one dimension.
\endgroup
\endgroup
\section{Markovian Dilations for local observables with independent memory}\label{sec:supp_star_chain}
As an application of the Lieb-Robinson bound, we analyze the star-to-chain transformation which approximates the independent environment by a finite number of discrete modes with the quality of the approximation governed by the number of modes. Typically, for a system with $n$ qubits, in order to approximate the system state to an $O(1)$ precision, we need a number of modes per independent bath of the environment that scale as $\text{poly}(n)$. However, for geometrically local many-body open quantum systems, we would expect that to approximate local observables to an $O(1)$ precision, the number of bath modes needed should be independent of $n$. In this section, we show that this qualitative expectation can be made precise using the Lieb Robinson bounds derived in this paper. We first describe the star-to-chain transformation, and review the analysis of its accuracy without using geometrical locality of the system \cite{trivedi2022description}. Then, we combine this analysis with the Lieb-Robinson bounds to provide tighter results for local observables.

\subsection{Preliminaries}

Throughout this section, we will assume that the environment is initially in vacuum state with $L=2$, and $B_{\alpha}(t)=A_{\mathcal{S}_{\alpha}}^{(\nu)}(t)$. For notational convenience, throughout this section we will write $\mathcal{S}_{\alpha}$ as $\alpha$ and $A_{\mathcal{S}_{\alpha}}^{(\nu)}(t)$ as $A_{\alpha}^{(\nu)}(t)$, respectively. Consequently, it follows that
\[
K_{A_{\alpha}^{(\nu)},A^{(\nu')}_{{\alpha'}}}(t)=\delta_{\alpha,{\alpha'}}\delta_{\nu,-}\delta_{\nu',+}V_{\alpha}(t),
\]
and $\text{V}_\alpha(t - t') = [A_{\alpha}(t), A_{\alpha}^{\dagger}(t')]$.
We will also define by $\text{V} \in \mathcal{K}_1(\mathbb{R})$ as the kernel that upper bounds $\text{V}_\alpha$ for all $\alpha$ and assume that $\text{TV}(\text{V}) < \infty$.
The system-environment coupling can be correspondingly written as
\[
H_{\text{SE}}(t)=\sum_\alpha \left(A_{\alpha}(t) R_\alpha^\dagger(t)+A_{\alpha}^{\dagger}(t) R_{\alpha}(t)\right).
\]

We introduce the \emph{spectral density function}, $\hat{\text{V}}_\alpha$, for the $\alpha^\text{th}$ bath as the fourier transform of the Kernel $\hat{\text{V}}_\alpha$
\[
\hat{\text{V}}_\alpha(\omega) = \int_{-\infty}^\infty {\text{V}}_\alpha(t)e^{i\omega t} dt.
\]
We note that $\hat{\text{V}}_\alpha$ is a bounded and non-negative function --- the fact that it is bounded follows straightforwardly:
\[
\abs{\hat{\text{V}}_\alpha(\omega)} \leq \int_{-\infty}^\infty \abs{\text{V}_\alpha(t)} dt = \text{TV}(\text{V}_\alpha).
\]
To see that $\hat{\text{V}}_\alpha(\omega)$ is non-negative, note that for any smooth function $h(t)$ that $\to 0$ faster than any polynomial in $1/t$ as $\abs{t}\to \infty$,
\begin{align*}
    0 \leq \int_{ -\infty}^\infty \int_{ -\infty}^\infty h(t)[A_{\alpha}(t), A_{\alpha}^{\dagger}(t')] h^*(t') dt' dt =  \int_{ -\infty}^\infty \int_{ -\infty}^\infty h(t) \text{V}_\alpha(t - t') h^*(t') dt' dt =  \int_{-\infty}^\infty \abs{\hat{h}(\omega)}^2 \hat{\text{V}}_\alpha(\omega) \frac{d\omega}{2\pi},
\end{align*}
where $\hat{h}(\omega) = \int_{-\infty}^{\infty}h(t) e^{-i\omega t} dt$. Now, since this is true for any $\hat{h}(\omega)$, we conclude that $\hat{\text{V}}_\alpha(\omega) \geq 0$ for all $\omega$. It will also be convenient to work with an explicit second quantization of the Hilbert space of the baths --- for describing the $\alpha^{\text{th}}$ bath's Hilbert space, it is always possible to introduce annihilation operators $a_{\alpha, \omega}$, where $\omega \in \mathbb{R}$, which satisfy the canonical commutation relations $[a_{\alpha, \omega}, a^\dagger_{\alpha', \omega'}] = \delta_{\alpha, \alpha'}\delta(\omega - \omega')$ such that
\[
A_{\alpha}(t) = \int_{-\infty}^\infty \hat{v}_\alpha(\omega) a_{\alpha, \omega} e^{-i\omega t} \frac{d\omega}{\sqrt{2\pi}},
\]
where $\hat{v}_\alpha(\omega) = \sqrt{\hat{\text{V}}_\alpha(\omega)}$ --- this can be verified by noting that this representation of $A_{\alpha}(t)$ reproduces the original commutation relation $[A_{\alpha}(t), A_{\alpha}^{\dagger}(t')] = \text{V}_\alpha(t - t')$. We will refer to $a_{\alpha, \omega}$ as the frequency-domain annihilation operator, at frequency $\omega$, for the $\alpha^\text{th}$ bath consistent with their physical interpretation as annihilating an excitation at frequency $\omega$ in the bath. In terms of $a_{\alpha,\omega}$, following section \ref{sec:regularization}, the regularized annihilation operator is given by
\[
A_{\alpha}^{\delta}(t) = \int_{-\infty}^\infty \eta_\delta(t - s) A_{\alpha}(s)ds = \int_{-\infty}^\infty \hat{v}_{\alpha}^\delta(\omega) a_{\alpha, \omega} e^{-i\omega t} \frac{d\omega}{\sqrt{2\pi}},
\]
where $\hat{v}^\delta_\alpha(\omega) = \hat{v}_\alpha(\omega) \hat{\eta}(\omega \delta)$ with $\hat{\eta}(\omega) = \int_{-1}^1 \eta(t)e^{i\omega t} dt$ being the Fourier transform of the mollifier $\eta$ introduced in section \ref{sec:regularization}. The mollified kernel, ${\text{V}}_\alpha^\delta(t - t') = [A_{\alpha}^{\delta}(t), A_{\alpha}^{\delta \dagger}(t')]$, is related to $\hat{v}_\alpha^\delta(\omega)$ via
\[
\text{V}^\delta_\alpha(\tau) = \int_{-\infty}^\infty \abs{\hat{v}_\alpha^\delta(\omega)}^2 e^{-i\omega \tau}\frac{d\omega}{2\pi}.
\]
We will also define $\text{V}_\alpha^{\delta, \delta'}(t - t') = [A_{\alpha}^{\delta}(t), A_{\alpha}^{\delta' \dagger}(t')]$ which can be expressed as
\[
\text{V}_\alpha^{\delta, \delta'}(\tau) = \int_{-\infty}^\infty \hat{v}^{\delta}_\alpha(\omega) \hat{v}^{\delta'*}_\alpha(\omega) e^{-i\omega \tau}\frac{d\omega}{2\pi}.
\]

\subsection{Analyzing error in star-to-chain transformation}
To perform the star-to-chain transformation, there will be three steps --- first, we will regularize the model as already described in section \ref{sec:regularization}. Next, we will introduce a hard frequency cut-off on the regularization. Finally, we will apply the star-to-chain transformation as laid out in Ref.~\cite{chin2010exact}. We will analyze the error incurred in each of these steps separately.

\emph{Regularization error}. We begin by using the result of section \ref{sec:regularization} to provide an explicit bound for the regularization error.
\begin{lemma}\label{lemma:star_to_chain_reg_error}
    For $\mathcal{B} \subseteq \mathcal{A}$, observable $O$ with $\norm{O} \leq 1$ and sufficiently small $\delta$,
    \begin{align*}
    &\smallabs{\vecbra{O, I_E}\mathcal{U}_\mathcal{B}^\delta(t, 0) \vecket{\rho_S(0), 0_E} - \lim_{\delta \to 0}\vecbra{O, I_E}\mathcal{U}_\mathcal{B}^\delta(t, 0) \vecket{\rho_S(0), 0_E}} \nonumber\\
    &\qquad \qquad \qquad \qquad \leq 4t\abs{\mathcal{B}}\lambda_0(2\delta)  + 4t\abs{\mathcal{B}}\big[1 + \big(2^2 t + 2^4t\textnormal{TV}(\textnormal{V})\big) \abs{\mathcal{B}}\big]\lambda_1(2\delta),
    \end{align*}
    where
    \[
    \lambda_0(\delta) = 4\sum_{t' \in \{0, t\}} \textnormal{TV}(\textnormal{V}_c; [t' -\delta, t' + \delta]\big) \text{ and }\lambda_1(\delta) = 2\delta \textnormal{TV}(\textnormal{V}),
    \]
    where $\textnormal{V}_c$ is the continuous part of $\textnormal{V}$.
\end{lemma}
\begin{proof}
    We begin by considering $\delta, \delta' > 0$ and consider the error $\Delta_{\delta, \delta'}$
    \begin{align}\label{eq:error_delta_deltap_decomp}
    \Delta_{\delta, \delta'} &=\abs{\vecbra{O, I_E} \mathcal{U}_\mathcal{B}^\delta(t, 0) - \mathcal{U}_\mathcal{B}^{\delta'}(t, 0)\vecket{\rho_S, 0_E}}, \nonumber \\
    &\leq \Theta_{\mathcal{B}}(\alpha) \bigg(\abs{\int_0^t \vecbra{O, I_E}\mathcal{U}_\mathcal{B}^{\delta'}(t, s)(A_{\alpha,\sigma}^{\delta}(s) - A_{\alpha,\sigma}^{\delta'}(s)) R_{\alpha, \sigma}^\dagger\mathcal{U}_\mathcal{B}^\delta (s, 0) \vecket{\rho_S(0), 0_E} ds} + \nonumber \\
    &\qquad \qquad \abs{\int_0^t \vecbra{O, I_E}\mathcal{U}_\mathcal{B}^{\delta'}(t, s)(A_{\alpha,\sigma}^{\delta\dagger}(s) - A_{\alpha,\sigma}^{\delta'\dagger}(s)) R_{\alpha, \sigma}\mathcal{U}_\mathcal{B}^\delta (s, 0) \vecket{\rho_S(0), 0_E} ds}\bigg), \nonumber \\
    &\leq \Theta_{\mathcal{B}}(\alpha) \big( \abs{\Gamma^{(1)}_{\alpha, \sigma}} + \abs{\Gamma^{(2)}_{\alpha, \sigma}} + \abs{\Gamma^{(3)}_{\alpha, \sigma}} + \abs{\Gamma^{(4)}_{\alpha, \sigma}}\big).
    \end{align}
    where, from Wick's theorem,
    \begin{align}\label{eq:gamma_split_def}
        &\Gamma_{\alpha, \sigma}^{(1)} = \int_0^t \int_0^s \big(\text{V}^{\delta, \delta}_{\alpha}(s - s') - \text{V}^{\delta, \delta'}_{\alpha}(s - s')\big)f_{\alpha, \sigma}^{(1)}(s, s') ds' ds \text{ with }\nonumber\\
        &f_{\alpha, \sigma}^{(1)}(s, s') = \vecbra{O, I_E} \mathcal{U}_\mathcal{B}^{\delta'}(t, s) R_{\alpha, \sigma}^\dagger \mathcal{U}_\mathcal{B}^\delta(s, s') R_{\alpha, l} \mathcal{U}_\mathcal{B}^{\delta}(s', 0) \vecket{\rho_S(0), 0_E},\nonumber \\
        &\Gamma_{\alpha, \sigma}^{(2)} = \int_0^t \int_{s}^t \big(\text{V}_{\alpha}^{\delta, \delta'}(s - s') - \text{V}_\alpha^{\delta', \delta'}(s-s') \big) f_{\alpha, \sigma}^{(2)}(s, s') ds'ds \text{ with }\nonumber \\
        &f_{\alpha, \sigma}^{(2)}(s, s') = \vecbra{O, I_E}\mathcal{U}_\mathcal{B}^{\delta'}(t, s') R_{\alpha, \sigma} \mathcal{U}_\mathcal{B}^{\delta'}(s', s) R_{\alpha, l}^\dagger \mathcal{U}^{\delta}_\mathcal{B}(s', 0) \vecket{\rho_S(0), 0_E},\nonumber \\
        &\Gamma_{\alpha, \sigma}^{(3)} = \int_0^t \int_0^s \big(\text{V}_{\alpha}^{\delta, \delta}(s' - s) - \text{V}_\alpha^{\delta, \delta'}(s' - s)\big)f_{\alpha, \sigma}^{(3)}(s, s')ds' ds \text{ where }\nonumber\\
        &f_{\alpha, \sigma}^{(3)}(s, s') = \vecbra{O, I_E} \mathcal{U}_\mathcal{B}^{\delta'}(t, s) R_{\alpha, \sigma} \mathcal{U}^\delta_\mathcal{B}(s, s') R_{\alpha, r}^\dagger \mathcal{U}_\mathcal{B}^\delta(s', 0) \vecket{\rho_S(0), 0_E},\nonumber \\
        &\Gamma_{\alpha, \sigma}^{(4)} = \int_0^t \int_{s}^t (\text{V}^{\delta',\delta'}_\alpha(s' - s) - \text{V}_\alpha^{\delta',\delta}(s' - s)) f_{\alpha, \sigma}^{(4)}(s, s') ds' ds \text{ where }\nonumber\\
        &f_{\alpha, \sigma}^{(4)}(s, s') = \vecbra{O, I_E} \mathcal{U}_{\mathcal{B}}^{\delta'}(t, s') R_{\alpha, \sigma}^\dagger \mathcal{U}^{\delta'}_\mathcal{B}(s', s)R_{\alpha,l}\mathcal{U}^\delta_{\mathcal{B}}(s, 0) \vecket{\rho_S(0), 0_E}.
    \end{align}
    We can now bound $\Gamma_{\alpha, \sigma}^{(j)}$ using lemma \ref{lemma:smallness_double_integral_kernel} --- the bound for each $j$ is identical, so we perform the calculation, to illustrate the procedure, only for $j = 1$. We begin by noting from lemma \ref{lemma:kernel_appx}, on the test-function space $\text{PC}_{\{0, t\}}^1(\mathbb{R})$ and for sufficiently small $\delta, \delta'$, $\text{V}_{\alpha}^{\delta, \delta} - \text{V}^{\delta, \delta'}_\alpha$ is $\lambda_0(\Delta), \lambda_1(\Delta)$-small where $\Delta = 2 \max(\delta, \delta')$ and
    \[
    \lambda_0(\Delta) = 4 \sum_{t' \in \{0, t\}} \text{TV}(\text{V}_{c}; [t' - \Delta, t' + \Delta]) \text{ and }\lambda_1(\Delta) = 2 \Delta \textnormal{TV}(\text{V}).
    \]
    To see this, note that for any $f \in \text{PC}_{\{0, t\}}^1(\mathbb{R})$
    \begin{align*}
        \abs{\int_{-\infty}^\infty \big(\text{V}_{\alpha}^{\delta, \delta}(s) - \text{V}^{\delta, \delta'}_\alpha(s)\big)f(s)ds} &\leq \abs{\int_{-\infty}^\infty \big(\text{V}_{\alpha}^{\delta, \delta}(s) - \text{V}_\alpha(s)\big)f(s)ds} + \abs{\int_{-\infty}^\infty \big( \text{V}^{\delta, \delta'}_\alpha(s) - \text{V}_{\alpha}(s)\big)f(s)ds},
    \end{align*}
    and then use lemma \ref{lemma:kernel_appx}. Next, we show that $f^{(1)}_{\alpha, \sigma}$ is both bounded and differentiable with respect to its arguments with bounded bounded partial derivatives. The fact that it is bounded follows immediately from Eq.~\ref{eq:gamma_split_def} --- in particular,
    \begin{align}\label{eq:bound_f_sigma}
        \smallabs{f^{(1)}_{\alpha, \sigma}(s, s')} = \smallabs{\vecbra{O, I_E} \mathcal{U}_\mathcal{B}^{\delta'}(t, s) R_{\alpha, \sigma}^\dagger \mathcal{U}_\mathcal{B}^\delta(s, s') R_{\alpha, l} \mathcal{U}_\mathcal{B}^{\delta}(s', 0) \vecket{\rho_S(0), 0_E}} \leq \norm{R_\alpha}^2 \leq 1.
    \end{align}
    Next, we compute its derivatives ---
    \begin{align}\label{eq:bound_deriv_f_1}
    \abs{\partial_1 f^{(1)}_{\alpha, \sigma}(s, s')} &\leq \Theta_{\mathcal{B}}(\alpha')\bigg(\abs{\vecbra{O, I_E}\mathcal{U}_\mathcal{B}^{\delta'}(t,s)[\mathcal{C}_{h_{\alpha'}}, R_{\alpha, \sigma}^\dagger ]\mathcal{U}^\delta_\mathcal{B}(s, s') R_{\alpha, l}\mathcal{U}^\delta_\mathcal{B}(s', 0)\vecket{\rho_S(0), 0_E}} +\nonumber \\
    &\qquad \quad \abs{\vecbra{O, I_E}\mathcal{U}_\mathcal{B}^{\delta'}(t,s)(A_{\alpha',\sigma'}^{\delta'}(s)R_{\alpha', \sigma'}^\dagger + A_{\alpha',\sigma'}^{\delta' \dagger}(s) R_{\alpha', \sigma'})R_{\alpha, \sigma}^\dagger\mathcal{U}^\delta_\mathcal{B}(s, s') R_{\alpha, l}\mathcal{U}^\delta_\mathcal{B}(s', 0)\vecket{\rho_S(0), 0_E}} + \nonumber\\
    &\qquad \quad \abs{\vecbra{O, I_E}\mathcal{U}_\mathcal{B}^{\delta'}(t,s)R_{\alpha, \sigma}^\dagger (A_{\alpha',\sigma'}^{\delta}(s)R_{\alpha', \sigma'}^\dagger + A_{\alpha',\sigma'}^{\delta \dagger}(s) R_{\alpha', \sigma'})\mathcal{U}^\delta_\mathcal{B}(s, s') R_{\alpha, l}\mathcal{U}^\delta_\mathcal{B}(s', 0)\vecket{\rho_S(0), 0_E}}\bigg), \nonumber\\
    &\leq 4 \Theta_{\mathcal{B}}(\alpha') \norm{h_{\alpha'}}\norm{R_\alpha}^2 + 2^4\Theta_{\mathcal{B}}(\alpha')\norm{R_{\alpha'}}^2\norm{R_\alpha}^2\text{TV}(\text{V}_\alpha), \nonumber \\
    &\leq \big(2^2 + 2^4\text{TV}(\text{V})\big) \abs{\mathcal{B}}.
    \end{align}
    Similarly,
    \begin{align}\label{eq:bound_deriv_f_2}
        \abs{\partial_2 f^{(1)}_{\alpha, \sigma}(s, s')} \leq  \big(2^2 + 2^4\text{TV}(\text{V}_\alpha)\big) \abs{\mathcal{B}}.
    \end{align}
    Using lemma \ref{lemma:smallness_double_integral_kernel} together with the estimates from Eqs.~\ref{eq:bound_f_sigma}, \ref{eq:bound_deriv_f_1} and \ref{eq:bound_deriv_f_2}, we then obtain that
    \begin{align}
        \Gamma^{(1)}_{\alpha, \sigma} \leq t\lambda_0(\Delta)  + \lambda_1(\Delta)\big[1 + \big(2^2 t + 2^4t\text{TV}(\text{V})\big) \abs{\mathcal{B}}\big].
    \end{align}
    A similar bound can be calculated for $\Gamma^{(2)}_{\alpha, \sigma}, \Gamma^{(3)}_{\alpha, \sigma}, \Gamma^{(4)}_{\alpha, \sigma}$. We then obtain from Eq.~\ref{eq:error_delta_deltap_decomp} that
    \[
    \abs{\vecbra{O, I_E}\mathcal{U}_\mathcal{B}^\delta(t, 0) - \mathcal{U}_X^{\delta'}(t, 0) \vecket{\rho_S(0), 0_E}} \leq 4t\abs{\mathcal{B}}\lambda_0(\Delta)  + 4t\abs{\mathcal{B}}\big[1 + \big(2^2 t + 2^4t\text{TV}(\text{V})\big) \abs{\mathcal{B}}\big]\lambda_1(\Delta).
    \]
    Taking the limit of $\delta' \to 0$, which amounts to setting $\Delta = 2\max(\delta, \delta')$ to $2\delta$, we obtain the lemma statement.
\end{proof}

\emph{Analyzing the frequency cutoff}.  Next, we will introduce a hard frequency cut-off on the regularized model. Given a frequency-cutoff $\omega_c$, we will introduce the bath annihilation operator
\[
A_{\alpha}^{\delta, \omega_c}(t) = \int_{-\omega_c}^{\omega_c} \hat{v}_\alpha^\delta(\omega) a_{\alpha, \omega}e^{-i\omega t}\frac{d\omega}{\sqrt{2\pi}},
\]
and define the Hamiltonian $H_\mathcal{B}^{\delta, \omega_c}$ by
\begin{align}\label{eq:hamil_freq_cutoff}
H^{\delta, \omega_c}_\mathcal{B}(t) = \sum_{\alpha \in \mathcal{B}} h_\alpha + \sum_{\alpha \in \mathcal{B}}\big(R_\alpha^\dagger A_{\alpha}^{\delta, \omega_c}(t) + \text{h.c.}\big).
\end{align}
We will now consider the error between the regularized model $H^\delta_\mathcal{B}(t)$ and the regularized model with frequency cut-off $H^{\delta, \omega_c}_\mathcal{B}(t)$. Physically, due to the regularization, $\hat{v}^\delta_\alpha(\omega)$ decreases to $0$ as $\omega \to \infty$ and consequently if $\omega_c$ is sufficiently large, we expect to incur a small error on introducing the frequency cut-off. For this purpose of analyzing this error, it is convenient to first introduce a bound on $\abs{\hat{v}_\alpha^\delta(\omega)}$.
\begin{lemma}\label{lemma:freq_cutoff}
For any $\omega_c > 0$,
\[
\int_{\abs{\omega} \geq \omega_c} \abs{\hat{v}_\alpha^\delta(\omega)}^2 \frac{d\omega}{2\pi} \leq \frac{\gamma_0 }{\delta^2 \omega_c}\textnormal{TV}(\textnormal{V}_\alpha) \exp\bigg[-\bigg(\frac{ {\omega_c}\delta}{16e}\bigg)^{1/2} \bigg],
\]
where $\gamma_0$ is a constant.
\end{lemma}
\begin{proof}
    To prove this lemma, we note that $\hat{\eta}(\omega)$, the fourier transform of the standard mollifier $\eta$ defined in Eq.~\ref{eq:standard_mollifiers}, is bounded above by any polynomial in $\omega^{-1}$ --- more precisely, for any $\omega$
    \[
    \abs{\hat{\eta}(\omega)} = \abs{\int_{-1}^1 \eta(t) e^{i\omega t} dt} \numeq{1} \abs{\frac{1}{\omega^k} \int_{-1}^{1}\eta^{(k)}(t)e^{i\omega t}dt} \leq \frac{2\norm{\eta^{(k)}}_\infty}{\abs{\omega}^k},
    \]
    where, in (1), we have used integration by parts $k$ times together with the fact that since $\eta(t)$ is a smooth and compact function on $[-1, 1]$, $\eta^{(k)}(\pm 1) = 0$. Now, for the standard mollifier, $\eta^{(k)}$ can be explicitly evaluated and upper bounded to obtain \cite{israel2015eigenvalue} (here, $A_0$ is the constant, defined in Eq.~\ref{eq:standard_mollifiers}, chosen such that $\int_{[-1, 1]} \eta(x) dx = 1$)
    \[
    \norm{\eta^{(k)}}_\infty \leq A_0 16^k k^{2k} \text{ and therefore }\abs{\hat{\eta}(\omega)} \leq \abs{\hat{\eta}(\omega)} \leq 2A_0 \bigg(\frac{16k^2}{\abs{\omega}}\bigg)^k.
    \]
    Since this is true for any $k$, we can choose $k$ to be the smallest integer such that $k \leq (\abs{\omega} / 16e)^{1/2}$ (and since $k$ is an integer, it is also true that $(\abs{\omega} / 16e)^{1/2} - 1\leq k$) and therefore
    \begin{align}\label{eq:upper_bound_rho}
    \abs{\hat{\eta}(\omega)} \leq \frac{2A_0}{e}\exp\bigg[-\bigg(\frac{\abs{\omega}}{16e}\bigg)^{1/2} \bigg].
    \end{align}
    With this, we can upper bound $\abs{\hat{v}_\alpha^\delta(\omega)}^2$ by noting that
    \begin{align}\label{eq:upper_bound_v_omega}
\abs{\hat{v}_\alpha^\delta(\omega)}^2 &= \abs{\hat{\textnormal{V}}_\alpha(\omega)} \abs{\hat{\eta}(\omega \delta)}^2 =\abs{\int_{-\infty}^\infty \hat{\text{V}}_\alpha(t) e^{i\omega t}dt}  \abs{\hat{\eta}(\omega \delta)}^2,\nonumber\\
&\leq  \text{TV}(\text{V}_\alpha) \abs{\hat{\eta}(\omega \delta)}^2 \leq \frac{4A_0^2}{e^2} \textnormal{TV}(\textnormal{V}_\alpha) \exp\bigg[-2\bigg(\frac{\abs{\omega}}{16e}\bigg)^{1/2}\bigg].
\end{align}
We can now use this bound to obtain the bound in the lemma statement:
\begin{align}
\abs{\int_{\abs{\omega} \geq \omega_c} \abs{\hat{v}_\alpha^\delta(\omega)}^2 \frac{d\omega}{2\pi}}&\leq\frac{2A_0^2}{e^2\pi}\text{TV}(\text{V}_\alpha){\int_{\abs{\omega} \geq \omega_c} \exp\bigg[-2\bigg(\frac{\delta \abs{\omega}}{16e}\bigg)^{1/2}\bigg]  d\omega}, \nonumber \\
        &\numleq{1} \frac{2^9  A_0^2 M_0}{e^2\delta^2 \pi}\text{TV}(\text{V}_\alpha) \int_{\abs{\omega}\geq \omega_c} \exp\bigg[-\bigg(\frac{\delta \abs{\omega}}{16e}\bigg)^{1/2}\bigg] \frac{d\omega}{\omega^2} , \nonumber \\
        &\leq \frac{2^9  A_0^2 M_0}{e^2 \delta^2 \pi}\text{TV}(\text{V}_\alpha) \exp\bigg[-\bigg(\frac{\delta {\omega_c}}{16e}\bigg)^{1/2}\bigg]\int_{\abs{\omega}\geq \omega_c}  \frac{d\omega}{\omega^2},\nonumber \\
        &= \frac{2^{10}  A_0^2 M_0}{e^2 \delta^2 \omega_c \pi}\text{TV}(\text{V}_\alpha) \exp\bigg[-\bigg(\frac{\delta {\omega_c}}{16e}\bigg)^{1/2}\bigg]
\end{align}
where in (1) we have introduced the constant $M_0 = \sup_{\omega > 0}\omega^2 e^{-\sqrt{\omega}}$ and used that $e^{-\sqrt{\delta \abs{\omega}/16e}} \leq 2^8 e^2 M_0 / \abs{\omega}^2 \delta^2$. Noting that $\gamma_0 = 2^{10}A_0^2 M_0 / e^2 \pi$ is a constant, we obtain the lemma statement.
\end{proof}
To analyze this error, it is convenient to introduce the following lemma.
\begin{lemma}\label{lemma:error_two_L2}
    Given $\mathcal{B}\subseteq \mathcal{A}$, suppose
    \[
    H_\mathcal{B}(t) = \sum_{\alpha \in \mathcal{B}} h_\alpha  + \sum_{\alpha \in \mathcal{B}} (A_{\alpha}(t) R_\alpha^\dagger + \textnormal{h.c.}) \text{ and }\tilde{H}_\mathcal{B}(t) = \sum_{\alpha \in \mathcal{B}} h_\alpha  + \sum_{\alpha \in \mathcal{B}} (\tilde{A}_{\alpha}(t) R_\alpha^\dagger + \textnormal{h.c.})
    \]
    such that $\textnormal{V}_\alpha(t) := [A_{\alpha}(t), A_{\alpha}^{\dagger}(0)], \tilde{\textnormal{V}}_\alpha(t) := [\tilde{A}_{\alpha}(t), \tilde{A}_{\alpha}^{\dagger}(0)]$ are both continuous functions, then for any system observable $O$ and initial system state $\rho(0)$,
    \[
    \abs{\vecbra{O, I_E}\mathcal{U}_\mathcal{B}(t, 0) - \tilde{\mathcal{U}}_\mathcal{B}(t, 0) \vecket{\rho_S(0), 0_E}} \leq {2}t^2 \norm{O} \sum_{\alpha \in \mathcal{B}} \bigg(\big(\sqrt{\textnormal{V}_\alpha(0)} + \sqrt{\tilde{\textnormal{V}}_\alpha(0)}\big)\sup_{s\in [0, t]}\sqrt{[\delta A_{\alpha}(s), \delta A_{\alpha}^{\dagger}(s)]}\bigg),
    \]
    where $\delta A_{\alpha}(s) = A_{\alpha}(s) - \tilde{A}_{\alpha}(s)$.
\end{lemma}
\begin{proof}
    We begin by the standard first-order perturbation theory expression:
\begin{align}\label{eq:l2_error_manipulation}
    &\abs{\vecbra{O, I_E}\mathcal{U}_\mathcal{B}(t, 0) - \tilde{\mathcal{U}}_\mathcal{B}(t, 0) \vecket{\rho_S(0), 0_E}}\nonumber\\
    &\qquad= \abs{\int_0^t \vecbra{O, I_E}\mathcal{U}_\mathcal{B}(t, s) \big(\mathcal{H}_\mathcal{B}(s) - \tilde{\mathcal{H}}_\mathcal{B}(s)\big)\tilde{\mathcal{U}}_\mathcal{B}(s, 0)\vecket{\rho(0)}ds}, \nonumber \\
    &\qquad\leq \Theta_{\mathcal{B}}(\alpha) \bigg(\abs{\int_0^t \vecbra{O, I_E}\mathcal{U}_\mathcal{B}(t, s)\delta\mathcal{A}_{\alpha,\sigma}(s)R_{\alpha, \sigma}^\dagger \tilde{\mathcal{U}}_\mathcal{B}(s, 0)\vecket{\rho_S(0), 0_E} ds} +\nonumber\\
    &\qquad \qquad \qquad \qquad \qquad \qquad \abs{\int_0^t \vecbra{O, I_E}\mathcal{U}_\mathcal{B}(t, s)\delta A_{\alpha,\sigma}^{\dagger}(s) R_{\alpha, \sigma} \tilde{\mathcal{U}}_\mathcal{B}(s, 0)\vecket{\rho_S(0), 0_E} ds}\bigg), \nonumber\\
    &\qquad \leq \Theta_{\mathcal{B}}(\alpha)\bigg(\abs{\int_{0}^t \int_s^t \vecbra{I_E}A_{\alpha,\sigma'}^{\dagger}(s')\delta A_{\alpha,r}(s)\vecket{0_E}\vecbra{O, I_E}\mathcal{U}_\mathcal{B}(t, s') R_{\alpha, \sigma'}\mathcal{U}_\mathcal{B}(s', s) R^\dagger_{\alpha, \sigma}\tilde{\mathcal{U}}_\mathcal{B}(s, 0)\vecket{\rho_S(0), 0_E}ds' ds} +\nonumber\\
    &\qquad \qquad \qquad \quad \ \abs{\int_0^t \int_0^s \vecbra{I_E}\delta A_{\alpha,\sigma}(s) \tilde{A}_{\alpha,l}^{\dagger}(s')\vecket{0_E}\vecbra{O, I_E}\mathcal{U}_\mathcal{B}(t, s) R_{\alpha, \sigma}^\dagger \tilde{\mathcal{U}}_\mathcal{B}(s, s') R_{\alpha, l}\mathcal{U}_\mathcal{B}(s', 0)\vecket{\rho_S(0), 0_E}ds' ds}+ \nonumber\\
    &\qquad \qquad \qquad \quad \abs{\int_{0}^t \int_s^t \vecbra{I_E}A_{\alpha,\sigma'}(s')\delta A_{\alpha,l}^{\dagger}(s)\vecket{0_E}\vecbra{O_X, I_E}\mathcal{U}_\mathcal{B}(t, s') R^\dagger_{\alpha, \sigma'}\mathcal{U}_\mathcal{B}(s', s) R_{\alpha, l}\tilde{\mathcal{U}}_\mathcal{B}(s, 0)\vecket{\rho_S(0), 0_E}ds' ds} + \nonumber\\
    &\qquad \qquad \qquad \quad \abs{\int_0^t \int_0^s \vecbra{I_E}\delta A_{\alpha,\sigma}^{\dagger}(s) \tilde{A}_{\alpha,r}(s')\vecket{0_E}\vecbra{O, I_E}\mathcal{U}_\mathcal{B}(t, s) R_{\alpha, \sigma} \tilde{\mathcal{U}}_\mathcal{B}(s, s') R_{\alpha, r}^\dagger\mathcal{U}_\mathcal{B}(s', 0)\vecket{\rho_S(0), 0_E}ds' ds}\bigg), \nonumber\\
    &\qquad \leq \Theta_{\mathcal{B}}(\alpha) \bigg(\sup_{s, s' \in [0, t]}\abs{\vecbra{I_E}A_{\alpha,\sigma'}^{\dagger}(s') \delta A_{\alpha,r}(s)\vecket{0_E}} + \sup_{s, s' \in [0, t]}\abs{\vecbra{I_E}\delta A_{\alpha,\sigma}(s)  \tilde{A}_{\alpha,l}^{\dagger}(s')\vecket{0_E}} + \nonumber\\
    &\qquad \qquad \qquad \qquad \qquad \sup_{s,s' \in [0, t]}\abs{\vecbra{I_E}A_{\alpha,\sigma'}(s') \delta A_{\alpha,l}^{\dagger}(s)\vecket{0_E}} + \sup_{s, s' \in [0, t]}\abs{\vecbra{I_E}\delta A_{\alpha,\sigma}^{\dagger}(s)  \tilde{A}_{\alpha,r}(s')\vecket{0_E}}\bigg) t^2\norm{R_\alpha}^2.
\end{align}
Now, we note that
\begin{align}\label{eq:l2_error_sup_bound_1}
\sup_{s, s' \in [0, t]}\abs{\vecbra{I_E}A_{\alpha,\sigma'}^{\dagger}(s') \delta A_{\alpha,r}(s)\vecket{0_E}} &= \sup_{s, s' \in [0, t]}\abs{\bra{0_E}\delta A_{\alpha}(s) A_{\alpha}^{\dagger}(s') \ket{0_E}} \nonumber\\
&\leq \sup_{s, s' \in [0, t]} \sqrt{\vecbra{0_E} \delta A_{\alpha}(s)\delta A_{\alpha}^{\dagger}(s) \ket{0_E}} \sqrt{\vecbra{0_E} A_{\alpha}(s) A_{\alpha}^{\dagger}(s) \ket{0_E}}, \nonumber\\
&\leq \sqrt{\text{V}_\alpha(0)} \sup_{s\in [0, t]}\sqrt{[\delta A_{\alpha}(s), \delta A_{\alpha}^{\dagger}(s)]}.
\end{align}
Proceeding similarly, we can also obtain that
\begin{align}\label{eq:l2_error_sup_bound_2}
&\sup_{s, s' \in [0, t]}\abs{\vecbra{I_E}A_{\alpha,\sigma'}^{\dagger}(s') \delta A_{\alpha,r}(s)\vecket{0_E}} \leq \sqrt{\text{V}_\alpha(0)} \sup_{s\in [0, t]}\sqrt{[\delta A_{\alpha}(s), \delta A_{\alpha}^{\dagger}(s)]}, \nonumber\\
&\sup_{s, s' \in [0, t]}\abs{\vecbra{I_E}\delta A_{\alpha,\sigma}^{\dagger}(s)  \tilde{A}_{\alpha,r}(s')\vecket{0_E}}, \sup_{s, s' \in [0, t]}\abs{\vecbra{I_E}\delta A_{\alpha,\sigma}(s)  \tilde{A}_{\alpha,l}^{\dagger}(s')\vecket{0_E}} \leq \sqrt{\tilde{\text{V}}_\alpha(0)} \sup_{s\in [0, t]}\sqrt{[\delta A_{\alpha}(s), \delta A_{\alpha}^{\dagger}(s)]}.
\end{align}
Substituting Eqs.~\ref{eq:l2_error_sup_bound_1} and \ref{eq:l2_error_sup_bound_2} into Eq.~\ref{eq:l2_error_manipulation}, we obtain the lemma statement.
\end{proof}
\noindent \textcolor{black}{We emphasize that a similar result also appears in Ref.~\cite{unifiedZhen2024}.} We can now use lemma \ref{lemma:freq_cutoff} and \ref{lemma:error_two_L2} to obtain the following lemma quantifying the error made on introducing frequency cut-off.
\begin{lemma}\label{lemma:freq_cutoff_error}
For any $\mathcal{B}\subseteq\mathcal{A}$ and $\omega_c > 0$:
    \begin{align}
        \abs{\vecbra{O, I_E}\mathcal{U}_{\mathcal{B}}^\delta(t, 0) - \mathcal{U}_{\mathcal{B}}^{\delta, \omega_c}(t, 0)\vecket{\rho_S(0), 0_E}} \leq \frac{4\sqrt{\gamma_0} }{\sqrt{\delta^3\omega_c}} t^2 \abs{\mathcal{B}} \norm{O}\textnormal{TV}(\textnormal{V})\exp\bigg[-\frac{1}{2}\bigg(\frac{ {\omega_c}\delta}{16e}\bigg)^{1/2}\bigg],
    \end{align}
    where $\gamma_0$ is the constant introduced in lemma \ref{lemma:freq_cutoff}.
\end{lemma}
\begin{proof}
    We will use lemma \ref{lemma:error_two_L2} --- we recall that $\text{V}_\alpha^\delta(t - t') = [A_{\alpha}^{\delta}(t), A_{\alpha}^{\delta \dagger}(t')] = \big(\text{V}_\alpha \star \eta_{\delta}\star \eta_\delta\big)(t - t')$ and hence we obtain that
    \[
 \abs{\text{V}^\delta_\alpha(t)} \leq \int_{-\infty}^\infty \abs{\text{V}_\alpha(t -\tau)}\abs{(\rho_\delta \star \rho_\delta)(\tau)}d\tau \leq \text{TV}(\text{V}_\alpha) \norm{\rho_\delta \star \rho_\delta}_\infty \leq \text{TV}(\text{V}_\alpha)\frac{\norm{\rho}_1}{\delta} = \frac{\text{TV}(\text{V}_\alpha)}{\delta},
    \]
    and
    \begin{align}\label{eq:upper_bound_v_omega_c}
    \abs{\text{V}_\alpha^{\delta, \omega_c}(t)} \leq \int_{-\omega_c}^{\omega_c} \abs{\hat{v}_{\alpha}^{\delta}(\omega)}^2 \frac{d\omega}{2\pi} \leq \int_{-\infty}^\infty \abs{\hat{v}_\alpha^\delta(\omega)}^2 \frac{d\omega}{2\pi} = \abs{\text{V}_\alpha^{\delta}(0)} \leq \frac{\text{TV}(\text{V}_\alpha)}{\delta}.
    \end{align}
    Next, we obtain from lemma \ref{lemma:freq_cutoff} that
    \begin{align}\label{eq:high_freq_cutoff}
        \abs{[\delta A_{\alpha}(s), \delta A_{\alpha}^{\dagger}(s)]} = \int_{\abs{\omega} \geq \omega_c} \abs{\hat{v}_\alpha^\delta(\omega)}^2 \frac{d\omega}{2\pi} \leq \frac{\gamma_0}{\delta^2 \omega_c}\text{TV}(\text{V}_\alpha) \exp\bigg[-\bigg(\frac{\omega_c \delta}{16e}\bigg)^{1/2} \bigg],
    \end{align}
    where $\gamma_0$ is the constant introduced in lemma \ref{lemma:freq_cutoff}. Using lemma \ref{lemma:error_two_L2}, we then obtain that
    \[
    \abs{\vecbra{O, I_E}\mathcal{U}^\delta(t, 0) - \mathcal{U}^{\delta, \omega_c}(t, 0) \vecket{\sigma_0, 0_E}} \leq \frac{4\sqrt{\gamma_0} }{\sqrt{\delta^3\omega_c}} t^2 \norm{O}\bigg(\sum_{\alpha} \text{TV}(\text{V}_\alpha)\bigg)\exp\bigg[-\frac{1}{2}\bigg(\frac{ {\omega_c}\delta}{16e}\bigg)^{1/2}\bigg].
    \]
    Upper bounding $\text{TV}(\text{V}_\alpha) \leq \text{TV}(\text{V})$, we obtain the lemma.
\end{proof}

\emph{Analysis of the star-to-chain transformation}. Having introduced this frequency cut-off, we now generate a set of $N_m$ modes in the bath by performing a Lanczos iteration \cite{chin2010exact} with annihilation operators $b_{\alpha, 1}, b_{\alpha, 2} \dots b_{\alpha, N_m}$ where
\[
b_{\alpha, k} = \frac{1}{\mathfrak{N}_{\alpha, k}} \int_{-\omega_c}^{\omega_c} p_{\alpha, k}(\omega)\hat{v}_{\alpha}^\delta(\omega)a_{\alpha, \omega} d\omega \text{ with }\mathfrak{N}_{\alpha, k} =\bigg[\int_{-\omega_c}^{\omega_c} \abs{p_{\alpha, k}(\omega)}^2 \abs{\hat{v}_{\alpha}^\delta(\omega)}^2 d\omega\bigg]^{1/2},
\]
and $p_{\alpha, k}(\omega)$ is a $(k - 1)^\text{th}$ degree polynomial generated recursively via the rule: $p_{\alpha, 1}(\omega) = 1$, $B_{\alpha, 0}  = 0$,
\begin{align}
&i) \ A_{\alpha, j} = \frac{\int_{-\omega_c}^{\omega_c}\omega p_{\alpha, j}^2(\omega) \abs{\hat{v}^\delta_{\alpha}(\omega)}^2 d\omega}{\int_{-\omega_c}^{\omega_c} p_{\alpha, j}^2(\omega) \abs{\hat{v}^\delta_{\alpha}(\omega)}^2 d\omega}, \nonumber\\
&ii) \ p_{\alpha, j + 1}(\omega) = (\omega - A_{\alpha, j})p_{\alpha, j}(\omega) - B_{\alpha, j - 1} p_{\alpha, j - 1}(\omega), \nonumber\\
&iii) \ B_{\alpha, j} = \frac{\int_{-\omega_c}^{\omega_c} p_{\alpha, j + 1}^2(\omega)\abs{\hat{v}_{\alpha}(\omega)}^2 d\omega}{\int_{-\omega_c}^{\omega_c}p_{\alpha, j}^2(\omega)\abs{\hat{v}_\alpha(\omega)}^2 d\omega},
\end{align}
where we use the convention that $p_{\alpha, 0}(\omega) = 0$. It can be easily verified that the modes so generated are independent i.e.~they satisfy the canonical commutation relation $[b_{\alpha, j}, b_{\alpha', j'}^\dagger] = \delta_{\alpha, \alpha'}\delta_{j, j'}$.

\emph{Finding the chain hamiltonian}. Having generated the $N_m$ modes as in step 2, we then obtain a Hamiltonian of the bath in between these modes --- for the $\alpha^\text{th}$ bath, this is given by
\[
H_{\alpha, E} = \sum_{j = 1}^{N_m}\omega_{\alpha, j}b_{\alpha, j}^\dagger b_{\alpha, j} + \sum_{j = 1}^{N_m - 1}t_{\alpha, j}\big(b_{\alpha, j}^\dagger b_{\alpha, j + 1} + \text{h.c.}\big) \text{ with }\omega_{\alpha, j} = A_{\alpha, j} \text{ and }t_{\alpha, j} = \sqrt{B_{\alpha, j - 1}}.
\]
The approximate system environment Hamiltonian corresponding to $\mathcal{B} \subseteq \mathcal{A}$ obtained after keeping $N_m$ modes of the star-to-chain transformation would then be given by
\begin{align}\label{eq:hamil_particle_num_trunc}
    H_\mathcal{B}^{\delta, \omega_c, N_m}(t) = \sum_{\alpha \in \mathcal{B}}h_\alpha + \sum_{\alpha \in \mathcal{B}} \big( R_{\alpha}^\dagger A_{\alpha}^{\delta, \omega_c, N_m}(t) + \text{h.c.}\big),
\end{align}
where
\[
A_{\alpha}^{\delta, \omega_c, N_m}(t) =\mathfrak{N}_{\alpha, 1} e^{iH_{\alpha, E}t} b_{\alpha, 1}e^{-iH_{\alpha, E}t}
\]
Next, we consider the error between the regularized model with a frequency cutoff, given by the Hamiltonian $H_\mathcal{B}^{\delta, \omega_c}(t)$, defined in Eq.~\ref{eq:hamil_freq_cutoff} and the Hamiltonian $H_\mathcal{B}^{\delta, \omega_c, N_m}(t)$ in Eq.~\ref{eq:hamil_particle_num_trunc}. For this, we will need the following lemma from Ref.
\begin{lemma}[Ref.~\cite{trivedi2022description}]\label{lemma:old_lemma_star_to_chain}
    Suppose $\hat{v}_\alpha^{\delta, \omega_c, N_m}(\omega; t)$ is such that
    \[
    A_{\alpha}^{\delta, \omega_c, N_m}(t) = \int_{-\omega_c}^{\omega_c} \hat{v}_\alpha^{\delta, \omega_c, N_m}(\omega; t)a_{\alpha, \omega} \frac{d\omega}{\sqrt{2\pi}},
    \]
    then
    \[
    \int_{-\omega_c}^{\omega_c}\abs{\hat{v}_\alpha^{\delta, \omega_c, N_m}(\omega; t) - \hat{v}_\alpha^\delta(\omega) e^{-i\omega t}}^2 \frac{d\omega}{2\pi} \leq 2N_m^2 \bigg(\frac{2e\omega_ct}{N_m}\bigg)^{N_m}\bigg(\int_{-\omega_c}^{\omega_c}\abs{\hat{v}_\alpha^{\delta}(\omega)}^2 \frac{d\omega}{2\pi}\bigg)
    \]
\end{lemma}
\begin{lemma}\label{lemma:star_to_chain_error}
For any $\mathcal{B} \subseteq \mathcal{A}$ and $N_m > 0$:
    \[\abs{\vecbra{O, I_E}\mathcal{U}_\mathcal{B}^{\delta, \omega_c}(t, 0) - \mathcal{U}_\mathcal{B}^{\delta, \omega_c, N_m}(t, 0)\vecket{\rho_S(0), 0_E}} \leq 2\sqrt{2}t^2 \abs{\mathcal{B}}\norm{O} \textnormal{TV}(\textnormal{V}) \frac{N_m}{\delta}\bigg(\frac{2e\omega_c t}{N_m}\bigg)^{N_m/2}
    \]
\end{lemma}
\begin{proof}
    We will again use lemma \ref{lemma:error_two_L2}. We begin  by noting that
    \[
    \abs{\text{V}_\alpha^{\delta, \omega_c, N_m}(0)} = \abs{\text{V}_\alpha^{\delta, \omega_c}(0)} \leq \frac{\text{TV}(\text{V}_\alpha)}{\delta}.
    \]
    To bound $\abs{[\delta A_{\alpha}(s), \delta A_{\alpha}^{\dagger}(s)]}$, where $\delta A_{\alpha}(s) = A_{\alpha}^{\delta, \omega_c}(s) - A_{\alpha}^{\delta, \omega_c, N_m}(s)$, we use lemma \ref{lemma:old_lemma_star_to_chain}:
    \[
    \abs{[\delta A_{\alpha}(s), \delta A_{\alpha}^{\dagger}(s)]} = \int_{-\omega_c}^{\omega_c} \abs{\hat{v}_\alpha^{\delta, \omega_c, N_m}(\omega; t) - \hat{v}^\delta_\alpha(\omega) e^{-i\omega t}}^2 \frac{d\omega}{2\pi} \leq 2 \text{V}_{\alpha}^{\delta, \omega_c}(0) N_m^2 \bigg(\frac{2e\omega_c t}{N_m}\bigg)^{N_m} \leq \frac{2\text{TV}(\text{V}_\alpha)}{\delta}  N_m^2  \bigg(\frac{2e\omega_c t}{N_m}\bigg)^{N_m}.
    \]
    From these estimates together with $\textnormal{TV}(\text{V}_\alpha) \leq \textnormal{TV}(\text{V})$, we obtain the lemma statement.
\end{proof}

\subsection{Proof of proposition 2}
Given a local observable $O_X$ supported on region $X$, and the initial state being vacuum in the independent environment. We will consider two models --- one is the original model described by the Hamiltonian $H(t)$, and the next is the model obtained by performing the star to chain transformation $H^{\delta, \omega_c, N_m}(t)$ and we are interested in analyzing the error between the expected value of the local observable in both of these models. We will estimate this error, $\text{Err}_{O_X}$ as a sum of three parts:
\[
\text{Err}_{O_X} \leq \Delta_{O_X}(t, 0; l) + \Delta_{O_X}^{\delta, \omega_c, N_m}(t, 0; l) + \text{Err}_{O_X}^{l},
\]
where $\Delta_{O_X}(t, 0; l)$ is the error in the local observable on truncating the Hamiltonian $H(t)$ locally as defined in Eq.~\ref{eq:delta_Ox_redef}, $\Delta_{O_X}(t, 0; l)$ is this truncation error but for the Hamiltonian $H^{\delta, \omega_c, N_m}(t)$ and $\text{Err}_{O_X}^{l}$ is the error in between the dynamics of the two truncated Hamiltonians which are both now supported only within the region of distance $l$ around $X$.

The error $\Delta_{O_X}(t, 0; l)$ can be directly estimated from proposition 1. Note now we have $L=2$. A good estimate of the error $\Delta_{O_X}^{\delta, \omega_c, N_m}(t, 0; l)$ is slightly more subtle, and requires a good estimate of the Lieb-Robinson velocity corresponding to the Hamiltonian $H^{\delta, \omega_c, N_m}(t)$, which in turn requires us to estimate the total variation of its kernels --- more specifically, if $\text{V}_\alpha^{\delta, \omega_c, N_m}(t - t') = [A_{\alpha}^{\delta, \omega_c, N_m}(t), A_{\alpha}^{\delta, \omega_c, N_m \dagger}(t')]$, then we define the upper bounding kernel $\text{V}^{\delta, \omega_c, N_m}$ as usual:
\[
\text{V}^{\delta, \omega_c, N_m}(\tau) = \sup_{\alpha \in \mathcal{A}}\  \abs{\text{V}_\alpha^{\delta, \omega_c, N_m}(\tau)}.
\]
Note that since we are only interested in the dynamics from $0$ to $t$, we will only estimate the quantity $\norminf{\bm\mu^{\textnormal{R}}[0,t]}$ of the kernels corresponding to $H^{\delta, \omega_c, N_m}(t)$, which equals to  $\text{TV}(\text{V}^{\delta, \omega_c, N_m}; [0, t])$. To do so, we will utilize the fact that $\text{V}_\alpha^\delta(t)$ is close to $\text{V}_\alpha^{\delta, \omega_c, N_m}(t)$. In particular,
\begin{align*}
\abs{\text{V}_\alpha^\delta(t) - \text{V}_\alpha^{\delta, \omega_c, N_m}(t)} \leq  \abs{\text{V}_\alpha^\delta(t) - \text{V}_\alpha^{\delta, \omega_c}(t)} + \abs{\text{V}_{\alpha, \omega_c}^\delta(t) - \text{V}_\alpha^{\delta, \omega_c, N_m}(t)}.
\end{align*}
Now, note that from lemma \ref{lemma:freq_cutoff} and $\text{TV}(\text{V}_\alpha) \leq \text{TV}(\text{V})$,
\begin{align*}
    \abs{\text{V}_\alpha^\delta(t) - \text{V}_\alpha^{\delta, \omega_c}(t)} = \abs{\int_{\abs{\omega} \geq \omega_c} \abs{\hat{v}_\alpha^\delta(\omega)}^2 \frac{d\omega}{2\pi}} \leq \frac{\gamma_0}{\delta^2 \omega_c} \text{TV}(\text{V}) \exp\bigg[-\bigg(\frac{\omega_c \delta}{16e} \bigg)^{1/2}\bigg],
\end{align*}
where we have used lemma \ref{lemma:freq_cutoff}. Next, consider [$\hat{v}_\alpha^{\delta, \omega_c, N_m}(\omega; t)$ is defined in lemma ]
\begin{align*}
    \abs{\text{V}_{\alpha}^{\delta, \omega_c}(t) - \text{V}_\alpha^{\delta, \omega_c, N_m}(t)} &= \abs{\int_{-\omega_c}^{\omega_c} \abs{\hat{v}^\delta_\alpha(\omega)}^2 e^{-i\omega t}\frac{d\omega}{2\pi} - \int_{-\omega_c}^{\omega_c}\big(\hat{v}_\alpha^{\delta, \omega_c, N_m}(\omega; 0)\big)^*\hat{v}_\alpha^{\delta, \omega_c, N_m}(\omega; t) \frac{d\omega}{2\pi}}, \nonumber\\
    &\numeq{1}\abs{\int_{-\omega_c}^{\omega_c} \big({\hat{v}^\delta_\alpha(\omega)}\big)^*\big(\hat{v}_\alpha^\delta(\omega) e^{-i\omega t} - \hat{v}^{\delta, \omega_c, N_m}_\alpha(\omega; t)\big)\frac{d\omega}{2\pi}}, \nonumber \\
    &\leq \bigg[\int_{-\omega_c}^{\omega_c} \abs{\hat{v}_\alpha^\delta(\omega)}^2 \frac{d\omega}{2\pi}\bigg]^{1/2}\bigg[\int_{-\omega_c}^{\omega_c}\abs{\hat{v}_\alpha^\delta(\omega) e^{-i\omega t} - \hat{v}^{\delta, \omega_c, N_m}_\alpha(\omega; t)}^2 \frac{d\omega}{2\pi}\bigg]^{1/2}, \nonumber\\
    &\numleq{2} \sqrt{2}N_m \bigg(\frac{2e\omega_c t}{N_m}\bigg)^{N_m/2} \int_{-\omega_c}^{\omega_c} \abs{\hat{v}_\alpha^\delta(\omega)}^2 \frac{d\omega}{2\pi}, \nonumber\\
    &=\sqrt{2}N_m \bigg(\frac{2e\omega_c t}{N_m}\bigg)^{N_m/2} \text{V}_\alpha^{\delta, \omega_c}(0) \numleq{3} \sqrt{2}N_m \bigg(\frac{2e\omega_c t}{N_m}\bigg)^{N_m/2} \frac{\text{TV}(\text{V})}{\delta},
\end{align*}
where in (1) we have used the fact that $\hat{v}_\alpha^{\delta, \omega_c, N_m}(\omega; 0) = \hat{v}_\alpha^\delta(\omega)$ for almost all $\omega \in [-\omega_c, \omega_c]$ --- this is implied by lemma \ref{lemma:old_lemma_star_to_chain}, in (2) we have used lemma \ref{lemma:old_lemma_star_to_chain} and in (3) we have used Eq.~\ref{eq:upper_bound_v_omega_c}. Thus, we have that, for any $t$,
\begin{align*}
&\text{TV}(\text{V}^{\delta, \omega_c, N_m}; [0, t]) = \int_0^t \sup_{\alpha \in \mathcal{A}} \abs{\text{V}_\alpha^{\delta, \omega_c, N_m}(s)}ds, \nonumber \\
&\qquad \leq \int_0^t \sup_{\alpha}\abs{\text{V}_\alpha^{\delta}(s)} ds  +  \frac{\gamma_0 t}{\delta^2 \omega_c} \text{TV}(\text{V}) \exp\bigg[-2\bigg(\frac{\omega_c \delta}{16e}\bigg)^{1/2}\bigg] +   \sqrt{2}N_m t \bigg(\frac{2e\omega_c t}{N_m}\bigg)^{N_m/2} \frac{\text{TV}(\text{V})}{\delta},, \nonumber\\
&\qquad \leq \text{TV}(\text{V}) + \frac{\gamma_0 t}{\delta^2 \omega_c} \text{TV}(\text{V}) \exp\bigg[-\bigg(\frac{\omega_c \delta}{16e}\bigg)^{1/2}\bigg] +   \sqrt{2}N_m t \bigg(\frac{2e\omega_c t}{N_m}\bigg)^{N_m/2} \frac{\text{TV}(\text{V})}{\delta}.
\end{align*}
which gives us that, for dynamics upto time $t$, a lieb-Robinson velocity of the model $H^{\delta, \omega_c, N_m}(t)$ is given by
\begin{align}\label{eq:lieb_robinson_star_to_chain}
v_\text{LR}^{\delta, \omega_c, N_m} = v_\text{LR} + O\bigg(\frac{t}{\delta^2\omega_c} \exp\big(-\Omega(\sqrt{\omega_c \delta})\big)\bigg) + O\bigg(\frac{N_m t}{\delta} \bigg(\frac{2e\omega_c t}{N_m}\bigg)^{N_m/2}\bigg),
\end{align}
where $v_\text{LR}$ is the lieb Robinson velocity of the original model.

Now, we estimate the number of modes $N_m$ required to reduce the total error in the observable, $\text{Err}_{O_X}$, to $\leq \varepsilon$. We first pick the length $l$ such that $\Delta_{O_X}(t, 0; l), \Delta_{O_X}^{\delta, \omega_c, N_m}(t, 0; l) \leq O(\varepsilon)$. Since our estimate for $v_\text{LR}^{\delta, \omega_c, N_m}$ is larger than $v_\text{LR}$, it is sufficient to ensure that $ \Delta_{O_X}^{\delta, \omega_c, N_m}(t, 0; l) \leq O(\varepsilon)$ --- from proposition 1, this is obtained by making the choice
\begin{align}\label{eq:choice_ell}
l = \Theta(\log\varepsilon^{-1}) + \Theta(v_\text{LR}^{\delta, \omega_c, N_m} t).
\end{align}
For this choice $l$, we next estimate $\text{Err}_{O_X}^l$ from lemmas \ref{lemma:star_to_chain_reg_error}, \ref{lemma:freq_cutoff_error} and \ref{lemma:star_to_chain_error}, we obtain that:
\begin{align}
    \text{Err}_{O_X}^l \leq O(l^d t \lambda_0(2\delta)) + O(t^2 l^{2d}\delta) + O\bigg(\frac{l^d t^2}{\sqrt{\delta^3 \omega_c}}  \exp(-\Omega(\sqrt{\omega_c \delta}))\bigg) + O\bigg(l^d t^2\frac{N_m}{\delta}\bigg(\frac{2e\omega_c t}{N_m}\bigg)^{N_m/2}\bigg).
\end{align}
We now make the following choices for the parameters $\omega_c$ and $\delta$ ---
\[
\omega_c = \frac{N_m}{2e^2t} \text{ and }\delta = \frac{N_m^{1 - \epsilon}}{\omega_c } = 2e^2 t N_m^{-\epsilon},
\]
where $\epsilon < 1$ is a small positive constant. With this choice together with Eqs.~\ref{eq:lieb_robinson_star_to_chain} and \ref{eq:choice_ell}, we then have that
\begin{align*}
l &= \Theta(\log \varepsilon^{-1}) + \Theta(t) + O\bigg(\frac{t}{N_m^{1-2\epsilon}} e^{-\Omega(N_m^{(1-\epsilon)/2})}\bigg) + O\big(N_m^{1 + \epsilon}e^{-N_m/2}\big),
\end{align*}
and
\[
\text{Err}^l_{O_X} \leq O\big(l^d t\lambda_0\big(4e^2 t N_m^{-\epsilon}\big) \big) + O\big(t^3 l^{2d}N_m^{-\epsilon} \big) + O\big( {l^d t N_m^{(3\epsilon - 1)/2}} e^{-\Omega(N_m^{(1 - \epsilon)/2})}\big) + O\big(l^d t N_m^{1 + \epsilon} e^{-N_m/2}\big).
\]
Clearly, the most dominant terms in this upper bound for $\text{Err}_{O_X}^l$ are the first two since the remaining two are exponentially suppressed in $N_m$. We can then pick $N_m$ to make these two terms smaller than $O(\varepsilon)$ to obtain
\[
N_m = \Theta\bigg(\bigg(\frac{l^{2d}t^3}{\varepsilon}\bigg)^{1 + o(1)}\bigg) +
\Theta\bigg(\bigg(t\lambda_0^{-1}\bigg(\frac{\varepsilon}{l^d t}\bigg)\bigg)^{1 + o(1)}\bigg),
\]
where, since $\epsilon$ can be chosen to be a constant arbitrarily close to 1, we have replaced the power of $1/\epsilon$ in the $O$-notation as $1 + o(1)$. With this choice of $N_m$, we see that as $t \to \infty, \varepsilon \to 0$, we can neglect the terms that are exponentially small in $N_m$ in both $l$ and $\text{Err}_{O_X}^l$ --- in particular, we then obtain that $l = \Theta(\log {\varepsilon}^{-1}) + \Theta(t)$ with which our estimate for $N_m$ becomes:
\[
N_m = \Theta\big(\big(\varepsilon^{-1}{t^{2d + 3}}\big)^{1 + o(1)}\big) + \Theta\big( \big({\varepsilon}^{-1}\log \varepsilon^{-1}\big)^{1 + o(1)}\big) + \Theta\big(\big(t\mu^{\textnormal{L}}_0\big(\varepsilon^{-1}\Theta(t^{d + 1}) + \Theta(t \log^d \varepsilon^{-1})\big)^{1 + o(1)} \big),
\]
where we have introduced $\mu^{\textnormal{L}}_0(x) = \lambda_0^{-1}(x^{-1})$. Since we have ensured $\Delta_{O_X}(t, 0; l), \Delta_{O_X}^{\delta, \omega_c, N_m}(t, 0, l), \text{Err}_{O_X}^l \leq O(\varepsilon)$, we obtain the proposition statement.
\newpage
\appendix
\section{Complete proof of Proposition \ref{lemma:mollification}}\label{appendix:prop_4}
\noindent We first present a useful technical lemma that we would need for proving proposition 4.
\begin{lemma}\label{lemma:smallness_double_integral_kernel}
    Suppose $f:\mathbb{R}^2 \to \mathbb{C}$ is continuous and differentiable with respect to either of its arguments then
    \begin{enumerate}
        \item[(a)] Consider a closed interval, $[a, b]$. Suppose that $\K \in \mathcal{K}_1(\mathbb{R})$ is $\lambda_0, \lambda_1-$small over a test function space $\textnormal{PC}_{\{0, b - a\}}^1(\mathbb{R})$. then
        \[
        \abs{\int_a^b \int_a^{s_1} \K(s_1 - s_2) f(s_1,s_2) ds_1 ds_2} \leq (\lambda_0 (b - a) + \lambda_1) \sup_{s_1, s_2 \in [a, b]} \abs{f(s_1, s_2)} + \lambda_1 (b - a) \max_{i \in \{1,2\}} \sup_{s_1, s_2 \in [a, b]}\smallabs{\partial_{s_i}f(s_1, s_2)}.
        \]
        \item[(b)] Consider two closed intervals, $[a_1, b_1]$ and $[a_2, b_2]$ which may or may not intersect with each other. Suppose that $\K \in \mathcal{K}_1(\mathbb{R})$ is $\lambda_0, \lambda_1-$small over a test function space $\textnormal{PC}_{\{a_1 - a_2, a_1 - b_2, b_1 - a_2, b_1 - b_2\}}^1(\mathbb{R})$ then
        \[
        \abs{\int_{a_1}^{b_1} \int_{a_2}^{b_2}\K(s_1 - s_2) f(s_1, s_2) ds_1 ds_2} \leq (\lambda_0 (b_2 - a_2) + 2\lambda_1) \sup_{\substack{s_1 \in [a_1, b_1] \\ s_2 \in [a_2, b_2]}} \abs{f(s_1, s_2)} + \lambda_1 \big (b_2 - a_2) \max_{i \in \{1,2\}}\sup_{\substack{s_1 \in [a_1, b_1] \\ s_2 \in [a_2, b_2]}} \abs{\partial_{s_i} f(s_1, s_2)}.
        \]
    \end{enumerate}
\end{lemma}
\begin{proof}
    (a) For notational convenience, introduce
    \begin{align}
        \phi_0 = \sup_{s_1,s_2 \in [a, b]} \smallabs{f(s_1, s_2)} \quad \text{and}\quad \phi_1 = \max_{i \in \{1, 2\}} \sup_{s_1, s_2 \in [a, b]} \smallabs{\partial_{s_i}f(s_1, s_2)}.
    \end{align}
    First, we rewrite the integral under consideration, by a change of variables $s_1, s_2 \to s = s_1 - s_2, s' = s_2$, as
\begin{align*}
    \int_a^b \int_a^{s_1} \K(s_1 - s_2) f(s_1,s_2) ds_1 ds_2 &= \int_{0}^{b - a} \K(s) \bigg[\int_{a}^{b - s} f(s + s', s') ds'\bigg] ds, \nonumber\\
    &= \int_{-\infty}^\infty \K(s) \underbrace{\bigg[\Theta_{[0, b - a]}(s)\int_a^{b - s}f(s+  s', s') ds'\bigg]}_{F(s)} ds,
\end{align*}
Clearly, we note that $F$ {has} a discontinuity/non-differentiability only at $s \in \{0, b - a\} \subseteq t^*$, and consequently $F \in \text{PC}_{\{0, b - a\}}^1(\mathbb{R})$. Furthermore, we also have that
\begin{align}\label{eq:part_a_first_integral_1}
\forall s \in \mathbb{R}:  \abs{F(s)} \leq \abs{\int_a^{b - s}f(s + s', s') ds'} \leq (b - a)\phi_0 \implies \norm{F}_\infty \leq (b - a) \phi_0,
\end{align}
and for $s \notin \{0, b - a\}$
\begin{align*}
    F'(s) = \begin{dcases}
        -f(b, b - s) + \int_{a}^{b - s}\partial_{1}f(s + s', s') ds' & \text{ if }s \in (0, b - a), \\
        0 & \text{ if }s\notin [0, b - a].
    \end{dcases}
\end{align*}
Therefore, we also have that
\begin{align}\label{eq:part_a_first_integral_2}
\forall s \notin \{0, b - a\} : \abs{F'(s)} \leq \phi_0 + (b - a) \phi_1 \implies \norm{F'}_\infty \leq \phi_0 + (b - a) \phi_1.
\end{align}
From Eqs.~\ref{eq:part_a_first_integral_1} and \ref{eq:part_a_first_integral_2} together with the fact that $\K$ is $\lambda_0, \lambda_1$-small on $\text{PC}_{t^*}^1(\mathbb{R})$, it follows that
\begin{align*}
    \abs{\int_a^b \int_a^{s_1} \K(s_1 - s_2) f(s_1, s_2) ds_1 ds_2 } = \abs{\int_{-\infty}^\infty \K(s) F(s) ds} \leq \lambda_0 (b - a) \phi_0 + \lambda_1 \big(\phi_0 + (b - a) \phi_1\big),
\end{align*}
which proves the lemma statement. \\

(b) For notational convenience, introduce
\begin{align}
    \phi_0 = \sup_{\substack{s_1 \in [a_1, b_1] \\ s_2 \in [a_2, b_2]}} \abs{f(s_1, s_2)} \quad \text{and}\quad \phi_1 = \max_{i \in \{1, 2\}} \sup_{s_1, s_2 \in [a, b]} \smallabs{\partial_{s_i}f(s_1, s_2)}.
\end{align}

We begin by rewriting the integral under consideration as
    \begin{align}
    \int_{a_1}^{b_1} \int_{a_2}^{b_2}\K(s_1 - s_2) f(s_1, s_2) ds_1 ds_2 &= \int_{-\infty}^\infty \int_{-\infty}^\infty \K(s_1 - s_2) f(s_1,s_2) \Theta_{[a_1, b_1]}(s_1) \Theta_{[a_2, b_2]}(s_2) ds_1 ds_2, \nonumber \\
    &=\int_{-\infty}^\infty \K(s) \underbrace{\bigg[\int_{a_2}^{b_2}f(s + s', s') \Theta_{[a_1, b_1]}(s' + s)  ds'\bigg]}_{F(s)}ds.
    \end{align}
    We will now establish that $F \in \text{PC}_{ \{a_1 - a_2, a_1 - b_2, b_1 - a_2, b_1 - b_2\}}^{1}(\mathbb{R})$. To do so, we need to show that all the points of discontinuity or non-differentiability of $F$ are in $ \{a_1 - a_2, a_1 - b_2, b_1 - a_2, b_1 - b_2\}$. To do so, let us attempt to differentiate $F(s)$ with respect to $s$ --- we obtain that
    \[
    F'(s) = \int_{a_2}^{b_2} \partial_{s_1}f(s + s', s') \Theta_{[a_1, b_1]}(s' + s) ds' + \int_{a_2}^{b_2} f(s + s', s') \big(\delta(s + s' - a_1) - \delta(s + s' - b_1)\big) ds',
    \]
    with the derivative being only defined when $s \notin \{a_1 - a_2, a_1 - b_2, b_1 - a_2, b_1 - b_2\}$. From this expression, it is clear that $F'(s)$ can only have a discontinuity/non-differentiability when $s \in \{a_1 - a_2, a_1 - b_2, b_1 - a_2, b_1 - b_2\}$, thus establishing that $F \in \text{PC}_{ \{a_1 - a_2, a_1 - b_2, b_1 - a_2, b_1 - b_2\}}^1(\mathbb{R})$. We can also see from this expression that for $s \notin \{a_1 - a_2, a_1 - b_2, b_1 - a_2, b_1 - b_2\}$,
    \begin{align}\label{eq:part_b_bound_deriv}
    \abs{F'(s)} \leq \phi_1 ( b_2 - a_2) + 2 \phi_0 \implies \norm{F'}_\infty \leq \phi_1 ( b_2 - a_2) + 2 \phi_0.
    \end{align}
    Furthermore, it is also clear that $\forall  s \notin  \{a_1 - a_2, a_1 - b_2, b_1 - a_2, b_1 - b_2\}$,
    \begin{align}\label{eq:part_b_bound_fun}
    \abs{F(s)} \leq \phi_0 (b_2 - a_2).
    \end{align}
    From Eqs.~\ref{eq:part_b_bound_deriv}, \ref{eq:part_b_bound_fun} and the fact that $\K$ is $\lambda_0, \lambda_1-$small on $\text{PC}_{ \{a_1 - a_2, a_1 - b_2, b_1 - a_2, b_1 - b_2\}}^1(\mathbb{R})$, we obtain that 
    \[
\abs{\int_{a_1}^{b_1} \int_{a_2}^{b_2}\K(s_1 - s_2) f(s_1, s_2) ds_1 ds_2} = \abs{\int_{-\infty}^\infty \K(s) F(s) ds} \leq \lambda_0 \phi_0 (b_2 - a_2) + \lambda_1 \big(\phi_1 (b_2 - a_2) + 2\phi_0\big),
    \]
    which proves the lemma.
    
\end{proof}
\begin{proof}[Proof of proposition \ref{lemma:mollification}]
It will be convenient to introduce, for $t > s$,
\begin{align}
    u(t, s) = \sup_{\alpha,\alpha'}\text{TV}(K_{\alpha, \alpha'};[s,t])\quad\text{and}\quad  u_c(t, s)= \sup_{\alpha, \alpha'} \int_s^t\smallabs{K_{\alpha, \alpha'; c}(s')}ds',
\end{align}
where $K_{\alpha, \alpha'; c}$ is the continuous part of $K_{\alpha, \alpha'}$. Note that since, by assumption, $K_{\alpha, \alpha'}$ are locally integrable (assumption \ref{assump:local_int}) and the indices $\alpha, \alpha'$ run over a finite set, $u(t, s), u_c(t, s) < \infty$. Furthermore, since $K_{\alpha, \alpha'; c}$ are continuous functions, so is $u_c(t, s)$ with $u_c(t, t) = 0$.

Now, consider $\delta_1, \delta_2 > 0$ --- we obtain that
    \begin{align*}
        &\abs{\vecbra{I_E, O_S}\bigg(\prod_{i = 1}^n {\mathcal{U}}^{\delta_1}_{\mathcal{A}_i}(t_i, s_i)\Omega_i \bigg) \vecket{\rho_E, \sigma_S} - \vecbra{I_E, O_S}\bigg(\prod_{i = 1}^n {\mathcal{U}}^{\delta_2}_{\mathcal{A}_i}(t_i, s_i)\Omega_i \bigg) \vecket{\rho_E, \sigma_S}} \nonumber\\
        &  \leq\sum_{j = 1}^n\abs{\vecbra{I_E, O_S} \bigg(\prod_{i = 1}^{j - 1}{\mathcal{U}}_{\mathcal{A}_i}^{\delta_1}(t_i, s_i) \Omega_i\bigg) \bigg({\mathcal{U}}^{\delta_1}_{\mathcal{A}_j}(t_j, s_j) - {\mathcal{U}}^{\delta_2}_{\mathcal{A}_j}(t_j, s_j)\bigg)\Omega_j \bigg(\prod_{i =  j + 1}^{n}{\mathcal{U}}_{\mathcal{A}_i}^{\delta_2}(t_i, s_i) \Omega_i\bigg)\vecket{\rho_E, \sigma_S}}, \nonumber\\
        & \leq \sum_{j = 1}^n \sum_{\substack{\alpha \in {\mathcal{A}_j} \\ \sigma}} \abs{\int_{s_j}^{t_j} \vecbra{I_E, O_S} \mathcal{V}^{\delta_1}_{1, j - 1} {\mathcal{U}}^{\delta_1}_{\mathcal{A}_j}(t_j, s) (B_{\alpha,\sigma}^{\delta_1}(s) - B_{\alpha,\sigma}^{\delta_2}(s)) R^{\sigma}_{\alpha}(s) {\mathcal{U}}_{\mathcal{A}_j}^{\delta_2}(s, s_j)\Omega_j \mathcal{V}^{\delta_2}_{j + 1, n}\bigg)\vecket{\rho_E, \sigma_S} ds}
    \end{align*}
where, for notational convenience, we have defined
\[
\mathcal{V}^{\delta}_{j, j'} = \prod_{i = j}^{j'}{\mathcal{U}}^\delta_{\mathcal{A}_i}(t_i, s_i)\Omega_i.
\]
Now, using the Wick's theorem, we obtain that
\begin{align}\label{eq:telescoping_sum}
    &\abs{\vecbra{I_E, O_S}\bigg(\prod_{i = 1}^n {\mathcal{U}}^{\delta_1}_{\mathcal{A}_i}(t_i, s_i) \bigg) \vecket{\rho_E, \sigma_S} - \vecbra{I_E, O_S}\bigg(\prod_{i = 1}^n {\mathcal{U}}^{\delta_2}_{\mathcal{A}_i}(t_i, s_i) \bigg) \vecket{\rho_E, \sigma_S}} \leq \nonumber\\
    &\qquad \qquad \qquad \qquad \qquad \sum_{\substack{j, j' \in [1:n] \\ j \neq j'}} \sum_{\substack{\alpha,\alpha' \in \mathcal{A}_{j} \\ \sigma; \sigma'}} \abs{\Gamma_{j, j', \alpha,\alpha'}^{\sigma; \sigma'}} + \sum_{j \in [1:n]} \sum_{\substack{\alpha,\alpha' \in \mathcal{A}_{j} \\ \sigma; \sigma'}} \bigg( \abs{\Gamma^{ \sigma; \sigma'}_{+, j, \alpha,\alpha'}} + \abs{\Gamma^{ \sigma;  \sigma'}_{-, j, \alpha,\alpha'}}\bigg) ,
\end{align}
where for $j \neq j'$:
\begin{align*}
    \Gamma_{j, j', \alpha,\alpha'}^{ \sigma; \sigma'} = \begin{dcases}\int_{s = s_j}^{t_j} \int_{s' = s_{j'}}^{t_{j'}} \Delta_{+; \alpha,\alpha'}^{ \sigma; \sigma'}(s' - s)f_{j, j', \alpha,\alpha'}^{ \sigma; \sigma'}(s', s) ds' ds & \text{ if }j' < j, \\
    \int_{s = s_j}^{t_j}\int_{s' = s_{j'}}^{t_{j'}}\Delta^{ \sigma'; \sigma}_{-, \alpha,\alpha'}(s' - s) f_{j, j', \alpha,\alpha'}^{ \sigma; \sigma'}(s', s) ds' ds & \text{ if }j' > j
    \end{dcases}
\end{align*}
with \begin{align*}
    &\qquad \qquad \qquad \Delta^{ \sigma;  \sigma'}_{+; \alpha,\alpha'}(s) = \K^{ \delta_1, \delta_1}_{\alpha',\alpha, \sigma}(s) - \K^{ \delta_1, \delta_2}_{\alpha',\alpha, \sigma}(s),\  \Delta^{\sigma; \sigma'}_{-; \alpha,\alpha'}(s) = \K^{\delta_1, \delta_2}_{\alpha,\alpha', \sigma'}(-s) - \K^{\delta_2, \delta_2}_{\alpha,\alpha', \sigma'}(-s), \nonumber\\
    &\text{For }j' < j:\nonumber \\
    &\qquad \qquad \qquad f_{j, j', \alpha,\alpha'}^{ \sigma; \sigma'}(s', s) = \vecbra{I_E, O_S}\mathcal{V}_{1, j' - 1}^{\delta_1}\mathcal{U}_{\mathcal{A}_{j'}}^{\delta_1}(t_{j'}, s') R_{\alpha', \sigma'}(s') \mathcal{U}_{\mathcal{A}_{j'}}^{\delta_1}(s', s_{j'}) \times \nonumber\\
    &\qquad \qquad \qquad \qquad \qquad \qquad \qquad \qquad\Omega_{j'}\mathcal{V}^{\delta_1}_{j' + 1, j - 1}\mathcal{U}^{\delta_1}_{\mathcal{A}_j}(t_j, s) R_{\alpha, \sigma}(s) \mathcal{U}^{\delta_2}_{\mathcal{A}_j}(s, s_j)\Omega_j \mathcal{V}^{\delta_2}_{j + 1, n}\vecket{\rho_E, \sigma_0}, \nonumber\\
    &\text{For }j < j': \nonumber \\
    &\qquad \qquad \qquad  f_{j, j', \alpha,\alpha'}^{ \sigma; \sigma'}(s', s) = \vecbra{I_E, O_S} \mathcal{V}^{\delta_1}_{1, j - 1}\mathcal{U}^{\delta_1}_{\mathcal{A}_j}(t_j, s) R_{\alpha, \sigma}(s)\mathcal{U}^{\delta_2}_{\mathcal{A}_j}(s, s_j)\times \nonumber\\
    &\qquad \qquad \qquad \qquad \qquad \qquad \qquad \qquad \Omega_j \mathcal{V}^{\delta_2}_{j + 1, j' - 1}\mathcal{U}^{\delta_2}_{\mathcal{A}_{j'}}(t_{j'}, s')R_{\alpha', \sigma'}(s')\mathcal{U}^{\delta_2}_{\mathcal{A}_{j'}}(s', s_{j'})\Omega_{j'}\mathcal{V}^{\delta_2}_{j' + 1, n}\vecket{\rho_E, \sigma_0}.
\end{align*}
and, 
\begin{align*}
    &\Gamma^{\sigma; \sigma'}_{+, j, \alpha,\alpha'} = \int_{s = s_j}^{t_j} \int_{s' =  s}^{t_j}\Delta_{+, \alpha,\alpha'}^{\sigma; \sigma'}(s' - s) f_{+, j, \alpha,\alpha'}^{\sigma; \sigma'}(s', s)ds' ds, \\
    &\Gamma^{\sigma; \sigma'}_{-, j, \alpha,\alpha'} = \int_{s = s_j}^{t_j} \int_{s' =  s_j}^{s}\Delta_{-, \alpha,\alpha'}^{\sigma; \sigma'}(s' - s) f_{-, j, \alpha,\alpha'}^{ \sigma; \sigma'}(s', s)ds' ds,
\end{align*}
where
\begin{align*}
    &f_{+, j, \alpha,\alpha'}^{\sigma; \sigma'}(s', s) = \vecbra{I_E, O_S} \mathcal{V}^{\delta_1}_{1, j - 1} {\mathcal{U}}^{\delta_1}_{\mathcal{A}_j}(t_j, s') R_{\alpha', \sigma'}(s'){\mathcal{U}}_{\mathcal{A}_j}^{\delta_1}(s', s)R_{\alpha, \sigma}(s){\mathcal{U}}^{\delta_2}_{\mathcal{A}_j}(s, s_j) \Omega_j \mathcal{V}_{j + 1, n}^{\delta_2}\vecket{\rho_E, \sigma_S}, \nonumber\\
    &f_{-, j, \alpha,\alpha'}^{\sigma; \sigma'}(s', s) =  \vecbra{I_E, O_S} \mathcal{V}^{\delta_1}_{1, j - 1} {\mathcal{U}}^{\delta_1}_{\mathcal{A}_j}(t_j, s) R_{\alpha, \sigma}(s){\mathcal{U}}_{\mathcal{A}_j}^{\delta_2}(s, s')R_{\alpha', \sigma'}(s'){\mathcal{U}}^{\delta_2}_{\mathcal{A}_j}(s', s_j) \Omega_j \mathcal{V}_{j + 1, n}^{\delta_2}\vecket{\rho_E, \sigma_S}.
\end{align*}
Next, we separately upper bound $\smallabs{\Gamma_{j, j', \alpha,\alpha'}^{ \sigma;  \sigma'}}$ and $\smallabs{\Gamma_{\pm, j, \alpha,\alpha'}^{ \sigma;  \sigma'}}$. For this, we will use lemmas \ref{lemma:kernel_appx} and \ref{lemma:smallness_double_integral_kernel}.

\underline{Smallness of $\Delta^{\sigma; \sigma'}_{\pm; \alpha,\alpha'}$}. First, we use lemma \ref{lemma:kernel_appx} to establish that $\Delta^{\sigma; \sigma'}_{\pm; \alpha,\alpha'}$ is ``small" --- the test function space that we will employ will be $\text{PC}_{t^*}^1(\mathbb{R})$ where $t^* = \{t_j - t_{j'}, s_j - s_{j'}, t_j - s_{j'}, s_j - t_{j'} : j, j' \in [1:n]\}$. We introduce $\tau=\max\abs{t}$ for $t\in t^{*}$. Now, consider $\Delta^{\sigma;\sigma'}_{+; \alpha,\alpha'}$ and $f \in \text{PC}_{t^*}^1(\mathbb{R})$:
\begin{align*}
    &\abs{\int_{-\infty}^\infty \Delta^{ \sigma; \sigma'}_{+; \alpha,\alpha'}(s) f(s)} = \abs{\int_{-\infty}^\infty (\K^{\delta_1, \delta_1}_{\alpha',\alpha, \sigma}(s) - \K^{\delta_1, \delta_2}_{\alpha',\alpha, \sigma}(s))f(s)ds}, \nonumber\\
    & \qquad\numleq{1} \abs{\int_{-\infty}^\infty \big(\K^{ \delta_1, \delta_1}_{\alpha',\alpha, \sigma}(s) - \K_{\alpha',\alpha, \sigma}(s)\big) f(s) ds} + \abs{\int_{-\infty}^\infty \big(\K^{ \delta_1, \delta_2}_{\alpha',\alpha, \sigma}(s) - \K_{\alpha',\alpha, \sigma}(s)\big) f(s) ds}, \nonumber\\
    & \qquad \numleq{2} 2\bigg(\sum_{t \in t^*} \text{TV}(\K_{\alpha',\alpha, \sigma;c}; [t - 2\delta_1, t + 2\delta_1]) + \text{TV}(\K_{\alpha',\alpha, \sigma;c}; [t - \delta_1 - \delta_2, t + \delta_1 +\delta_2])\bigg)\norm{f}_\infty + \nonumber \\
    &\quad\quad\quad\quad\bigg(2\delta_1 \text{TV}(\K_{\alpha',\alpha, \sigma};[-\tau-2\delta_1, \tau+2\delta_1]+(\delta_1 + \delta_2) \text{TV}(\K_{\alpha',\alpha, \sigma};[-\tau-\delta_1-\delta_2, \tau+\delta_1+\delta_2]\bigg)\norm{f'}_\infty, \nonumber\\
    & \qquad \numleq{3} 4 \bigg(\sum_{t\in t^*} u_c(t-2\delta,t+2\delta)\norm{f}_\infty + \delta\ u(-\tau-2\delta,\tau+2\delta)\bigg) \norm{f'}_\infty,
    \end{align*}
    where in (1) we have used the triangle inequality, in (2) we have used lemma \ref{lemma:kernel_appx} and in (3) we have re-expressed our results in terms of the upper bounding on the total variation Assumption \ref{assump:local_int} and $\delta = \max(\delta_1, \delta_2)$. A similar analysis would yield an identical upper bound on the kernel $\Delta^{\nu, \sigma; \nu', \sigma'}_{-; \alpha}$. Thus, we have that 
    \begin{align}\label{eq:smallness_delta}
     \Delta^{\sigma; \sigma'}_{\pm; \alpha,\alpha'} \text{ is } \bigg(4\sum_{t\in t^*} u_c(t-2\delta,t+2\delta), 4\delta\ u(-\tau-2\delta, \tau+2\delta)\bigg)\text{-close on the test-function space }\text{PC}_{t^*}^1(\mathbb{R}).
    \end{align}

\underline{$\smallabs{\Gamma^{\sigma; \sigma'}_{j, j', \alpha,\alpha'}}, \smallabs{\Gamma^{\sigma; \sigma'}_{\pm, j, \alpha,\alpha'}}$ as $\delta_1, \delta_2 \to 0$}. We will now use the smallness of $\Delta_{\pm}^{ \sigma; \sigma'}$ together with lemma \ref{lemma:smallness_double_integral_kernel} to show that $\smallabs{\Gamma^{\sigma; \sigma'}_{j, j'; \alpha,\alpha'}}, \smallabs{\Gamma^{\sigma; \sigma'}_{\pm, j, \alpha,\alpha'}} \to 0$ as $\delta_1, \delta_2 \to 0$. In order to apply lemma \ref{lemma:smallness_double_integral_kernel}, we need to establish that $f^{ \sigma; \sigma'}_{j, j', \alpha,\alpha'}, f^{ \sigma;  \sigma'}_{\pm, j, \alpha,\alpha'}$ are bounded functions and have bounded derivatives with respect to either of its coordinates. We will show this in detail for $f_{j, j', \alpha,\alpha'}^{\sigma; \sigma'}$ (with $j' < j$) --- a similar calculation will hold for $f_{j, j', \alpha,\alpha'}^{\sigma; \sigma'}$ (with $j' > j$) and $f_{\pm, j, \alpha,\alpha'}^{\sigma; \sigma'}$.

We begin by noting that, for any $s, s'$, 
\[
\smallabs{f^{\sigma; \sigma'}_{j, j', \alpha,\alpha'}(s', s)} \leq \smallnorm{R_{\alpha', \sigma'}(s')} \smallnorm{R_{\alpha, \sigma}(s)} \leq 1,
\]
i.e. $f_{j, j', \alpha,\alpha'}^{\sigma; \sigma'}(s', s)$ is bounded above by a constant uniform in $\delta_1, \delta_2$. The same holds for the derivative of $f^{\sigma; \sigma'}_{j, j', \alpha,\alpha'}$ with respect to either of its arguments. For instance, consider the derivative with respect to the first argument:
\begin{align}\label{eq:od_der_first_coord}
&\abs{\partial_1 f^{\sigma; \sigma'}_{j, j', \alpha,\alpha'}(s', s)} \nonumber\\
&\qquad\leq |\vecbra{I_E, O_S} \mathcal{V}^{\delta_1}_{1, j' - 1}\mathcal{U}^{\delta_1}_{\mathcal{A}_{j'}}(t_{j'}, s')\big(R_{\alpha’, \sigma'})'(s')\mathcal{U}^{\delta_1}_{\mathcal{A}_{j'}}(s', s_{j'})\Omega_{j'}\mathcal{V}^{\delta_1}_{j' + 1, j - 1} \times\nonumber\\
&\qquad\qquad \qquad \qquad \qquad \qquad \qquad \qquad \qquad \qquad \mathcal{U}^{\delta_1}_{\mathcal{A}_j}(t_j, s)R_{\alpha, \sigma}(s) \mathcal{U}^{\delta_2}_{\mathcal{A}_j}(s, s_j) \mathcal{V}^{\delta_2}_{j + 1, n}\vecket{\rho_E, \sigma_0}| +\nonumber \\
&\qquad \quad \Theta_{\mathcal{A}_{S_{\alpha'}}}(\alpha'')|\vecbra{I_E, O_S}\mathcal{V}^{\delta_1}_{1, j' - 1}\mathcal{U}_{\mathcal{A}_{j'}}^{\delta_1}(t_{j'}, s')[\mathcal{C}_{h_{\alpha''}(s')}, R_{\alpha', \sigma'}(s')] {\mathcal{U}}_{\mathcal{A}_{j'}}^{\delta_1}(s', s_{j'})\Omega_{j'} \mathcal{V}^{\delta_1}_{j' + 1, j - 1}\times\nonumber\\
&\qquad\qquad \qquad \qquad \qquad \qquad \qquad \qquad \qquad \qquad\mathcal{U}^{\delta_1}_{\mathcal{A}_j}(t_j, s){R}_{\alpha, \sigma}(s)\mathcal{U}^{\delta_2}_{\mathcal{A}_j}(s, s_j) \mathcal{V}^{\delta_2}_{j  +1, n}\vecket{\rho_E, \sigma_s}| + \nonumber \\
&\qquad \quad \Theta_{\mathcal{A}_{S_{\alpha'}}}(\alpha'') |\vecbra{I_E, O_S} \mathcal{V}^{\delta_1}_{1, j' - 1} {\mathcal{U}}_{\mathcal{A}_j}^{\delta_1}(t_{j'}, s')B_{\alpha'',\sigma''}^{\delta_1}(s')[ R_{\alpha'', \sigma''}(s'), R_{\alpha', \sigma'}(s')] {\mathcal{U}}_{\mathcal{A}_{j'}}^{\delta_1}(s', s_{j'}) \Omega_{j'}\mathcal{V}^{\delta_1}_{j' + 1, j - 1}\times \nonumber\\
&\qquad\qquad \qquad \qquad \qquad \qquad \qquad \qquad \qquad \qquad\mathcal{U}^{\delta_1}_{\mathcal{A}_j}(t_j, s){R}_{\alpha, \sigma}(s)\mathcal{U}^{\delta_1}_{\mathcal{A}_j}(s, s_j) \mathcal{V}^{\delta_2}_{j + 1, n}\vecket{\rho_E, \sigma_S}|, \nonumber \\
&\qquad\numleq{1}  |\vecbra{I_E, O_S} \mathcal{V}^{\delta_1}_{1, j' - 1}\mathcal{U}^{\delta_1}_{\mathcal{A}_{j'}}(t_{j'}, s')\big(R_{\alpha', \sigma'})'(s')\mathcal{U}^{\delta_1}_{\mathcal{A}_{j'}}(s', s_{j'})\Omega_{j'}\mathcal{V}^{\delta_1}_{j' + 1, j - 1} \times\nonumber\\
&\qquad\qquad \qquad \qquad \qquad \qquad \qquad \qquad \qquad \qquad \mathcal{U}^{\delta_1}_{\mathcal{A}_j}(t_j, s)R_{\alpha, \sigma}(s) \mathcal{U}^{\delta_2}_{\mathcal{A}_j}(s, s_j) \mathcal{V}^{\delta_2}_{j + 1, n}\vecket{\rho_E, \sigma_0}| +\nonumber \\
&\qquad\quad \Theta_{\mathcal{A}_{S_{\alpha'}}}(\alpha'')|\vecbra{I_E, O_S}\mathcal{V}^{\delta_1}_{1, j' - 1}{\mathcal{U}}_{\mathcal{A}_{j'}}^{\delta_1}(t_{j'}, s')[\mathcal{C}_{h_{\alpha''}(s')}, R_{\alpha', \sigma'}(s')] {\mathcal{U}}_{\mathcal{A}_{j'}}^{\delta_1}(s', s_{j'})\Omega_{j'} \mathcal{V}^{\delta_1}_{j' + 1, j - 1}\times\nonumber\\
&\qquad\qquad \qquad \qquad \qquad \qquad \qquad \qquad \qquad \qquad\mathcal{U}^{\delta_1}_{\mathcal{A}_j}(t_j, s){R}_{\alpha, \sigma}(s)\mathcal{U}^{\delta_2}_{\mathcal{A}_j}(s, s_j) \mathcal{V}^{\delta_2}_{j  +1, n}\vecket{\rho_E, \sigma_s}| + \nonumber \\
&\qquad \quad \Theta_{\mathcal{A}_{S_{\alpha'}}}(\alpha'') |\wickleft\big(\vecbra{I_E, O_S} \mathcal{V}^{\delta_1}_{1, j' - 1} {\mathcal{U}}_{\mathcal{A}_j}^{\delta_1}(t_{j'}, s'); B_{\alpha'',\sigma''}^{\delta_1}(s')\big)[ R_{\alpha'', \sigma''}(s'), R_{\alpha', \sigma'}(s')] {\mathcal{U}}_{\mathcal{A}_{j'}}^{\delta_1}(s', s_{j'}) \Omega_{j'}\mathcal{V}^{\delta_1}_{j' + 1, j - 1}\times \nonumber\\
&\qquad\qquad \qquad \qquad \qquad \qquad \qquad \qquad \qquad \qquad\mathcal{U}^{\delta_1}_{\mathcal{A}_j}(t_j, s){R}_{\alpha, \sigma}(s)\mathcal{U}^{\delta_2}_{\mathcal{A}_j}(s, s_j) \mathcal{V}^{\delta_2}_{j + 1, n}\vecket{\rho_E, \sigma_S}| + \nonumber\\
&\qquad \quad \Theta_{\mathcal{A}_{S_{\alpha'}}}(\alpha'') |\vecbra{I_E, O_S} \mathcal{V}^{\delta_1}_{1, j' - 1} {\mathcal{U}}_{\mathcal{A}_j}^{\delta_1}(t_{j'}, s') [ R_{\alpha'', \sigma''}(s'), R_{\alpha', \sigma'}(s')] \times \nonumber\\
&\qquad\qquad \qquad \qquad \qquad   \qquad\wickright\big(B_{\alpha'',\sigma''}^{\delta_1}(s'); {\mathcal{U}}_{\mathcal{A}_{j'}}^{\delta_1}(s', s_{j'}) \Omega_{j'}\mathcal{V}^{\delta_1}_{j' + 1, j - 1}\mathcal{U}^{\delta_1}_{\mathcal{A}_j}(t_j, s){R}_{\alpha, \sigma}(s)\mathcal{U}^{\delta_2}_{\mathcal{A}_j}(s, s_j) \mathcal{V}^{\delta_2}_{j + 1, n}\vecket{\rho_E, \sigma_S}\big)|, \nonumber \\
&\qquad \leq \smallnorm{(R_{\alpha', \sigma'}(s'))'} + 4 \abs{\mathcal{A}_{S_{\alpha'}}} + 4\Theta_{\mathcal{A}_{S_{\alpha'}}}(\alpha'')\smallnorm{\wickleft\big(\vecbra{I_E, O_S} \mathcal{V}^{\delta_1}_{1, j' - 1} {\mathcal{U}}_{\mathcal{A}_j}^{\delta_1}(t_{j'}, s'); B_{\alpha'',\sigma''}^{\delta_1}(s')\big)} + \nonumber\\
&\qquad \qquad \qquad  4 \Theta_{\mathcal{A}_{S_{\alpha'}}}(\alpha'')\norm{\wickright\big(B_{\alpha'',\sigma''}^{\delta_1}(s'); {\mathcal{U}}_{\mathcal{A}_{j'}}^{\delta_1}(s', s_{j'}) \Omega_{j'}\mathcal{V}^{\delta_1}_{j' + 1, j - 1}\mathcal{U}^{\delta_1}_{\mathcal{A}_j}(t_j, s){R}_{\alpha, \sigma}(s)\mathcal{U}^{\delta_2}_{\mathcal{A}_j}(s, s_j) \mathcal{V}^{\delta_2}_{j + 1, n}\vecket{\rho_E, \sigma_S}\big)}_1,
\end{align}
where we obtained (1) by applying the Wick's theorem. Now, note from Eq.~\ref{eq:right_W_1} that
\begin{align*}
&\norm{\wickright\big(B_{\alpha'',\sigma''}^{\delta_1}(s'); {\mathcal{U}}_{\mathcal{A}_{j'}}^{\delta_1}(s', s_{j'}) \Omega_{j'}\mathcal{V}^{\delta_1}_{j' + 1, j - 1}\mathcal{U}_{\mathcal{A}_j}^{\delta_1}(t_j,s)R_{\alpha,\sigma}(s)\mathcal{U}_{\mathcal{A}_j}^{\delta_2}(s,s_j) \mathcal{V}^{\delta_2}_{j + 1, n}\vecket{\rho_E, \sigma_S}\big)}_1  \nonumber\\
&\qquad\leq \sum_{\sigma''',\alpha'''} \bigg(\abs{\int_{s_{j'}}^{s'}\abs{\vecbra{I_E}B_{\alpha'',\sigma''}^{\delta_1}(s')B_{\alpha''',\sigma'''}^{\delta_1}(s''')\vecket{\rho_E}}ds'''}  + \sum_{l = j - 1}^{j' + 1}\abs{\int_{s_l}^{t_l} \abs{\vecbra{I_E}B_{\alpha'',\sigma''}^{\delta_1}(s')B_{\alpha''',\sigma'''}^{\delta_1}(s''') \vecket{\rho_E}}ds'''} + \nonumber\\
&\qquad\qquad \qquad  \abs{\int_{s}^{t_j} \abs{\vecbra{I_E}B_{\alpha'',\sigma''}^{\delta_1}(s')B_{\alpha''',\sigma'''}^{\delta_1}(s''') \vecket{\rho_E}}ds'''}+\abs{\int_{s_j}^{s} \abs{\vecbra{I_E}B_{\alpha'',\sigma''}^{\delta_1}(s')B_{\alpha''',\sigma'''}^{\delta_2}(s''') \vecket{\rho_E}}ds'''} + \nonumber\\
&\qquad\qquad \qquad   \sum_{l = j + 1}^{n}\abs{\int_{s_l}^{t_l} \abs{\vecbra{I_E}B_{\alpha'',\sigma''}^{\delta_1}(s')B_{\alpha''',\sigma'''}^{\delta_2}(s''') \vecket{\rho_E}}ds'''}\bigg),\nonumber \\
&\qquad\numleq{1} 4\big(n - j' + 6\big)L\abs{\mathcal{A}}u(-\tau-2\delta,\tau+2\delta),
\end{align*}
where, in (1), we have used Eq.~\ref{eq:total_variation_mollified} to conclude that for any $\delta, \delta' \in \{\delta_1, \delta_2\}$ and interval $[a, b]\subseteq[-\tau,\tau]$,
\[
\abs{\int_a^b \vecbra{I_E}B_{\alpha,\sigma}^{\delta}(s)B_{\alpha',\sigma'}^{\delta'}(s') \vecket{\rho_E} ds'} \leq \text{TV}(\K^{ \delta, \delta'}_{\alpha,\alpha', \sigma'};[a-2\delta,b+2\delta]) \leq u(-\tau-2\delta,\tau+2\delta).
\]
Similarly, we obtain that
\begin{align*}
    &\norm{\wickleft(\vecbra{I_E, O_S}\mathcal{V}_{1, j' - 1}^{\delta_1}{\mathcal{U}}^{\delta_1}_{\mathcal{A}_j}(t_{j'}, s'); B_{\alpha'',\sigma''}^{\delta_1}(s')\big)} \nonumber\\
    &\qquad \leq \sum_{\alpha''', \sigma'''} \bigg(\abs{\int_{s'}^{t_{j'}}\abs{\vecbra{I_E} B_{\alpha''',\sigma'''}^{\delta_1}(s''')B_{\alpha'',\sigma''}^{\delta_1}(s''')\vecket{\rho_E}}ds''' } + \sum_{l = 1}^{j' - 1}\abs{\int_{s_l}^{t_l}\abs{\vecbra{I_E}B_{\alpha''',\sigma'''}^{\delta_1}(s''')B_{\alpha'',\sigma''}^{\delta_1}(s')\vecket{\rho_E}}ds'''}\bigg), \nonumber\\
    &\qquad \leq 4j'L\abs{\mathcal{A}} u(-\tau-2\delta, \tau+2\delta).
\end{align*}
Combining these bounds with Eq.~\ref{eq:od_der_first_coord}, we obtain that
\[
\abs{\partial_1 f_{j, j', \alpha,\alpha'}^{ \sigma;  \sigma'}(s', s)} \leq 4\smallabs{\mathcal{A}_{S_{\alpha'}}}+ 16(n - 6)L \smallabs{\mathcal{A}_{S_{\alpha'}}}\smallabs{\mathcal{A}}\ u(-\tau-1,\tau+1) + \sup_{s' \in\bigcup_{i = 1}^n[s_i, t_i]}\smallnorm{(R_{\alpha', \sigma'}(s'))'},
\]
which is bounded above uniformly in $\delta_1, \delta_2$ and $s, s'$ if $\delta<1/2$. A similar calculation can be performed to also check that $\smallabs{\partial_2 f_{j, j', \alpha,\alpha'}^{\sigma,  \sigma'}(s', s)}$ is also bounded above uniformly in $\delta_1, \delta_2$ and $s, s'$. Since all $\smallabs{f_{j, j', \alpha,\alpha'}^{\sigma,  \sigma'}(s', s)}, \smallabs{\partial_1 f_{j, j', \alpha,\alpha'}^{\sigma, \sigma'}(s', s)}$ and $\smallabs{\partial_2 f_{j, j', \alpha,\alpha'}^{\sigma, \sigma'}(s', s)}$ are uniformly bounded, it then follows from the smallness of $\Delta_{+, \alpha,\alpha'}^{\sigma; \sigma'}$ and lemma \ref{lemma:smallness_double_integral_kernel}b that $\Gamma^{\sigma; \sigma'}_{j,j', \alpha,\alpha'} \to 0$ as $\delta_1, \delta_2 \to 0$ for $j > j'$. A similar argument can be repeated for $\Gamma_{j, j', \alpha,\alpha'}^{\sigma; \sigma'}$ (when $j < j'$) as well as for $\Gamma^{\sigma;  \sigma'}_{\pm, j, \alpha,\alpha'}$ (where we would use lemma \ref{lemma:smallness_double_integral_kernel}a instead of lemma \ref{lemma:smallness_double_integral_kernel}b). Together with Eq.~\ref{eq:telescoping_sum}, this implies that
\[
\abs{\vecbra{I_E, O_S}\bigg(\prod_{i = 1}^n {\mathcal{U}}^{\delta_1}_{\mathcal{A}_i}(t_i, s_i) \bigg) \vecket{\rho_E, \sigma_S} - \vecbra{I_E, O_S}\bigg(\prod_{i = 1}^n {\mathcal{U}}^{\delta_2}_{\mathcal{A}_i}(t_i, s_i) \bigg) \vecket{\rho_E, \sigma_S}}\to 0 \text{ as }\delta_1, \delta_2 \to 0,
\]
which, by the completeness of complex numbers, implies the lemma statement.
\end{proof}
\section{Addressing the integrability and lower Lipschitz boundedness of the commutator bound}\label{appendix:math_considerations}
We now provide rigorous proofs justifying the discussion in \emph{consideration }1 and \emph{consideration }2 of section \ref{sec:defining_commutator_bound}. In particular, lemma \ref{lemma:int_gamma} proves that $\gamma^{X, Y}_{\bm\zeta, \mathcal{B}}$ is integrable in the sense discussed in \emph{consideration }1. Finally, we also provide a proof of lemma \ref{lemma:upper_bound_lemma_stat} from section \ref{sec:defining_commutator_bound} which addresses consideration 2.

To prove lemma \ref{lemma:int_gamma}, we first prove a useful technical lemma.
\begin{lemma}\label{lemma:measurability_lemma}
    Suppose $A_x$ is a family of possibly uncountable sets parameterized by $x \in [a, b]$ and $g:A \times [a, b] \to [0, 1]$, where $A = \cup_{x\in[a, b]} A_x$, is a function which satisfy the following properties:
    \begin{enumerate}
        \item[1.] If $x \leq x'$, $A_x \subseteq A_{x'}$.
        \item[2.] $\forall \phi \in A$, $g(\phi, \cdot)$ is an $\ell-$Lipschitz continuous function, with $\ell$ is uniform in $\phi$ i.e.~$\forall\ \phi \in A, \forall \ x, x' \in[a, b]$, $\abs{g(\phi, x) - g(\phi, x')} \leq \ell \abs{x - x'}$
    \end{enumerate}
    Then the function $G: [a, b] \to \mathbb{R}$ defined by
    \[
    G(x) = \sup_{\phi \in A_x} g(\phi, x),
    \]
    is Lesbesgue integrable.
\end{lemma}
\begin{proof}
    We note that $G$, by construction, is a bounded function over a compact interval $[a, b]$. Consequently, as long as $G$ is measurable, it is also integrable. We will use two facts from measure theory to prove the measurability of $G$ ---
    \begin{enumerate}[leftmargin=1.5cm]
        \item[Fact 1:] Suppose $h:[a, b]^2 \to \mathbb{R}$ is a measurable function, then the function $H: [a, b] \to \mathbb{R}$ defined by $H(x) = h(x, x)$ is also a measurable function.
        \item[Fact 2:] Suppose $\{f_n:\mathcal{S} \to \mathbb{R}\}_{n\in \mathbb{N}}$, where $\mathcal{S}$ together with an appropriate $\sigma-$algebra is a measurable space, is a sequence of measurable functions which converges pointwise to $f:\mathcal{S}\to \mathbb{R}$ i.e. $f(x) = \lim_{n\to \infty} f_n(x)$, then $f$ is also a measurable function.
    \end{enumerate}
    To prove this lemma, we begin by defining $\tilde{G}:[a, b]^2 \to [0, 1]$ via
    \[
    \tilde{G}(x, y) = \sup_{\phi \in A_x} g(\phi, y),
    \]
    then, clearly, $G(x) = \tilde{G}(x, x)$. Furthermore, $\tilde{G}(x, y)$ satisfies two important properties --- \emph{first}, it is a non-decreasing function of $x$ for any fixed $y \in [a, b]$. To see this, simply note that for $x \leq x'$, since $A_x \subseteq A_{x'}$
    \[
    \tilde{G}(x, y) = \sup_{\phi \in A_x} g(\phi, y) \leq \sup_{\phi \in A_{x'}} g(\phi, y) = \tilde{G}(x', y).
    \]
    \emph{Second}, $\tilde{G}(x, y)$ is continuous in $y$ for any fixed $x \in [a, b]$. To see this, note that for any $y, y'$
    \[
    \tilde{G}(x, y) = \sup_{\phi \in A_x} g(\phi, y) \leq \sup_{\phi \in A_x} \bigg[g(\phi, y') + \ell\abs{y - y'}\bigg] \leq \tilde{G}(x, y') + \ell \abs{y - y'}.
    \]
    Since this holds for any $y, y'$, it also implies that $\tilde{G}(x, y') \leq \tilde{G}(x, y) + \ell \abs{y - y'}$ and thus $\abs{\tilde{G}(x, y') - \tilde{G}(x, y)}\leq \ell \abs{y - y'}$, which implies continuity of $\tilde{G}(x, \cdot)$.

    Next, note that since $\tilde{G}(x, y)$ is continuous in $y$, we can approximate by a sequence of functions $\tilde{G}_n(x, y)$ where
    \[
    \tilde{G}_n(x, y) = \sum_{j = -\infty}^\infty \Theta_{[(j - 1)/n, j/n)}(y)\tilde{G}\bigg(x, \frac{j}{n}\bigg),
    \]
    i.e. for every $x$, we perform a step-like approximation of $\tilde{G}(x, \cdot)$. It follows from the continuity of $\tilde{G}(x, \cdot)$ that
    \[
    \lim_{n \to \infty} \tilde{G}_n(x, y) = G(x, y).
    \]
    Now, we establish that $\tilde{G}_n(x, y)$ is measurable. Recall that this is equivalent to showing that $\tilde{G}_n^{-1}((-\infty, a])$ is measurable --- $\tilde{G}_n^{-1}((-\infty, a])$ can be explicitly calculated to be
    \[
    \tilde{G}_n^{-1}((-\infty, a]) = \bigcup_{j = -\infty}^\infty \mathcal{X}_{n, j}\times \bigg[\frac{j - 1}{n}, \frac{j}{n}\bigg) \text{ where }\mathcal{X}_{n, j} = \bigg\{x : G\bigg(x, \frac{j}{n}\bigg) \leq a\bigg\}.
    \]
    We note that $\mathcal{X}_{n, j} \subseteq [a, b]$ is measurable --- to see this, note that since $G(\cdot, j/n)$ is a non-decreasing function, if $x \in \mathcal{X}_{n, j}$ and $x' \leq x$, then $x' \in \mathcal{X}_{n, j}$ and hence $\mathcal{X}_{n, j}$ can only be $(-\infty, b]$ or $(-\infty, b)$ for some b or $(-\infty, \infty)$. From the measurability of $\mathcal{X}_{n, j}$, it immediately follows that $\tilde{G}_n^{-1}((-\infty, a])$ is measurable implying that $\tilde{G}_n$ is a measurable function.

    Finally, we then use Fact 2 to conclude that $\tilde{G}(x, y)$, which is a point-wise limit of $\tilde{G}_n(x, y)$, is also measurable and use Fact 1 to finally conclude that $G(x) = \tilde{G}(x, x)$ is measurable.
\end{proof}

\begin{lemma}\label{lemma:int_gamma}
    Given $s_0 < s_1$ and $\bm{\zeta}=\{\zeta_\alpha\}$ with $\zeta_\alpha:\mathbb{R}\to [0, \infty)$ a locally bounded continuous function,
    \begin{enumerate}
        \item[(a)] Given $s_0 < s_1$, the function $\gamma^{X, Y}_{\bm\zeta + \beta_0\bm{\mu}_{[s_0, \cdot]}^R, \mathcal{B}}(s_1, \cdot): [s_0, s_1] \to \mathbb{R}$ is Lesbesgue integrable.
        \item[(b)] Given a smooth and integrable function $k:\mathbb{R} \to \mathbb{C}$, the function $h:\mathbb{R} \to \mathbb{C}$ defined via
        \[
        h(s) = \int_{s_0}^{s_1} k(s - s')\gamma^{X, Y}_{\bm\zeta + \beta_0 \bm{\mu}_{[s_0, s']}^R, \mathcal{B}}(s_1, s')ds',
        \]
        is integrable.
    \end{enumerate}
\end{lemma}
\begin{proof}
    (a) For notational convenience, we begin by defining $\Gamma(\cdot) = \gamma^{X, Y}_{\bm\zeta + \beta_0\bm{\mu}_{[s_0, \cdot]}^R, \mathcal{B}}(s_1, \cdot)$ and rewriting it as
    \[
    \Gamma(s) = \sup_{\phi \in A_s} g(\phi, s) \text{ where }A_s = \mathcal{S}_{\bm\zeta + \beta_0 \bm{\mu}^{\textnormal{R}}_{[s_0, s]}}(\rho_E) \text{ and } g(\phi, s) = \sup_{\norm{O_X}, \norm{\mathcal{L}_Y}_\diamond \leq 1} \abs{\vecbra{O_X, I_E}\mathcal{U}_\mathcal{B}(s_1, s) \mathcal{L}_Y \vecket{\phi}}.
    \]
    We now show that the sets $A_s$ and the function $g$ satisfy the properties required for an application of lemma \ref{lemma:measurability_lemma}. To show property 1 of lemma \ref{lemma:measurability_lemma}, we first note that since for $s < s'$ and $\forall \tau \in \mathbb{R}$,
    \[
    \zeta_\alpha(\tau) + \beta_0 {\mu}^{\textnormal{R}}_{\alpha,[s_0, s]}(\tau) = \zeta_\alpha(\tau) + \beta_0{\mu}^{\textnormal{R}}_{\alpha}(\tau -s, \tau-s_0) \numleq{1} \zeta_\alpha(\tau) + \beta_0 \mu^{\textnormal{R}}_{\alpha}(\tau - s', \tau - s_0) = \zeta_\alpha(\tau) + \beta_0 \mu^{\textnormal{R}}_{\alpha,[s_0, s']}(\tau),\quad\forall\alpha,
    \]
    where in (1) we have used the definition of $\mu^{\textnormal{R}}_\alpha$ in Eq. \ref{eq:mu}. Therefore, from the definition of $S_{\bm\zeta}(\rho_E)$ [Definition~\ref{def:S_rho_zeta}], it follows that $A_{s} \subseteq A_{s'} $. To show property 2 of lemma \ref{lemma:measurability_lemma}, suppose $ \phi \in \cup_{s \in [s_0, s_1]} A_s$---we note that for $s, s' \in [s_0, s_1]$ with $s' < s$, we have that
    \begin{align}\label{eq:lipschitz_cont}
        &\abs{\vecbra{O_X, I_E}\mathcal{U}_\mathcal{B}(s_1, s) \mathcal{L}_Y \vecket{\phi} - \vecbra{O_X, I_E}\mathcal{U}_\mathcal{B}(s_1, s') \mathcal{L}_Y \vecket{\phi}} \nonumber\\
        &\qquad = \abs{\int_{s'}^s\vecbra{O_X, I_E} \mathcal{U}_\mathcal{B}(s_1, \tau) \mathcal{H}_{\mathcal{B}}(\tau) \mathcal{L}_Y\vecket{\phi} d\tau}, \nonumber\\
        &\qquad \leq \sum_{\alpha}\bigg( \Theta_{\mathcal{B}}(\alpha) \abs{\int_{s'}^s \vecbra{O_X, I_E} \mathcal{U}_\mathcal{B}(s_1, \tau) \mathcal{C}_{h_\alpha(\tau)} \mathcal{L}_Y\vecket{\phi} d\tau } + \nonumber\\
        &\qquad \qquad \qquad \qquad \qquad \qquad \Theta_{\mathcal{B}}(\alpha) \abs{\int_{s'}^s \vecbra{O_X, I_E}\mathcal{U}_\mathcal{B}(s_1, \tau)B_{\alpha,\sigma}(\tau)R_{\alpha, \sigma}(\tau)\mathcal{L}_Y\vecket{\phi} d\tau}\bigg), \nonumber \\
        &\qquad \numleq{1} 2\abs{\mathcal{B}} \norm{O_X}\smallnorm{\mathcal{L}_Y}_\diamond  \abs{s - s'} +\sum_{\alpha,\sigma}\Theta_\mathcal{B}(\alpha)\bigg[ \int_{s'}^s \smallnorm{\wickleft\big(\vecbra{O_X, I_E}\mathcal{U}_\mathcal{B}(s_1, \tau); B_{\alpha,\sigma}(\tau)\big)}\smallnorm{R_{\alpha, \sigma}(\tau)}_\diamond \smallnorm{\mathcal{L}_Y \vecket{\phi}}_1 d\tau + \nonumber\\
        &\qquad \qquad \qquad \qquad \qquad \qquad \int_{s'}^s \smallnorm{\vecbra{O_X, I_E}\mathcal{U}_\mathcal{B}(s_1, \tau)} \smallnorm{R_{\alpha, \sigma}(\tau)}_\diamond \smallnorm{\wickright\big(B_{\alpha,\sigma}(\tau); \mathcal{L}_Y \vecket{\phi}\big)}_1d\tau\bigg],
    \end{align}
    where we have used the Wick's theorem in (1). Next, we note that $\vecbra{O_X, I_E} \mathcal{U}_\mathcal{B}(s_1, \tau) \in \norm{O_X} \mathcal{Q}_{\bm{\mu}^{\textnormal{L}}_{[\tau, s_1]}}(\rho_E)$---from lemma \ref{lemma:closure_left_wick}, it then follows that $\forall \tau \in [s', s]$:
    \[
    \smallnorm{\wickleft\big(\vecbra{O_X, I_E}\mathcal{U}_\mathcal{B}(s_1, \tau); B_{\alpha,\sigma}(\tau)\big)} \leq 2 \norm{O_X} \mu^{\textnormal{L}}_{\alpha,\sigma,[\tau, s_1]}(\tau) = 2 \norm{O_X} \mu^{\textnormal{L}}_{\alpha,\sigma}(0, s_1 - \tau) \leq 2\mu^{\textnormal{L}}_{\alpha,\sigma}(s_0-s_1, s_1 -s_0).
    \]
    Similarly, suppose that $\phi \in A_{\tilde{s}}$ for some $\tilde{s}$---then, note that $\mathcal{L}_Y \vecket{\phi} \in \norm{\mathcal{L}_Y}_\diamond \mathcal{S}_{\bm\zeta + \beta_0\bm{\mu}^{\textnormal{R}}_{[s_0, \tilde{s}]}}(\rho_E)$. It then follows from lemma \ref{lemma:closure_right_wick}c that $\forall \tau \in [s', s]$:
    \begin{align*}
    \smallnorm{\wickright\big(B_{\alpha,\sigma}(\tau);\mathcal{L}_Y\vecket{\phi}\big)}_1 &\leq \norm{\mathcal{L}_Y}_\diamond (\zeta_\alpha(\tau)+ \beta_0\mu^{\textnormal{R}}_{[s_0, \tilde{s}]}(\tau) \nonumber \\
    &\leq \norm{\mathcal{L}_Y}_\diamond (\zeta_\alpha(\tau) + \beta_0 \mu^{\textnormal{R}}_\alpha(\tau -\tilde{s}, \tau - s_0) ), \nonumber \\
    &\leq \norm{\mathcal{L}_Y}_\diamond \left(\sup_{\tau\in[s_0,s_1]}\zeta_\alpha(\tau) +\beta_0\mu^{\textnormal{R}}_\alpha(s_0-s_1,s_1-s_0)\right).
    \end{align*}
    From these estimates and Eq.~\eqref{eq:lipschitz_cont}, it follows that
    \[
    \abs{\vecbra{O_X, I_E}\mathcal{U}_\mathcal{B}(s_1, s) \mathcal{L}_Y \vecket{\phi} - \vecbra{O_X, I_E}\mathcal{U}_\mathcal{B}(s_1, s') \mathcal{L}_Y \vecket{\phi}} \leq \ell \abs{s - s'},
    \]
    where
    \[
    \ell = \abs{\mathcal{B}}\bigg(2 + 4L\norm{\bm\mu^{\textnormal{L}}(s_0-s_1,s_1-s_0)}_\infty+2L\beta_0\norm{\bm\mu^{\textnormal{R}}(s_0-s_1,s_1-s_0)}_\infty+2L \sup_{\alpha,\tau \in [s_0, s_1]}\zeta_\alpha(\tau)\bigg).
    \]
    Consider now
    \begin{align*}
    g(\phi, s) &= \sup_{\norm{O_X}\leq 1, \norm{\mathcal{L}_Y}_\diamond \leq 1} \abs{\vecbra{O_X, I_E}\mathcal{U}_\mathcal{B}(s_1, s) \mathcal{L}_Y \vecket{\phi}}, \nonumber\\
    &\numleq{1} \sup_{\norm{O_X}\leq 1, \norm{\mathcal{L}_Y}_\diamond \leq 1}\bigg(\abs{\vecbra{O_X, I_E}\mathcal{U}_\mathcal{B}(s_1, s') \mathcal{L}_Y \vecket{\phi}} + \ell  \abs{s - s'}\bigg), \nonumber\\
    &\leq g(\phi, s') + \ell \abs{s - s'},
\end{align*}
where in (1) we have used Eq.~\ref{eq:lipschitz_cont}. A similar analysis would also yield that $g(\phi, s') \leq g(\phi, s) + \ell \abs{s - s'}$ and consequently we obtain that $\abs{g(\phi, s') - g(\phi,s )} \leq \ell \abs{s - s'}$. This implies that $g(\phi, \cdot)$ is a continuous function. Finally, since the sets $A_s$ and the function $g$ satisfy the properties of the lemma \ref{lemma:measurability_lemma}, it then follows that $\Gamma$ is also an integrable function.

(b) The proof of part a justifies that $\gamma_{\bm\zeta + \beta_0 \bm{\mu}^{\textnormal{R}}_{[s_0, \cdot]}, \mathcal{B}}(s_1, \cdot) : [s_0, s_1] \to [0, 1]$ is an integrable function and thus allows us to sensibly define and use integrals of the form
\[
\int_{s_0'}^{s_1'} g(s)\gamma^{X, Y}_{\bm\zeta + \beta_0 \bm{\mu}^{\textnormal{R}}_{[s_0, s]}, \mathcal{B}}(s_1, s) ds,
\]
where $s_0 \leq s_0' \leq s_1' \leq s_1$ and $g:[s_0, s_1] \to \mathbb{R}$ is an integrable and function. Consider now the function
\[
h(s) = \int_{s_0'}^{s_1'}k(s - s')\gamma^{X, Y}_{\bm\zeta + \beta_0\bm{\mu}^{\textnormal{R}}_{[s_0, s']}, \mathcal{B}}(s_1, s') ds',
\]
where $k$ is a smooth and integrable function. We note that $h(s)$ is well defined since $k(s - (\cdot))\gamma^{X, Y}_{\bm\zeta + \beta_0\bm{\mu}^{\textnormal{R}}_{[s_0, \cdot]}, \mathcal{B}}(s_1, \cdot)$ is integrable in the interval $[s_0', s_1']$. Furthermore, as a function of $s, s'$, $k(s - s')\gamma^{X, Y}_{\bm\zeta + \beta_0\bm{\mu}^{\textnormal{R}}_{[s_0, s']}, \mathcal{B}}(s_1, s')$ is absolutely integrable and hence by Fubini's theorem, $h$ is itself an integrable function.
\end{proof}

\noindent Next, we address \emph{consideration} 2 from section \ref{sec:defining_commutator_bound} by providing a proof of lemma \ref{lemma:upper_bound_lemma_stat}.
\begin{proof}[Proof of lemma \ref{lemma:upper_bound_lemma_stat}] The proof of this lemma follows by repeated application of Wick's theorem. We begin by noting that for $\phi \in \mathcal{S}_{\bm\zeta + \beta_0 \bm{\mu}^{\textnormal{R}}_{[s_0, s - \varepsilon]}}(\rho_E)$
\begin{align*}
&\abs{\vecbra{O_X, I_E} \mathcal{U}_\mathcal{B}(t, s - \varepsilon) \mathcal{L}_{Y}\vecket{\phi}} \nonumber\\
&\qquad = \abs{\vecbra{O_X, I_E} \mathcal{U}_\mathcal{B}(t, s) \mathcal{U}_\mathcal{B}(s, s-\varepsilon)\mathcal{L}_{Y} \vecket{{\phi}}}, \\
&\qquad =\abs{\vecbra{O_X, I_E} \mathcal{U}_\mathcal{B}(t, s) \mathcal{L}_{Y}\vecket{\phi}-i\int_{s -\varepsilon}^{s} \vecbra{O_X, I_E} \mathcal{U}_\mathcal{B}(t, s')\mathcal{H}(s') \mathcal{U}(s', s - \varepsilon) \mathcal{L}_{Y} \vecket{{\phi}} ds'}, \\
&\qquad \leq \abs{\vecbra{O_X, I_E} \mathcal{U}_\mathcal{B}(t, t') \mathcal{L}_{Y} \vecket{\phi}} +  \Theta_{\mathcal{B}}(\alpha)\int_{s - \varepsilon}^{s}\abs{\vecbra{O_X, I_E} \mathcal{U}_\mathcal{B}(t, s')B_{\alpha,\sigma}(s') R_{\alpha, \sigma}(s') \vecket{\tilde{\phi}(s')}}ds'+2\varepsilon\abs{\mathcal{B}},
\end{align*}
where, for $s' \in [s-\varepsilon, s]$, $\vecket{\tilde{\phi}(s')} = \mathcal{U}_\mathcal{B}(s', s-\varepsilon) \mathcal{L}_Y\vecket{\phi}$. First, note that since $\phi \in \mathcal{S}_{\bm\zeta + \beta_0 \bm{\mu}^{\textnormal{R}}_{[s_0, s-\varepsilon]}}(\rho_E) \subseteq \mathcal{S}_{\bm\zeta + \beta_0 \bm{\mu}^{\textnormal{R}}_{[s_0, s]}}(\rho_E)$, it follows that $\abs{\vecbra{O_X, I_E} \mathcal{U}_\mathcal{B}(t, s) \mathcal{L}_Y\vecket{\phi}} \leq \norm{O_X}\norm{\mathcal{L}_Y}_\diamond \gamma_{\bm\zeta + \beta_0 \bm{\mu}^{\textnormal{R}}_{[s_0, s]}}^{X, Y}(t, s)$ and therefore
\begin{align}\label{eq:bound_1_lemma_cont}
    &\abs{\vecbra{O_X, I_E} \mathcal{U}_\mathcal{B}(t, s - \varepsilon) \mathcal{L}_{Y}\vecket{\phi}}\nonumber\\
    &\qquad \qquad \leq \norm{O_X}\norm{\mathcal{L}_Y}_\diamond \gamma_{\bm\zeta + \beta_0 \bm{\mu}^{\textnormal{R}}_{[s_0, s]}}^{X, Y}(t, s) + \Theta_{\mathcal{B}}(\alpha)\int_{s - \varepsilon}^{s}\abs{\vecbra{O_X, I_E} \mathcal{U}_\mathcal{B}(t, s')B_{\alpha,\sigma}(s') R_{\alpha, \sigma}(s') \vecket{\tilde{\phi}(s')}}ds'+2\varepsilon\abs{\mathcal{B}}.
\end{align}
Next, note that since $\phi \in \mathcal{S}_{\bm\zeta + \beta_0 \bm{\mu}^{\textnormal{R}}_{[s_0, s - \varepsilon]}}(\rho_E)$,
\[
\tilde{\phi}(s') \in \mathcal{S}_{\bm\zeta + \beta_0 \bm{\mu}^{\textnormal{R}}_{[s_0, s - \varepsilon]} + \bm{\mu}^{\textnormal{R}}_{[s - \varepsilon, s']}}(\rho_E) \subseteq \mathcal{S}_{\bm\zeta + (\beta_0+1) \bm{\mu}^{\textnormal{R}}_{[s_0, s ']}}(\rho_E)  \subseteq \mathcal{S}_{\bm\zeta + (\beta_0+1) \bm{\mu}^{\textnormal{R}}_{[s_0, s]}}(\rho_E),
\]
and consequently,
\begin{align}\label{eq:bound_2_lemma_cont}
&\abs{\vecbra{O_X, I_E} \mathcal{U}_\mathcal{B}(t, s')B_{\alpha,\sigma}(s') R_{\alpha, \sigma}(s') \vecket{\tilde{\phi}(s')}}\nonumber \\
&\qquad \numleq{1} \abs{\overleftarrow{\text{W}}(\vecbra{O_X, I_E} \mathcal{U}_\mathcal{B}(t, s'), B_{\alpha,\sigma}(s')) R_{\alpha, \sigma}(s')\mathcal{L}_Y\vecket{\tilde{\phi}({s'})}} + \abs{\vecbra{O_X, I_E} \mathcal{U}_\mathcal{B}(t, s') R_{\alpha, \sigma} \mathcal{L}_{Y} \overrightarrow{\text{W}}\big(B_{\alpha,\sigma}(s'), \tilde{\phi}(s')\big)}, \nonumber\\
&\qquad \leq \norm{\overleftarrow{\text{W}}(\vecbra{O_X, I_E} \mathcal{U}_\mathcal{B}(t, s'), B_{\alpha,\sigma}(s'))} \norm{\mathcal{L}_Y}_\diamond +  \norm{\overrightarrow{\text{W}}(B_{\alpha,\sigma}(s'), \vecket{\tilde{\phi}({s'})})}_1 \norm{\mathcal{L}_Y}_\diamond \norm{O_X}, \nonumber\\
&\qquad \numleq{2} \big(2\norm{\bm\mu^{\textnormal{L}}(0, t-s')}_\infty + \zeta_\alpha(s')+(\beta_0+1)\mu_{\alpha,[s_0,s]}^R(s')\big) \norm{O_X} \norm{\mathcal{L}_Y}_\diamond'\nonumber\\
&\qquad \numleq{3} \big(2\norm{\bm{\mu}^{\textnormal{L}}(0,t-s_0)}_\infty+\sup_{\alpha,s'\in[s_0,t]}\zeta_\alpha(s')+(\beta_0+1)\norm{\bm{\mu}^{\textnormal{R}}(s_0-t,t-s_0)}_\infty \big) \norm{O_X} \norm{\mathcal{L}_Y}_\diamond,
\end{align}
where in (1) we have used the Wick's theorem, in (2) we have used lemma \ref{lemma:closure_left_wick} together with $\vecbra{O_X, I_E}\mathcal{U}(t, s') \in \mathcal{Q}_{\bm\mu^{\textnormal{L}}_{[s', t]}}(\rho_E)$ to upper bound $\norm{\overleftarrow{\text{W}}(\vecbra{O_X, I_E} \mathcal{U}_\mathcal{B}(t', s'), B_{\alpha,\sigma}(s'))}$ and lemma \ref{lemma:closure_right_wick}c to upper bound $\norm{\overrightarrow{\text{W}}(B_{\alpha,\sigma}(s'), \vecket{\tilde{\phi}({s'})})}_1$. In (3), we have used the fact that the function $\mu^{\textnormal{R}}_{\alpha,[s_0, s]}(s')$, for any $s_0<s,s'<t$, is upper bounded by $\mu^{\textnormal{R}}_\alpha(s_0-t,t-s_0)$.
From Eqs.~\ref{eq:bound_1_lemma_cont} and \ref{eq:bound_2_lemma_cont}, we then obtain that
\[
\gamma^{X, Y}_{\zeta + \alpha_0\mu^{\textnormal{R}}_{[s_0, s - \varepsilon]}, _\mathcal{B}}(t, s - \varepsilon) \leq \gamma^{X, Y}_{\zeta + \alpha_0\mu^{\textnormal{R}}_{[s_0, s]}, _\mathcal{B}}(t, s ) + C \varepsilon,
\]
with
\[
C=\abs{\mathcal{B}}\left(2+4\norm{\bm{\mu}^{\textnormal{L}}(-t, t)}_\infty+2\norm{\bm\zeta_{[0, t]}}_\infty+2\norm{\bm{\mu}^{\textnormal{R}}(-t, t)}_\infty \right)
\]
which proves the lemma statement.
\end{proof}
\section{Proof of lemma \ref{lemma:taylor_expansion_bound}}\label{appendix:taylor_expansion}
\begin{proof}[Proof of lemma \ref{lemma:taylor_expansion_bound}]
Define $\vecket{\tilde{\phi}} = \mathcal{U}_Y(t' - \varepsilon, t')\mathcal{U}_\mathcal{B}(t', t' - \varepsilon) \vecket{\phi}$ and note that $\vecket{\tilde{\phi}} \in \mathcal{S}_{\bm\zeta + 2\bm\mu^{\textnormal{R}}_{[t' - \varepsilon, t']}}(\rho_E)$. We will assume that $\norm{\mathcal{L}_Y}_\diamond, \norm{O_X} = 1$ without loss of generality since the quantity that we want to bound is linear in $\mathcal{L}_Y, O_X$. Using the triangle inequality, we then obtain that the error of interest has two contributions:
\begin{subequations}\label{eq:lemma_5_equation}
\begin{align}
    &\abs{\vecbra{O_X, I_E} \mathcal{U}_\mathcal{B}(t, t' - \varepsilon)\mathcal{L}_Y \vecket{\phi} - \vecbra{O_X, I_E} \mathcal{U}_\mathcal{B}(t, t')\bigg(I - i\int_{t' - \varepsilon}^{t'}\mathcal{H}_{ \mathcal{B} \cap \mathcal{A}_Y}(s)ds\bigg)\mathcal{L}_Y \vecket{\tilde{\phi}}} \leq \Delta^{(1)}(t, t') + \Delta^{(2)}(t, t')
\end{align}
where
\begin{align}\label{eq:delta_12_defs}
    &\Delta^{(1)}(t, t') = \abs{\vecbra{O_X, I_E} \mathcal{U}_\mathcal{B}(t, t')\big(\mathcal{U}_\mathcal{B}(t', t'-\varepsilon) \mathcal{L}_Y \mathcal{U}_\mathcal{B}(t' - \varepsilon, t')-\mathcal{U}_{\mathcal{B}\cap \mathcal{A}_Y}(t', t'-\varepsilon) \mathcal{L}_Y \mathcal{U}_{\mathcal{B}\cap \mathcal{A}_Y}(t' - \varepsilon, t') \big)\mathcal{U}_\mathcal{B}(t', t'-\varepsilon) \vecket{\phi}},  \\
    &\Delta^{(2)}(t, t') = \abs{\vecbra{O_X, I_E} \mathcal{U}_\mathcal{B}(t, t')\bigg(\mathcal{U}_{\mathcal{B}\cap \mathcal{A}_Y}(t', t'-\varepsilon) \mathcal{L}_Y - \bigg(I -  i\int_{t' - \varepsilon}^{t'}\mathcal{H}_{\mathcal{B}\cap \mathcal{A}_Y}(s) ds\bigg)\mathcal{L}_Y  \bigg) \vecket{\tilde{\phi}}}.
\end{align}
\end{subequations}
We will analyze both of these errors separately and provide $O(\varepsilon^2)$ upper bounds on them.

\underline{Analyzing $\Delta^{(1)}(t, t')$}. This error can be analyzed by performing a first-order Dyson series expansion of $\mathcal{U}_\mathcal{B}(t', t'-\varepsilon)\mathcal{L}_Y \mathcal{U}_\mathcal{B}(t' -\varepsilon, t')$ and $\mathcal{U}_Y(t', t'-\varepsilon)\mathcal{L}_Y \mathcal{U}_Y(t'-\varepsilon, t')$:
    \begin{align}\label{eq:dyson_expansion}
        &\mathcal{U}_\mathcal{B}(t', t'-\varepsilon)\mathcal{L}_Y \mathcal{U}_\mathcal{B}(t'-\varepsilon, t') = \mathcal{L}_Y -i\int_{t' - \varepsilon}^{t'}[\mathcal{H}(s'), \mathcal{L}_Y] ds'  + \mathcal{R}^{(2)}(t', t'-\varepsilon),\nonumber \\
        &\mathcal{U}_Y(t', t'-\varepsilon)\mathcal{L}_Y \mathcal{U}_Y(t'-\varepsilon, t') = \mathcal{L}_Y -i\int_{t' - \varepsilon}^{t'}[\mathcal{H}_Y(s'), \mathcal{L}_Y] ds' + \mathcal{R}_Y^{(2)}(t', t'-\varepsilon),
    \end{align}
where $\mathcal{R}^{(2)}(t', t'-\varepsilon), \mathcal{R}^{(2)}_Y(t', t'-\varepsilon)$ are the respective remainders that are given by
    \begin{align*}
        &\mathcal{R}^{(2)}(t', t'-\varepsilon) = (-i)^2\int_{t' - \varepsilon}^{t'}\int_{s_1'}^{t'}\mathcal{U}_\mathcal{B}(t', s_2')[\mathcal{H}(s_2'), [\mathcal{H}(s_1'), \mathcal{L}_Y]] \mathcal{U}_\mathcal{B}(s_2', t')ds_2' ds_1', \nonumber \\
        &\mathcal{R}_Y^{(2)}(t', t'-\varepsilon) = (-i)^2\int_{t' - \varepsilon}^{t'}\int_{s_1'}^{t'}\mathcal{U}_{\mathcal{B}\cap \mathcal{A}_Y}(t', s_2')[\mathcal{H}_{\mathcal{B}\cap \mathcal{A}_Y}(s_2'), [\mathcal{H}_{\mathcal{B}\cap \mathcal{A}_Y}(s_1'), \mathcal{L}_Y]] \mathcal{U}_{\mathcal{B}\cap \mathcal{A}_Y}(s_2', t') ds_2' ds_1'. \nonumber
    \end{align*}
We remind the reader that
\begin{align}\label{eq:redef_H_Y}
\mathcal{H}_{\mathcal{B}\cap \mathcal{A}_Y}(s) = \mathcal{H}_{\mathcal{B}\cap \mathcal{A}_Y}^S(s) + \mathcal{H}_{\mathcal{B}\cap \mathcal{A}_Y}^{SE}(s) \text{ where }\mathcal{H}_{\mathcal{B}\cap \mathcal{A}_Y}^S(s) = \Theta_{{\mathcal{B}\cap \mathcal{A}_Y}}(\alpha) \mathcal{C}_{h_\alpha(s)},  \mathcal{H}_Y^{SE}(s)=  \Theta_{{\mathcal{B}\cap \mathcal{A}_Y}}(\alpha)(-1)^\sigma B_{\alpha,\sigma}(s) R_{\alpha, \sigma}(s),
\end{align}
Consequently, $[\mathcal{H}_{\mathcal{B}\cap \mathcal{A}_Y}(\alpha)(s), \mathcal{L}_Y] = [\mathcal{H}(s), \mathcal{L}_Y]$. Therefore, $\Delta^{(1)}(t, t')$ [Eq.~\ref{eq:lemma_5_equation}b] has contributions only from the second order remainder in the Dyson expansion in Eq.~\ref{eq:dyson_expansion} i.e.
    \begin{align}\label{eq:delta_1_to_rem}
    \Delta^{(1)}(t, t') \leq \abs{\vecbra{O_X, I_E} \mathcal{U}_\mathcal{B}(t, t') \mathcal{R}_Y^{(2)}(t, t') \mathcal{U}_\mathcal{B}(t', t' - \varepsilon) \vecket{\phi}} +  \abs{\vecbra{O_X, I_E} \mathcal{U}_\mathcal{B}(t, t') \mathcal{R}^{(2)}(t, t') \mathcal{U}_\mathcal{B}(t', t' - \varepsilon) \vecket{\phi}}.
    \end{align}
    Consider first bounding $\abs{\vecbra{O_X, I_E} \mathcal{U}_\mathcal{B}(t, t') \mathcal{R}_Y^{(2)}(t, t') \mathcal{U}_\mathcal{B}(t', t' - \varepsilon) \vecket{\phi}}$ --- note that
    \[
    \abs{\vecbra{O_X, I_E} \mathcal{U}_\mathcal{B}(t, t') \mathcal{R}_Y^{(2)}(t, t') \mathcal{U}_\mathcal{B}(t', t' - \varepsilon) \vecket{\phi}} = \abs{\int_{t' - \varepsilon}^{t'}\int_{s_1'}^{t'} \vecbra{\hat{\theta}(s_2')}[\mathcal{H}_{\mathcal{B}\cap \mathcal{A}_Y}(\alpha)(s_2'), [\mathcal{H}_{\mathcal{B}\cap \mathcal{A}_Y}(s_1'), \mathcal{L}_Y]] \vecket{\hat{\phi}(s_2')}ds_2' ds_1'},
    \]
    where, for $s_2' \in [t'-\varepsilon, t']$, $\vecbra{\hat{\theta}(s_2')} = \vecbra{O_X, I_E} \mathcal{U}_\mathcal{B}(t, t') \mathcal{U}_{\mathcal{B}\cap \mathcal{A}_Y}(t', s_2') \in \mathcal{Q}_{\bm\mu^{\textnormal{L}}_{[t' - \varepsilon, t]}}(\rho_E)$ and $\vecket{\hat{\phi}(s_2')} = \mathcal{U}_{\mathcal{B}\cap \mathcal{A}_Y}(s_2', t') \mathcal{U}_\mathcal{B}(t', t'-\varepsilon)\vecket{\phi} \in \mathcal{S}_{\bm\zeta + 2\bm\mu^{\textnormal{R}}_{[t' - \varepsilon, t']}}(\rho_E)$. Now, note that
    \begin{align*}
        &\vecbra{\hat{\theta}(s_2')} [\mathcal{H}_Y(s_2'), [\mathcal{H}_{\mathcal{B}\cap \mathcal{A}_Y}(s_1'), \mathcal{L}_Y]\vecket{\hat{\phi}(s_2')} = \vecbra{\hat{\theta}(s_2')} [\mathcal{H}_{\mathcal{B}\cap \mathcal{A}_Y}^S(s_2'), [\mathcal{H}_{\mathcal{B}\cap \mathcal{A}_Y}^S(s_1'), \mathcal{L}_Y]\vecket{\hat{\phi(s_2')}} + \nonumber\\
        &\qquad \qquad \qquad \vecbra{\hat{\theta}(s_2')} [\mathcal{H}_{\mathcal{B}\cap \mathcal{A}_Y}^S(s_2'), [\mathcal{H}_{\mathcal{B}\cap \mathcal{A}_Y}^{SE}(s_1'), \mathcal{L}_Y]\vecket{\hat{\phi}(s_2')} + \nonumber \\
        &\qquad \qquad \qquad \vecbra{\hat{\theta}(s_2')} [\mathcal{H}_{\mathcal{B}\cap \mathcal{A}_Y}^{SE}(s_2'), [\mathcal{H}_{\mathcal{B}\cap \mathcal{A}_Y}^{S}(s_1'), \mathcal{L}_Y]\vecket{\hat{\phi}(s_2')} + \vecbra{\hat{\theta(s_2')}} [\mathcal{H}_{\mathcal{B}\cap \mathcal{A}_Y}^{SE}(s_2'), [\mathcal{H}_{\mathcal{B}\cap \mathcal{A}_Y}^{SE}(s_1'), \mathcal{L}_Y]\vecket{\hat{\phi}(s_2')}.
    \end{align*}
    We can bound these terms one by one. For calculating these bounds, it will be convenient to note that, from lemmas \ref{lemma:closure_right_wick}c and \ref{lemma:closure_left_wick}, it follows
    \begin{align}\label{eq:bounds_wick_contr}
        &\smallnorm{\wickleft\big(\vecbra{\hat{\theta}(s_2')}, \{B_{\alpha_j,\sigma_j}(s_j')\}_{j \in \{1, 2\}}\big)} \leq 2^2\prod_{j=1,2}\smallnorminf{\bm{\mu}^{\textnormal{L}}_{[t'-\varepsilon,t]}(s_j')} \leq 4 \smallnorminf{\bm{\mu}^{\textnormal{L}}(-t,t)}^2, \nonumber \\
        &\smallnorm{\wickright\big(\{B_{\alpha_j,\sigma_j}(s_j')\}_{j \in \{1, 2\}}, \vecket{\hat{\phi}(s_2')} \big)}_1 \leq  \prod_{j=1,2}\left(2\smallnorminf{\bm{\mu}^{\textnormal{R}}_{[t'-\varepsilon,t']}(s_j')}+\smallnorminf{\bm\zeta(s_j')}\right) \leq \left(2\smallnorminf{\bm{\mu}^{\textnormal{R}}(-t,t)} + \smallnorminf{\bm{\zeta}_{[0,t]}}\right)^2, \nonumber\\
        &\smallnorm{\wickleft\big(\vecbra{\hat{\theta}(s_2')}, B_{\alpha_j,\sigma_j}(s_j')\big)} \leq 2\smallnorminf{\bm{\mu}^{\textnormal{L}}_{[t'-\varepsilon,t]}(s_j')} \leq 2 \left(\smallnorminf{\bm{\mu}^{\textnormal{L}}(-t,t)}\right), \nonumber \\
        &\smallnorm{\wickleft\big( B_{\alpha_j,\sigma_j}(s_j'), \vecket{\hat{\phi}(s_2')}\big)}_1 \leq \smallnorminf{2\bm\mu^{\textnormal{R}}_{[t' - \varepsilon, t']}(s_j')} + \smallnorminf{\bm\zeta(s_j')} \leq \left(2\smallnorminf{\bm\mu^{\textnormal{R}}(-t,t)}+\smallnorminf{\bm\zeta_{[0,t]}}\right).
    \end{align}
    Now, note that
    \begin{align}\label{eq:delta_1_S_S_bound}
        &\abs{\vecbra{\hat{\theta}(s_2')} [\mathcal{H}_{\mathcal{B}\cap \mathcal{A}_Y}^S(s_2'), [\mathcal{H}_{\mathcal{B}\cap \mathcal{A}_Y}^S(s_1'), \mathcal{L}_Y]]\vecket{\hat{\phi}(s_2')}} \nonumber\\
        &\qquad \qquad \leq \Theta_{\mathcal{B}\cap \mathcal{A}_Y}(\alpha_1) \Theta_{\mathcal{B}\cap \mathcal{A}_Y}(\alpha_1) \abs{\vecbra{\theta(s_2')} [\mathcal{C}_{h_{\alpha_2}(s_2')}, [\mathcal{C}_{h_{\alpha_1}(s_1')},  \mathcal{L}_Y]]\vecket{\hat{\phi}(s_2')}}\leq 2^4\abs{\mathcal{A}_Y}^2.
    \end{align}
    Next,
    \begin{align}\label{eq:delta_1_S_SE_bound}
        &\abs{\vecbra{\hat{\theta}(s_2')}[\mathcal{H}^S_{\mathcal{B}\cap \mathcal{A}_Y}(s_2'), [\mathcal{H}^{SE}_{\mathcal{B}\cap \mathcal{A}_Y}(s_1'), \mathcal{L}_Y]]\vecket{\hat{\phi}(s_2')}}\nonumber\\
        &\qquad \qquad \leq \Theta_{\mathcal{B}\cap \mathcal{A}_Y}(\alpha_1)\Theta_{\mathcal{B}\cap \mathcal{A}_Y}(\alpha_2)\abs{\vecbra{\hat{\theta}(s_2')}[\mathcal{C}_{h_{\alpha_2}(s_2')}, [B_{\alpha_1,\sigma}(s_1')R_{\alpha_1, \sigma}(s_1'), \mathcal{L}_Y]]\vecket{\hat{\phi}(s_2')}}, \nonumber \\
        &\qquad \qquad \numleq{1} \Theta_{\mathcal{B}\cap \mathcal{A}_Y}(\alpha_1)\Theta_{\mathcal{B}\cap \mathcal{A}_Y}(\alpha_2) \abs{\wickleft(\vecbra{\hat{\theta}(s_2')}, B_{\alpha_1,\sigma}(s_1'))[\mathcal{C}_{h_{\alpha_2}}(s_2'), [R_{\alpha_1, \sigma}(s_1'), \mathcal{L}_Y]]\vecket{\hat{\phi}(s_2')}} + \nonumber \\
        &\qquad \qquad \qquad \qquad\Theta_{\mathcal{B}\cap \mathcal{A}_Y}(\alpha_1)\Theta_{\mathcal{B}\cap \mathcal{A}_Y}(\alpha_2)\abs{\vecbra{\hat{\theta}(s_2')}[\mathcal{C}_{h_{\alpha_2}(s_2')}, [R_{\alpha_1, \sigma}(s_1'), \mathcal{L}_Y]]\wickright(B_{\alpha_2, \sigma}, \vecket{\hat{\phi}(s_2')})}, \nonumber \\
        &\qquad \qquad \numleq{2} 2^4\abs{\mathcal{A}_Y}^2L \left(2\norminf{\bm\mu^{\textnormal{L}}(-t,t)}+2\norminf{\bm\mu^{\textnormal{R}}(-t,t)}+\norminf{\bm\zeta_{[0,t]}}\right),
    \end{align}
    where, in (1), we have used the Wick's theorem and in (2) we have used the bounds in Eq.~\ref{eq:bounds_wick_contr}. A similar analysis yields the same bound for  $\vecbra{\hat{\theta}(s_2')} [\mathcal{H}_{Y}^{SE}(s_2'), [\mathcal{H}_{Y}^{S}(s_1'), \mathcal{L}_Y]\vecket{\hat{\phi}(s_2')}$ ---
    \begin{align}\label{eq:delta_1_SE_S_bound}
        \vecbra{\hat{\theta}(s_2')} [\mathcal{H}_{Y}^{SE}(s_2'), [\mathcal{H}_{Y}^{S}(s_1'), \mathcal{L}_Y]\vecket{\hat{\phi}(s_2')} \leq 2^4\abs{\mathcal{A}_Y}^2 L\left(2\norminf{\bm\mu^{\textnormal{L}}(-t,t)}+2\norminf{\bm\mu^{\textnormal{R}}(-t,t)}+\norminf{\bm\zeta_{[0,t]}}\right).
    \end{align}
    Finally, consider  $\vecbra{\hat{\theta}(s_2')} [\mathcal{H}_{Y}^{SE}(s_2'), [\mathcal{H}_{Y}^{SE}(s_1'), \mathcal{L}_Y]\vecket{\hat{\phi}(s_2')}$. We note that
    \begin{align}\label{eq:explicit_expansion_second_order}
        &\vecbra{\hat{\theta}(s_2')}[\mathcal{H}_{\mathcal{B}\cap \mathcal{A}_Y}^{SE}(s_2'), [\mathcal{H}_{\mathcal{B}\cap \mathcal{A}_Y}^{SE}(s_1'), \mathcal{L}_Y]] \vecket{\hat{\phi}(s_2')}   \nonumber\\
        &= \Theta_{\mathcal{B}\cap \mathcal{A}_Y}(\alpha_1)\Theta_{\mathcal{B}\cap \mathcal{A}_Y}(\alpha_2) (-1)^{\sigma_1 + \sigma_2} \bigg(\vecbra{\hat{\theta}(s_2')}B_{\alpha_2,\sigma_2}(s_2') B_{\alpha_1,\sigma_1}(s_1') R_{\alpha_2, \sigma_2}(s_2')[R_{\alpha_1, \sigma_1}(s_1'), \mathcal{L}_Y] \vecket{\hat{\phi}(s_2')} -\nonumber\\
        &\qquad \qquad \qquad \qquad \qquad \qquad  \qquad \qquad \qquad \qquad \vecbra{\hat{\theta}(s_2')} B_{\alpha_1,\sigma_1}(s_1')B_{\alpha_2,\sigma_2}(s_2') [R_{\alpha_1, \sigma_1}(s_1'), \mathcal{L}_Y]R_{\alpha_2, \sigma_2}(s_2') \vecket{\hat{\phi}(s_2')}\bigg), \nonumber \\
        &\numeq{1}\Theta_{\mathcal{B}\cap \mathcal{A}_Y}(\alpha_1)\Theta_{\mathcal{B}\cap \mathcal{A}_Y}(\alpha_2) (-1)^{\sigma_1 + \sigma_2}\bigg(\wickleft\big(\vecbra{\hat{\theta}(s_2')}, \{B_{\alpha_j,\sigma_j}(s_j')\}_{j \in \{1, 2\}}\big)[R_{\alpha_2, \sigma_2}(s_2'),[R_{\alpha_1, \sigma_1}(s_1'), \mathcal{L}_Y]] \vecket{\hat{\phi}(s_2')} + \nonumber\\
        &\qquad \qquad \qquad \qquad \qquad \qquad  \vecbra{\hat{\theta}(s_2')}[R_{\alpha_2, \sigma_2}(s_2'),[R_{\alpha_1, \sigma_1}(s_1'), \mathcal{L}_Y]]\wickright\big(\{B_{\alpha_j,\sigma_j}(s_j')\}_{j \in \{1, 2\}}, \vecket{\hat{\phi}(s_2')}\big) +\nonumber\\
        &\qquad \qquad \qquad \qquad \qquad \qquad   \wickleft\big(\vecbra{\hat{\theta}(s_2')}, B_{\alpha_2,\sigma_2}(s_2')\big)[R_{\alpha_2, \sigma_2}(s_2'), [R_{\alpha_1, \sigma_1}(s_1'), \mathcal{L}_Y]]\wickright\big(B_{\alpha_1,\sigma_1}(s_1'), \vecket{\hat{\phi}(s_2')}\big) + \nonumber \\
        &\qquad \qquad \qquad \qquad \qquad \qquad \wickleft\big(\vecbra{\hat{\theta}(s_2')}, B_{\alpha_1,\sigma_1}(s_1')\big)[R_{\alpha_2, \sigma_2}(s_2'), [R_{\alpha_1, \sigma_1}(s_1'), \mathcal{L}_Y]]\wickright\big(B_{\alpha_2,\sigma_2}(s_2'), \vecket{\hat{\phi}(s_2')}\big)\bigg) +\nonumber \\
        &\quad \Theta_{\mathcal{B}\cap \mathcal{A}_Y}(\alpha_1)\Theta_{\mathcal{B}\cap \mathcal{A}_Y}(\alpha_2) \bigg((-1)^{\sigma_1}\K_{\alpha_2,\alpha_1 ,\sigma_1}(s_2' - s_1') \vecbra{\hat{\theta}(s_2')} \mathcal{C}_{R_{\alpha_2}(s_2')}[R_{\alpha, \sigma_1}(s_1'), \mathcal{L}_Y] \vecket{\hat{\phi}(s_2')} + \nonumber\\
        &\qquad \qquad \qquad \qquad \qquad \qquad \qquad \qquad (-1)^{\sigma_2}\K_{\alpha_1,\alpha_2, \sigma_2}(s_1' - s_2')\vecbra{\hat{\theta}(s_2')}[\mathcal{C}_{R_{\alpha_1}(s_1')}, \mathcal{L}_Y]R_{\alpha, \sigma_2}(s_2') \vecket{\hat{\phi}(s_2')}\bigg),
    \end{align}
    where (1) is obtained by using the Wick's theorem and writing out all possible wick contractions of the bath operators $\{B_{\alpha_j,\sigma_j}(s_j')\}_{j \in \{1, 2\}}$, with the contractions being with either $\vecket{\hat{\phi}(s_2')}, \vecbra{\hat{\theta}(s_2')}$ or between themselves (i.e.~the last line of Eq.~\ref{eq:explicit_expansion_second_order}).
    Next, using the bounds in Eq.~\ref{eq:bounds_wick_contr} together with the fact that $\abs{\K_{\alpha_2,\alpha_1, \sigma_1}(s_2' - s_1')},\abs{\K_{\alpha_1, \alpha_2, \sigma_2}(s_1' - s_2')} \leq \norm{K}_\infty$ (lemma \ref{lemma:mollified_kernel}a), we obtain that
    \begin{align}\label{eq:delta_1_SE_SE_bound}
        \abs{\vecbra{\hat{\theta}(s_2')} [\mathcal{H}^{SE}_{\mathcal{B}\cap \mathcal{A}_Y}(s_2'), [\mathcal{H}^{SE}_{\mathcal{B}\cap \mathcal{A}_Y}(s_1'), \mathcal{L}_Y]]} \leq  &\abs{\mathcal{A}_Y}^2L^2\bigg(2^6\norminf{\bm\mu^{\textnormal{L}}(-t,t)}^2+2^4(2\norminf{\bm{\mu}^{\textnormal{R}}(-t,t)}+\norminf{\bm\zeta_{[0,t]}})^2+\nonumber\\
        &\quad\quad2^6\norminf{\bm{\mu}^{\textnormal{L}}(-t,t)}(2\norminf{\bm\mu^{\textnormal{R}}(-t,t)}+\norminf{\bm\zeta_{[0,t]}})+2^4\norminf{K}\bigg).
    \end{align}
    Combining the bounds in Eqs.~\ref{eq:delta_1_S_S_bound}, \ref{eq:delta_1_S_SE_bound}, \ref{eq:delta_1_SE_S_bound} and \ref{eq:delta_1_SE_SE_bound}, we obtain that
    \begin{subequations}\label{eq:remainder_Y_bound}
    \begin{align}
        &\abs{\vecbra{O_X, I_E} \mathcal{U}_\mathcal{B}(t, t') \mathcal{R}_Y^{(2)}(t, t') \mathcal{U}_\mathcal{B}(t', t' - \varepsilon) \vecket{\phi}} \leq \nonumber\\
        &\qquad \qquad \qquad \varepsilon^2\abs{\mathcal{A}_Y}^2 \big(2^4L^2 \left((2\norminf{\bm{\mu}^{\textnormal{L}}(-t,t)}+2\norminf{\bm{\mu}^{\textnormal{R}}(-t,t}+\norminf{\bm\zeta_{[0,t]}})^2+\norminf{K}\right)+\nonumber\\
        &\qquad \qquad \qquad2^5L\left(2\norminf{\bm\mu^{\textnormal{L}}(-t,t)}+2\norminf{\bm\mu^{\textnormal{R}}(-t,t)}+\norminf{\bm\zeta_{[0,t]}}\right)+2^4\big).
    \end{align}
    A similar analysis can be carried out to bound $\abs{\mathcal{R}^{(2)}(t', t'-\varepsilon)}$ and obtain
    \begin{align}
         &\abs{\vecbra{O_X, I_E} \mathcal{U}_\mathcal{B}(t, t') \mathcal{R}^{(2)}(t, t') \mathcal{U}_\mathcal{B}(t', t' - \varepsilon) \vecket{\phi}} \leq \nonumber\\
        &\qquad \qquad \qquad \varepsilon^2\abs{\mathcal{A}_Y}\abs{\mathcal{A}_Y'} \big(2^4L^2 \left((2\norminf{\bm{\mu}^{\textnormal{L}}(-t,t)}+2\norminf{\bm{\mu}^{\textnormal{R}}(-t,t}+\norminf{\bm\zeta_{[0,t]}})^2+\norminf{K}\right)+\nonumber\\
        &\qquad \qquad \qquad2^5L\left(2\norminf{\bm\mu^{\textnormal{L}}(-t,t)}+2\norminf{\bm\mu^{\textnormal{R}}(-t,t)}+\norminf{\bm\zeta_{[0,t]}}\right)+2^4\big).
    \end{align}
    \end{subequations}
    where $\mathcal{A}'_Y = \{\alpha' : \mathcal{S}_{\alpha} \cap S_{\alpha'} \neq \emptyset \text{ for some } \alpha \in \mathcal{A}_Y\}$. From Eq.~\ref{eq:delta_1_to_rem}, we then obtain that
    \begin{align*}
    \Delta^{(1)}(t, t') \leq  &\varepsilon^2  \big(\abs{\mathcal{A}_Y}(\abs{\mathcal{A}_Y} +\abs{\mathcal{A}'_Y})\big(2^4L^2 \left((2\norminf{\bm{\mu}^{\textnormal{L}}(-t,t)}+2\norminf{\bm{\mu}^{\textnormal{R}}(-t,t}+\norminf{\bm\zeta_{[0,t]}})^2+\norminf{K}\right)+\nonumber\\
        &\qquad 2^5L\left(2\norminf{\bm\mu^{\textnormal{L}}(-t,t)}+2\norminf{\bm\mu^{\textnormal{R}}(-t,t)}+\norminf{\bm\zeta_{[0,t]}}\right)+2^4\big).
    \end{align*}

    \underline{Analyzing $\Delta^{(2)}(t, t')$}. Next, we bound $\Delta^{(2)}(t, t')$ --- starting from Eq.~\ref{eq:delta_12_defs} can be expressed as
    \[
    \Delta^{(2)}(t, t') = \abs{\int_{t' - \varepsilon}^{t'} \int_{s_1'}^{t'} \vecbra{\hat{\theta}(s_2')}\mathcal{H}_{\mathcal{B}\cap \mathcal{A}_Y}(s_2') \mathcal{H}_{\mathcal{B}\cap \mathcal{A}_Y}(s_1')\mathcal{L}_Y\vecket{\tilde{\phi}} ds_2' ds_1'},
    \]
    where as defined above, $\vecbra{\hat{\theta}(s_2')} = \vecbra{O_X, I_E} \mathcal{U}_\mathcal{B}(t, t') \mathcal{U}_Y(t', s_2') \in \mathcal{Q}_{\bm\mu^{\textnormal{L}}_{[s_2', t]}}(\rho_E) \subseteq \mathcal{Q}_{\bm\mu^{\textnormal{L}}_{[t' - \varepsilon, t]}}(\rho_E)$. Using the explicit expression for $\mathcal{H}_Y(s)$, we obtain that
    \begin{align*}
&\vecbra{\hat{\theta}(s_2')}\mathcal{H}_Y(s_2') \mathcal{H}_Y(s_1')\mathcal{L}_Y\vecket{\tilde{\phi}} = \vecbra{\hat{\theta}(s_2')}\mathcal{H}_Y^S(s_2') \mathcal{H}_Y^S(s_1')\mathcal{L}_Y\vecket{\tilde{\phi}} + \vecbra{\hat{\theta}(s_2')}\mathcal{H}_Y^S(s_2') \mathcal{H}_Y^{SE}(s_1')\mathcal{L}_Y\vecket{\tilde{\phi}} +\nonumber\\
&\qquad \qquad \qquad  \vecbra{\hat{\theta}(s_2')}\mathcal{H}_Y^{SE}(s_2') \mathcal{H}_Y^S(s_1')\mathcal{L}_Y\vecket{\tilde{\phi}} + \vecbra{\hat{\theta}(s_2')}\mathcal{H}_Y^{SE}(s_2') \mathcal{H}_Y^{SE}(s_1')\mathcal{L}_Y\vecket{\tilde{\phi}}.
    \end{align*}
In addition to the bounds in Eq.~\ref{eq:bounds_wick_contr}, it will be useful to also note that
    \begin{align}\label{eq:bounds_wick_contr_new}
        &\smallnorm{\wickright\big( B_{\alpha_j,\sigma_j}(s_j'); \vecket{\tilde{\phi}}\big)}_1 \leq 2\smallnorminf{\bm\mu^{\textnormal{R}}_{[t' - \varepsilon, t']}(s_j')} + \smallnorminf{\bm\zeta(s_j')} \leq \big(2\smallnorminf{\bm\mu^{\textnormal{R}}(-t,t)}+\smallnorminf{\bm\zeta_{[0,t]}}\big), \nonumber\\
        &\smallnorm{\wickright\big( \{B_{\alpha_j,\sigma_j}(s_j')\}_{j \in \{1, 2\}}; \vecket{\tilde{\phi}}\big)}_1 \leq \prod_{j=1,2}\left(2\smallnorminf{\bm\mu^{\textnormal{R}}_{[t' - \varepsilon, t']}(s_j')} + \smallnorminf{\bm\zeta(s_j')}\right) \leq \big(2\smallnorminf{\bm\mu^{\textnormal{R}}(-t,t)}+\smallnorminf{\bm\zeta_{[0,t]}}\big)^2.
    \end{align}
Now, first note that we have the bound
    \begin{align}
        \abs{\vecbra{\hat{\theta}(s_2')}\mathcal{H}_Y^S(s_2') \mathcal{H}_Y^S(s_1')\mathcal{L}_Y\vecket{\tilde{\phi}}} \leq \Theta_{\mathcal{A}_Y}(\alpha_1)\Theta_{\mathcal{A}_Y}(\alpha_2) \abs{\vecbra{\hat{\theta}(s_2')} \mathcal{C}_{h_{\alpha_2}(s_2')}\mathcal{C}_{h_{\alpha_1}(s_1')}\mathcal{L}_Y\vecket{\tilde{\phi}}} \leq 4 \abs{\mathcal{A}_Y}^2.
    \end{align}
Next, consider
    \begin{align}
        &\abs{\vecbra{\hat{\theta}(s_2')}\mathcal{H}_Y^{SE}(s_2') \mathcal{H}_Y^S(s_1')\mathcal{L}_Y\vecket{\tilde{\phi}}} \leq \Theta_{\mathcal{A}_Y}(\alpha_1)\Theta_{\mathcal{A}_Y}(\alpha_2)\abs{\vecbra{\hat{\theta}(s_2')} R_{\alpha_2, \sigma_2}(s_2')B_{\alpha_2,\sigma_2}(s_2') \mathcal{C}_{h_{\alpha_1}(s_1)} \mathcal{L}_Y\vecket{\tilde{\phi}}}, \nonumber\\
        & \qquad \qquad  \numleq{1} \Theta_{\mathcal{A}_Y}(\alpha_1)\Theta_{\mathcal{A}_Y}(\alpha_2) \bigg(\abs{\wickleft({\vecbra{\hat{\theta}(s_2')}, B_{\alpha_2,\sigma_2}(s_2')})  R_{\alpha_2, \sigma_2}(s_2') \mathcal{C}_{h_{\alpha_1}(s_1')} \mathcal{L}_Y \vecket{\tilde{\phi}}} + \nonumber\\
        &\qquad \qquad \qquad \qquad \qquad \qquad \abs{\vecbra{\hat{\theta}(s_2')}  R_{\alpha_2, \sigma_2}(s_2') \mathcal{C}_{h_{\alpha_1}(s_1')}\mathcal{L}_Y \wickright(B_{\alpha_2,\sigma_2}(s_2'), \vecket{\tilde{\phi}})}\bigg),\nonumber \\
        &\qquad \qquad \numleq{2} 2^2\abs{\mathcal{A}_Y}^2L\big(2\norminf{\bm\mu^{\textnormal{R}}(-t,t)}+\norminf{\bm\zeta_{[0,t]}}+2\norminf{\bm\mu^{\textnormal{L}}(-t,t)} \big).
    \end{align}
    Similarly,
    \begin{align}
        \abs{\vecbra{\hat{\theta}(s_2')}\mathcal{H}_Y^{S}(s_2')\mathcal{H}_Y^{SE}(s_1')\mathcal{L}_Y\vecket{\tilde{\phi}}} \leq 2^2\abs{\mathcal{A}_Y}^2L\big(2\norminf{\bm\mu^{\textnormal{R}}(-t,t)}+\norminf{\bm\zeta_{[0,t]}}+2\norminf{\bm\mu^{\textnormal{L}}(-t,t)} \big).
    \end{align}
    Finally, consider
    \begin{align}
        &\vecbra{\hat{\theta}(s_2')} \mathcal{H}_Y^{SE}(s_2')\mathcal{H}_Y^{SE}(s_1')\mathcal{L}_Y \vecket{\tilde{\phi}} \nonumber\\
        &\qquad \qquad= (-1)^{\sigma_1}(-1)^{\sigma_2}\Theta_{\mathcal{A}_Y}(\alpha_1)\Theta_{\mathcal{A}_Y}(\alpha_2)\vecbra{\hat{\theta}(s_2')}B_{\alpha_2,\sigma_2}(s_2')B_{\alpha_1,\sigma_1}(s_1')R_{\alpha_2, \sigma_2}(s_2') R_{\alpha_1, \sigma_1}(s_1')\mathcal{L}_Y\vecket{\tilde{\phi}}, \nonumber \\
        &\qquad \qquad= (-1)^{\sigma_1}(-1)^{\sigma_2}\Theta_{\mathcal{A}_Y}(\alpha_1) \Theta_{\mathcal{A}_Y}(\alpha_2) \bigg(\wickleft\big(\vecbra{\hat{\theta}(s_2')}, \{B_{\alpha_j,\sigma_j}(s_j')\}_{j \in \{1, 2\}}\big)R_{\alpha_2, \sigma_2}(s_2')R_{\alpha_1, \sigma_1}(s_1')\mathcal{L}_Y\vecket{\tilde{\phi}} + \nonumber \\
        &\qquad \qquad \qquad \qquad \qquad \qquad \quad \vecbra{\hat{\theta}(s_2')}R_{\alpha_2, \sigma_2}(s_2') R_{\alpha_1, \sigma_1}(s_1')\mathcal{L}_Y \wickright\big(\{B_{\alpha_j,\sigma_j}(s_j')\}_{j \in \{1, 2\}}, \vecket{\tilde{\phi}}\big) + \nonumber \\
        &\qquad \qquad \qquad \qquad \qquad \qquad \quad \wickleft\big(\vecbra{\hat{\theta}(s_2')}, B_{\alpha_2,\sigma_2}(s_2')\big) R_{\alpha_2, \sigma_2}(s_2') R_{\alpha_1, \sigma_1}(s_1')\mathcal{L}_Y\wickright\big(B_{\alpha_1,\sigma_1}(s_1'), \vecket{\tilde{\phi}}\big) + \nonumber \\
         &\qquad \qquad \qquad \qquad \qquad \qquad \quad \wickleft\big(\vecbra{\hat{\theta}(s_2')}, B_{\alpha_1,\sigma_1}(s_1')\big) R_{\alpha_2, \sigma_2}(s_2') R_{\alpha_1, \sigma_1}(s_1')\mathcal{L}_Y\wickright\big(B_{\alpha_2,\sigma_2}(s_2'), \vecket{\tilde{\phi}}\big)\bigg) + \nonumber \\
         &\quad \qquad \qquad (-1)^{\sigma_1}(-1)^{\sigma_2}\Theta_{\mathcal{A}_Y}(\alpha_1)\Theta_{\mathcal{A}_Y}(\alpha_2)\K_{\alpha_2, \alpha_1,\sigma_1}(s_2' - s_1') \vecbra{\hat{\theta}(s_2')}R_{\alpha, \sigma_2}(s_2') R_{\alpha, \sigma_1}(s_1') \mathcal{L}_Y\vecket{\tilde{\phi}}.
    \end{align}
    Therefore, using the bounds in Eqs.~\ref{eq:bounds_wick_contr} and \ref{eq:bounds_wick_contr_new} together with the fact that $\abs{\K_{\alpha_2,\alpha_1, \sigma_1}(s_2' - s_1')} \leq \norm{K}_\infty$, we obtain that
    \begin{align*}
    \abs{\vecbra{\hat{\theta}(s_2')} \mathcal{H}_Y^{SE}(s_2')\mathcal{H}_Y^{SE}(s_1')\mathcal{L}_Y \vecket{\tilde{\phi}}} \leq &2^2\abs{\mathcal{A}_Y}^2L^2\big(2^2\norminf{\bm\mu^{\textnormal{L}}(-t,t)}^2+(2\norminf{\bm\mu^{\textnormal{R}}(-t,t)}+\norminf{\bm\zeta_{[0,t]}})^2+\nonumber\\
    &\qquad 2^2\norminf{\bm\mu^{\textnormal{L}}(-t,t)}(2\norminf{\bm\mu^{\textnormal{R}}(-t,t)}+\norminf{\bm\mu^{\textnormal{L}}_{[0,t]}})+\norminf{K}\big)
    \end{align*}
    Therefore, we obtain that
    \begin{align*}
    \Delta^{(2)}(t, t') \leq &\varepsilon^2\abs{\mathcal{A}_Y}^2\big(2^2+2^3L\left(2\norminf{\bm\mu^{\textnormal{R}}(-t,t)}+\norminf{\bm\zeta_{[0,t]}}+2\norminf{\bm\mu^{\textnormal{L}}(-t,t)}\right)+\nonumber\\
    &\qquad2^2L^2\left(\norminf{K}+(2\norminf{\bm\mu^{\textnormal{L}}(-t,t)}+2\norminf{\bm\mu^{\textnormal{R}}(-t,t)}+\norminf{\bm\zeta_{[0,t]}})^2\right)\big).
    \end{align*}
	    Combining the upper bounds for $\Delta^{(1)}(t, t')$ and $\Delta^{(2)}(t, t')$, we obtain the lemma statement.
	\end{proof}
    \section{Proof of Lemma \ref{lemma:coupled_both_summations}}\label{app:coupled_bath_summations}
    \begin{proof}[Proof of Lemma \ref{lemma:coupled_both_summations}]
(a) We note that, from Eq.~\eqref{eq:SM_environment_distance} that
\[
d(\textnormal{x}_\alpha,\textnormal{x}_{\alpha'}) \leq l \implies d(\mathcal{S}_\alpha, \mathcal{S}_{\alpha'}) \leq p_+ l + q_+ \quad \text{and}\quad d(\textnormal{x}_\alpha,\textnormal{x}_{\alpha'})  \geq l \implies  d(\mathcal{S}_\alpha, \mathcal{S}_{\alpha'}) \geq \frac{1}{p_-}(l - q_-).
\]
Therefore,
\begin{align}\label{eq:upper_bound_dxx_sum}
\sum_{\alpha': d(\textnormal{x}_\alpha,\textnormal{x}_{\alpha'}) \leq l} 1 \leq \sum_{\alpha':d(\mathcal{S}_\alpha,\mathcal{S}_{\alpha'}) \leq p_+l + q_+} 1  = \sum_{n = 0}^{\lfloor p_{+}l + q_+ \rfloor} \abs{\{\alpha': d(\mathcal{S}_\alpha, \mathcal{S}_{\alpha'}) = n\}}.
\end{align}
Next, note that
\[
 \abs{\{\alpha': d(\mathcal{S}_\alpha, \mathcal{S}_{\alpha'}) = n\}} \numleq{1}\mathcal{Z} \abs{\{x \in \mathbb{Z}^d: d(x, \mathcal{S}_{\alpha}) = n\}} \numleq{2} \mathcal{Z} \abs{\mathcal{S}_{\alpha}}\abs{\{x \in \mathbb{Z}^d: d(x, 0) = n\}},
\]
where in (1) we have used the fact that any point $x \in \mathbb{Z}^d$ can be in at-most $\mathcal{Z}$ sets $\mathcal{S}_\alpha$ and (2) follows from the fact that, for any set $\mathcal{S}$, $d(x, \mathcal{S}) = n \implies \exists y \in \mathcal{S}: d(x, y) = n$ and thus $\abs{\{x \in \mathbb{Z}^d: d(x, \mathcal{S}) = n\}} \leq \abs{\{x \in \mathbb{Z}^d, y \in \mathcal{S}: d(x, y) = n\}} \leq \abs{\mathcal{S}}\abs{\{x \in \mathbb{Z}^d: d(x, y) = n\}}$. Next, each coordinate span of $\mathcal S_\alpha$ is at most $a_0$, so $\abs{\mathcal S_\alpha}\leq(a_0+1)^d$. Also, $\abs{\{x\in\mathbb Z^d:d(x,0)=n\}}\leq2^{d} d^{d-1}(1+n)^{d-1}$ for every integer $n\geq0$. Define $C_{d,a_0}=(a_{0}+1)^d2^{d} d^{d-1}$ to collect these geometric factors. We obtain
\begin{align}
\abs{\{\alpha': d(\mathcal{S}_\alpha, \mathcal{S}_{\alpha'}) = n\}} \leq \mathcal{Z}C_{d,a_0}(1+n)^{d - 1}.
\end{align} 
Returning to Eq.~\eqref{eq:upper_bound_dxx_sum}, we then obtain
\begin{align}
\sum_{\alpha': d(\textnormal{x}_{\alpha}, \textnormal{x}_{\alpha'}) \leq l} 1 \leq \mathcal ZC_{d,a_0}\sum_{n = 0}^{\lfloor p_+ l + q_+ \rfloor}(1+n)^{d-1} \leq \mathcal ZC_{d,a_0}(1+p_+l+q_+)^d \numleq{1}\mathcal ZC_{d,a_0}(1+p_++q_+)^d l^d,
\end{align}
where (1) follows from the assumption that $l \geq 1$. Setting $D_0 = \mathcal ZC_{d,a_0}(1+p_++q_+)^d$, we obtain the first inequality in part (a) of the lemma statement.

Following a similar analysis, we have that
\begin{align*}
\sum_{\alpha':l\leq d(\textnormal{x}_\alpha, \textnormal{x}_{\alpha'})  <l + 1} 1 &\leq 
\sum_{\alpha': \substack{d(\mathcal{S}_{\alpha}, \mathcal{S}_{\alpha'}) \leq p_+ (l+1) + q_+ \\ d(\mathcal{S}_{\alpha}, \mathcal{S}_{\alpha'}) \geq (l - q_-)/p_- }} 1,\nonumber\\
&\leq \sum_{n \in \mathbb{Z}_{\geq0}: n \in [(l - q_-)/p_-, p_+(l + 1)+q_+]} \abs{\{\alpha': d(\mathcal{S}_\alpha, \mathcal{S}_{\alpha'}) = n\}}, \nonumber \\
&\leq \mathcal ZC_{d,a_0} \sum_{n \in \mathbb{Z}_{\geq0}: n \in [(l - q_-)/p_-, p_+(l + 1)+q_+]} (1+n)^{d-1}, \nonumber\\
&\leq \mathcal ZC_{d,a_0} \bigg(\abs{p_+ - \frac{1}{p_-}}l + 1+p_+ + q_+ + \frac{q_-}{p_-}\bigg) \big(1+p_+ (l+1) + q_+\big)^{d-1},\nonumber\\
&\leq \mathcal ZC_{d,a_0} \abs{p_+ - \frac{1}{p_-}} (1+2p_+ + q_+)^{d-1} l^d +  \mathcal ZC_{d,a_0} \bigg(1+p_+ + q_+ + \frac{q_-}{p_-}\bigg) (1+2p_+ + q_+)^{d-1}l^{d-1}.
\end{align*}
Here we bound the number of integers in the interval by its length plus one, and use $l\geq1$ in the last step. Setting $D_{\partial 0} =  \mathcal ZC_{d,a_0} (1+p_+ + q_+ + q_-/p_-) (1+2p_+ + q_+)^{d-1}$ and $D_{\partial 1} =\mathcal ZC_{d,a_0} \abs{p_+ - 1/p_-} (1+2p_+ + q_+)^{d-1}$ proves the second inequality in part (a) of the lemma.

(b) We begin by noting that for any $c_0 \in (0, 1)$,
\[
\sum_{\alpha'} \frac{1}{(1 + d(\text{x}_{\alpha}, \text{x}_{\alpha'}))^\eta} \leq \sum_{\alpha'} \frac{1}{(1 + c_0 d(\text{x}_\alpha, \text{x}_{\alpha'}))^\eta}.
\]
Picking $c_0 < \min(1, p_+/q_+)$ and using Eq.~\eqref{eq:SM_environment_distance}, we obtain that
\begin{align}
\sum_{\alpha'} \frac{1}{(1 + d(\text{x}_{\alpha}, \text{x}_{\alpha'}))^\eta} &\leq \sum_{\alpha'} \frac{1}{((1 - c_0 q_+/p_+) + (c_0/p_+)d(\mathcal{S}_\alpha, \mathcal{S}_{\alpha'}))^\eta},\nonumber\\
& \leq \frac{1}{\min((1 - c_0 q_+/p_+), c_0/p_+)^\eta} \sum_{\alpha'}\frac{1}{(1 + d(\mathcal{S}_\alpha, \mathcal{S}_{\alpha'}))^\eta}.
\end{align}
We now proceed similar to the analysis in part (a) to obtain
\begin{align}
\sum_{\alpha'}\frac{1}{(1 + d(\mathcal{S}_\alpha, \mathcal{S}_{\alpha'}))^\eta} &\leq \sum_{n = 0}^\infty \frac{1}{(1 + n)^\eta} \abs{\{\alpha': d(\mathcal{S}_\alpha, \mathcal{S}_{\alpha'}) = n\}} \nonumber \\
&\leq \mathcal ZC_{d,a_0} \sum_{n = 0}^\infty \frac{(1+n)^{d-1}}{(1 + n)^\eta} \nonumber \\
&\leq \mathcal ZC_{d,a_0} \bigg(1 + \frac{1}{\eta - d}\bigg).
\end{align}
Setting $c_0 = \min(1, p_+ / q_+)/2$ and
\[
B_\eta =\frac{2^\eta\mathcal ZC_{d,a_0}}{\min((1 - c_0 q_+/p_+), c_0/p_+)^\eta}  \bigg(1 + \frac{1}{\eta - d}\bigg),
\]
we obtain part (b) of the lemma.

\end{proof}

\end{document}